\documentclass[fullpage, 11pt]{article}
\newcommand{\remove}[1]{}
\newcommand{\ignore}[1]{}

\usepackage{amsmath}
\usepackage{amsfonts}
\usepackage{graphicx}
\usepackage{amssymb}
\usepackage{amsthm}

\newtheorem{thm}{Theorem}[section]
\newtheorem{claim}[thm]{Claim}
\newtheorem{lemma}[thm]{Lemma}

\newtheorem{cor}[thm]{Corollary}

\newcommand{\lastfootnote}{\footnotemark[\value{footnote}]}
\newcounter{validFootnote}
\theoremstyle{definition}
\newtheorem{defn}[thm]{Definition}
\theoremstyle{remark}

\usepackage[T1]{fontenc}
\usepackage{yfonts}

\newcommand{\If}{{\bf if}\space}
\newcommand{\Then}{\space{\bf then:}\space \space}
\newcommand{\Else}{{\bf else:}\space \space}
\newcommand{\ElseIf}{{\bf else if}\space}

\newcommand{\For}{{\bf for every}\space}

\newcommand{\Send}{$\mathsf{send} \thickspace$}
\newcommand{\Receive}{$\mathsf{receive} \thickspace$}
\usepackage{fullpage}
\title{\huge \sf \bf Authenticated Adversarial Routing \ \\ \ }
\author{Yair Amir\thanks{Johns Hopkins University Department of Computer Science. Email: {\tt yairamir@cs.jhu.edu} Part of this work was done while visiting IPAM and supported in part by NSF grant 0430254.} \and Paul Bunn\thanks{UCLA Department of Mathematics.  Email: {\tt bunn@math.ucla.edu} Supported in part by NSF grants  0430254, 0716835, 0716389 and 0830803.} \and Rafail Ostrovsky\thanks{UCLA Departments of Computer Science and Department of Mathematics.  Email: {\tt rafail@cs.ucla.edu} Part of this work was done while visiting IPAM and supported in part by IBM Faculty Award, Xerox Innovation Group Award, NSF grants 0430254, 0716835, 0716389, 0830803 and U.C. MICRO grant.}}
\date{ }
\begin{document}
\maketitle

\begin{abstract}
\ignore{
The Shannon Coding theorem \cite{Shannon} states that given a noisy channel with channel capacity $C$ and information rate $R$,  if $R < C$ then there exist codes that allow the probability of error at the receiver to be made arbitrarily small.  This classical theorem was extended by Rajagopalan and Schulman in 1994 \cite{RS} to an interactive coding theorem for static topology networks, where they showed that in noisy-edge networks a constant overhead in communication can also be achieved (provided none of the processors are malicious), thus establishing an optimal-rate coding theorem for static-topology networks consisting of trustworthy nodes.
}

The aim of this paper is to demonstrate
the feasibility of authenticated throughput-efficient routing in an unreliable and dynamically changing synchronous network in which the majority of malicious insiders try to destroy and alter messages or disrupt communication in any way.  More specifically, in this paper we seek to answer the following question:  Given a network in which the majority of nodes are controlled by a node-controlling adversary and whose topology is changing every round, is it possible to develop a protocol with polynomially-bounded memory per processor that guarantees throughput-efficient and correct end-to-end communication? We answer the question affirmatively for extremely general corruption patterns: we only request that the topology of the network and the corruption pattern of the adversary leaves at least one path each round connecting the sender and receiver through honest nodes (though this path may change at every round).  Out construction works in the public-key setting and enjoys bounded memory per processor (that does not depend on the amount of traffic and is polynomial in the network size.) Our protocol achieves {\em optimal transfer rate} with negligible decoding error.  We stress that our protocol assumes no knowledge of which nodes are corrupted nor which path is reliable at any round, and is also fully distributed with nodes making decisions locally, so that they need not know the topology of the network at any time.
%

The optimality that we prove for our protocol is very strong.  Given any routing protocol, we evaluate its efficiency (rate of message delivery) in the ``worst case,'' that is with respect to the {\em worst} possible graph and against the {\em worst} possible (polynomially bounded) adversarial strategy (subject to the above mentioned connectivity constraints).  Using this metric, we show that there does not exist {\em any} protocol that can be asymptotically superior (in terms of throughput) to ours in this setting.

We remark that the aim of our paper is to demonstrate via explicit example the feasibility of throughput-efficient authenticated adversarial routing.  However, we stress that out protocol is {\it not} intended to provide a practical solution, as due to its complexity, no attempt thus far has been made to make the protocol practical by reducing constants or the large (though polynomial) memory requirements per processor.

Our result is related to recent work of Barak, Goldberg and Xiao in 2008 \cite{BGX} who studied fault localization in networks assuming a private-key trusted setup setting. Our work, in contrast, assumes a public-key PKI setup and aims at not only fault localization, but also transmission optimality. Among other things, our work answers one of the open questions posed in the Barak et.\ al.\ paper regarding fault localization on multiple paths. The use of a public-key setting to achieve strong error-correction results in networks was inspired by the work of Micali, Peikert, Sudan and Wilson \cite{MPSW} who showed that classical error-correction against a polynomially-bounded adversary can be achieved with surprisingly high precision. Our work is also related to an interactive coding theorem of Rajagopalan and Schulman \cite{RS} who showed that in noisy-edge static-topology networks  a constant overhead in communication can also be achieved (provided none of the processors are malicious), thus establishing an optimal-rate routing theorem for static-topology networks.  Finally, our work is closely related and builds upon to the problem of End-To-End Communication in distributed networks, studied by Afek and Gafni \cite{AG}, Awebuch, Mansour, and Shavit \cite{AMS}, and Afek, Awerbuch, Gafni, Mansour, Rosen, and Shavit \cite{Slide}, though none of these papers consider or ensure correctness in the setting of a node-controlling adversary that may corrupt the majority of the network.

\end{abstract}
%
%
\section{Introduction}
\indent \indent Our goal is to design a routing protocol for an unreliable and dynamically changing synchronous network that is resilient against malicious insiders who may try to destroy and alter messages or disrupt communication in any  way. We model the network as a communication graph $G=(V,E)$ where each vertex is a processor and each edge is a communication link. We do not assume that the topology of this graph is fixed or known by the processors. Rather, we assume a complete graph on $n$ vertices, where some of the edges are ``up'' and some are ``down'', and the status of each edge can change dynamically at any time.

We concentrate on the most basic task, namely how two processors in the network can exchange information. Thus, we assume that there are two designated vertices, called the sender $S$ and the receiver $R$, who wish to communicate with each other.  The sender has an infinite read-once input tape of {\em packets} and the receiver has an infinite write-once {\em output tape} which is initially empty.  We assume that packets are of some bounded size, and that any edge in the system that is ``up'' during some round can transmit only one packet (or control variables, also of bounded size) per round.

We will evaluate our protocol using the following three considerations:
	\begin{enumerate}\setlength{\itemsep}{4pt} \setlength{\parskip}{0pt} \setlength{\parsep}{0pt}
	\item {\bf Correctness.}  A protocol is {\it correct} if the sequence of packets output by the receiver is a prefix of packets appearing on the sender's input tape, without duplication or omission.
	
	\item {\bf Throughput.}  This measures the number of packets on the output tape as a function of the number of rounds that have passed.
	
	\item {\bf Processor Memory.}  This measures the memory required of each node by the protocol, independent of the number of packets to be transferred.
	\end{enumerate}
All three considerations will be measured in the worst-case scenario as standards that are guaranteed to exist regardless of adversarial interference.  One can also evaluate a protocol based on its dependence on global information to make decisions.  In the protocol we present in this paper, we will not assume there is any global view of the network available to the internal nodes.  Such protocols are termed ``local control,'' in that each node can make all routing decisions based only the local conditions of its adjacent edges and neighbors.

Our protocol is designed to be resilient against a malicious, polynomially-bounded adversary who may attempt to impact the {\it correctness}, {\it throughput}, and {\it memory} of our protocol by disrupting links between the nodes or taking direct control over the nodes and forcing them to deviate from our protocol in any manner the adversary wishes.  In order to relate our work to previous results and to clarify the two main forms of adversarial interference, we describe two separate (yet coordinated with each other) adversaries\footnote{The separation into two separate adversaries is artificial: our protocol is secure whether edge-scheduling and corruption of nodes are performed by two separate adversaries that have different capabilities yet can coordinate their actions with each other, or this can be viewed as a single coordinated adversary.}:
	\begin{enumerate}\setlength{\itemsep}{4pt} \setlength{\parskip}{0pt} \setlength{\parsep}{0pt}
	\item[] \textsf{\large Edge-Scheduling Adversary.}  This adversary controls the {\it links} between nodes every round.  More precisely, at each round, this adversary decides which edges in the network are up and which are down.  We will say that the edge-scheduling adversary is {\em conforming} if for every round there is at least one path from the sender to the receiver (although the path may change each round)\footnote{A more general definition of an edge-scheduling adversary would be to allow completely arbitrary edge failures, with the exception that in the limit there is no permanent cut between the sender and receiver. However, this definition (while more general) greatly complicates the exposition, including the definition of throughput rate, and we do not treat it here.}.  The adversary can make any arbitrary poly-time computation to maximize interference in routing, so long as it remains {\em conforming}.

	\item[] \textsf{\large Node-Controlling Adversary.}  This adversary controls the {\it nodes} of the network that it has corrupted.  More precisely, each round this adversary decides which nodes to corrupt.  Once corrupted, a node is forever under complete adversarial control and can behave in an arbitrary malicious manner. We say that the node-controlling adversary is {\em conforming} if every round there is a connection between the sender and receiver consisting of edges that are ``up'' for the round (as specified by the edge-scheduling adversary) and that passes through {\it uncorrupted} nodes.  We emphasize that this path can change each round, and there is no other restriction on which nodes the node-controlling adversary may corrupt (allowing even a vast majority of corrupt nodes).
	\end{enumerate}
There is another reason to view these adversaries as distinct: we deal with the challenges they pose to correctness, throughput, and memory in different ways.  Namely, aside from the conforming condition, the edge-scheduling adversary cannot be controlled or eliminated.  Edges themselves are not inherently ``good'' or ``bad,'' so identifying an edge that has failed does not allow us to forever refuse the protocol to utilize this edge, as it may come back up at any time (and indeed it could form a crucial link on the path connecting the sender and receiver that the conforming assumption guarantees).  In sum, we cannot hope to control or alter the behavior of the edge-scheduling adversary, but must come up with a protocol that works well regardless of the behavior of the ever-present (conforming) edge-scheduling adversary.

By contrast, our protocol will limit the amount of influence the node-controlling adversary has on correctness, throughput, and memory.  Specifically, we will show that if a node deviates from the protocol in a sufficiently destructive manner (in a well-defined sense), then our protocol will be able to identify it as corrupted in a timely fashion. Once a corrupt node has been identified, it will be {\it eliminated} from the network.  Namely, our protocol will call for honest nodes to refuse all communication with nodes that have been identified as corrupt\footnote{The {\it conforming} assumption guarantees that the sender and receiver are incorruptible, and our protocol places the responsibility of identifying and eliminating corrupt nodes on these two nodes.}.  Thus, there is an inherent difference in how we handle the edge-scheduling adversary verses how we handle the node-controlling adversary.  We can restrict the influence of the latter by eliminating the nodes it has corrupted, while the former must be dealt with in a more ever-lasting manner.
\subsection{Previous Work}
\indent \indent To motivate the importance of the problem we consider in this paper, and to emphasize the significance of our result, it will be useful to highlight recent works in related areas.  To date, routing protocols that consider adversarial networks have been of two main flavors: {\it End-to-End Communication} protocols that consider dynamic topologies (a notion captured by our ``edge-scheduling adversary''), and {\it Fault Detection and Localization} protocols, which handle devious behavior of nodes (as modeled by our ``node-controlling adversary'').

\bigskip
\noindent {\sc End-to-End Communication:} One of the most relevant research directions to our paper is the notion of End-to-End communication in distributed networks, considered by Afek and Gafni \cite{AG}, Awerbuch, Mansour and Shavit \cite{AMS},  Afek, Awebuch,  Gafni, Mansour, Rosen, and Shavit \cite{Slide}, and Kushilevitz, Ostrovsky and Rosen \cite{KO} . Indeed, our starting point is the {\em Slide} protocol\footnote{Also known in practical works as ``gravitational flow'' routing.} developed in these works. It was designed to perform end-to-end communication with bounded memory in a model where (using our terminology) an edge-scheduling adversary controls the edges (subject to the constraint that there is no permanent cut between the sender and receiver).  The Slide protocol has proven to be incredibly useful in a variety of settings, including multi-commodity flow (Awerbuch and Leigthon \cite{AL}) and in developing routing protocols that compete well (in terms of packet loss) against an online bursty adversary (\cite{AKOR}).  However, prior to our work there was no version of the Slide protocol that could handle malicious behavior of the nodes.  A comparison of various versions of the Slide protocol and our protocol is featured in Figure \ref{tableComparison} of Section \ref{ourResults} below.

\bigskip
\noindent {\sc Fault Detection and Localization Protocols}: At the other end, there have been a number of works that explore the possibility of a node-controlling adversary that can corrupt nodes.  In particular, there is a recent line of work that considers a network consisting of a {\em single path} from the sender to the receiver, culminating in the recent work of Barak, Goldberg and Xiao \cite{BGX} (for further background on fault localization see references therein). In this model, the adversary can corrupt any node on the path (except the sender and receiver) in a dynamic and malicious manner. Since corrupting any node on the path will sever the honest connection between $S$ and $R$, the goal of a protocol in this model is {\em not} to guarantee that all messages sent to $R$ are received. Instead, the goal is to {\it detect} faults when they occur and to {\it localize} the fault to a single edge.

There have been many results that provide Fault Detection (FD) and Fault Localization (FL) in this model.  In Barak et.\ al.\ \cite{BGX}, they formalize the definitions in this model and the notion of a secure FD/FL protocol, as well as providing lower bounds in terms of communication complexity to guarantee accurate fault detection/location in the presence of a node-controlling adversary. While the Barak et.\ al.\ paper has a similar flavor to our paper, we emphasize that their protocol does not seek to guarantee successful or efficient routing between the sender and receiver. Instead, their proof of security guarantees that if a packet is deleted, malicious nodes cannot collude to convince $S$ that no fault occurred, nor can they persuade $S$ into believing that the fault occurred on an honest edge. Localizing the fault in their paper relies on cryptographic tools, and in particular the assumption that one-way functions exist. Although utilizing these tools (such as MACs or Signature Schemes) increases communication cost, it is shown by Goldberg, Xiao, Barak, and  Redford \cite{GXBR} that the existence of a protocol that is able to securely detect faults (in the presence of a node-controlling adversary) implies the existence of one-way functions, and it is shown in Barak et.\ al.\ \cite{BGX} that {\it any} protocol that is able to securely localize faults necessarily requires the intermediate nodes to have a trusted setup. The proofs of these results do not rely on the fact that there is a single path between $S$ and $R$, and we can therefore extend them to the more general network encountered in our model to justify our use of cryptographic tools and a trusted setup assumption (i.e.\ PKI) to identify malicious behavior.

Another paper that addresses routing in the Byzantine setting is the work of Awerbuch, Holmes, Nina-Rotary and Rubens \cite{AHNR}, though this paper does not have a fully formal treatment of security, and indeed a counter-example that challenges its security is discussed in the appendix of \cite{BGX}.

\bigskip
\noindent{\sc Error-correction in the active setting:} Due to space considerations, we will not be able to give a comprehensive account of all the work in this area.  Instead we highlight some of the most relevant works and point out how they differ from our setting and results. For a lengthy treatment of error-correcting codes against polynomially bounded adversaries, we refer to the work of Micali at. al \cite{MPSW} and references therein.  It is important to note that this work deals with a graph with a single ``noisy'' edge, as modelled by an adversary who can partially control and modify information that crosses the edge.  In particular, it does not address throughput efficiency or memory considerations in a full communication network, nor does it account for malicious behavior at the vertices.  Also of relevance is the work on Rajagopalan and Schulman on error-correcting network coding \cite{RS}, where they show how to correct noisy edges during distributed computation. Their work does not consider actively  malicious nodes, and thus is different from our setting. It should also be noted that their work utilizes Schulman's tree-codes \cite{S} that allow length-flexible online error-correction. The important difference between our work and that of Schulman is that in our network setting, the amount of malicious activity of corrupt nodes is not restricted.
\subsection{Our Results} \label{ourResults}
\indent \indent  To date, there has not been a protocol that has considered simultaneously a network susceptible to faults occurring due to edge-failures {\it and} faults occurring due to malicious activity of corrupt nodes.  The end-to-end communication works are not secure when the nodes are allowed to become corrupted by a node-controlling adversary, and the fault detection and localization works focus on a {\it single path} for some duration of time, and do not consider a fully distributed routing protocol that utilizes the entire network and attempts to maximize throughput efficiency while guaranteeing correctness in the presence of edge-failures and corrupt nodes.  Indeed, our work answers one of the open questions posed in the Barak et.\ al.\ paper regarding fault localization on multiple paths.  In this paper we bridge the gap between these two research areas and obtain the first routing protocol simultaneously secure against both an edge-scheduling adversary and a node-controlling adversary, even if these two adversaries attack the network using an arbitrary coordinated poly-time strategy.  Furthermore, our protocol achieves comparable efficiency standards in terms of throughput and processor memory as state-of-the-art protocols that are not secure against a node-controlling adversary and does so using {\em local-control} protocols.  An informal statement of our result and comparison of our protocol to existing protocols can be found below.  Although not included in the table, we emphasize that the linear transmission rate that we achieve (assuming at least $n^2$ messages are sent) is asymptotically optimal, as {\it any} protocol operating in a network with a single path connecting sender and receiver can do no better than one packet per round.

\bigskip
\noindent {\bf A ROUTING THEOREM FOR ADVERSARIAL NETWORKS (Informal):} {\sf If one-way functions exist, then for any $n$-node graph and $k$ sufficiently large, there exists a trusted-setup {\em linear throughput} transmission protocol that can send $n^2$ messages in $O(n^2)$ rounds with $O(n^4 (k+\log n))$ memory per processor that is resilient against {\bf any} poly-time conforming Edge-Scheduling Adversary and {\bf any} conforming poly-time Node-Controlling Adversary, with negligible (in $k$) probability of failure or decoding error.}\\
%
\begin{figure}[h!t]
\begin{center}
$\begin{array}{|c||c|c|c|c|}  \hline   & & & & \\ & \mbox{\footnotesize Secure Against} & \mbox{\footnotesize Secure Against} & \mbox{\footnotesize Processor} & \mbox{\footnotesize Throughput Rate} \\ & \mbox{\footnotesize Edge-Sched.\ Ad?} & \mbox{\footnotesize Node-Cntr.\ Ad?} & \mbox{\footnotesize Memory}& {\scriptstyle x \thickspace} \mbox{\footnotesize rounds} {\scriptstyle \rightarrow f(x) \thickspace} \mbox{\footnotesize packets}\\ \hline \hline \mbox{\footnotesize Slide Protocol of \cite{Slide}} &  YES & NO & O(n^2 \log n) & f(x) = O(x-n^2) \\ \hline \mbox{\footnotesize Slide Protocol of \cite{KO}} & YES & NO & O(n \log n) & f(x) = O(x/n-n^2) \\ \hline \mbox{\footnotesize (folklore)} & & & & \\ \mbox{\footnotesize (Flooding $+$ Signatures)} & YES & YES & O(1) & f(x) = O(x/n-n^2)\\ \hline \mbox{\footnotesize (folklore)} & & & & \\ \mbox{\footnotesize (Signatures $+$ Sequence No.'s)} & YES & YES & unbounded & f(x) = O(x - n^2) \\ \hline \hline & & & &  \\
\mbox{\footnotesize Our Protocol} & YES & YES & O(n^4 {\scriptstyle (k}\mbox{\footnotesize +}{\scriptstyle \log n)}) & f(x) = O(x - n^2) \\
\hline
\end{array}$
\end{center}
\caption{Comparison of Our Protocol to Related Existing Protocols and Folklore. }
\label{tableComparison}
\end{figure}
\vspace{-.5cm}
\section{Challenges and Na\"{i}ve Solutions}
\indent \indent Before proceeding, it will be useful to consider a couple of na\"{i}ve %
solutions that achieve the goal of {\it correctness} (but perform poorly in terms of {\it throughput}), and help to illustrate some of the technical challenges that our theorem resolves.  Consider the approach of having the sender continuously {\em flood} a single signed packet into the network for $n$ rounds.  Since the {\it conforming} assumption guarantees that the network provides a path between the sender and receiver through honest nodes at every round, this packet will reach the receiver within $n$ rounds, regardless of adversarial interference.  After $n$ rounds, the sender can begin flooding the network with the next packet, and so forth\footnote{An alternative approach would have the sender continue flooding the first packet, and upon receipt, the receiver floods confirmation of receipt.  This alternative solution requires sequence numbers to accompany packets/confirmations, and the rule that internal nodes only keep and broadcast the packet and confirmation with largest sequence number.  Although this alternative may potentially speed things up, in the worst-case it will still take $O(n)$ rounds for a single packet/confirmation pair to be transmitted.}.  Notice that this solution will require each processor to store and continuously broadcast a single packet at any time, and hence this solution achieves excellent efficiency in terms of {\it processor memory}. However, notice that the {\em throughput} rate is sub-linear, namely after $x$ rounds, only $O(x/n)$ packets have been outputted by the receiver.

One idea to try to improve the throughput rate might be to have the sender streamline the process, sending packets with ever-increasing sequence numbers without waiting for $n$ rounds to pass (or signed acknowledgments from the receiver) before sending the next packet.  In particular, across each of his edges the sender will send every packet once, waiting only for the neighboring node's confirmation of receipt before sending the next packet across that edge.  The protocol calls for the internal nodes to act similarly.  Analysis of this approach shows that not only has the attempt to improve throughput failed (it is still $O(x/n)$ in the worst-case scenario), but additionally this modification requires arbitrarily large (polynomial in $n$ and $k$) processor memory, since achieving correctness in the dynamic topology of the graph will force the nodes to remember all of the packets they see until they have broadcasted them across all adjacent edges or seen confirmation of their receipt from the receiver.
\subsection{Challenges in Dealing with Node-Controlling Adversaries}
\indent \indent In this section, we discuss some potential strategies that the node-controlling and edge-scheduling adversaries\footnote{We give a formal definition of the adversary in Section \ref{adversary}.} may incorporate to disrupt network communication.  Although our theorem will work in the presence of {\it arbitrary} malicious activity of the adversarial controlled nodes (except with negligible probability), it will be instructive to list a few obvious forms of devious behavior that our protocol must protect against.  It is important to stress that this list is {\it not} intended to be exhaustive. Indeed, we do not claim to know all the specific ways an arbitrary polynomially bounded adversary may force nodes to deviate from a given protocol, and in this paper we rigorously prove that our protocol is secure against all possible deviations.
	\begin{itemize}\setlength{\itemsep}{2pt} \setlength{\parskip}{0pt} \setlength{\parsep}{0pt}
	\item \textsf{Packet Deletion/Modification.}  Instead of forwarding a packet, a corrupt node ``drops it to the floor'' (i.e.\ deletes it or effectively deletes it by forever storing it in memory), or modifies the packet before passing it on.  Another manifestation of this is if the sender/receiver requests fault localization information of the internal nodes, such as providing documentation of their interactions with neighbors.  A corrupt node can then block or modify information that passes through it in attempt to hide malicious activity or implicate an honest node.

	\item \textsf{Introduction of Junk/Duplicate Packets.} The adversary can attempt to disrupt communication flow and ``jam'' the network by having corrupted nodes introduce junk packets or re-broadcast old packets.  Notice that junk packets can be handled by using cryptographic signatures to prevent introduction of ``new'' packets, but this does not control the re-transmission of old, correctly signed packets.

	\item \textsf{Disobedience of Transfer Rules.}  If the protocol specifies how nodes should make decisions on where to send packets, etc., then corrupt nodes can disregard these rules.  This includes ``lying'' to adjacent nodes about their current state.

	\item \textsf{Coordination of Edge-Failures.}  The edge-scheduling adversary can attempt to disrupt communication flow by scheduling edge-failures in any manner that is consistent with the {\it conforming} criterion.  Coordinating edge failures can be used to impede correctness, memory, and throughput in various ways: e.g.\ packets may become lost across a failed edge, stuck at a suddenly isolated node, or arrive at the receiver out of order.  A separate issue arises concerning fault localization: when the sender/receiver requests documentation from the internal nodes, the edge-scheduling adversary can slow progress of this information, as well as attempt to protect corrupt nodes by allowing them to ``play-dead'' (setting all of its adjacent edges to be {\it down}), so that incriminating evidence cannot reach the sender.
	\end{itemize}
\subsection{Highlights of Our Solution} \label{highlights}
\indent \indent Our starting point is the Slide protocol \cite{Slide}, which has enjoyed practical success in networks with dynamic topologies, but is not secure against nodes that are allowed to behave maliciously.  We provide a detailed description of our version of the Slide protocol in Section \ref{nAP}, but highlight the main ideas here.  Begin by viewing the edges in the graph as consisting of two directed edges, and associate to each end of a directed edge a {\it stack} data-structure able to hold $2n$ packets and to be maintained by the node at that end.  The protocol specifies the following simple, local condition for transferring a packet across a directed edge: if there are more packets in the stack at the originating end than the terminating end, transfer a packet across the edge.  Similarly, within a node's local stacks, packets are shuffled to average out the stack heights along each of its edges.  Intuitively, packet movement is analogous to the flow of water: high stacks create a pressure that force packets to ``flow'' to neighboring lower stacks.  At the source, the sender maintains the pressure by filling his outgoing stacks (as long as there is room) while the receiver relieves pressure by consuming packets and keeping his stacks empty.  Loosely speaking, packets traveling to nodes ``near'' the sender will therefore require a very large potential, packets traveling to nodes near the receiver will require a small potential, and packet transfers near intermediate nodes will require packages to have a moderate potential.  Assuming these potential requirements exist, packets will pass from the sender with a high potential, and then ``flow'' downwards across nodes requiring less potential, all the way to the receiver.

Because the Slide protocol provides a fully distributed protocol that works well against an edge-scheduling adversary, our starting point was to try to extend the protocol by using digital signatures\footnote{In this paper we use  public-key operations to sign individual packets with control information. Clearly, this is too expensive to do per-packet in practice.  There are methods of amortizing the cost of signatures by signing ``batches'' of packets;  using private-key initialization \cite{BGX,GXBR}, or using a combination of private-key and public key operations, such as ``on-line/off-line'' signatures \cite{SIG1,SIG2}.  For the sake of clarity and since the primary focus of our paper is theoretical feasibility, we restrict our attention to the straight-forward public-key setting without considering these additional cost-saving techniques.} to provide resilience against Byzantine attacks and arbitrary malicious behavior of corrupt nodes. This proved to be a highly nontrivial task that required us to develop a lot of additional machinery, both in terms of additional protocol ideas and novel techniques for proving correctness. We give a detailed explanation of our techniques in Section \ref{aP} and formal pseudo-code in Section \ref{MpseudoCode}, as well as providing rigorous proofs of security in Section \ref{MProofs}.  However, below we first give a sample of some of the key ideas we used in ensuring our additional machinery would be provably secure against a node-controlling adversary, and yet not significantly affect throughput or memory, compared to the original Slide protocol:
	\begin{itemize}\setlength{\itemsep}{3pt} \setlength{\parskip}{0pt} \setlength{\parsep}{0pt}
	\item \textsc{Addressing the ``Coordination of Edge-Scheduling'' Issues.}  In the absence of a node-controlling adversary, previous versions of the Slide protocol (e.g.\ \cite{Slide}) are secure and efficient against an edge-scheduling adversary, and it will be useful to discuss how some of the challenges posed by a network with a dynamic topology are handled.  First, note that the total capacity of the stack data-structure is bounded by $4n^3$.  That is, each of the $n$ nodes can hold at most $2n$ packets in each of their $2n$ stacks (along each directed edge) at any time.
		\begin{itemize}\setlength{\itemsep}{2pt} \setlength{\parskip}{0pt} \setlength{\parsep}{0pt}
		\item To handle the loss of packets due to an edge going down while transmitting a packet, a node is required to maintain a copy of each packet it transmits along an edge until it receives confirmation from the neighbor of successful receipt.

		\item To handle packets becoming stuck in some internal node's stack due to edge failures, {\it error-correction} is utilized to allow the receiver to decode a full message without needing every packet.  In particular, if an error-correcting code allowing a fraction of $\lambda$ faults is utilized, then since the capacity of the network is $4n^3$ packets, if the sender is able to pump $4n^3/\lambda$ codeword packets into the network and there is no malicious deletion or modification of packets, then the receiver will necessarily have received enough packets to decode the message.
		
		\item The Slide protocol has a natural bound in terms of memory per processor of $O(n^2 \log n)$ bits, where the bottleneck is the possibility of a node holding up to $2n^2$ packets in its stacks, where each packet requires $O(\log n)$ bits to describe its position in the code.
		\end{itemize}
	Of course, these techniques are only valid if nodes are acting honestly, which leads us to our first extension idea.

	\item \textsc{Handling Packet Modification and Introduction of Junk Packets.}  Before inserting any packets into the network, the sender will authenticate each packet using his digital signature, and intermediate nodes and the receiver never accept or forward messages not appropriately signed.  This simultaneously prevents honest nodes becoming bogged down with junk packets, as well as ensuring that if the receiver has obtained enough authenticated packets to decode, a node-controlling adversary cannot impede the successful decoding of the message as the integrity of the codeword packets is guaranteed by the inforgibility of the sender's signature.

	\item \textsc{Fault Detection.}  In the absence of a node-controlling adversary, our protocol looks almost identical to the Slide protocol of \cite{Slide}, with the addition of signatures that accompany all interactions between two nodes.  First, the sender attempts to pump the $4n^3/\lambda$ codeword packets of the first message into the network, with packet movement exactly as in the original Slide protocol.  We consider all possible outcomes:
		\begin{enumerate}\setlength{\itemsep}{2pt} \setlength{\parskip}{0pt} \setlength{\parsep}{0pt}
		\item \textsf{The sender is able to insert all codeword packets and the receiver is able to decode.}  In this case, the message was transmitted successfully, and our protocol moves to transfer the next message.

		\item \textsf{The sender is able to insert all codeword packets, but the receiver has not received enough to decode.}  In this case, the receiver floods the network with a single-bit message indicating {\it packet deletion} has occurred.

		\item \textsf{The sender is able to insert all codeword packets, but the receiver cannot decode because he has received duplicated packets.}  Although the sender's authenticating signature guarantees the receiver will not receive junk or modified packets, a corrupt node is able to duplicate valid packets.  Therefore, the receiver may receive enough packets to decode, but cannot because he has received duplicates.  In this case, the receiver floods the network with a single message indicating the label of a duplicated packet.

		\item \textsf{After some amount of time, the sender still has not inserted all codeword packets.}  In this case, the duplication of old packets is so severe that the network has become jammed, and the sender is prevented from inserting packets even along the honest path that the conforming assumption guarantees.  If the sender believes the jamming cannot be accounted for by edge-failures alone, he will halt transmission and move to localizing a corrupt node\footnote{We emphasize here the importance that the sender is able to distinguish the case that the jamming is a result of the edge-scheduling adversary's controlling of edges verses the case that a corrupt node is duplicating packets.  After all, in the case of the former, there is no reward for ``localizing'' the fault to an edge that has failed, as {\it all} edges are controlled by the edge-scheduling adversary, and therefore no edge is inherently better than another.  But in the case a node is duplicating packets, if the sender can identify the node, it can eliminate it and effectively reduce the node-controlling adversary's ability to disrupt communication in the future.}.  One contribution this paper makes is to prove a lower bound on the insertion rate of the sender for the Slide protocol {\it in the absence of the node-controlling adversary}.  This bound not only alerts the sender when the jamming he is experiencing exceeds what can be expected in the absence of corrupt nodes, but it also provides a mechanism for localizing the offending node(s).
		\end{enumerate}
	The above four cases exhaust all possibilities.  Furthermore, if a transmission is not successful, the sender is not only able to {\it detect} the fact that malicious activity has occured, but he is also able to distinguish the {\it form} of the malicious activity, i.e.\ which case 2-4 he is in.  Meanwhile, for the top case, our protocol enjoys (within a constant factor) an equivalent throughput rate as the original Slide protocol.

	\item \textsc{Fault Localization.}  Once a fault has been detected, it remains to describe how to {\it localize} the problem to the offending node.  To this end, we use digital signatures to achieve a new mechanism we call ``Routing with Responsibility.''  By forcing nodes to sign key parts of every communication with their neighbors during the transfer of packets, they can later be held accountable for their actions.  In particular, once the sender has identified the reason for failure (cases 2-4 above), he will request all internal nodes to return {\it status reports}, which are signatures on the relevant parts of the communication with their neighbors.  We then prove in each case that with the complete status report from every node, the sender can with overwhelming probability identify and eliminate a corrupt node. Of course, malicious nodes may choose not to send incriminating information. We handle this separately as explained below.

	\item \textsc{Processor Memory.}  The signatures on the communication a node has with its neighbors for the purpose of fault localization is a burden on the memory required of each processor that is not encountered in the original Slide protocol.  One major challenge was to reduce the amount of signed information each node must maintain as much as possible, while still guaranteeing that each node has maintained ``enough'' information to identify a corrupt node in the case of arbitrary malicious activity leading to a failure of type 2-4 above.  The content of Theorem \ref{memoryThm} in Section \ref{aP} demonstrates that the extra memory required of our protocol is a factor of $n^2$ higher than that of the original Slide protocol.

	\item \textsc{Incomplete Information.}  As already mentioned, we show that regardless of the reason of failure 2-4 above, once the sender receives the status reports from every node, a corrupt node can be identified.  However, this relies on the sender obtaining all of the relevant information; the absence of even a single node's information can prevent the localization of a fault.  We address this challenge in the following ways:
		\begin{enumerate}\setlength{\itemsep}{2pt} \setlength{\parskip}{0pt} \setlength{\parsep}{0pt}
		\item We minimize the amount of information the sender requires of each node.  This way, a node need not be connected to the sender for very many rounds in order for the sender to receive its information.  Specifically, regardless of the reason for failure 2-4 above, a status report consists of only $n$ pieces of information from each node, i.e.\ one packet for each of its edges.

		\item If the sender does not have the $n$ pieces of information from a node, it cannot afford to wait indefinitely.  After all, the edge-scheduling adversary may keep the node disconnected indefinitely, or a corrupt node may simply refuse to respond.  For this purpose, we create a {\it blacklist} for non-responding nodes, which will disallow them from transferring codeword packets in the future.  This way, anytime the receiver fails to decode a codeword as in cases 2-4 above, the sender can request the information he needs, blacklist nodes not responding within some short amount of time, and then re-attempt to transmit the codeword using only non-blacklisted nodes.  Nodes should not transfer codeword packets to blacklisted nodes, but they do still communicate with them to transfer the information the sender has requested.  If a new transmission again fails, the sender will only need to request information from nodes that were participating, i.e.\ he will {\it not} need to collect new information from blacklisted nodes (although the nodes will remain blacklisted until the sender gets the original information he requested of them).  Nodes will be removed from the blacklist and re-allowed to route codeword packets as soon as the sender receives their information.
		\end{enumerate}
	\item \textsc{The Blacklist.}  Blacklisting nodes is a delicate matter; we want to place malicious nodes ``playing-dead'' on this list, while at the same time we don't want honest nodes that are temporarily disconnected from being on this list for too long.  We show in Theorem \ref{mainAdThm} and Lemma \ref{maxWasted} that the occasional honest node that gets put on the blacklist won't significantly hinder packet transmission. Intuitively, this is true because any honest node that is an important link between the sender and receiver will not remain on the blacklist for very long, as his connection to the sender guarantees the sender will receive all requested information from the node in a timely manner.

	\hspace{.5cm} Ultimately, the blacklist allows us to control the amount of malicious activity a single corrupt node can contribute to.  Indeed, we show that each failed message transmission (cases 2-4 above) can be localized (eventually) to (at least) one corrupt node.  More precisely, the blacklist allows us to argue that malicious activity can cause at most $n$ failed transmissions before a corrupt node can necessarily be identified and eliminated.  Since there are at most $n$ corrupt nodes, this bounds the number of failed transmissions at $n^2$.  The result of this is that other than at most $n^2$ failed message transmissions, our protocol enjoys the same throughput efficiency of the old Slide protocol.  The formal statement of this fact can be found in Theorem \ref{mainAdThm} in Section \ref{aP}, and its proof can be found in Section \ref{MProofs}.
	\end{itemize}
\section{The Formal Model}\label{firstApp}
\indent \indent It will be useful to describe two models in this section, one in the presence of an edge-scheduling adversary (all nodes act ``honestly''), and one in the presence of an adversary who may ``corrupt'' some of the nodes in the network.  In Section \ref{nAP} we present an efficient protocol (``Slide'') that works well in the edge-scheduling adversarial model, and we then extend this protocol in Section \ref{aP} to work in the additional presence of the node-controlling adversary.
\subsection{The Edge-Scheduling Adversarial Model} \label{model}
\indent \indent We model a communication network by an undirected graph $G = (V,E)$, where $|V| = n$.  Each vertex (or {\it node}) represents a processor that is capable of storing information (in its {\it buffers}) and passing information to other nodes along the edges.  We distinguish two nodes, the {\it sender}, denoted by $S$, and the {\it receiver}, denoted by $R$.  In our model, $S$ has an input stream of {\it messages} $\{ m_1, m_2, \dots \}$ of uniform size that he wishes to transmit through the network to $R$.  As mentioned in the Introduction, the three commodities we care about are {\it Correctness}, {\it Throughput}, and {\it Processor Memory}.

We assume a {\it synchronous} network, so that there is a universal clock that each node has access to\footnote{Although synchronous networks are difficult to realize in practice, we can further relax the model to one in which there is a known upper-bound on the amount of time an active edge can take to transfer a packet.}.  The global time is divided into discrete chunks, called {\it rounds}, during which nodes communicate with each other and transfer packets.  Each round consists of two equal intervals of unit time called {\it stages}, so that all nodes are synchronized in terms of when each stage begins and ends.  We assume that the edges have some fixed capacity $P$ in terms of the amount of information that can be transmitted across them per stage.  The messages will be sub-divided into packets of uniform size $P$, so that exactly one packet can be transferred along an edge per stage\footnote{Our protocol for the node-controlling adversarial model will require the packets to include signatures from a cryptographic signature scheme.  The security of such schemes depend on the security parameter $k$, and the size of the resulting signatures have size $O(k)$.  Additionally, error-correction will require packets to carry with them an index of $O(\log n)$ bits.  Therefore, we assume that $P \geq (k+ \log n)$, so that in each time step a complete packet (with signature and index) can be transferred.}.

The sole purpose of the network is to transmit the messages from $S$ to $R$, so $S$ is the only node that introduces new messages into the network, and $R$ is the only node that removes them from the network (although below we introduce a node-controlling adversary who may corrupt the intermediate nodes and attempt to disrupt the network by illegally deleting/introducing messages).  Although the edges in our model are bi-directional, it will be useful to consider each link as consisting of two directed edges.  Except for the {\it conforming} restriction (see below), we allow the edges of our network to fail and resurrect arbitrarily.  We model this via an {\it Edge-Scheduling Adversary}, who controls the status of each edge of the network, and can alter the state of any edge at any time.   We say that an edge is {\it active} during a given stage/round if the edge-scheduling adversary allows that edge to remain ``up'' for the entirety of that stage/round.  We impose one restriction on the failure of edges:
	\begin{defn} An edge-scheduling adversary is {\em conforming} if for every round of the protocol, there exists at least one path between $S$ and $R$ consisting of edges that {\it active} for the entirety of the round.\end{defn}
For a given round $\mathtt{t}$, we will refer to the path guaranteed by the conforming assumption as the {\it active path} of round $\mathtt{t}$.  Notice that although the conforming assumption guarantees the existence of an active path for each round, it is {\it not} assumed that any node (including $S$ and $R$) is aware of what that path is.  Furthermore, this path may change from one round to the next.  The edge-scheduling adversary cannot affect the network in any way other than controlling the status of the edges.  In the next section, we introduce a node-controlling adversary who can take control of the nodes of the network\footnote{The distinction between the two kinds of adversaries is made solely to emphasize the contribution of this paper.  Edge-scheduling adversaries (as described above) are commonly used to model edge failures in networks, while the contribution of our paper is in controlling a node-controlling adversary, which has the ability to corrupt the {\it nodes} of the network.}.
\subsection{The Node-Controlling $+$ Edge-Scheduling Adversarial Model} \label{adversary}
\indent \indent This model begins with the edge-scheduling adversarial model described above, and adds a polynomially bounded Node-Controlling Adversary that is capable of corrupting {\it nodes} in the network.  The node-controlling adversary is {\it malicious}, meaning that the adversary can take complete control over the nodes he corrupts, and can therefore force them to deviate from any protocol in whatever manner he likes.  We further assume that the adversary is {\it dynamic}, which means that he can corrupt nodes at any stage of the protocol, deciding which nodes to corrupt based on what he has observed thus far\footnote{Although the node-controlling adversary is {\it dynamic}, he is still constrained by the conforming assumption.  Namely, the adversary may not corrupt nodes that have been, or will be, part of any active path connecting sender and receiver.}.  For a thorough discussion of these notions, see \cite{Gol} and references therein.

As in Multi-Party Computation (MPC) literature, we will need to specify an ``access-structure'' for the adversary.
	\begin{defn} A node-controlling adversary is {\it conforming} if he does not corrupt any nodes who have been or will be a part of any round's active path.\end{defn}
Apart from this restriction, the node-controlling adversary may corrupt whoever he likes (i.e.\ it is not a threshold adversary).  Note that the {\it conforming} assumption implicitly demands that $S$ and $R$ are incorruptible, since they are always a part of any active path.  Also, this restriction on the adversary is really more a statement about when our results remain valid.  This is similar to e.g.\ {\it threshold adversary} models, where the results are only valid if the number of corrupted nodes does not exceed some threshold value $t$.  Once corrupted, a node is forever considered to be a corrupt node that the adversary has total control over (although the adversary may choose to have the node act honestly).

Notice that because correctness, throughput, and memory are the only commodities that our model values, an {\it honest-but-curious} adversary is completely benign, as privacy does not need to be protected\footnote{If desired, privacy can be added trivially by encrypting all packets.} (indeed, any intermediate node is presumed to be able to read any packet that is passed through it).  Our techniques for preventing/detecting malicious behavior will be to incorporate a {\it digital signature} scheme that will serve the dual purpose of validating information that is passed between nodes, as well as holding nodes accountable for information that their signature committed them to.

We assume that there is a Public-Key Infrastructure (PKI) that allows digital signatures.  In particular, before the protocol begins we choose a security parameter $k$ sufficiently large and run a key generation algorithm for a digital signature scheme, producing $n = |G|$ (secret key, verification key) pairs $(sk_N, vk_N)$.  As output to the key generation, each processor $N \in G$ is given its own private signing key $sk_N$ and a list of all $n$ signature verification keys $vk_{\widehat{N}}$ for all nodes $\widehat{N} \in G$.  In particular, this allows the sender and receiver to sign messages to each other that cannot be forged (except with negligible probability in the security parameter) by any other node in the system.
\section{Routing Protocol in the Edge-Scheduling Adversarial Model} \label{nAP}
\indent \indent In this section we formally describe our edge-scheduling protocol, which is essentially the ``Slide'' protocol of \cite{Slide}.
\subsection{Definitions and High-Level Ideas} \label{nonAdProtocol}
\indent \indent The goal of the protocol is to transmit a sequence of messages $\{m_1, m_2, \dots \}$ of uniform size from the sender $S$ to the receiver $R$ (refer to Section \ref{model} for a complete description of the model).  Each node will maintain a stack (i.e.\ FILO buffers) along each of its (directed) edges that can hold up to $2n$ packets concurrently.  To allow for packets to become stuck in the buffers, we will utilize {\it error-correction} (see e.g.\ \cite{Gol}).  Specifically, the messages $\{m_1, m_2, \dots \}$ are converted into codewords $\{c_1, c_2, \dots\}$, allowing the receiver to decode a message provided he has received an appropriate number (depending on the {\it information rate} and {\it error-rate} of the code) of bits of the corresponding codeword.  In this paper, we assume the existence of a error-correcting code with information rate $\sigma$ and error rate $\lambda$.

As part of the setup of our protocol, we assume that the messages $\{m_1, m_2, \dots\}$ have been partitioned to have uniform size $M = \frac{6 \sigma Pn^3}{\lambda}$ (recall that $P$ is the capacity of each edge and $\sigma$ and $\lambda$ are the parameters for the error-correction code).  The messages are expanded into codewords, which will have size $C = \frac{M}{\sigma} = \frac{6Pn^3}{\lambda}$.  The codewords are then divided into $\frac{C}{P}= \frac{6n^3}{\lambda}$ packets of size $P$.  We emphasize this quantity for later use:\vspace{-4pt}
	\begin{equation} \label{capacity}
	D := \frac{6n^3}{\lambda} = \mbox{\it number of packets per codeword}.
	\end{equation}
Note that the only ``noise'' in our network results from undelivered packets or out-dated packets (in the edge-scheduling adversarial model, any packet
that $R$ receives has not been altered).  Therefore, since each codeword consists of $D = \frac{6n^3}{\lambda}$ packets, by definition of $\lambda$, if $R$ receives $(1-\lambda)D = (1-\lambda)\left(\frac{6n^3}{\lambda}\right)$ packets corresponding to the same codeword, he will be able to decode.  We emphasize this fact:
	\begin{itemize}\setlength{\itemsep}{2pt} \setlength{\parskip}{0pt} \setlength{\parsep}{0pt}
	\item[] {\bf Fact 1.}  If the receiver has obtained $D-6n^3 = (1-\lambda)\left(\frac{6n^3}{\lambda}\right)$ packets from any codeword, he will be able to decode the codeword to obtain the corresponding message.
	\end{itemize}
Because our model allows for edges to go up/down, we force each node to keep incoming and outgoing buffers for every {\it possible} edge, even if that edge isn't part of the graph at the outset.  We introduce now the notion of {\it height} of a buffer, which will be used to determine when packets are transferred and how packets are re-shuffled between the internal buffers of a given node between rounds.
\begin{defn} The {\em height} of an incoming/outgoing buffer is the number of packets currently stored in that buffer.\end{defn}
The presence of an edge-scheduling adversary that can force edges to fail at any time complicates the interaction between the nodes.  Note that our model does not assume that the nodes are aware of the status of any of its adjacent edges, so failed edges can only be detected when information that was supposed to be passed along the edge does not arrive.  We handle potential edge failures as follows.  First, the incoming/outgoing buffers at either end of an edge will be given a ``status'' (normal or problem).  Also, to account for a packet that may be lost due to edge failure during transmission across that edge, a node at the receiving end of a failed edge may have to leave room in its corresponding incoming buffer.  We refer to this gap as a {\it ghost packet}, but emphasize that the {\it height} of an incoming buffer is {\it not} affected by ghost packets (by definition, {\it height} only counts packets that are present in the buffer).  Similarly, when a sending node ``sends'' a packet across an edge, it actually only sends a copy of the packet, leaving the original packet in its outgoing buffer.  We will refer to the original copy left in the outgoing buffer as a {\it flagged packet}, and note that flagged packets continue to contribute to the height of an outgoing buffer until they are deleted.

The codewords will be transferred sequentially, so that at any time, the sender is only inserting packets corresponding to a single codeword.  We will refer to the rounds for which the sender is inserting codeword packets corresponding to the $i^{th}$ codeword as the $i^{th}$ {\it transmission.}  Lemma \ref{obsOne} below states that after the sender has inserted $D-2n^3$ packets corresponding to the same codeword, the receiver can necessarily decode.  Therefore, when the sender has inserted this many packets corresponding to codeword $b_i$, he will clear his outgoing buffers and begin distributing packets corresponding to the next codeword $b_{i+1}$.
\subsection{Detailed Description of the Edge-Scheduling Protocol} \label{nonAdDesc}
\indent \indent We describe now the two main parts of the edge-scheduling adversarial routing protocol: the Setup and the Routing Phase.  For a formal presentation of the pseudo-code, see Section \ref{pseudoCode}.\hspace*{\fill} \\ \hspace*{\fill} \\
\noindent {\bf Setup.}
Each internal (i.e.\ not $S$ or $R$) node has the following buffers:
	\begin{enumerate}\setlength{\itemsep}{2pt} \setlength{\parskip}{0pt} \setlength{\parsep}{0pt}
	\item {\it Incoming Buffers.}  Recall that we view each bi-directional edge as consisting of two directed edges.  Then for each incoming edge, a node will have a buffer that has the capacity to hold $2n$ packets at any given time.  Additionally, each incoming buffer will be able to store a ``Status'' bit, the label of the ``Last-Received'' packet, and the ``Round-Received'' index (the round in which this incoming buffer last {\em accepted} a packet, see Definition \ref{receive} below).  The way that this additional information is used will be described in the ``Routing Rules for Receiving Node'' section below.

	\item {\it Outgoing Buffers.}  For each outgoing edge, a node will have a buffer that has the capacity to hold $2n$ packets at any given time.  Like incoming buffers, each outgoing buffer will also be able to store a status bit, the index label of one packet (called the ``Flagged'' packet), and a ``Problem-Round'' index (index of the most recent round in which the status bit switched to 1).
	\end{enumerate}
The receiver will only have incoming buffers (with capacity of one) and a large {\it Storage Buffer} that can hold up to $D$ packets.  Similarly, the sender has only outgoing buffers (with capacity $2n$) and the input stream of messages $\{ m_1, m_2, \dots \}$ which are encoded into the codewords and divided into packets, the latter then distributed to the sender's outgoing buffers.

Also as part of the Setup, all nodes learn the relevant parameters ($P$, $n$, $\lambda$, and $\sigma$).\\ \hspace*{\fill} \\
\noindent {\bf Routing Phase.}
As indicated in Section \ref{model}, we assume a synchronous network, so that there are well-defined rounds in which information is passed between nodes.  Each round consists of two units of time, called {\it Stages}.  The formal treatment of the Routing Phase can be found in the pseudo-code of Section \ref{pseudoCode}.  Informally, Figure \ref{fig418} below considers a directed edge $E(A,B)$ from $A$ (including $A=S$) to $B$ (including $B=R$), and describes what communication each node sends in each stage.
\begin{figure}[ht]
\begin{center}
$\begin{array}{c|lcc} \mbox{\underline{\Large Stage}} & \mbox{\hspace{1cm} \underline{\hspace{1.5cm} \Large A \hspace{1.5cm}}} & & \mbox{\underline{\hspace{1.5cm} \Large B \hspace{1.5cm}}} \\  & & & \\ & H_A:=\mbox{Height of buffer along $E(A,B)$} & & \\ 1 & \mbox{Height of flagged p$.$ (if there is one)} & \longrightarrow & \\ & \mbox{Round prev$.$ packet was sent} & & \\ & & \longleftarrow & \begin{array}{l} H_B:=\mbox{Height of buffer along $E(A,B)$} \\ \mbox{Round prev$.$ packet was received} \end{array}\\ \hline & \mbox{Send packet if:} & & \\2 & \hspace{.5cm} \bullet \hspace{.1cm} H_A > H_B \quad OR & \longrightarrow & \\ & \hspace{.5cm} \bullet \hspace{.1cm} \mbox{$B$ didn't rec$.$ prev$.$ packet sent} & & \end{array}$
\end{center}
\vspace{-.6cm}
\caption{\it Description of communication exchange along directed edge $E(A,B)$ during the Routing Phase of any round.}
\label{fig418}
\end{figure}
\newpage
In addition to this communication, each node must update its internal state based on the communication it receives.  In particular, from the communication $A$ receives from $B$ in Stage 1 of any round, $A$ can determine if $B$ has received the most recent packet $A$ sent.  If so, $A$ will delete this packet and switch the status of the outgoing buffer along this edge to ``normal.''  If not, $A$ will keep the packet as a flagged packet, and switch the status of the outgoing buffer along this edge to ``problem.''  At the other end, if $B$ does not receive $A$'s Stage 1 communication {\it or} $B$ does not receive a packet it was expecting from $A$ in Stage 2, then $B$ will leave a gap in its incoming buffer (termed a ``ghost packet'') and will switch this buffer's status to ``problem.''  On the other hand, if $B$ successfully receives a packet in Stage 2, it will switch the buffer back to ``normal'' status.\\ \hspace*{\fill} \\
%
\noindent {\bf Re-Shuffle Rules.} At the end of each round, nodes will shuffle the packets they are holding according to the following rules:
	\begin{enumerate}\setlength{\itemsep}{2pt} \setlength{\parskip}{0pt} \setlength{\parsep}{0pt} \vspace{-5pt}
	\item Take a packet from the fullest buffer and shuffle it to the emptiest buffer, provided the difference in height is at least two (respectively one) when the packet is moved between two buffers of the same type (respectively when the packet moves from an incoming buffer to an outgoing buffer).  Packets will never be re-shuffled from an outgoing buffer to an incoming buffer.  If two (or more) buffers are tied for having the most packets, then a packet will preferentially be chosen {\it from} incoming buffers over outgoing buffers (ties are broken in a round-robin fashion).  Conversely, if two (or more) buffers are tied for the emptiest buffer, then a packet will preferentially be {\it given to} outgoing buffers over incoming buffers (again, ties are broken in a round-robin fashion).

	\item Repeat the above step until the difference between the fullest buffer and the emptiest buffer does not meet the criterion outlined in Step 1.
	\end{enumerate}
\vspace{-2pt} Recall that when a packet is shuffled locally between two buffers, packets travel in a FILO manner, so that the top-most packet of one buffer is shuffled to the top spot of the next buffer.  When an outgoing buffer has a flagged packet or an incoming buffer has a ghost packet, we use instead the following modifications to the above re-shuffle rules.  Recall that in terms of measuring a buffer's height, flagged packets are counted but ghost packets are not.
\begin{itemize}\setlength{\itemsep}{2pt} \setlength{\parskip}{0pt} \setlength{\parsep}{0pt} \vspace{-5pt}
\item[-] Outgoing buffers do not shuffle {\it flagged} packets.  In particular, if Rule 1 above selects to transfer a packet {\it from} an outgoing buffer, the top-most {\it non-flagged} packet will be shuffled.  This may mean that a gap is created between the flagged packet and the next non-flagged packet.

\item[-] Incoming buffers do not re-shuffle ghost packets.  In particular, ghost packets will remain in the incoming buffer that created them, although we do allow ghost packets to slide {\it down} within its incoming buffer during re-shuffling.  Also, packets shuffled {\it into} an incoming buffer are not allowed to occupy the same slot as a ghost packet (they will take the first non-occupied slot)\footnote{Note that because ghost packets do not count towards height, there appears to be a danger that the re-shuffle rules may dictate a packet gets transferred into an incoming buffer, and this packet either has no place to go (because the ghost packet occupies the top slot) or the packet increases in height (which would violate Claim \ref{packetHeight} below).  However, because only incoming buffers are allowed to re-shuffle packets into other incoming buffers, and the difference in height must be at least two when this happens, neither of these troublesome events can occur.}.
\end{itemize}
\vspace{-2pt} The sender and receiver have special rules for re-shuffling packets.  Namely, during the re-shuffle phase the sender will fill each of his outgoing buffers (in an arbitrary order) with packets corresponding to the current codeword.  Meanwhile, the receiver will empty all of its incoming buffers into its storage buffer.  If at any time $R$ has received enough packets to decode a codeword $b_i$ (Fact 1 says this amount is at most $D-6n^3$), then $R$ outputs message $m_i$ and deletes all packets corresponding to codeword $b_i$ from its storage buffer that he receives in later rounds.
\subsection{Analysis of the Edge-Scheduling Adversarial Protocol} \label{nonAdAnalysis}
\indent \indent We now evaluate our edge-scheduling protocol in terms of our three measurements of performance: {\it correctness}, {\it throughput}, and {\it processor memory}.  The throughput standard expressed in Theorem \ref{nonAdTheorem} below will serve an additional purpose when we move to the node-controlling adversary setting: The sender will know that malicious activity has occurred when the throughput standard of Theorem \ref{nonAdTheorem} is not observed.  Both of the theorems below will be proved rigorously in Sections \ref{proofs} and \ref{bareBones}, after presenting the pseudo-code in Section \ref{pseudoCode}.
\begin{thm} \label{nonAdTheorem} Each message $m_i$ takes at most $3D$ rounds to pass from the sender to the receiver.  In particular, after $O(xD)$ rounds, $R$ will have received at least $O(x)$ messages.  Since each message has size $M$=$\frac{6 \sigma}{\lambda}Pn^3$=$\space O(n^3)$ and $D$=$\frac{6n^3}{\lambda}$=$\space O(n^3)$, after $O(x)$ rounds, $R$ has received $O(x)$ bits of information, and thus our edge-scheduling adversarial protocol enjoys a linear throughput rate.
\end{thm}
\vspace{-3pt} The above theorem implicitly states that our edge-scheduling protocol is {\it correct}.  For completeness, we also state the memory requirements of our edge-scheduling protocol, which is bottle-necked by the $O(n^2)$ packets that each internal node has the capacity to store in its buffers.\vspace{-3pt}
\begin{thm} \label{nonAdMem} The edge-scheduling protocol described in Section \ref{nonAdDesc} (and formally in the pseudo-code of Section \ref{pseudoCode}) requires at most $O(n^2 \log n)$ bits of memory of the internal processors.
\end{thm}
\newpage
\section{Pseudo-Code for the Edge-Scheduling Adversarial Protocol} \label{pseudoCode}
\begin{figure}[h!t]
\begin{center}
\framebox{\footnotesize
\begin{minipage}[b]{6.0in}
\begin{footnotesize}
\begin{tabbing}
aaaa\=aaaa\=aaaa\=aaaa\=aaaa\=aaaaaaaa\=aaaaaaaaaaaa\=aaaaaa\=aaaaaa\=aaaaaaaaaaaaaaaaaa\=aaaaaa\=aaaaaaaa \kill
{\bf Setup} \\
\>{\scriptsize \bf DEFINITION OF VARIABLES}:\\[-1pt]
01\>\>$n :=$ Number of nodes in $G$;\\
02\>\>$D :=\frac{6n^3}{\lambda}$;\\
03\>\>$\mathtt{T} :=$ Transmission index;\\
04\>\>$\mathtt{t} :=$ Stage/Round index;\\
05\>\>$P :=$ Capacity of edge (in bits);\\
06\>\>\For $N \in G$\\
07\>\>\>\For {\it outgoing} {\bf edge} $E(N,B) \in G, B \neq S$ and $N \neq R$\\
08\>\>\>\>$\mathsf{OUT} \in [2n] \times \{0,1\}^{P}$; \>\>\>\>{\scriptsize \#\# Outgoing Buffer able to hold $2n$ packets}\\
09\>\>\>\>$\tilde{p} \in \{0,1\}^P \cup \bot$; \>\>\>\>{\scriptsize \#\# Copy of packet to be sent}\\
10\>\>\>\>$sb \in \{0,1\}$; \>\>\>\>{\scriptsize \#\# Status bit}\\
11\>\>\>\>$d \in \{0,1\}$; \>\>\>\>{\scriptsize \#\# Bit indicating if a packet was sent in the previous round}\\
12\>\>\>\>$FR \in [0..6D] \cup \bot$; \>\>\>\>{\scriptsize \#\# Flagged Round (index of round $N$ first tried to send $\tilde{p}$ to $B$)}\\
13\>\>\>\>$H \in [0..2n]$; \>\>\>\>{\scriptsize \#\# Height of $\mathsf{OUT}$. Also denoted $H_{OUT}$ when there's ambiguity}\\
14\>\>\>\>$H_{FP} \in [1..2n] \cup \bot$; \>\>\>\>{\scriptsize \#\# Height of Flagged Packet}\\
15\>\>\>\>$RR \in [-1..6D] \cup \bot$; \>\>\>\>{\scriptsize \#\# Round Received index (from adjacent incoming buffer)}\\
16\>\>\>\>$H_{IN}\in [0..2n] \cup \bot$; \>\>\>\>{\scriptsize \#\# Height of incoming buffer of $B$}\\
17\>\>\>\For {\it incoming} {\bf edge} $E(A,N) \in G, A \neq R$ and $N \neq S$\\
18\>\>\>\>$\mathsf{IN} \in [2n] \times \{0,1\}^{P}$; \>\>\>\>{\scriptsize \#\# Incoming Buffer able to hold $2n$ packets}\\
19\>\>\>\>$p \in \{0,1\}^P \cup \bot$; \>\>\>\>{\scriptsize \#\# Packet just received}\\
20\>\>\>\>$sb \in \{0,1\}$; \>\>\>\>{\scriptsize \#\# Status bit}\\
21\>\>\>\>$RR \in [-1..6D]$; \>\>\>\>{\scriptsize \#\# Round Received (index of round $N$ last rec'd a p$.$ from $A$)}\\
22\>\>\>\>$H \in [0..2n]$; \>\>\>\>{\scriptsize \#\# Height of $\mathsf{IN}$. Also denoted $H_{IN}$ when there's ambiguity}\\
23\>\>\>\>$H_{GP} \in [1..2n] \cup \bot$; \>\>\>\>{\scriptsize \#\# Height of Ghost Packet}\\
24\>\>\>\>$H_{OUT}\in [0..2n] \cup \bot$; \>\>\>\>{\scriptsize \#\# Height of outgoing buffer, or height of Flagged Packet of $A$}\\
25\>\>\>\>$sb_{OUT} \in \{0,1\}$; \>\>\>\>{\scriptsize \#\# Status Bit of outgoing buffer of $A$}\\
26\>\>\>\>$FR \in [0..6D] \cup \bot$; \>\>\>\>{\scriptsize \#\# Flagged Round index (from adjacent outgoing buffer)}\\
\>{\scriptsize \bf INITIALIZATION OF VARIABLES}:\\[-1pt]
27\>\>\For $N \in G$\\
28\>\>\>\For {\it incoming} {\bf edge} $E(A,N) \in G, A \neq R$ and $N \neq S$\\
29\>\>\>\>{\it Initialize $\mathsf{IN}$}; \>\>\>\>{\scriptsize \#\# Set each entry in $\mathsf{IN}$ to $\bot$}\\
30\>\>\>\>$p, FR, H_{GP} = \bot$;\\
31\>\>\>\>$sb, sb_{OUT}, H, H_{OUT} = 0$; $RR = -1$;\\
32\>\>\>\For {\it outgoing} {\bf edge} $E(N,B) \in G, B \neq S$ and $N \neq R$\\
33\>\>\>\>{\it Initialize $\mathsf{OUT}$}; \>\>\>\>{\scriptsize \#\# Set each entry in $\mathsf{OUT}$ to $\bot$}\\
34\>\>\>\>$\tilde{p}, H_{FP}, RR, FR = \bot$;\\
35\>\>\>\>$sb, d, H, H_{IN} = 0$;\\
{\bf End Setup}
\end{tabbing}
\end{footnotesize}
\end{minipage} }
\end{center}
\vspace{-.6cm}
\caption{Pseudo-Code for Internal Nodes' Setup for the Edge-Scheduling Adversarial Model}
\label{setupCode}
\end{figure}
\newpage
\begin{figure}[h!t]
\begin{center}
\framebox{\footnotesize
\begin{minipage}[b]{6.0in}
\begin{footnotesize}
\begin{tabbing}
aa\=aaaa\=aaaa\=aaaa\=aaaa\=aaaaaaaaaaaaaaaaaaa\=a \kill
{\bf Sender and Receiver's Additional Setup} \\
\>{\bf DEFINITION OF ADDITIONAL VARIABLES FOR SENDER}:\\
36\>\>$\mathcal{M} := \{m_1, m_2, \dots \} = $ Input Stream of Messages; \\
37\>\>$\kappa \in [0..D]$ = Number of packets corresponding to current codeword the sender has {\it knowingly} inserted;\\
\\
\>{\bf INITIALIZATION OF SENDER'S VARIABLES}:\\
38\>\>{\bf \em Distribute Packets};\>\>\>\>{\scriptsize \#\# See Figure \ref{routingRules2}}\\
39\>\>$\kappa = 0$;\\
\\
\\
\>{\bf DEFINITION OF ADDITIONAL VARIABLES FOR RECEIVER}:\\
40\>\>$I_R \in [D] \times (\{0,1\}^P \cup \bot)$ = Storage Buffer to hold packets corresponding to current codeword;\\
41\>\>$\kappa \in [0..D]$ := Number of packets received corresponding to current codeword;\\
\\
\>{\bf INITIALIZATION OF RECEIVER'S VARIABLES}:\\
42\>\>$\kappa = 0$;\\
43\>\>{\it Initialize $I_R$};\>\>\>\>{\scriptsize \#\# Sets each element of $I_R$ to $\bot$}\\
{\bf End Sender and Receiver's Additional Setup}
\end{tabbing}
\end{footnotesize}
\end{minipage} }
\end{center}
\vspace{-.6cm}
\caption{Additional Code for Sender and Receiver Setup}
\label{setupCode2}
\end{figure}
\vspace{1cm}
%
\begin{figure}[h!t]
\begin{center}
\leavevmode
\framebox{\footnotesize
\begin{minipage}[b]{6.0in}
\begin{footnotesize}
\begin{tabbing}
aaa\=aa\=aa\=aaa\=aaa\=aaa\=aaa\=aaa\=aaa\=aaaaaaaaaaaaaaaaaaaaa\=aaaaaaaaa\=aaaa \kill
{\bf Transmission $\mathtt{T}$}\\
01\>\For $N \in G$\\
02\>\>\For $\mathtt{t} < 2*(3D)$\>\>\>\>\>\>\>\>\>{\#\# The factor of 2 is for the 2 stages per round}\\
03\>\>\>\If $\mathtt{t}$ (mod 2) = 0  \Then \>\>\>\>\>\>\>\>{\#\# STAGE 1}\\
04\>\>\>\>\For {\it outgoing} {\bf edge} $E(N,B) \in G, N \neq R, B \neq S$\\
05\>\>\>\>\>\If $H_{FP}=\bot$: \Send $(H, \bot, \bot)$; \quad \Else \Send $(H-1, H_{FP}, FR)$;\\
06\>\>\>\>\>\Receive $(H_{IN},RR)$;\\
07\>\>\>\>\>{\bf \em Reset Outgoing Variables};\\
08\>\>\>\>\For {\it incoming} {\bf edge} $E(A,N) \in G, N \neq S, A \neq R$\\
09\>\>\>\>\>\Send $(H, RR)$;\\
10\>\>\>\>\>$sb_{OUT}=0$; $FR=\bot$;\\
11\>\>\>\>\>\Receive $(H, \bot, \bot)$ \boldmath$\mathsf{or}$\unboldmath $\thickspace (H, H_{FP}, FR)$;\>\>\>\>\>\>{\#\# If $H = \bot$ or $FR > RR$, set $sb_{OUT}$=$1$; and}\\
\>\>\>\>\>\>\>\>\>\>\>{\#\# $H_{OUT}$=$H_{FP}$; O.W.\ set $H_{OUT}$=$H$; $sb_{OUT}$=$0$;}\\
12\>\>\>\ElseIf $\mathtt{t}$ (mod 2) = 1 \Then \>\>\>\>\>\>\>\>{\#\# STAGE 2}\\
13\>\>\>\>\For {\it outgoing} {\bf edge} $E(N,B) \in G, N \neq R, B \neq S$\\
14\>\>\>\>\>\If $H_{IN}\neq \bot$ \Then \>\>\>\>\>\>{\#\# Received $B$'s info.}\\
15\>\>\>\>\>\>{\bf \em Create Flagged Packet};\\
16\>\>\>\>\>\>\If ($sb$=$1$ \boldmath$\mathsf{or}$\unboldmath $\thickspace (sb$=$0$ \boldmath$\mathsf{and}$\unboldmath $\thickspace H_{OUT} > H_{IN})$) \Then \\
17\>\>\>\>\>\>\>{\bf \em Send Packet};\\
18\>\>\>\>\For {\it incoming} {\bf edge} $E(A,N) \in G, N \neq S, A \neq R$\\
19\>\>\>\>\>{\bf \em Receive Packet};\\
20\>\>\>\>\If $N \notin \{S,R\}$ \Then {\bf \em Re-Shuffle};\\
21\>\>\>\>\ElseIf $N = R$ \Then {\bf \em Receiver Re-Shuffle};\\
22\>\>\>\>\ElseIf $N = S$ \Then {\bf \em Sender Re-Shuffle};\\
\\
23\>\>\>\>\If $\mathtt{t}=2(3D)-1$ \Then {\bf \em End of Transmission Adjustments};\\
{\bf End Transmission $\mathtt{T}$}
\end{tabbing}
\end{footnotesize}
\end{minipage} }
\end{center}
\vspace{-.6cm}
\caption{Routing Rules I for Edge-Scheduling Adversarial Model}
\label{routingRules}
\end{figure}
\newpage $\thickspace$
\begin{figure}[h!t]
\begin{center}
\leavevmode
\framebox{\footnotesize
\begin{minipage}[b]{6.0in}
\begin{footnotesize}
\begin{tabbing}
aaa\=aa\=aaa\=aaa\=aaa\=aaa\=aaa\=aaaaaaaaaaaaaaaaaaaaaaaaaaaaaaaaaaa\=aaaa\=aaaa \kill
24\>{\bf \em Reset Outgoing Variables}\\
25\>\>\If $d = 1$:\>\>\>\>\>\>{\scriptsize \#\# $N$ sent a packet previous round}\\
26\>\>\>$d=0$;\\
27\>\>\>\If $RR = \bot$ \boldmath$\mathsf{or}$\unboldmath $\thickspace \bot \neq FR > RR$\>\>\>\>\>{\scriptsize \#\# Didn't receive conf.\ of packet receipt}\\
28\>\>\>\>$sb=1$;\\
29\>\>\If $RR \neq \bot$:\\
30\>\>\>\If $\bot \neq FR \leq RR$:\>\>\>\>\>{\scriptsize \#\# $B$ rec'd most recently sent packet}\\
31\>\>\>\>\If $N=S$ \Then $\kappa = \kappa +1$;\\
32\>\>\>\>$\mathsf{OUT}[H_{FP}] = \bot$; {\it Fill Gap};\>\>\>\>{\scriptsize \#\# Remove $\tilde{p}$ from $\mathsf{OUT}$, shifting}\\
\>\>\>\>\>\>\>\>{\scriptsize \#\# down packets on top of $\tilde{p}$ if necessary}\\
33\>\>\>\>$FR,\tilde{p},H_{FP}=\bot$; $sb=0$; $H=H-1$;\\
34\>\>\If $\bot \neq RR < FR$ \boldmath$\mathsf{and}$\unboldmath $\thickspace \bot \neq H_{FP} < H$:\>\>\>\>\>\>{\scriptsize \#\# $B$ did {\em not} receive most recently sent packet}\\
35\>\>\>{\it Elevate Flagged Packet};\>\>\>\>\>{\scriptsize \#\# Swap $\mathsf{OUT}[H]$ and $\mathsf{OUT}[H_{FP}]$; Set $H_{FP}=H$;}\\
\\
36\>{\bf \em Create Flagged Packet}\\
37\>\>\If $sb=0$ \boldmath$\mathsf{and}$\unboldmath $\thickspace H > H_{IN}$:\>\>\>\>\>\>{\scriptsize \#\# Normal Status, will send top packet}\\
38\>\>\>$\tilde{p} = \mathsf{OUT}[H]$; $H_{FP} = H$; $FR = \mathtt{t}$;\\
\\
39\>{\bf \em Send Packet}\\
40\>\>$d=1$;\\
41\>\>\Send $(\tilde{p}, FR)$;\\
\\
42\>{\bf \em Receive Packet}\\
43\>\>\Receive $(p, FR)$;\\
44\>\>\If $H_{OUT} = \bot$:\>\>\>\>\>\>{\scriptsize \#\# Didn't Rec. $A$'s height info.}\\
45\>\>\>$sb = 1$;\\
46\>\>\>\If $H_{GP} > H$ \boldmath$\mathsf{or}$\unboldmath $\thickspace (H_{GP} = \bot$ \boldmath$\mathsf{and}$\unboldmath $\thickspace H <2n)$: $\quad H_{GP} = H+1$;\\
47\>\>\ElseIf $sb_{OUT}=1$ \boldmath$\mathsf{or}$\unboldmath $\thickspace H_{OUT} > H$:\>\>\>\>\>\>{\scriptsize \#\# A packet should've been sent}\\
48\>\>\>\If $p =\bot$:\>\>\>\>\>{\scriptsize \#\# Packet wasn't rec'd}\\
49\>\>\>\>$sb = 1$;\\
50\>\>\>\>\If $H_{GP} > H$ \boldmath$\mathsf{or}$\unboldmath $\thickspace (H_{GP} = \bot$ \boldmath$\mathsf{and}$\unboldmath $\thickspace H <2n)$: $\quad H_{GP} = H+1$;\\
51\>\>\>\ElseIf $RR<FR$:\>\>\>\>\>{\scriptsize \#\# Packet was rec'd and should keep it}\\
52\>\>\>\>\If $H_{GP} = \bot$: $\quad H_{GP} = H+1$;\>\>\>\>{\scriptsize \#\# If no slot is saved for $p$, put it on top}\\
53\>\>\>\>$sb=0$; $\mathsf{IN}[H_{GP}]=p$; $H = H+1$; $H_{GP} = \bot$; $RR = \mathtt{t}$;\\
54\>\>\>\Else\>\>\>\>\>{\scriptsize \#\# Packet was rec'd, but already had it}\\
55\>\>\>\>$sb=0$; {\it Fill Gap}; $H_{GP}=\bot$;\>\>\>\>{\scriptsize \#\# See comment about {\it Fill Gap} on line 57 below}\\
56\>\>\Else\>\>\>\>\>\>{\scriptsize \#\# A packet should NOT have been sent}\\
57\>\>\>$sb=0$; {\it Fill Gap}; $H_{GP} = \bot$;\>\>\>\>\>{\scriptsize \#\# If packets occupied slots {\it above} the}\\
\>\>\>\>\>\>\>\>{\scriptsize \#\# Ghost Packet, then {\it Fill Gap} will Slide}\\
\>\>\>\>\>\>\>\>{\scriptsize \#\# those packets down one slot}\\
58\>{\bf \em End of Transmission Adjustments}\\
59\>\>\For {\it outgoing} edge $E(N,B) \in G$, $N \neq R$, $B \neq S$:\\
60\>\>\>\If $H_{FP} \neq \bot$: \\
61\>\>\>\>$OUT[H_{FP}] = \bot$; {\it Fill Gap};\>\>\>\>{\scriptsize \#\# Remove any flagged packet $\tilde{p}$ from $\mathsf{OUT}$, shifting}\\
\>\>\>\>\>\>\>\>{\scriptsize \#\#  down packets on top of $\tilde{p}$ if necessary}\\
62\>\>\>\>$d, sb = 0;$ $FR,H_{FP}, \tilde{p}=\bot$; $H=H-1$;\\
63\>\>\For {\it incoming} edge $E(A,N) \in G$, $N \neq S$, $A \neq R$:\\
64\>\>\>$H_{GP} = \bot$; $sb =0$; $RR = -1$; {\it Fill Gap};\\
65\>\>\If $N=S$ \Then {\bf \em Distribute Packets};\\
\end{tabbing}
\end{footnotesize}
\end{minipage} }
\end{center}
\vspace{-.6cm}
\caption{Routing Rules for Edge-Scheduling Adversarial Model (continued)}
\label{routingRules2}
\end{figure}
%
\newpage
\begin{figure}[h!t]
\begin{center}
\leavevmode
\framebox{\footnotesize
\begin{minipage}[b]{6.0in}
\begin{footnotesize}
\begin{tabbing}
aaa\=aa\=aaa\=aaa\=aaa\=aaa\=aaa\=aaaaaaaaaaaaaaaaa\=aaaaaaaaaaaaaaa\=aaaa \kill
71\>{\bf \em Re-Shuffle}\\
72\>\>$(M,B_F)=$ {\it Find Maximum Buffer}\>\>\>\>\>\>\> {\scriptsize \#\# Node $N$ finds its fullest buffer $B_F$ with height $M$,}\\
\>\>\>\>\>\>\>\>\> {\scriptsize \#\# breaking ties by 1) selecting incoming buffers over}\\
\>\>\>\>\>\>\>\>\> {\scriptsize \#\# outgoing buffers, then 2) Round-Robin}\\
73\>\>$(m,B_T)=$ {\it Find Minimum Buffer}\>\>\>\>\>\>\> {\scriptsize \#\# Node $N$ finds its emptiest buffer $B_T$ with height $m$,}\\
\>\>\>\>\>\>\>\>\> {\scriptsize \#\# breaking ties by 1) selecting outgoing buffers over}\\
\>\>\>\>\>\>\>\>\> {\scriptsize \#\# incoming buffers, then 2) Round-Robin}\\
74\>\>\If {\it Packet Should Be Re-Shuffled:}\>\>\>\>\>\>\> {\scriptsize \#\# A packet should be re-shuffled if $M-m >1$ \boldmath$\mathsf{or}$\unboldmath}\\
\>\>\>\>\>\>\>\>\> {\scriptsize \#\# $M-m=1$ \boldmath$\mathsf{and}$\unboldmath $\thickspace \left\{ \begin{array}{l} B_F \mbox{ is an Inc. Buffer}\\ B_T \mbox{ is an Out. Buffer} \end{array} \right\}$}\\
75\>\>\>{\bf \em Adjust Heights}\>\>\>\>\>\>{\scriptsize \#\# Adjust $M$, $m$ to account for Ghost, Flagged packets.}\\
{\bf 76}\>\>\>$SIG_{N,N} = SIG_{N,N} + (M-m-1)$;\>\>\>\>\>\>{\scriptsize \#\# Only used for (node-contr$.$ $+$ edge-sched$.$) protocol}\\
77\>\>\>{\bf \em Shuffle Packet}\\
78\>\>\>{\bf \em Re-Shuffle}\\
\\
79\>{\bf \em Adjust Heights}\\
80\>\>\If $B_F$ is an Out. Buffer \boldmath$\mathsf{and}$\unboldmath $\thickspace H_{FP} \geq H_{OUT}$:\>\>\>\>\>\>\>{\scriptsize \#\# $H_{FP}$ and $H_{OUT}$ refer to $B_F$'s info.  If true,}\\
81\>\>\>$M = M-1$;\>\>\>\>\>\>{\scriptsize \#\# then a Flagged packet is top-most non-null packet}\\
82\>\>\If $B_F$ is an Inc. Buffer \boldmath$\mathsf{and}$\unboldmath $\thickspace \mathsf{IN}[H_{IN}+1] \neq \bot$:\>\>\>\>\>\>\>{\scriptsize \#\# $\mathsf{IN}$ and $H_{IN}$ refer to $B_F$'s info.  If true,}\\
83\>\>\>$M = M+1$;\>\>\>\>\>\>{\scriptsize \#\# then there is a Ghost Packet creating a gap}\\
84\>\>\If $B_T$ is an Out. Buffer \boldmath$\mathsf{and}$\unboldmath $\thickspace \mathsf{OUT}[H_{OUT}] = \bot$:\>\>\>\>\>\>\>{\scriptsize \#\# $\mathsf{OUT}$ and $H_{OUT}$ refer to $B_T$'s info.  If true,}\\
85\>\>\>$m = m-1$;\>\>\>\>\>\>{\scriptsize \#\# then there is a Flagged packet creating a gap}\\
86\>\>\If $B_T$ is an Inc. Buffer \boldmath$\mathsf{and}$\unboldmath $\thickspace H_{GP} \neq \bot$:\>\>\>\>\>\>\>{\scriptsize \#\# $H_{GP}$ and $H_{IN}$ refer to $B_T$'s info.  If true,}\\
87\>\>\>$m = m+1$;\>\>\>\>\>\>{\scriptsize \#\# then there is a Ghost Packet creating a gap}\\
\\
88\>{\bf \em Shuffle Packet}\\
89\>\>$B_T[m+1] = B_F[M]$;\\
90\>\>$B_F[M]=\bot$;\\
91\>\>$H_{B_T} = H_{B_T} +1$;\>\>\>\>\>\>\>{\scriptsize \#\# $H_{B_T}$ is the height of $B_T$}\\
92\>\>$H_{B_F} = H_{B_F} -1$;\>\>\>\>\>\>\>{\scriptsize \#\# $H_{B_F}$ is the height of $B_F$}\\
93\>\>\If $B_F$ is an Inc. Buffer \boldmath$\mathsf{and}$\unboldmath $\thickspace \bot \neq H_{GP} > H_{IN}$, \Then \>\>\>\>\>\>\>{\scriptsize \#\# $H_{GP}$ and $H_{IN}$ refer to $B_F$'s info.  Since $B_F$ lost a}\\
94\>\>\>$H_{GP} = H_{IN}+1$;\>\>\>\>\>\>{\scriptsize \#\# packet, slide Ghost Packet down into top slot}\\
\\
95\>{\bf \em Sender Re-Shuffle}\\
96\>\>{\it Fill Packets}; \>\>\>\>\>\>\>{\scriptsize \#\# Fills each outgoing buffer with codeword packets not}\\
\>\>\>\>\>\>\>\>\>{\scriptsize \#\# yet distributed, adjusting each $H_{OUT}$ appropriately}\\
97\>{\bf \em Receiver Re-Shuffle}\\
98\>\>\For {\it incoming} edge $E(A,R) \in G$:\>\>\>\>\>\>\>{\scriptsize \#\# Reset $R$'s Inc. Buffer to be open}\\
99\>\>\>\If $H_{IN} > 0$:\>\>\>\>\>\>{\scriptsize \#\# $R$ rec'd a packet along this edge this round}\\
100\>\>\>\>\If $IN[1]$ is a packet for current codeword:\>\>\>\>\>{\scriptsize \#\# Also, see comments on 104 below}\\
101\>\>\>\>\>$I_R[\kappa] = IN[1]$; $\kappa = \kappa +1$;\\
102\>\>\>$H_{IN}=0$; $IN[1] = \bot$; $H_{GP} = \bot$;\\
103\>\>\If $\kappa \geq D-3n^3$ \Then \>\>\>\>\>\>\> {\scriptsize \#\# $R$ can decode by Fact 1}\\
104\>\>\>Decode and output message; \>\>\>\>\>\>{\scriptsize \#\# Also, only keep codeword packets corresponding}\\
\>\>\>\>\>\>\>\>\>{\scriptsize \#\#  to {\it next} message in future rounds}
\end{tabbing}
\end{footnotesize}
\end{minipage} }
\end{center}
\vspace{-.6cm}
\caption{Re-Shuffle Rules for both Edge-Scheduling {\it and} (Node-Controlling $+$ Edge-Scheduling) Protocols}
\label{reShuffleRules}
\end{figure}
\newpage
\section{Edge-Scheduling Adversary Model:\\ Proofs of Lemmas and Theorems} \label{proofs}
\indent \indent Before proving the main two theorems for the edge-scheduling adversarial protocol (Theorems \ref{nonAdTheorem} and \ref{nonAdMem}), we will first state and prove a sequence of claims that follow immediately from the Routing and Re-Shuffle Rules of Section \ref{nonAdProtocol}.  We have included pseudo-code in Section \ref{pseudoCode}, and when appropriate the proofs will refer to specific line numbers in the pseudo-code.  In particular, we will reference a line in pseudo-code by writing ({\bf X.YY}), where X refers to the Figure and YY to the line number.  We pushed the claims and proofs that rely heavily on the pseudo-code (but are unlikely to add insight) to Section \ref{bareBones} so as not to distract the reader with the gory details of these proofs.  Logically, these claims need to be proven first as the claims and proofs below will rely on them (even though Section \ref{bareBones} appears below, the proofs there do not rely on proofs here, so there is no danger of circularity).

We state and prove here the claims that will lead to Theorem \ref{nonAdTheorem} and Theorem \ref{nonAdMem}.
\begin{claim} \label{capacity0} The capacity of the internal buffers of the network (not counting $S$ or $R$'s buffers) is $4n(n-2)^2$.\end{claim}
\begin{proof} Each node has $(n-2)$ outgoing buffers (one to each node except itself and $S$, {\bf \ref{setupCode}.07}) and $(n-2)$ incoming buffers (one from each node except itself and $R$, {\bf \ref{setupCode}.17}), and thus a total of $2(n-2)$ buffers.  Each of these buffers has capacity $2n$ (Lemma \ref{pseudo}, parts 5, 6, and 9), and there are $n-2$ internal nodes, so the internal buffer capacity of the network is $4n(n-2)^2$.\end{proof}
\begin{claim} \label{capacity1} The maximum amount of potential\footnote{See Definition \ref{potDef}.} in the internal buffers of the network at any time is $2n(2n+1)(n-2)^2$.\end{claim}
\begin{proof} A buffer contributes the most to network potential when it is full, in which case it contributes $\sum_{i=1}^{2n} i = n(2n+1)$.  Since there are $2(n-2)$ buffers per internal node, and $n-2$ internal nodes, the maximum amount of potential in the internal buffers is as claimed.\end{proof}
\noindent We define the {\it height} of a packet in an incoming/outgoing buffer to be the spot it occupies in that buffer.
\begin{claim} \label{balancing} After re-shuffling, (and hence at the very end/beginning of each round), all of the buffers of each node are {\bf balanced}.  In particular, there are no incoming buffers that have height strictly bigger than any outgoing buffers, and the difference in height between any two buffers is at most one. \end{claim}
\begin{proof}[Proof of Claim \ref{balancing}.]  We prove this using induction (on the round index), noting that all buffers are balanced at the outset of the protocol (lines ({\bf \ref{setupCode}.29}) and ({\bf \ref{setupCode}.33})).  Consider any node $N$ in the network, and assume that its buffers are all balanced at the end of some round $\mathtt{t}$.  We need to show the buffers of $N$ will remain balanced at the end of the next round $\mathtt{t}+1$.  Let $B_1$ and $B_2$ denote any two buffers of $N$, and let $h_1$ be the variable denoting the height of $B_1$ and $h_2$ the height of $B_2$.  Suppose for the sake of contradiction that $h_1 \geq h_2+2$ at the end of round $\mathtt{t}+1$ (after re-shuffling).  Let $H$ denote the height of the maximum buffer in $N$ at the end of $\mathtt{t}+1$, so $H \geq h_1 \geq h_2+2$.  Also let $h$ denote the height of the minimum buffer in $N$ at the end of $\mathtt{t}+1$, so $h \leq h_2 \leq H-2$.  But then Re-Shuffle Rules dictate that $N$ should've kept re-shuffling ({\bf \ref{reShuffleRules}.72-74}), a contradiction.

Similarly, assume for contradiction that there exists an incoming buffer whose height $h_2$ is bigger than that of some outgoing buffer that has height $h_1$.  Let $H$ and $h$ be as defined above, so we have that $h \leq h_1 < h_2 \leq H$.  In the case that $h_2 = H$, Re-Shuffle Rules ({\bf \ref{reShuffleRules}.72}) guarantee that an {\it incoming} buffer will be selected to take a packet from.  Also, if $h = h_1$, then Re-Shuffle Rules ({\bf \ref{reShuffleRules}.73}) guarantee that an {\it outgoing} buffer will be chosen to give a packet to.  Therefore, in this case a packet should have been re-shuffled ({\bf \ref{reShuffleRules}.74}), and hence we have contradicted the fact that we are at the end of the Re-Shuffle phase of round $\mathtt{t}$.  On the other hand, if $h \neq h_1$ or $H \neq h_2$, then $H-h \geq 2$, and again Re-Shuffling should not have terminated ({\bf \ref{reShuffleRules}.74}).
\end{proof}
\noindent The following observation is the formalization of the concept of packets ``flowing downhill'' that was introduced in Section \ref{nonAdProtocol}.
\begin{claim} \label{packetHeight} Every packet is inserted into one of the sender's outgoing buffers at some initial height.  When (a copy of) the packet goes between any two buffers $B_1 \neq B_2$ (either across an edge or locally during re-shuffling), its height in $B_2$ is less than or equal to the height it had in $B_1$.  If $B_1 = B_2$, the statement remains true EXCEPT for on line ({\bf \ref{routingRules2}.35}).\end{claim}
\begin{proof} See Section \ref{bareBones}, where we restate and prove this in Lemma \ref{packetHeight2}.\end{proof}
\begin{defn} \label{receive} We will say that a packet is {\bf accepted} by a buffer $B$ in round $\mathtt{t}$ if $B$ receives and {\it stores} that packet in round $\mathtt{t}$, either due to a packet transfer or re-shuffling (as on ({\bf \ref{routingRules2}.53}) or ({\bf \ref{reShuffleRules}.89})).\end{defn}
\begin{defn} \label{insert} We say that the sender {\bf inserts} a packet into the network in round $\mathtt{t}$ if any internal node (or $R$) accepts the packet (as in Definition \ref{receive}) in round $\mathtt{t}$.  Note that this definition does not require that $S$ receives the verification of receipt (i.e.\ that $S$ receives the communication on ({\bf \ref{routingRules}.06}) indicating $RR \geq FR$), so $S$ may not be aware that a packet was inserted.\end{defn}
\noindent Notice that in terms of transferring packets, the above definition distinguishes between the case that a packet is accepted by $B$ in round $\mathtt{t}$ (as defined above) and the case that a packet arrives at $B$ in round $\mathtt{t}$ but is deleted by $B$ (by failing the conditional statement on line ({\bf \ref{routingRules2}.51})).  As emphasized in the Introduction, {\it correctness} and {\it throughput rate} are two of the three commodities with which we will evaluate a given routing protocol.  In our protocol, we will need to show that packets are not lost en route from $S$ to $R$ to ensure {\it correctness}, and meanwhile we will want to show that packets are not (overly) duplicated (since transferred packets are actually copies of the original, some packet duplication is necessary) to allow a ``fast'' {\it throughput rate}.  The following two claims guarantee that packet duplication won't become problematic while simultaneously guaranteeing that packets are never deleted completely (except by $R$).
\begin{claim} \label{packetProliferation2} Before the end of transmission $\mathtt{T}$, any packet that was inserted into the network during transmission $\mathtt{T}$ is either in some buffer (perhaps as a flagged packet) or has been received by $R$.\end{claim}
\begin{proof}See Section \ref{bareBones}, where we restate and prove this in Lemma \ref{packetProliferation22}. \end{proof}
\begin{claim} \label{packetProliferation3} Not counting flagged packets, there is at most one copy of any packet in the network at any time (not including packets in the sender or receiver's buffers).  Looking at all copies (flagged and un-flagged) of any given packet present in the network at any time, at most one copy will {\bf ever} be accepted (as in Definition \ref{receive}) by another node.\end{claim}
\begin{proof}See Section \ref{bareBones}, where we restate and prove this in Lemma \ref{packetProliferation32}. \end{proof}
The following claim won't be needed until we introduce the protocol for the (Node-Controlling $+$ Edge-Scheduling) adversarial model, but follows from the Routing Rules outlined in Section \ref{nonAdProtocol}.
\begin{claim} \label{packetProliferation} At any time, an outgoing buffer has at most one flagged packet.\end{claim}
\begin{proof}See Section \ref{bareBones}, where we restate and prove this in Corollary \ref{packetProliferationA2}. \end{proof}
The following definition formalizes the notion of ``potential,'' and will be necessary to prove throughput performance bounds.
\begin{defn} \label{potDef} For any buffer\footnote{Packets in one of the sender or receiver's buffers do not count towards potential.} $B \neq S,R$ that has height $h$ at time $t$, define the {\em potential} of $B$ at time $t$, denoted by $\Phi^B_t$, to be:
	\begin{equation*}
	\Phi^B_t := \sum_{i=1}^h i = \frac{h(h+1)}{2}.
	\end{equation*}
	For any internal node $N \in \mathcal{P} \setminus \{R,S\}$, define the node's potential $\Phi^N_t$ to be the sum of its buffer's potentials:
		\begin{equation*}
		\Phi^N_t := \hspace{-.4cm} \sum_{\mbox{\scriptsize Buffers $B$ of $N$}} \hspace{-.4cm}\Phi^B_t
		\end{equation*}
	Define the {\em network potential} $\Phi_t$ at time $t$ to be the sum of all the internal buffers' potentials:
	\begin{equation*}
	\begin{array}{r} {\displaystyle \Phi_t := \thickspace \sum \thickspace \Phi^N_t} \\ {\scriptstyle N \in \mathcal{P} \setminus \{R,S\} } \thickspace \end{array}
	\end{equation*}
\end{defn}
It will be useful to break an internal node's potential into two parts.  The first part, which we will term \textsf{packet duplication potential}, will be the sum of the heights of the flagged packets in the node's outgoing buffers {\it that have already been accepted} by the neighboring node (as in Definition \ref{receive}).  Recall that a flagged packet is a packet that was sent along an outgoing edge, but the sending node is maintaining a copy of the packet until it gets confirmation of receipt.  Therefore, the contribution of packet duplication potential to overall network potential is the extraneous potential; it represents the over-counting of duplicated packets.  We emphasize that not all flagged packets count towards packet-duplication potential, since packets are flagged as soon as the sending node determines a packet should be sent (see line {\bf \ref{routingRules2}.38}), but the flagged packet's height does not count towards packet duplication potential until the receiving node has accepted the packet as on line ({\bf \ref{routingRules2}.53}) (which may happen in a later round or not at all).  The other part of network potential will be termed \textsf{non-duplicated potential}, and is the sum of the heights of all non-flagged packets together with flagged packets that have not yet been accepted.  Note that the separation of potential into these two parts is purely for analysis of our protocol, indeed the nodes are not able to determine if a given flagged packet contributes to {\it packet duplication} or {\it non-duplicated} potential.  For convenience, we will often refer to (network) non-duplicated potential simply as (network) potential (the meaning should be clear from context).

Notice that when a node accepts a packet, its own (non-duplicated) potential instantaneously increases by the height that this packet assumes in the corresponding incoming buffer.  Meanwhile, the sending node's {\it non-duplicated potential} drops by the height that the packet occupied in its outgoing buffer, and there is a simultaneous and equivalent {\it increase} in this sending node's {\it packet duplication} potential.  Separating overall network potential into these two categories will be necessary to state and prove the following Lemma:
\begin{lemma} \label{item3} Every change in network potential comes from one of the following 3 events:\vspace{-.3cm}
    \begin{enumerate}\setlength{\itemsep}{1pt} \setlength{\parskip}{0pt} \setlength{\parsep}{0pt}
    \item $S$ inserts a packet into the network.

    \item $R$ receives a packet.

    \item A packet that was sent from one internal node to another is accepted; the verification of packet receipt is received by the sending node; a packet is shuffled between buffers of the same node; or a packet is moved within a buffer.
    \end{enumerate}
\vspace{-.2cm}Furthermore, changes in network potential due to item 1) are strictly non-negative and changes due to item 2) are strictly non-positive.  Also, changes in network {\bf \em non-duplicated potential} due to item 3) are strictly non-positive.  Finally, at all times, network {\bf \em packet duplication potential} is bounded between zero and $2n^3-8n^2+8n$.
\end{lemma}
\begin{proof} Since network potential counts the heights of the internal nodes' buffers, it only changes when these heights change, which in turn happens exclusively when there is packet movement.  By reviewing the pseudo-code, we see that this happens only on lines ({\bf \ref{routingRules2}.32}), ({\bf \ref{routingRules2}.35}), ({\bf \ref{routingRules2}.53}), ({\bf \ref{routingRules2}.55}), ({\bf \ref{routingRules2}.57}), ({\bf \ref{routingRules2}.61}), ({\bf \ref{routingRules2}.64}), and ({\bf \ref{reShuffleRules}.89-90}).  Each of these falls under one of the three items listed in the Lemma, thus proving the first statement in the Lemma.  That network potential changes due to packet insertion by $S$ are strictly non-negative is obvious (either the receiving node's potential increases by the height the packet assumed, as on ({\bf \ref{routingRules2}.53}), or the receiving node is $R$ and the packet does not contribute to potential).  Similarly, that potential change upon packet receipt by $R$ is strictly non-positive is clear, since packets at $R$ do not count towards potential (see Definition \ref{potDef}).  Also, since only flagged packets (but not necessarily all of them) contribute to network packet duplication potential, the biggest it can be is the maximal number of flagged packets that can exist in the network at any given time, times the maximum height each flagged packet can have.  By Claim \ref{packetProliferation}, there are at most $(n-2)^2$ flagged packets in the network at any given time, and each one has maximal height $2n$ (Lemma \ref{pseudo}, part 9), so network packet duplication potential is bounded by $2n^3 - 8n^2+8n$.

It remains to prove that changes in network non-duplicated potential due to item 3) are strictly non-positive.  To do this, we look at all lines on which there is packet movement, and argue each will result in a non-positive change to non-duplicated potential.  Clearly potential changes on lines ({\bf \ref{routingRules2}.32}), ({\bf \ref{routingRules2}.55}), ({\bf \ref{routingRules2}.57}), ({\bf \ref{routingRules2}.61}), and ({\bf \ref{routingRules2}.64}) are non-positive.  Also, if ({\bf \ref{routingRules2}.35}) is reached, if $R$ has already accepted the packet, then that packet's potential will count towards {\it duplicated} potential within the outgoing buffer, and so the change in potential as on ({\bf \ref{routingRules2}.35}) will not affect non-duplicated potential.  If on the other hand $R$ has {\it not} already accepted the packet, then the flagged packet still counts towards non-duplication potential in the outgoing buffer.  Since the result of ({\bf \ref{routingRules2}.35}) is simply to swap the flagged packet with the top packet in the buffer, the net change in non-duplication potential is zero.  That changes in potential due to re-shuffling packets ({\bf \ref{reShuffleRules}.89-90}) are strictly non-positive follows from Claim \ref{packetHeight}.  It remains to check the cases that a packet that was transferred between two internal nodes is accepted ({\bf \ref{routingRules2}.53}).  Notice that upon receipt there are two changes to network non-duplicated potential: it increases by the height the packet assumes in the incoming buffer it arrived at ({\bf \ref{routingRules2}.53}), and it decreases by the height the packet had in the corresponding outgoing buffer (this decrease is because the flagged packet in the outgoing buffer will count towards packet duplication potential instead of non-duplicated potential the instant the packet is accepted).  The decrease outweighs the increase since the packet's height in the incoming buffer is less than or equal to the height it had in the corresponding outgoing buffer (Claim \ref{packetHeight}).
\end{proof}
The following Lemma will be useful in bounding the number of rounds in which no packets are inserted.  We begin with the following definition:
\begin{defn} The sender is {\em blocked} from inserting any packets in some round $\mathtt{t}$ if the sender is not able to insert any packets in $\mathtt{t}$ (see Definition \ref{insert}).  Let $\beta_{\mathtt{T}}$ denote the number of rounds in a transmission $\mathtt{T}$ that the sender was blocked.\end{defn}
\begin{lemma} \label{potentialDrop} If at any point in any transmission $\mathtt{T}$, the number of blocked rounds is $\beta_{\mathtt{T}}$, then there has been a {\bf \em decrease} in the network's non-duplicated potential by at least\footnote{An initial guess that the minimal potential drop equals ``$2n$'' for each blocked round is incorrect.  Consider the case where the active path consists of all $n-2$ intermediate nodes with the following current state: the first two nodes' buffers all have height $2n$, the next pair's buffers all have height $2n-1$, and so forth, down to the last pair of internal nodes, whose buffers all have height $n+2$.  Then the drop in the network's non-duplicated potential is only $n+2$ for this round.} $n \beta_{\mathtt{T}}$.\end{lemma}
The intuition of the proof is to argue that each blocked round creates a drop in non-duplicated potential of at least $n$ as follows.  If the sender is blocked from inserting a packet, the node $N$ adjacent to the sender (along the active honest path) will necessarily have a full incoming buffer along its edge to the sender.  By the fact that buffers are {\it balanced} (Lemma \ref{balancing}),  this implies that all of $N$'s outgoing buffers are also full.  Meanwhile, at the opposite end of the active honest path, the node adjacent to the receiver will necessarily send a packet to the receiver if there is anything in its outgoing buffer along this edge, and this will result in a drop of potential of whatever height the packet had in the outgoing buffer.  Therefore, at the start of the active honest path, the buffers are full, while at the end of the path, a packet will be transferred to height zero (in the receiver's buffer).  Intuitively,  it therefore seems that tracking all packet movements along the active honest path should result in a drop of potential of at least $2n$.  As the counter-example in the footnote shows, this argument does not work exactly (we are only guaranteed a drop of $n$), but the structure of the proof is guided by this intuition.  We begin with the following lemma.
\begin{lemma} \label{chainL2} Let $\mathcal{C} = N_1 N_2 \dots N_l$ be a path consisting of $l$ nodes, such that $R = N_l$ and $S \notin \mathcal{C}$.  Suppose that in round $\mathtt{t}$, all edges $E(N_i,N_{i+1})$, $1 \leq i <l$ are {\it active} for the entire round.  Let $\phi$ denote the change in the network's non-duplicated potential caused by:
	\begin{enumerate}\setlength{\itemsep}{1pt} \setlength{\parskip}{0pt} \setlength{\parsep}{0pt}
	\item (For $1 \leq i < l$) Packet transfers across $E(N_i, N_{i+1})$ in round $\mathtt{t}$,
	\item (For $1 < i < l$) Re-shuffling packets {\bf \em into} $N_i$'s outgoing buffers during $\mathtt{t}$,
	\end{enumerate}
Then if $O_{N_1, N_2}$ denotes $N_1$'s outgoing buffer along $E(N_1, N_2)$ and $O$ denotes its height at the outset of $\mathtt{t}$, we have:
	\begin{enumerate}\setlength{\itemsep}{1pt} \setlength{\parskip}{0pt} \setlength{\parsep}{0pt}
	\item[-] If $O_{N_1, N_2}$ has a flagged packet that has already been accepted by $N_2$ {\bf \em before} round $\mathtt{t}$, then:
		\begin{equation}
		\phi \leq -O + l-1
		\end{equation}
	\item[-] Otherwise,
		\begin{equation}
		\phi \leq -O + l-2
		\end{equation}
	\end{enumerate}
\end{lemma}
\begin{proof} The proof of this lemma is rather involved and relies heavily on the pseduo-code, so we have pushed its proof to Section \ref{bareBones}, where it is restated and proved as Lemma \ref{chainL}.\end{proof}
We can prove Lemma \ref{potentialDrop} as a Corollary.
\begin{proof}[Proof of Lemma \ref{potentialDrop}.]  For every blocked round $\mathtt{t}$, by the {\em conforming} assumption there exists a chain $\mathcal{C}_{\mathtt{t}}$ connecting the sender and receiver that satisfies the hypothesis of Lemma \ref{chainL2}.  Letting $N_1$ denote the first node on this chain (not including the sender), the fact that the round was blocked means that $N_1$'s incoming buffer was full, and then by Lemma \ref{balancing}, so was $N_1$'s outgoing buffer along $E(N_1,N_2)$.  Since the length of the chain $l$ is necessarily less than or equal to $n$, Lemma \ref{chainL2} says that the change in non-duplicated potential contributions of $\phi$ (see notation there) satisfy:
	\begin{equation}
	\phi \leq -O_{N_1, N_2} +l - 1  \leq -2n + n -1 < -n
	\end{equation}
Since $\phi$ only records some of the changes to non-duplicated potential, we use Statement 3 of Lemma \ref{item3} to argue that the contributions not counted will only help the bound since they are strictly non-positive.  Since we are not double counting anywhere, each blocked round will correspond to a drop in non-duplicated potential of at least $-n$, which then yields the lemma.
\end{proof}
The following Lemma will bound the number of rounds that $S$ needs to insert packets corresponding to the same codeword.
\begin{lemma} \label{obsOne} If at any time $D-2n^3$ distinct packets corresponding to some codeword $b_i$ have been inserted into the network, then $R$ can necessarily decode message $m_i$.\end{lemma}
\begin{proof} Every packet that has been inserted into the network has either reached $R$ or is in the incoming/outgoing buffer of an internal node (Claim \ref{packetProliferation2}).  Since the maximum number of packets that are in the latter category is less than $4n^3$ (Claim \ref{capacity0}), if $D-2n^3$ distinct packets corresponding to $b_i$ have been inserted, then $R$ has necessarily received $D - 6n^3 = (1-\lambda)\left(\frac{6n^3}{\lambda}\right)$ of these, and so by Fact 1 $R$ can decode message $m_i$.\end{proof}
\noindent We can now (restate and) prove the two main theorems of Section \ref{nonAdAnalysis}.\\ \hspace*{\fill} \\
\noindent{\bf Theorem \ref{nonAdTheorem}.}  {\it Each message $m_i$ takes at most $3D$ rounds to pass from the sender to the receiver.  In particular, after $O(xD)$ rounds, $R$ will have received at least $O(x)$ messages.  Since each message has size $M$=$\frac{6 \sigma}{\lambda}Pn^3$=$\space O(n^3)$ and $D$=$\frac{6n^3}{\lambda}$=$\space O(n^3)$, after $O(x)$ rounds, $R$ has received $O(x)$ bits of information, and thus our edge-scheduling adversarial protocol enjoys a linear throughput rate.}
\begin{proof}[{\it Proof of Theorem \ref{nonAdTheorem}}.]  Let $\mathtt{t}$ denote the round that $S$ first tries to insert packets corresponding to a new codeword $b_i$ into the network.  In each round between $\mathtt{t}$ and $\mathtt{t}+3D$, either $S$ is able to insert a packet or he isn't.  By the pigeonhole principle, either $D$ rounds pass in which $S$ can insert a packet, or $2D$ rounds pass in which no packets are inserted.  In the former case, $R$ can decode by Lemma \ref{obsOne}.  It remains to prove the theorem in the latter case.  Lemma \ref{potentialDrop} says that the network non-duplicated potential drops by at least $n$ in each of the $2D$ rounds in which no packets are inserted, a total drop of $2nD$.  Meanwhile, Lemma \ref{item3} guarantees that the {\it increase} to network potential between $\mathtt{t}$ and $\mathtt{t} + 3D$ caused by {\it duplicated potential} is at most by $2n^3-8n^2+8n$.  Combining these two facts, we have that (not counting changes in potential caused by packet insertions) the network potential {\em drops} by at least $2nD - 2n^3+8n^2-8n$ between $\mathtt{t}$ and $\mathtt{t} + 3D$.  Since network potential can never be negative, we must account for this (non-duplicated) potential drop with positive contributions to potential change.  The potential already in the network at the start of $\mathtt{t}$ adds to the potential at most $4n^4-14n^3+8n^2+8n$ (Claim \ref{capacity1}).  Therefore, packet insertions must account for the remaining change in potential of $(2nD - 2n^3+8n^2-8n) -(4n^4-14n^3+8n^2+8n) = 2nD-4n^4 +(12n^3-16n) \geq 2nD - 4n^4$ (where the last inequality assumes $n \geq 3$).  Lemma \ref{item3} states that the only way network potential can increase (other than the contribution of packet duplication potential which has already been accounted for) is when $S$ inserts a packet (a maximum increase of $2n$ per packet), so it must be that $S$ inserted at least $(2nD-4n^4)/2n = D-2n^3$ packets into the network between $\mathtt{t}$ and $\mathtt{t}+3D$, and again $R$ can decode by Lemma \ref{obsOne}. \end{proof}
\noindent {\bf Theorem \ref{nonAdMem}.} {\it The edge-scheduling protocol described in Section \ref{nonAdDesc} (and formally in the pseudo-code of Section \ref{pseudoCode}) requires at most $O(n^2 \log n)$ bits of memory of the internal processors.}
\begin{proof}[{\it Proof of Theorem \ref{nonAdMem}}.]  Packets have size $\log n$ to allow the packets to be indexed.  Since each internal node needs to hold at most $O(n^2)$ packets at any time (it has $2(n-2)$ buffers, each able to hold $2n$ packets), the theorem follows.
\end{proof}
\section{Edge-Scheduling Protocol:\\ Pseudo-Code Intensive Claims and Proofs} \label{bareBones}
\indent \indent In this section we prove that our pseudo-code is consistent with the claimed properties that our protocol enjoys.
	
The following lemma is the first attempt to link the pseudo-code with the high-level description of what our protocol is doing.  Recall that a buffer is in {\em normal} (respectively {\em problem}) status whenever its status bit $sb$ is zero (respectively one).  Also, an outgoing buffer is said to have a {\em flagged packet} if $H_{FP} \neq \bot$, and the flagged packet is the packet in the outgoing buffer at height $H_{FP}$.  Notice that because the pseudo-code is written sequentially, things that conceptually happen simultaneously appear in the pseudo-code as occurring consecutively.  In particular, when packets are moved between buffers, updating the buffers' contents and updating the height variables does not happen simultaneously in the code, which explains the wording of the first sentence in the following lemma.
\begin{lemma}\label{pseudo} At all times (i.e.\ all lines of code in Figures \ref{routingRules}, \ref{routingRules2}, and \ref{reShuffleRules}) EXCEPT when packets travel between buffers (({\bf \em \ref{routingRules2}.32-33}), ({\bf \em \ref{routingRules2}.52-53}), and ({\bf \em \ref{reShuffleRules}.89-90})), along any (directed) edge $E(A,B)$ for any pair of internal nodes $(A,B)$, we have that:
    \begin{enumerate}\setlength{\itemsep}{1pt} \setlength{\parskip}{0pt} \setlength{\parsep}{0pt}
    \item If $H_{GP} > H_{IN}$ or $H_{GP} = \bot$, then $H_{GP} = H_{IN}+1$ or $H_{GP} = \bot$ and $\mathsf{IN}[i] \neq \bot \thickspace \forall i \in [1..H_{IN}]$ and $\mathsf{IN}[i] = \bot \thickspace \forall i \in [H_{IN}+1..2n]$.

    \item If $H_{GP} \leq H_{IN}$, then $\mathsf{IN}[i] \neq \bot \thickspace \forall i \in [1..H_{GP}-1]$ and $\forall i \in [H_{GP}+1..H_{IN}+1]$, and $\mathsf{IN}[i] = \bot \thickspace \forall i \in [H_{IN}+2..2n]$ and $\mathsf{IN}[H_{GP}] = \bot$.

    \item If $H_{FP} > H_{OUT}$, then $sb=1$ and $\mathsf{OUT}[i] \neq \bot \thickspace \forall i \in [1..H_{OUT}-1]$ and $\mathsf{OUT}[H_{FP}] \neq \bot$.

    \item If $H_{FP} = \bot$ or $H_{FP} \leq H_{OUT}$, then $\mathsf{OUT}[i] \neq \bot \thickspace \forall i \in [1..H_{OUT}]$.

    \item The height of $\mathsf{IN}$, as defined by the number of packets (i.e.\ non-null entries) of $\mathsf{IN}$, is equal to the value of $H_{IN}$.

    \item The height of $\mathsf{OUT}$, as defined by the number of packets (i.e.\ non-null entries) of $\mathsf{OUT}$, is equal to the value of $H_{OUT}$.

    \item Whenever ({\bf \ref{routingRules2}.53}) is reached, $H_{GP} \in [1..2n]$ and $H_{IN} \in [0..2n-1]$.

    \item Whenever ({\bf \ref{routingRules2}.32}) is reached, $H_{FP} \neq \bot$ and $H_{OUT} \in [1..2n]$.

    \item At {\it all} times (even those listed in the hypothesis above), $H_{IN}, H_{OUT} \in [0..2n]$ and $H_{GP}, H_{FP} \in \bot \cup [1..2n]$ (so the domains of these variables are well-defined).
    \end{enumerate}
Additionally, during any call to {\it Re-Shuffle}:
    \begin{itemize}\setlength{\itemsep}{1pt} \setlength{\parskip}{0pt} \setlength{\parsep}{0pt}
    \item[10.] Whenever the conditional statement on line ({\bf \em \ref{reShuffleRules}.74}) is satisfied, one packet will pass between buffers.  In particular, there will be a buffer that was storing the packet before the call to {\it Re-Shuffle} that will not be storing (that instance of) the packet after the reshuffle.  Similarly, there will be another buffer that has filled a vacant slot with (an instance of) the packet in question.

    \item[11.] Flagged packets do not move.  More precisely, if $H_{FP} \neq \bot$ just before any call to {\it Re-Shuffle}, then $H_{FP}$ and $\mathsf{OUT}[H_{FP}]$ will not change during that call to {\it Re-Shuffle}.

    \item[12.] Either $H_{GP}$ does not change during re-shuffling or $H_{GP}$ has {\bf \em decreased} to equal $H_{IN}+1$.  Also, if $H_{GP} \neq \bot$, then $\mathsf{IN}[H_{GP}]$ does not get filled at any point during re-shuffling.

    \item[13.] If $H_{IN} < 2n$ before Re-Shuffling, then $H_{IN} < 2n$ after Re-Shuffling.
    \end{itemize}
\end{lemma}
\noindent {\it Proof of Lemma \ref{pseudo}.} We prove each Statement of the Lemma above simultaneously by using induction on the round and line number as follows.  We first prove the Lemma holds at the outset of the protocol (base case).  We then notice that the above variables only change their value in the lines excluded from the Lemma and lines ({\bf \ref{routingRules2}.35}), ({\bf \ref{routingRules2}.38}), ({\bf \ref{routingRules2}.46}), ({\bf \ref{routingRules2}.50}), ({\bf \ref{routingRules2}.55}), ({\bf \ref{routingRules2}.57}), ({\bf \ref{routingRules2}.61-62}), ({\bf \ref{routingRules2}.64}), and ({\bf \ref{reShuffleRules}.91-94}).  In particular, we use the induction hypothesis to argue that as long as the statement of the Lemma is true going into each set of excluded lines and lines ({\bf \ref{routingRules2}.35}), ({\bf \ref{routingRules2}.38}), ({\bf \ref{routingRules2}.46}), ({\bf \ref{routingRules2}.50}), ({\bf \ref{routingRules2}.55}), ({\bf \ref{routingRules2}.57}), ({\bf \ref{routingRules2}.61-62}), ({\bf \ref{routingRules2}.64}), and ({\bf \ref{reShuffleRules}.91-94}), then it will remain true when the protocol leaves each of those lines.  Using this technique, we now prove each Statement listed above.\\ \hspace*{\fill} \\
\textsc{Base Case.}  At the outset of the protocol, $H_{GP}$ and $H_{FP} = \bot$, $H_{IN}$ and $H_{OUT}=0$, and all entries of $\mathsf{IN}$ and $\mathsf{OUT}$ are $\bot$ ({\bf \ref{setupCode}.29-31} and {\bf \ref{setupCode}.33-35}) so Statements 1-6 and 9 are true.\\ \hspace*{\fill} \\
\textsc{Induction Step.}  We now prove that each of the above Statements hold after leaving lines ({\bf \ref{routingRules2}.32-33}), ({\bf \ref{routingRules2}.35}), ({\bf \ref{routingRules2}.38}), ({\bf \ref{routingRules2}.46}), ({\bf \ref{routingRules2}.50}), ({\bf \ref{routingRules2}.52-53}), ({\bf \ref{routingRules2}.55}), ({\bf \ref{routingRules2}.57}), ({\bf \ref{routingRules2}.61-62}), ({\bf \ref{routingRules2}.64}), ({\bf \ref{reShuffleRules}.89-90}), and ({\bf \ref{reShuffleRules}.91-94}), provided they held upon entering these lines.\\ \hspace*{\fill} \\
\underline{Lines ({\bf \ref{routingRules2}.32-33}).} The variables in Statements 1, 2, 5, 7, do not change in these lines, and hence these Statements remain valid by the induction hypothesis.  Statement 3 is vacuously true, since $H_{FP}$ is set to $\bot$ at the end of line ({\bf \ref{routingRules2}.33}).  Also, Statement 9 will remain valid as long as Statement 8 does, as $H_{FP}$ is set to $\bot$ on line ({\bf \ref{routingRules2}.33}), and $H_{OUT} \in [0..2n]$ would follow from Statement 8 since upon entering these lines, $H_{OUT} \in [1..2n]$ (Statement 8), and so subtracting 1 from $H$ on line ({\bf \ref{routingRules2}.33}) ensures that $H_{OUT}$ will remain in $[0..2n-1] \subseteq [0..2n]$.  The first part of Statement 8, that $H_{FP} \neq \bot$ when ({\bf \ref{routingRules2}.32}) is reached, follows immediately from Claim \ref{HFPFR} below together with the fact that ({\bf \ref{routingRules2}.30}) must have been satisfied to reach ({\bf \ref{routingRules2}.32}).

We next prove Statement 6.  Anytime lines ({\bf \ref{routingRules2}.32-33}) are reached, the decrease of one by $H_{OUT}$ on ({\bf \ref{routingRules2}.33}) represents the fact that $\mathsf{OUT}$ should be deleting a packet on these lines.  Since the induction hypothesis (applied to Statement 6) guarantees that $H_{OUT}$ matches the number of packets (non-bottom entries) of $\mathsf{OUT}$ {\it before} lines ({\bf \ref{routingRules2}.32-33}), the changes to $H_{OUT}$ and the height of $\mathsf{OUT}$ on these lines will exactly match/cancel provided $\mathsf{OUT}$ {\it does} actually decrease in height by 1 (i.e.\ provided $\mathsf{OUT}[H_{FP}] \neq \bot$).  Since $H_{FP}$ is changed ({\bf \ref{routingRules2}.33}) {\it after} deleting a packet ({\bf \ref{routingRules2}.32}), we may apply the induction hypothesis to Statements 3 and 4 to argue that $\mathsf{OUT}[H_{FP}] \neq \bot$ as long as the value of $H_{FP}$ was not $\bot$ when line ({\bf \ref{routingRules2}.32}) was reached.  This was proven above for the first part of Statement 8.

Statement 4 follows from the argument above as follows.  Upon leaving ({\bf \ref{routingRules2}.33}), $H_{FP}=\bot$, so we must show $\mathsf{OUT}[i] \neq \bot \thickspace \forall i \in [1..H_{OUT}]$.  As was argued above, $H_{FP} \neq \bot$ when ({\bf \ref{routingRules2}.32}) is reached.  If $H_{FP} > H_{OUT}$ when ({\bf \ref{routingRules2}.32}) is reached, then by the induction hypothesis applied to Statement 3, on that same line $\mathsf{OUT}[i] \neq \bot \thickspace \forall i \in [1..H_{OUT}-1]$ and $\mathsf{OUT}[H_{FP}] \neq \bot$.  The packet at height $H_{FP}$ will be deleted on ({\bf \ref{routingRules2}.32}), so that $\mathsf{OUT}[i] \neq \bot \thickspace \forall i \in [1..H_{OUT}-1]$, but $\mathsf{OUT}[i] = \bot$ for all $i \geq H_{OUT}$.  Then when $H_{OUT}$ is reduced by one on ({\bf \ref{routingRules2}.33}), we will have that $\mathsf{OUT}[i] \neq \bot \thickspace \forall i \in [1..H_{OUT}]$, as required.

If on the other hand $H_{FP} \leq H_{OUT}$ when ({\bf \ref{routingRules2}.32}) is reached, then by the induction hypothesis applied to Statement 4, on that same line $\mathsf{OUT}[i] \neq \bot \thickspace \forall i \in [1..H_{OUT}]$.  The packet at height $H_{FP}$ will be deleted on ({\bf \ref{routingRules2}.32}) and the packets on top of it shifted down one if necessary, so that after ({\bf \ref{routingRules2}.32}) but before ({\bf \ref{routingRules2}.33}), we will have that $\mathsf{OUT}[i] \neq \bot \thickspace \forall i \in [1..H_{OUT}-1]$, but $\mathsf{OUT}[i] = \bot$ for all $i \geq H_{OUT}$.  Then when $H_{OUT}$ is reduced by one on ({\bf \ref{routingRules2}.33}), we will have that $\mathsf{OUT}[i] \neq \bot \thickspace \forall i \in [1..H_{OUT}]$, as required.

The second part of Statement 8 also follows from the arguments above as follows.  First, it was shown in the proof of Statement 6 that $\mathsf{OUT}[H_{FP}] \neq \bot$ when ({\bf \ref{routingRules2}.32}) is reached.  In particular, the height of $\mathsf{OUT}$ is at least one going into ({\bf \ref{routingRules2}.32}), and then the induction hypothesis applied to Statement 6 implies that $H_{OUT} \geq 1$ when ({\bf \ref{routingRules2}.32}) is reached, and the induction hypothesis applied to Statement 9 implies that $H_{OUT} \leq 2n$ when ({\bf \ref{routingRules2}.32}) is reached.\\ \hspace*{\fill} \\
\underline{Line ({\bf \ref{routingRules2}.35}).}  Since only $H_{FP}$ and $\mathsf{OUT}$ are modified on ({\bf \ref{routingRules2}.35}), we need only verify Statements 3, 4, 6, and 9 remain true after leaving ({\bf \ref{routingRules2}.35}).  Since $H_{FP}$ is gets the value $max(H_{OUT},H_{FP})$ on ({\bf \ref{routingRules2}.35}), Statement 9 will be true by the induction hypothesis (applied to Statement 9).  Also, the height of $\mathsf{OUT}$ does not change, as ({\bf \ref{routingRules2}.35}) only swaps the location of two packets already in $\mathsf{OUT}$, so Statement 6 will remain true.

Statement 3 is only relevant if $H_{FP} > H_{OUT}$ before reaching ({\bf \ref{routingRules2}.35}), since otherwise $H_{FP} = H_{OUT}$ upon leaving ({\bf \ref{routingRules2}.35}), and Statement 3 will be vacuously true.  On the other hand, if $H_{FP}>H_{OUT}$, then line ({\bf \ref{routingRules2}.35}) is not reached since ({\bf \ref{routingRules2}.34}) will be false.

In order to reach ({\bf \ref{routingRules2}.35}), $H_{FP} \neq \bot$ on ({\bf \ref{routingRules2}.34}), and so both $H_{OUT}$ and $H_{FP}$ are not equal to $\bot$ when ({\bf \ref{routingRules2}.35}) is entered (Claim \ref{HFPFR}), and hence $H_{FP} \neq \bot$ upon leaving ({\bf \ref{routingRules2}.35}).  Also, since ({\bf \ref{routingRules2}.35}) is only reached if $H_{FP} < H_{OUT}$ ({\bf \ref{routingRules2}.34}), we use the induction hypothesis (applied to Statement 4) to argue that before reaching ({\bf \ref{routingRules2}.35}), we had that $\mathsf{OUT}[i] \neq \bot \thickspace \forall i \in [1..H_{OUT}]$.  In particular, both $\mathsf{OUT}[H_{FP}]$ and $\mathsf{OUT}[H_{OUT}]$ are storing a packet, and the call to {\it Elevate Flagged Packet} simply swaps these packets, so that after the swap, it is still the case that $\mathsf{OUT}[i] \neq \bot \thickspace \forall i \in [1..H_{OUT}]$.  Since in this case $H_{FP} = H_{OUT}$ after line ({\bf \ref{routingRules2}.35}), Statement 4 will remain true.
\\ \hspace*{\fill} \\
\underline{Line ({\bf \ref{routingRules2}.38}).}  $H_{FP}$ is the only relevant value changed on ({\bf \ref{routingRules2}.38}), so it remains to prove the relevant parts of Statements 3, 4 and 9.  We will show that whenever ({\bf \ref{routingRules2}.38}) is reached, $H_{OUT} \in [1..2n]$ and $\mathsf{OUT}[H_{OUT}] \neq \bot$.  If we can show these two things, we will be done, since when $H_{FP}$ is set to $H_{OUT}$ on ({\bf \ref{routingRules2}.38}), Statement 9 will be true, Statement 4 will follow from the induction hypothesis applied to either Statement 3 or 4, and Statement 3 will not be relevant.  By the induction hypothesis (applied to Statement 9), $H_{OUT} \in [0..2n]$ when ({\bf \ref{routingRules2}.38}) is reached.  The fact that ({\bf \ref{routingRules2}.38}) was reached means that the conditional statement on the line before ({\bf \ref{routingRules2}.37}) was satisfied, and thus $\mathsf{OUT}$ is in normal status ($sb=0$) and $H_{OUT} \in [1..2n]$.  By the induction hypothesis (applied to Statement 3), the fact that $sb=0$ going into ({\bf \ref{routingRules2}.37}) implies that $H_{FP} = \bot$ or $H_{FP} \leq H_{OUT}$ going into ({\bf \ref{routingRules2}.37}), and then the induction hypothesis (applied to Statement 4) says that $\mathsf{OUT}[H_{OUT}] \neq \bot$ when ({\bf \ref{routingRules2}.38}) is entered.\\ \hspace*{\fill} \\
\underline{Lines ({\bf \ref{routingRules2}.46}) and ({\bf \ref{routingRules2}.50}).}  The parts of Statements 1, 2, and 9 involving changes to $H_{GP}$ are the only Statements that are affected by these lines.  If the conditional statement on these lines are not satisfied, then no values change, and there is nothing to prove.  We therefore consider the case that the conditional statement is satisfied.  Then $H_{GP}$ is set to $H_{IN} + 1$ on these lines, and hence Statement 2 is vacuously satisfied.  Since we are assuming $H_{GP}$ changes value on ({\bf \ref{routingRules2}.46}) or ({\bf \ref{routingRules2}.50}), the conditional statement says that $H_{GP} = \bot$ or $H_{GP} > H_{IN}$ going into ({\bf \ref{routingRules2}.46}) (respectively ({\bf \ref{routingRules2}.50})).  By the induction hypothesis (applied to Statement 1), $\mathsf{IN}[i] \neq \bot$ for all $1 \leq i \leq H_{IN}$, and $\mathsf{IN}[i] = \bot$ for all $i > H_{IN}$.  Therefore, since $\mathsf{IN}$ and $H_{IN}$ do not change on ({\bf \ref{routingRules2}.46}) or ({\bf \ref{routingRules2}.50}), Statement 1 will remain true upon leaving these lines.  Finally, for Statement 9, we need only show $H_{GP} \in [1..2n]$ upon leaving line ({\bf \ref{routingRules2}.46}) (respectively line ({\bf \ref{routingRules2}.50})).  If $H_{GP} > H_{IN}$ going into line ({\bf \ref{routingRules2}.46}) (respectively line ({\bf \ref{routingRules2}.50})), then the change to $H_{GP}$ is non-positive, and so the induction hypothesis applied to Statements 1 and 9 guarantee $H_{GP}$ will be in $[1..2n]$ upon leaving these lines.  On the other hand, if $H_{GP} = \bot$ going into either of these lines, then $H_{IN} < 2n$, and the induction hypothesis applied to Statement 9 indicates that $H_{IN} \in [0..2n-1]$ going into these lines, and hence $H_{GP} \in [1..2n]$ upon leaving either line.
\\ \hspace*{\fill} \\
\underline{Lines ({\bf \ref{routingRules2}.52-53}).} Notice that $H_{GP}$ necessarily equals $\bot$ when leaving ({\bf \ref{routingRules2}.53}), so Statement 2 above is vacuously satisfied.  Also, neither $H_{OUT}$, $H_{FP}$, nor $\mathsf{OUT}$ is modified in these lines, so Statements 3, 4, 6, 8, and the parts of Statement 9 concerning these variables will remain valid by the induction hypothesis.

We prove Statement 1 first.  Recall that the {\it height} of an incoming buffer refers to the number of (non-ghost) packets the buffer currently holds.  Since $H_{GP}$ will necessarily equal $\bot$ when leaving line ({\bf \ref{routingRules2}.53}), we must show that $\mathsf{IN}[i] \neq \bot \thickspace \forall i \in [1..H_{IN}]$ and $\mathsf{IN}[i] = \bot \thickspace \forall i \in [H_{IN}+1..2n]$ upon leaving line ({\bf \ref{routingRules2}.53}).  Both of these follow immediately from the induction hypothesis applied to Statements 1 and 2, as follows.  By the induction hypothesis applied to Statements 1, 2, and 9, either $H_{GP} = \bot$, $1 \leq H_{GP} \leq H_{IN}$, or $H_{GP} = H_{IN}+1 \leq 2n$ when line ({\bf \ref{routingRules2}.52}) is reached.  We consider each case:
	\begin{itemize}
	\item If $H_{GP} = H_{IN}+1$ when we reach line ({\bf \ref{routingRules2}.52}), then by the induction hypothesis (applied to Statement 1) it will also be true that $\mathsf{IN}[i] \neq \bot \thickspace \forall i \in [1..H_{IN}]$ and $\mathsf{IN}[i] = \bot \thickspace \forall i \in [H_{IN}+1..2n]$ when this line is reached.  While on line ({\bf \ref{routingRules2}.53}), first $\mathsf{IN}[H_{GP}] = \mathsf{IN}[H_{IN}+1]$ is filled with a packet, and then $H_{IN}$ is increased by one, and so Statement 1 will remain true by the end of line ({\bf \ref{routingRules2}.53}).

	\item If $1 \leq H_{GP} \leq H_{IN}$ when the protocol reaches ({\bf \ref{routingRules2}.52}), then also when this line is reached we have that (by the induction hypothesis applied to Statement 2) $\mathsf{IN}[i] \neq \bot \thickspace \forall i \in [1..H_{GP}-1]$ and $\forall i \in [H_{GP}+1..H_{IN}+1]$, and $\mathsf{IN}[i] = \bot \thickspace \forall i \in [H_{IN}+2..2n]$ and $\mathsf{IN}[H_{GP}] = \bot$.  When a packet is inserted into slot $H_{GP}$ and $H_{IN}$ is increased by one on line ({\bf \ref{routingRules2}.53}), we will therefore have that all slots between 1 and (the new value of) $H_{IN}$ will have a packet, and all other slots will be $\bot$, and thus Statement 1 will hold.

	\item If $H_{GP}= \bot$ going into line ({\bf \ref{routingRules2}.52}), then $H_{GP}$ will be set to $H_{IN}+1$ on this line, and then we can repeat the argument of the top bullet point, provided $H_{IN}+1 \leq 2n$.  If $sb_{OUT} =1$, then Statement 4 of Lemma \ref{subclaim6} states that $H_{GP} \neq \bot$ when ({\bf \ref{routingRules2}.52}) is reached, contradicting the fact we are in the case $H_{GP}= \bot$.  So we may assume $sb_{OUT}=0$, and then the fact that ({\bf \ref{routingRules2}.52}) was reached means that ({\bf \ref{routingRules2}.47}) must have been satisfied because $H_{OUT} > H_{IN}$.  Since both of these variables live in $[0..2n]$ by the induction hypothesis applied to Statement 9, we conclude $H_{IN} < 2n$ on ({\bf \ref{routingRules2}.47}), and it cannot change value between then and ({\bf \ref{routingRules2}.52}).
	\end{itemize}
The first part of Statement 7 is proven in the above three bullet points.  For the second part, if $sb_{OUT} = 0$ when ({\bf \ref{routingRules2}.47}) was evaluated earlier in the round, then the fact that ({\bf \ref{routingRules2}.53}) was reached means $H_{OUT} > H_{IN}$, and then the second part of Statement 7 follows from the induction hypothesis applied to Statement 9.  If on the other hand $sb_{OUT}=1$ when ({\bf \ref{routingRules2}.47}) was evaluated, then the second part of Statement 7 follows from Statement 5 of Lemma \ref{subclaim6}.

We now prove Statement 5.  There are two relevant changes made on line ({\bf \ref{routingRules2}.53}) that affect Statement 5: a packet is added to $\mathsf{IN}[H_{GP}]$ and $H_{IN}$ is increased by one.  The argument in the preceding paragraph showed that when ({\bf \ref{routingRules2}.53}) is reached, $H_{GP} \in [1..2n]$ and $\mathsf{IN}[H_{GP}] = \bot$, and therefore the net effect of ({\bf \ref{routingRules2}.53}) is to increase the number of packets stored in $\mathsf{IN}$ by one and to increase $H_{IN}$ by one.  Therefore, since Statement 5 was true going into line ({\bf \ref{routingRules2}.53}) by the induction hypothesis, it will remain true upon leaving ({\bf \ref{routingRules2}.53}).

It remains to prove the parts of Statement 9 not yet proven, namely that at {\it all} times $H_{IN} \in [0..2n]$ and $H_{GP} \in \bot \cup [1..2n]$.  As was proven in the third bullet point above, if ({\bf \ref{routingRules2}.52}) is satisfied, then $H_{IN} < 2n$, and hence the change there does not threaten the domain of $H_{GP}$.  Also, ({\bf \ref{routingRules2}.53}) sets $H_{GP}$ to $\bot$, which is again in the valid domain.  Meanwhile, on ({\bf \ref{routingRules2}.53}) $H_{IN}$ is changed to $H_{IN} +1\leq 2n$, where the inequality follows from the induction hypothesis applied to Statement 7.\\ \hspace*{\fill} \\
\underline{Line ({\bf \ref{routingRules2}.55}), ({\bf \ref{routingRules2}.57}), and ({\bf \ref{routingRules2}.64}).}  Since $\mathsf{IN}$ and $H_{GP}$ are the only relevant quantities that change value on these lines, only the relevant parts of Statements 1, 2, and 9 must be proven.  Since $H_{GP}$ is set to $\bot$ on these lines, Statement 9 is immediate and Statement 2 is vacuously true.  It remains to prove Statement 1.  If $H_{GP} = \bot$ going into ({\bf \ref{routingRules2}.55}), ({\bf \ref{routingRules2}.57}), or ({\bf \ref{routingRules2}.64}), then $H_{GP}$ and $\mathsf{IN}$ will not change, and the inductive hypothesis (applied to Statement 1) will ensure that Statement 1 will continue to be true upon exiting any of these lines.  If $1 \leq H_{GP} \leq H_{IN}$ when ({\bf \ref{routingRules2}.55}), ({\bf \ref{routingRules2}.57}), or ({\bf \ref{routingRules2}.64}) is entered, then we may apply the induction hypothesis to Statement 2 to conclude that $\mathsf{IN}[i] \neq \bot \thickspace \forall i \in [1..H_{GP}-1]$ and $\forall i \in [H_{GP}+1..H_{IN}+1]$, and $\mathsf{IN}[i] = \bot \thickspace \forall i \in [H_{IN}+2..2n]$ and $\mathsf{IN}[H_{GP}] = \bot$.  In particular, there is a gap in $\mathsf{IN}$ storing a ``ghost packet,'' and this gap will be filled when {\it Fill Gap} is called on ({\bf \ref{routingRules2}.55}), ({\bf \ref{routingRules2}.57}) or ({\bf \ref{routingRules2}.64}).  Namely, this will shift all the packets from height $H_{GP}+1$ through $H+1$ {\it down} one spot, so that after {\it Fill Gap} is called, $\mathsf{IN}[i] \neq \bot \thickspace \forall i \in [1..H_{IN}]$ and $\mathsf{IN}[i] = \bot \thickspace \forall i \in [H_{IN}+1..2n]$, which is Statement 1.  Finally, if $H_{GP} > H_{IN}$ when ({\bf \ref{routingRules2}.55}), ({\bf \ref{routingRules2}.57}) or ({\bf \ref{routingRules2}.64}) is entered, then {\it Fill Gap} will not do anything, and so $\mathsf{IN}$ will not change.  Since Statement 1 was true going into these lines (by our induction hypothesis), it will remain true upon exiting these lines.\\ \hspace*{\fill} \\
\underline{Line ({\bf \ref{routingRules2}.61-62}).}  The only relevant variables to change values on these lines are $sb_{OUT}$, $H_{OUT}$, $H_{FP}$, and $\mathsf{OUT}$, so we need only verify Statements 3, 4, 6, and 9 remain true after leaving ({\bf \ref{routingRules2}.61-62}).  First note that $H_{FP} \neq \bot$ upon reaching ({\bf \ref{routingRules2}.61}) (since ({\bf \ref{routingRules2}.60}) must be satisfied to reach ({\bf \ref{routingRules2}.61-62})), so the induction hypothesis (applied to Statements 3 and 4) implies that $\mathsf{OUT}[H_{FP}] \neq \bot$ when ({\bf \ref{routingRules2}.61}) is reached.  Therefore, $H_{OUT} \geq 1$ when ({\bf \ref{routingRules2}.61}) is reached, and hence $H_{OUT} \in [1..2n]$ upon reaching ({\bf \ref{routingRules2}.61}) by the induction hypothesis (applied to Statement 9).  In particular, when $H_{OUT}$ is reduced by one on ({\bf \ref{routingRules2}.62}), we will have that $H_{OUT} \in [0..2n-1]$ upon leaving ({\bf \ref{routingRules2}.62}), as required.  Also, $H_{FP}$ will be set to $\bot$ upon leaving ({\bf \ref{routingRules2}.62}), so Statement 9 remains true.

Statement 6 also follows from the fact that $\mathsf{OUT}[H_{FP}] \neq \bot$ when ({\bf \ref{routingRules2}.61}) is reached, as follows.  Since (by induction) Statement 6 was true upon reaching ({\bf \ref{routingRules2}.61}), the packet deleted from $\mathsf{OUT}$ on ({\bf \ref{routingRules2}.61}) is accounted for by the drop in $H_{OUT}$ on ({\bf \ref{routingRules2}.62}).

Statement 3 is vacuously true upon leaving ({\bf \ref{routingRules2}.62}), so it remains to prove Statement 4.  This argument is identical to the one used to prove Statement 4 in lines ({\bf \ref{routingRules2}.32-33}) above.\\ \hspace*{\fill} \\
\underline{Lines ({\bf \ref{reShuffleRules}.89-94}).}  We first prove Statements 10-13, and then address Statements 1-9.  We first prove Statement 10, i.e.\ that before {\bf \em Shuffle Packet} is called on ({\bf \ref{reShuffleRules}.77}), we have that $B_F[M] \neq \bot$ and $B_T[m+1] = \bot$.
	\begin{itemize}\setlength{\itemsep}{1pt} \setlength{\parskip}{0pt} \setlength{\parsep}{0pt}
	\item If $B_F$ is an outgoing buffer and $H_{FP} = \bot$ or $H_{FP} < H_{B_F}$, then $M = H_{B_F}$ (the conditional statement on line ({\bf \ref{reShuffleRules}.80}) will fail), and then $B_F[M] \neq \bot$ by the induction hypothesis applied to Statement 4.

	\item If $B_F$ is an outgoing buffer and $H_{FP} \geq H_{B_F}$, then $M = H_{B_F}-1$ (the conditional statement on line ({\bf \ref{reShuffleRules}.80}) will pass), and then $B_F[M] \neq \bot$ by the induction hypothesis applied to Statement 3 or 4 (that $M = H_{B_F}-1$ is greater than zero follows from the fact that $H_{FP} \neq \bot$ implies ${\scriptstyle FP \neq \bot}$ (Claim \ref{HFPFR}), and then the induction hypothesis applied to Statement 6 says $H_{B_F}>0$).

	\item If $B_F$ is an incoming buffer and $B_F[M+1] \neq \bot$, then ({\bf \ref{reShuffleRules}.82}) is satisfied and $M$ is set to $M+1$ on line ({\bf \ref{reShuffleRules}.82}), and then by construction $B_F[M] \neq \bot$ after line ({\bf \ref{reShuffleRules}.83}).

	\item Suppose $B_F$ is an incoming buffer and $B_F[M+1] = \bot$.  Notice that the induction hypothesis applied to Statement 2 and the fact that $B_F[M+1] = \bot$ imply that $H_{GP} > H_{IN} = M$.  Therefore, the induction hypothesis applied to Statement 1 implies that $B_F[M] \neq \bot$.

	\item If $B_T$ is an outgoing buffer and $B_T[m] = \bot$, then the conditional statement on line ({\bf \ref{reShuffleRules}.84}) will be satisfied, and hence $m$ is set to $m-1$.  Thus after line ({\bf \ref{reShuffleRules}.85}), $B_T[m+1] = \bot$.

	\item If $B_T$ is an outgoing buffer and $B_T[m] \neq \bot$, then the induction hypothesis applied to Statements 3, 4, and 6 imply that $B_T[m+1] = \bot$.

	\item If $B_T$ is an incoming buffer and $H_{GP} = \bot$, then the value of $m$ is not changed on line ({\bf \ref{reShuffleRules}.86}), and so $m+1 = H_{IN}+1$.  The induction hypothesis applied to Statement 1 then implies that $B_T[m+1] = \bot$.
	
	\item If $B_T$ is an incoming buffer and $H_{GP} \neq \bot$, then $B_T[H_{IN}+2] = \bot$ by the induction hypothesis applied to Statements 1 and 2, and thus after $m$ is changed to $m+1$ on ({\bf \ref{reShuffleRules}.87}), we have that $B_T[m+1] = B_T[H_{IN}+2] = \bot$, as required.
	\end{itemize}
For Statements 11-13, we need to change notation slightly, since Re-Shuffling can occur between two buffers of any types (except outgoing to incoming).  To prove these statements, we therefore treat 4 cases: 1) $B_F$ is an outgoing buffer, 2) $B_F$ is an incoming buffer, 3) $B_T$ is an outgoing buffer, 4) $B_T$ is an incoming buffer.  We then prove the necessary Statements in each case.
	\begin{enumerate}\setlength{\itemsep}{4pt} \setlength{\parskip}{0pt} \setlength{\parsep}{0pt}
	\item[] \textsc{Case 1.}  The value of $B_F[M]=\mathsf{OUT}[M]$ is changed on line ({\bf \ref{reShuffleRules}.90}), and hence Statement 11 will hold provided $M \neq H_{FP}$.  The top two bullet points above guarantee that this is indeed the case.  Statements 12 and 13 are not relevant unless $B_T$ is an incoming buffer, which will be handled in case 4 below.

	\item[] \textsc{Case 2.}  For Statement 13, the only relevant change to $H_{IN}$ is on line ({\bf \ref{reShuffleRules}.92}), where $H_{IN}$ {\it decreases} in value, and hence Statement 13 will remain true.  For the first part of Statement 12, the only place $H_{GP}$ can change is line ({\bf \ref{reShuffleRules}.94}).  But if $H_{GP}$ does change value here, then the conditional statement on the previous line guarantees that that $H_{GP}$ {\it decreases} to $H_{IN}+1$.  Statement 11 and the second part of Statement 12 are not relevant to this case.

	\item[] \textsc{Case 3.}  The value of $B_T[m+1]=\mathsf{OUT}[m+1]$ is changed on line ({\bf \ref{reShuffleRules}.89}), and hence Statement 11 will hold provided $m+1 \neq H_{FP}$.  But we have already shown Statement 10 remains true, and in particular the slot that is filled on line ({\bf \ref{reShuffleRules}.89}) was vacant.  If $H_{FP} \neq \bot$, then by the induction hypothesis applied to Statements 3 and 4, $\mathsf{OUT}[H_{FP}] \neq \bot$, and hence $\mathsf{OUT}[m+1] = \bot$ implies that $m+1 \neq H_{FP}$.  Statements 12 and 13 are not relevant to this case.

	\item[] \textsc{Case 4.}  Since $B_T$ is an incoming buffer, the condition on line ({\bf \ref{reShuffleRules}.74}) implies that the value of $m$ (which is the height of $B_T$) on line ({\bf \ref{reShuffleRules}.73}) must be at most $2n-2$ ($M-m >1$ and $M,m \in [0..2n]$ by induction hypothesis applied to Statement 9).  Therefore, when the height of $B_T$ is increased by one on line ({\bf \ref{reShuffleRules}.91}), it will be at most $2n-1$, and so Statement 13 will remain true.  For the second part of Statement 12, we must show that the value of $m+1$ on line ({\bf \ref{reShuffleRules}.89}) is not equal to $H_{GP}$.  In the case that $H_{GP} \neq \bot$ on line ({\bf \ref{reShuffleRules}.86}), the value of $m$ will change to $H_{IN}+1$ on line ({\bf \ref{reShuffleRules}.87}), and then the induction hypothesis applied to Statement 1 implies that $H_{GP} \leq H_{IN}+1 =m$ and so $H_{GP} \neq m+1$ on line ({\bf \ref{reShuffleRules}.89}).  Statement 11 and the first part of Statement 12 are not relevant for this case.
	\end{enumerate}
It remains to verify Statements 1-9.  There are two parts to proving Statements 1 and 2; we must show they hold when $B_F$ is an incoming buffer and also when $B_T$ is an incoming buffer.  For the latter part, Statements 1 and 2 will be true if we can show that anytime an incoming buffer's slot is filled as on line ({\bf \ref{reShuffleRules}.89}), the slot was either slot $H_{IN}+1$ (in the case that $H_{GP}=\bot$) or $H_{IN}+2$ (in the case that $H_{GP} \neq \bot$).  These facts follow immediately from the definition of $m$ on line ({\bf \ref{reShuffleRules}.73}) and lines ({\bf \ref{reShuffleRules}.86-87}) and ({\bf \ref{reShuffleRules}.89}).  For the former part, Statements 1 and 2 will remain true provided the packet taken from $B_F$ on line ({\bf \ref{reShuffleRules}.89}) is the top-most packet in $B_F$.  Looking at the conditional statement on line ({\bf \ref{reShuffleRules}.82}), if $\mathsf{IN}[H_{IN}+1] \neq \bot$, then by the induction hypothesis applied to Statements 1 and 2, we must have that $\mathsf{IN}[H_{IN}+1]$ is the top-most non-null packet, which is the packet that will be taken from $B_F$ on line ({\bf \ref{reShuffleRules}.89}) (since in this case $M=H_{IN}$ is changed to $H_{IN}+1$ on line ({\bf \ref{reShuffleRules}.83})).  On the other hand, if $\mathsf{IN}[H_{IN}+1] = \bot$ on line ({\bf \ref{reShuffleRules}.82}), then the induction hypothesis applied to statements 1 and 2 imply that $\mathsf{IN}[H_{IN}]$ is the top-most non-null packet, which is exactly the packet taken on line ({\bf \ref{reShuffleRules}.89}) (since the conditional statement on line ({\bf \ref{reShuffleRules}.82}) won't be satisfied, and hence the value of $M$ won't be change on line ({\bf \ref{reShuffleRules}.83})).

Similarly, there are two parts to proving Statements 3 and 4; we must show they hold when $B_F$ is an outgoing buffer and also when $B_T$ is an outgoing buffer.  The former part will be true provided the packet taken from $B_F$ on line ({\bf \ref{reShuffleRules}.79}) is the top-most {\it non-flagged} packet.  If $H_{FP} = \bot$, then there is no flagged packet, and hence the packet taken from $B_F$ should be the top packet, i.e.\ the packet in index $B_F[H_{OUT}]$.  Investigating the definition of $M$ on line ({\bf \ref{reShuffleRules}.72}) and lines ({\bf \ref{reShuffleRules}.80-81}) and ({\bf \ref{reShuffleRules}.89}) shows that this will be the case if $H_{FP}=\bot$.  If $H_{FP} \neq \bot$ and $H_{FP} < H_{OUT}$, then investigating those same lines also shows the top packet will be taken from $B_F$ (which is not flagged since $H_{FP} < H_{OUT}$ by assumption).  If $H_{FP} \geq H_{OUT}$, then line ({\bf \ref{reShuffleRules}.80}) will be satisfied, shifting the value of $M$ to $H_{OUT}-1$ on line ({\bf \ref{reShuffleRules}.81}).  By the induction hypothesis applied to Statement 3, this new value of $M$ corresponds to the top-most non-flagged packet of $B_F$.  The latter part will be true provided the packet given to $B_T$ takes the first free slot in $B_T$ (in particular, the packet will not over-write a flagged packet's spot).  If $B_T[H_{OUT}] \neq \bot$ on line ({\bf \ref{reShuffleRules}.84}), then the induction hypothesis applied to Statements 3, 4, and 6 imply that all slots of $B_T$ between $[1..H_{OUT}]$ are non-$\bot$, and all spots above $H_{OUT}$ are $\bot$.  Therefore, (since in this case the conditional statement on line ({\bf \ref{reShuffleRules}.84}) fails and hence the value of $m$ does not change on the next line) the definition of $m$ on line ({\bf \ref{reShuffleRules}.73}) and line ({\bf \ref{reShuffleRules}.89}) show that the first free slot of $B_T$ will be filled.  On the other hand, if $B_T[H_{OUT}] = \bot$ on line ({\bf \ref{reShuffleRules}.84}), then by the induction hypothesis, we must have that $B_T[H_{OUT}]$ is the first free slot of $B_T$, and by investigating lines ({\bf \ref{reShuffleRules}.73}), ({\bf \ref{reShuffleRules}.84-85}), and ({\bf \ref{reShuffleRules}.89}), this is exactly the spot that is filled.

Statements 5 and 6 remain true by the fact that Statement 10 was proven true and lines ({\bf \ref{reShuffleRules}.91}) and ({\bf \ref{reShuffleRules}.92}).  To satisfy the condition on line ({\bf \ref{reShuffleRules}.74}), it must be that $H_{B_F} = M \geq 1$ and $H_{B_T} = m <2n$, and hence the changes made to $H_{B_F}$ and $H_{B_T}$ on lines ({\bf \ref{reShuffleRules}.91}) and ({\bf \ref{reShuffleRules}.92}) will guarantee the parts of Statement 9 regarding $H_{OUT}$ and $H_{IN}$ remain true.  Also, $H_{GP}$ remains in the appropriate demain by induction applied to Statements 9, 12, and 13.  Statements 7, 8, are not relevant.
\hspace*{\fill} \hspace*{-12pt}$\scriptstyle{\blacksquare}$
%
\begin{lemma} \label{variables} The domains of all of the variables in Figures \ref{setupCode} and \ref{setupCode2} are appropriate.  In other words, the protocol never calls for more information to be stored in a node's variable (buffer, packet, etc$.$) than the variable has room for.
\end{lemma}
\begin{proof} Below we fix a node $N \in G$ and track changes to each of its variables.
	\begin{itemize}
	\item[] \textsf{Outgoing Buffers $OUT$} ({\bf \ref{setupCode}.08}).  Each entry of $OUT$ is initialized to $\bot$ on ({\bf \ref{setupCode}.33}).  After this point, Statement 6 of Lemma \ref{pseudo} above guarantees $OUT$ will need to hold at most $H_{OUT}$ packets, and since $H_{OUT}$ is always between 0 and $2n$ (by Statement 9 of Lemma \ref{pseudo}) and packets have size $P$, the domain for $OUT$ is as indicated.

	\item[] \textsf{Copy of Packet to be Sent $\tilde{p}$} ({\bf \ref{setupCode}.09}).  This is initialized to $\bot$ on ({\bf \ref{setupCode}.34}), and is only modified afterwards on ({\bf \ref{routingRules2}.38}), ({\bf \ref{routingRules2}.33}), and ({\bf \ref{routingRules2}.62}).  By Statements 3, 4, and 9 of Lemma \ref{pseudo}, $OUT[H] \neq \bot$ when $\tilde{p}$ is set on ({\bf \ref{routingRules2}.38}), and the changes on ({\bf \ref{routingRules2}.33}) and ({\bf \ref{routingRules2}.62}) reset $\tilde{p}$ to $\bot$.  Therefore, the domain of $\tilde{p}$ is as indicated.

	\item[] \textsf{Outgoing Status Bit $sb$} ({\bf \ref{setupCode}.10}).  This is initialized to 0 on ({\bf \ref{setupCode}.35}), and is only modified afterwards on lines ({\bf \ref{routingRules2}.33}), ({\bf \ref{routingRules2}.28}), and ({\bf \ref{routingRules2}.62}), all of which change $sb$ to 0 or 1, as required.

	\item[] \textsf{Packet Sent Bit $d$} ({\bf \ref{setupCode}.11}).  This is initialized to 0 on ({\bf \ref{setupCode}.35}), and is only modified afterwards on lines ({\bf \ref{routingRules2}.26}), ({\bf \ref{routingRules2}.40}), and ({\bf \ref{routingRules2}.62}), each of which change $d$ to 0 or 1, as required.

	\item[] \textsf{Flagged Round Index $FR$} ({\bf \ref{setupCode}.12}).  This is initialized to $\bot$ on ({\bf \ref{setupCode}.34}), and is only modified afterwards on lines ({\bf \ref{routingRules2}.38}), ({\bf \ref{routingRules2}.33}), and ({\bf \ref{routingRules2}.62}).  The latter two lines reset $FR$ to $\bot$, while ({\bf \ref{routingRules2}.38}) sets $FR$ to the index of the current stage and round $\mathtt{t}$, and since there are $3D$ rounds per transmission and 2 stages per round ({\bf \ref{routingRules}.02}), so when $FR$ is set to $\mathtt{t}$ on ({\bf \ref{routingRules2}.38}), it will be in $[0..6D]$, as required.

	\item[] \textsf{Height of Outgoing Buffer $H$} ({\bf \ref{setupCode}.13}).  This is initialized to 0 on ({\bf \ref{setupCode}.35}).  After this point, Statement 9 of Lemma \ref{pseudo} above guarantees $H \in [0..2n]$, as required.

	\item[] \textsf{Height of Flagged Packet $H_{FP}$} ({\bf \ref{setupCode}.14}).  Statement 9 of Lemma \ref{pseudo} guarantees that $H_{FP}$ will lie in the appropriate domain at all times.

	\item[] \textsf{Round Adjacent Node Last Received a Packet $RR$} ({\bf \ref{setupCode}.15}).  This is initialized to $\bot$ on ({\bf \ref{setupCode}.34}), and is only modified afterwards when it is received on ({\bf \ref{routingRules}.06}), where it is either set to the received value or $\bot$ if nothing was received.  As discussed below, the incoming buffer's value for $RR$ always lies in the appropriate domain domain, and hence so will the value received on ({\bf \ref{routingRules}.06}).

	\item[] \textsf{Outgoing Buffer's Value for Adjacent Node's Incoming Buffer Height $H_{IN}$} ({\bf \ref{setupCode}.16}).  This is initialized to 0 on ({\bf \ref{setupCode}.35}), and is only modified afterwards on line ({\bf \ref{routingRules}.06}), where it is set to the value sent on ({\bf \ref{routingRules}.09}) by the adjacent node, or $\bot$ in case no value was received.  Since the value sent on ({\bf \ref{routingRules}.09}) will always be between 0 and $2n$ (by Statement 9 of Lemma \ref{pseudo}), $H_{IN}$ has the required domain.

	\item[] \textsf{Incoming Buffers $IN$} ({\bf \ref{setupCode}.18}).  Each entry of $IN$ is initialized to $\bot$ on ({\bf \ref{setupCode}.29}).  After this point, Statement 5 of Lemma \ref{pseudo} above guarantees $IN$ will need to hold at most $H_{IN}$ packets, and since $H_{IN}$ is always between 0 and $2n$ (by Statement 9 of Lemma \ref{pseudo}) and packets have size $P$, the domain for $IN$ is as indicated.

	\item[] \textsf{Packet Just Received $p$} ({\bf \ref{setupCode}.19}).  This is initialized to $\bot$ on ({\bf \ref{setupCode}.30}), and is only modified afterwards on ({\bf \ref{routingRules2}.43}), where it either is set to the value sent on ({\bf \ref{routingRules2}.41}) or $\bot$ in the case no value was received.  Since the value sent on ({\bf \ref{routingRules2}.41}) has the appropriate domain (i.e.\ the size of a packet, $P$), in either case $p$ has the appropriate domain.

	\item[] \textsf{Incoming Status Bit $sb$} ({\bf \ref{setupCode}.20}).  This is initialized to 0 on ({\bf \ref{setupCode}.31}), and is only modified afterwards on lines ({\bf \ref{routingRules2}.45}), ({\bf \ref{routingRules2}.49}), ({\bf \ref{routingRules2}.53}), ({\bf \ref{routingRules2}.55}), ({\bf \ref{routingRules2}.57}), and ({\bf \ref{routingRules2}.64}), all of which change $sb$ to 0 or 1 as required.

	\item[] \textsf{Round Received Index $RR$} ({\bf \ref{setupCode}.21}).  This is initialized to $-1$ on ({\bf \ref{setupCode}.31}), and is only modified afterwards on lines ({\bf \ref{routingRules2}.53}) and ({\bf \ref{routingRules2}.64}).  The former sets $RR$ to the index of the current stage and round $\mathtt{t}$, and since there are $3D$ rounds per transmission and 2 stages per round ({\bf \ref{routingRules}.02}), setting $RR = \mathtt{t}$ as on ({\bf \ref{routingRules2}.53}) will put $RR$ in $[0..6D]$ as required.  Meanwhile, ({\bf \ref{routingRules2}.64}) resets $RR$ to $-1$.  Thus, at all times $RR \in \{0,1\}^{6D}$, as required.

	\item[] \textsf{Height of Incoming Buffer $H$} ({\bf \ref{setupCode}.22}).  This is initialized to 0 on ({\bf \ref{setupCode}.31}).  After this point, Statement 9 of Lemma \ref{pseudo} above guarantees $H \in [0..2n]$, as required.

	\item[] \textsf{Height of Ghost Packet $H_{GP}$} ({\bf \ref{setupCode}.23}).  Statement 9 of Lemma \ref{pseudo} guarantees that $H_{GP}$ will lie in the appropriate domain at all times.

	\item[] \textsf{Incoming Buffer's Value for Adjacent Node's Outgoing Buffer Height $H_{OUT}$} ({\bf \ref{setupCode}.24}).  This is initialized to 0 on ({\bf \ref{setupCode}.31}), and is only modified afterwards on line ({\bf \ref{routingRules}.11}), where it is set to be one of the values sent on ({\bf \ref{routingRules}.05}) by the adjacent node, or $\bot$ in case no value was received.  Since the value sent on ({\bf \ref{routingRules}.05}) (either $H_{OUT}$ or $H_{FP}$) will always be $\bot$ or a number between 1 and $2n$ (see domain argument above for an outgoing buffer's height of flagged packet variable $H_{FP}$), $H_{OUT}$ has the required domain.

	\item[] \textsf{Incoming Buffer's Value for Adjacent Node's Status Bit $sb_{OUT}$} ({\bf \ref{setupCode}.25}).  This is initialized to 0 on ({\bf \ref{setupCode}.31}), and is only modified afterwards on lines ({\bf \ref{routingRules}.10}) and ({\bf \ref{routingRules}.11}).  Both changes assign $sb_{OUT}$ to `0' or `1', as required.
	
	\item[] \textsf{Incoming Buffer's Value for Adjacent Node's Flagged Round Index $FR$} ({\bf \ref{setupCode}.26}).  This is initialized to $\bot$ on ({\bf \ref{setupCode}.30}), and is only modified afterwards on lines ({\bf \ref{routingRules}.10-11}) and ({\bf \ref{routingRules2}.43}).  Each of these times, $FR$ is either set to the value sent by the adjacent node, or $\bot$ in the case nothing was received.  Since the values sent on ({\bf \ref{routingRules}.05}) and ({\bf \ref{routingRules2}.45}) live in $[0..6D] \cup \bot$ (see argument above for an outgoing buffer's variable $FR$ living in the appropriate domain), so does $FR$.

	\item[] \textsf{Sender's Count of Packets Inserted $\kappa$} ({\bf \ref{setupCode2}.37}).  We want to argue that at all times, $\kappa$ corresponds to the number of packets (corresponding to the current codeword) that the sender has {\it knowingly} inserted.  Lines ({\bf \ref{setupCode2}.39}) and ({\bf \ref{routingRules2}.68}) guarantee that $\kappa=0$ at the outset of any transmission.  The only other place $\kappa$ is modified is ({\bf \ref{routingRules2}.31}) where it is incremented by one, so we must argue that ({\bf \ref{routingRules2}.31}) is reached exactly once for every packet the sender knowingly inserts.  By ``knowingly'' inserting a packet, we means that the sender has received verification that the adjacent node has received and stored the packet, and hence the sender can delete the packet.
	
	Suppose that in some round $\mathtt{t}$, the sender sends a packet $p$ as on ({\bf \ref{routingRules2}.41}).  By Claim \ref{obsBelowHere} below, the sender will continue to try and send this packet to its neighbor until he receives confirmation of receipt.  There are two things to show: 1) If the sender does not receive confirmation of receipt, then $\kappa$ is never incremented as on ({\bf \ref{routingRules2}.31}), and 2) If the sender {\it does} receive confirmation of receipt, then $\kappa$ is incremented {\it exactly once}.  By ``receiveing confirmation of receipt,'' we mean that line ({\bf \ref{routingRules2}.30}) is satisfied in some round $\mathtt{t}'$ when the sender's value for $\tilde{p}$ equals the packet $p$ sent in round $\mathtt{t}$ (see Definition \ref{confRec} below).  Clearly, 1) will be true since ({\bf \ref{routingRules2}.31}) will never be reached if ({\bf \ref{routingRules2}.30}) is never satisfied.  For 2), suppose that in some later round $\mathtt{t}' > \mathtt{t}$ the sender gets confirmation of receipt for $p$.  Clearly line ({\bf \ref{routingRules2}.31}) is reached this round, and $\kappa$ is incremented by one there.  We must show $\kappa$ will not be incremented due to $p$ ever again.  However, $p$ will be deleted on line ({\bf \ref{routingRules2}.32}) of round $\mathtt{t}'$, and therefore this packet can cause the sender to reach ({\bf \ref{routingRules2}.31}) at most once.  Thus, at all times $\kappa$ corresponds to the number of packets (corresponding to the current codeword) that the sender has {\it knowingly} inserted, as desired.  Since each codeword has $D$ packets, the domain for $\kappa$ is as required.

	\item[] \textsf{Receiver's Storage Buffer $I_R$} ({\bf \ref{setupCode2}.40}).  Each entry of $I_R$ is initialized to $\bot$ on ({\bf \ref{setupCode2}.43}), after which it is only modified on lines ({\bf \ref{reShuffleRules}.101}) and ({\bf \ref{routingRules2}.66}).  The latter resets $I_R$, while the former sets entry $\kappa$ of $I_R$ to the packet in $IN[1]$.  We show below that $\kappa$ will always accurately represent the number of current codeword packets the receiver has received, and hence will be a value between 0 and $D$.  It remains to show that $IN[1]$ will always hold a packet when ({\bf \ref{reShuffleRules}.101}) is reached.  We use Claim \ref{hIN} below which states that for the receiver, anytime $H_{IN} >0$, $H_{GP} = \bot$.  Therefore, whenever ({\bf \ref{reShuffleRules}.99}) is satisfied, Statement 1 of Lemma \ref{pseudo} (together with the argument that $\mathsf{IN}$ has the appropriate domain) state that $IN[1]$ will hold a packet, as required.

	\item[] \textsf{Receiver's Number of Packets Received $\kappa$} ({\bf \ref{setupCode2}.41}).  We want to show that $\kappa$ always equals the number of packets corresponding to the current codeword the receiver has received so far.  Lines ({\bf \ref{setupCode2}.42}) and ({\bf \ref{routingRules2}.66}) guarantee that $\kappa=0$ at the outset of any transmission.  The only other place $\kappa$ is modified is ({\bf \ref{reShuffleRules}.101}) where it is incremented by one, so we must argue that ({\bf \ref{reShuffleRules}.101}) is reached exactly once for every packet (corresponding to the current codeword) that the receiver receives.  By Statement 1 of Lemma \ref{pseudo} and Claim \ref{hIN} below, anytime ({\bf \ref{reShuffleRules}.101}) is reached, $\mathsf{IN}[1]$ necessarily stores a packet.  This packet is added to $I_R$ on ({\bf \ref{reShuffleRules}.101}) and then is promptly deleted from $\mathsf{IN}$ on ({\bf \ref{reShuffleRules}.102}).  By Claim \ref{packetProliferation3}, the receiver will never enter ({\bf \ref{reShuffleRules}.100}) twice due to the same packet, and hence ({\bf \ref{reShuffleRules}.101}) is reached exactly once for every distinct packet corresponding to the current codeword (see comments on {\bf \ref{reShuffleRules}.100} and {\bf \ref{reShuffleRules}.104}).  Therefore, $\kappa$ always equals the number of packets corresponding to the current codeword the receiver has received so far, as desired.  Since there are $D$ packets per codeword, $\kappa \in [0..D]$, as required.
	\end{itemize}
\vspace{-22pt}
\end{proof}
\begin{claim} \label{hIN} For any of the receiver's buffers $IN$, $H_{IN} = 0$ at the start of every round.  Also, anytime $H_{IN} > 0$, $H_{GP} = \bot$.
\end{claim}
\begin{proof} $H = H_{IN}$ is set to 0 at the outset of the protocol ({\bf \ref{setupCode}.31}).  The first statement follows immediately from line ({\bf \ref{reShuffleRules}.102}), where each of the receiver's incoming buffers $IN$ have $H_{IN}$ reset to zero during the re-shuffle phase of every round.  For the second statement, we will show that whenever $H$ changes value from 0 in any round $\mathtt{t}$, that $H_{GP}$ will be set to $\bot$ at the same time, and neither will change value until the end of the round when $H$ will be reset to zero during re-shuffling.  In particular, the only place $H$ can change from zero is on ({\bf \ref{routingRules2}.53}).  Suppose ({\bf \ref{routingRules2}.53}) is reached in some round $\mathtt{t}$, changing $H$ from zero to 1, and also changing $H_{GP}$ to $\bot$.  Looking at the pseudo-code, neither $H$ nor $H_{GP}$ can change value until line ({\bf \ref{reShuffleRules}.102}), where $H$ is reset to zero.  Therefore, $H$ can only be non-zero between lines ({\bf \ref{routingRules2}.53}) and ({\bf \ref{routingRules}.21}) (when {\em Receiver Re-Shuffle} is called) of a given round, and at these times $H_{GP}$ is always equal to $\bot$.
\end{proof}
\begin{claim} \label{HFPFR} Let $OUT$ be any outgoing buffer, and $H_{FP}$, $FR$, and $sb$ denote the height of it flagged packet, round the packet was flagged, and status bit, respectively (see ({\bf \ref{setupCode}.10, \ref{setupCode}.12, \ref{setupCode}.14})).  Then $H_{FP} = \bot \Leftrightarrow FR = \bot$.  Also, anytime $OUT$ has no flagged packets (i.e.\ $H_{FP} = \bot$), $OUT$ has normal status (i.e.\ $sb=0$).
\end{claim}
\begin{proof} The first statement is true at the outset of the protocol ({\bf \ref{setupCode}.34}), so it will be enough to make sure that anytime $H_{FP}$ or $FR$ changes value from $\bot$ to non-$\bot$ (or vice-versa), the other one also changes.  Examining the pseudo-code, these changes occur only on lines ({\bf \ref{routingRules2}.33}), ({\bf \ref{routingRules2}.38}), and ({\bf \ref{routingRules2}.62}), where it is clear $H_{FP}$ takes on a non-$\bot$ (respectively $\bot$) value if and only if $FR$ does.

The second statement is true at the outset of the protocol ({\bf \ref{setupCode}.34-35}).  So it is enough to show: 1) anytime $H_{FP}$ is set to $\bot$, $sb$ is equal to zero, and 2) anytime $sb$ changes to one, $H_{FP} \neq \bot$.  The former is true since anytime $H_{FP}$ changes to $\bot$, $sb$ is set to zero on the same line (({\bf \ref{routingRules2}.33}) and ({\bf \ref{routingRules2}.62})), while the latter is true since $sb$ only changes to one on ({\bf \ref{routingRules2}.28}), which can only be reached if $FR \neq \bot$ ({\bf \ref{routingRules2}.27}), which by the first statement of this claim implies $H_{FP} \neq \bot$.
\end{proof}
\begin{claim} \label{stupid} \hspace*{\fill}
    \begin{enumerate}\setlength{\itemsep}{1pt} \setlength{\parskip}{0pt} \setlength{\parsep}{0pt}
    \item Anytime $sb_{OUT}$ is equal to 1 when {\it Create Flagged Packet} is called on line ({\bf \ref{routingRules}.15}), $H_{FP} \neq \bot$.

    \item Anytime {\it Send Packet} is called on line ({\bf \ref{routingRules}.17}), the flagged packet has height at least one (i.e.\ $H_{FP}$ is at least one anytime {\it Send Packet} is called).
    \end{enumerate}
\end{claim}
\begin{proof}
We prove the 2$^{nd}$ statement by separating the proof into the following two cases.
        \begin{itemize}\setlength{\itemsep}{6pt} \setlength{\parskip}{0pt} \setlength{\parsep}{0pt}
        \item[] \textsf{Case 1: $sb_{OUT} = 0$ at the start of Stage 2}.  Since {\it Send Packet} is called, the conditional statement on line ({\bf \ref{routingRules}.16}) was satisfied.  Therefore, since we are in the case $sb_{OUT} =0$ on that line, then $H_{OUT}>H_{IN}$.  Tracing $H_{IN}$ backwards, it was received on line ({\bf \ref{routingRules}.06}) and represents the value of $H_{IN}$ that was sent on line ({\bf \ref{routingRules}.09}).  Using the induction hypothesis applied to Statement 9 of Lemma \ref{pseudo}, $H_{IN} \geq 0$ and hence the value of $H_{OUT}$ on ({\bf \ref{routingRules}.16}) must be at least one.  Since $H_{OUT}$ and $H_{IN}$ cannot change between lines ({\bf \ref{routingRules}.15}) and ({\bf \ref{routingRules}.16}) of any round, when {\it Create Flagged Packet} was called, it was still true that $sb_{OUT}=0$ and $H_{OUT} > H_{IN} \geq 0$.  Therefore, line ({\bf \ref{routingRules2}.37}) will be satisfied and ({\bf \ref{routingRules2}.38}) will set $H_{FP} = H_{OUT} \geq 1$ as required.

        \item[] \textsf{Case 2: $sb_{OUT} = 1$ at the start of Stage 2}.  Let $\mathtt{t}$ denote some round where $sb_{OUT} = 1$ at the start of Stage 2.  Our strategy will be to find the most recent round that $sb_{OUT}$ switched from 0 to 1, and argue that the value that $H_{FP}$ acquired in that round has not changed.  So let $\mathtt{t_0}+1$ denote the most recent round that $sb_{OUT}$ had the value 0 at any stage of the round.  We argue that $sb_{OUT}=1$ by the end of $\mathtt{t}_0+1$, and $sb_{OUT}=0$ at the start of Stage 2 of round $\mathtt{t}_0$ (the round {\it before} $\mathtt{t}_0+1$) as follows:
		\begin{itemize}
		\item If $sb_{OUT}$ equals 0 by the end of round $\mathtt{t_0}+1$, then it will at the start of round $\mathtt{t_0}+2$, contradicting the choice of $\mathtt{t_0}+1$.

		\item If $sb_{OUT} =1$ at the start of Stage 2 of round $\mathtt{t}_0$, then $sb_{OUT}$ must have changed its value to 0 sometime between Stage 2 of round $\mathtt{t}_0$ and the end of round $\mathtt{t}_0+1$ (since $sb_{OUT}=0$ at some point of round $\mathtt{t}_0+1$ by definition).  This can only happen on line ({\bf \ref{routingRules2}.33}) inside the {\it Reset Outgoing Variables} function of round $\mathtt{t}_0+1$ (this is the only place that $sb_{OUT}$ can be set to zero).  However, since $sb_{OUT}$ cannot change between the time that {\it Reset Outgoing Variables} is called on line ({\bf \ref{routingRules}.07}) and the end of the round, it must be that $sb_{OUT}$ was equal to zero at the start of round $\mathtt{t}_0+2$, contradicting the choice of $\mathtt{t}_0+1$.
		\end{itemize}
	Now since $sb_{OUT}=0$ at the start of round $\mathtt{t_0}+1$ (it cannot change between Stage 2 of $\mathtt{t}_0$ and the start of $\mathtt{t}_0+1$), and $sb_{OUT}=1$ by the end of $mathtt{t}_0+1$, it must have changed on line ({\bf \ref{routingRules2}.28}) of round $\mathtt{t}_0+1$ (this is the only line that sets $sb_{OUT}$ to 1).  In particular, the conditional statements on lines ({\bf \ref{routingRules2}.25}) and ({\bf \ref{routingRules2}.27}) must have been satisfied, and so $d$ was equal to 1 on line ({\bf \ref{routingRules2}.25}) of round $\mathtt{t}_0+1$.  Since $d$ is reset to zero during Stage 1 of every round ({\bf \ref{routingRules2}.26}), it must be that $d$ was switched from 0 to 1 on line ({\bf \ref{routingRules2}.40}) of round $\mathtt{t}_0$ (this is the only place $d$ is set to one).  Thus, we have that {\it Send Packet} was called on line ({\bf \ref{routingRules}.17}) of round $\mathtt{t}_0$.  We are now back in \textsf{Case 1} above (but for round $\mathtt{t}_0$ instead of $\mathtt{t}$), and thus $H_{FP}$ was set to a value of at least 1 on line ({\bf \ref{routingRules2}.38}) of round $\mathtt{t}_0$.  It remains to argue that $H_{FP}$ does not decrease in value between round $\mathtt{t}_0$ and line ({\bf \ref{routingRules}.17}) of round $\mathtt{t}$.  But $H_{FP}$ can only change value on lines ({\bf \ref{routingRules2}.33}), ({\bf \ref{routingRules2}.35}), and ({\bf \ref{routingRules2}.38}).  For round $\mathtt{t}_0$, the former two of these lines have both passed when the latter is called (setting $H_{FP} \geq 1$ as in Case 1).  Meanwhile, between $\mathtt{t}_0+1$ and $\mathtt{t}$, we know that ({\bf \ref{routingRules2}.33}) and ({\bf \ref{routingRules2}.38}) cannot be reached, as this would imply the value of $sb_{OUT}$ is zero sometime after $\mathtt{t}_0+1$, contradicting the choice of $\mathtt{t}_0+1$.  The only other place $H_{FP}$ can change is ({\bf \ref{routingRules2}.35}), which can only {\it increase} $H_{FP}$.  Thus in any case, $\bot \neq H_{FP} \geq 1$ when {\it Send Packet} is called on ({\bf \ref{routingRules}.17}) of round $\mathtt{t}$.
        \end{itemize}
The proof of the 1$^{st}$ statement follows from the proof given in \textsf{Case 2} above.
\end{proof}
\begin{defn} \label{confRec} We will say that an outgoing buffer gets {\bf \em confirmation of receipt} for a packet $p$ that it sent across its adjacent edge whenever line ({\bf \ref{routingRules2}.30}) (alternatively line ({\bf \ref{routingRules2M}.46}) for the node-controlling $+$ edge-controlling protocol of Sections \ref{aP}-\ref{mBareBones}) is reached and satisfied and the packet subsequently deleted on ({\bf \ref{routingRules2}.32}) (respectively ({\bf \ref{routingRules2M}.50})) is (a copy of) $p$.
\end{defn}
\begin{claim} \label{obsBelowHere} Suppose (an instance of) a packet $p$ is accepted by node $B$ in round $\mathtt{t}$ (using the definition of ``accepted'' from Definition \ref{receive}).  Then:
	\begin{enumerate}
	\item Let $\mathtt{t}'$ be the first round after\footnote{The Claim remains valid even if $\mathtt{t}'$ is a round in a {\it different} transmission than $\mathtt{t}$.} $\mathtt{t}$ in which $B$ attempts to send (a copy of) this packet across any outgoing edge.  Then the corresponding outgoing buffer $\mathsf{OUT}$ of $B$ will necessarily have normal status at the start of Stage 2 of $\mathtt{t}'$.
	\item If $B$ fails to get confirmation of receipt for the packet in the following round (i.e.\ either $RR$ is not received on ({\bf \ref{routingRules}.06}) of round $\mathtt{t}'+1$, or it is received but $RR < FR$), then $\mathsf{OUT}$ enters problem status as on ({\bf \ref{routingRules2}.28}) of round $\mathtt{t}'+1$.  $\mathsf{OUT}$ will remain in problem status until the end of the transmission or until the round in which it gets confirmation of receipt (i.e.\ until $RR$ is received as on ({\bf \ref{routingRules}.06}) with $RR \geq \mathtt{t}'$).
	\item From the time $p$ is first flagged as on ({\bf \ref{routingRules2}.38}) of round $\mathtt{t}'$ through the time $B$ does get confirmation of receipt (or through the end of the transmission, whichever comes first), $B$ will not have any other flagged packets, i.e.\ $\tilde{p}, \mathsf{OUT}[H_{FP}] = p$ and $FR=\mathtt{t}'$.
	\end{enumerate}
\end{claim}
\begin{proof} We prove Statement 1 by contradiction.  Let $\mathtt{t}'$ denote the first round after $\mathtt{t}$ in which $B$ attempts to send (a copy of) $p$ across an edge $E(B,C)$, i.e.\ $\mathtt{t}'$ is the first round after $\mathtt{t}$ that {\it Send Packet} is called by $B$'s outgoing buffer $\mathsf{OUT}$ such that the $\tilde{p}$ that appears on line ({\bf \ref{routingRules2}.38}) of that round corresponds to $p$.  For the sake of contradiction, assume that $sb_{OUT}=1$ at the start of Stage 2 of round $\mathtt{t}'$.  Since $sb_{OUT}$ cannot change between the start of Stage 2 and the time that {\it Create Flagged Packet} is called on line ({\bf \ref{routingRules}.15}), we must have that $sb_{OUT}=1$ on line ({\bf \ref{routingRules2}.37}) of round $\mathtt{t}'$, and hence ({\bf \ref{routingRules2}.38}) is not reached that round.  In particular, when {\it Send Packet} is called on line ({\bf \ref{routingRules}.17}) (as it must be by the fact that $p$ was sent during round $\mathtt{t}'$), the packet $\tilde{p}$ that is sent (which is $p$) was set in some previous round.  Let $\widetilde{\mathtt{t}}$ denote the most recent round for which $\tilde{p}$ was set to $p$ as on ({\bf \ref{routingRules2}.38}) (this is the only line which sets $\tilde{p}$).  Then by assumption $\widetilde{\mathtt{t}} < \mathtt{t}'$, and $\mathsf{OUT}$ had normal status at the start of Stage 2 of round $\widetilde{\mathtt{t}}$ (in order for ({\bf \ref{routingRules2}.38}) to be reached).  Since $\mathsf{OUT}$ had normal status at the start of Stage 2 of round $\widetilde{\mathtt{t}}$, but by assumption $\mathsf{OUT}$ had problem status at the start of Stage 2 of round $\mathtt{t}'$, let $\widehat{\mathtt{t}}$ denote the first round such that $\widetilde{\mathtt{t}} < \widehat{\mathtt{t}} \leq \mathtt{t}'$ and such that $\mathsf{OUT}$ had problem status at the start of Stage 2 of $\widehat{\mathtt{t}}$.  Since the only place $\mathsf{OUT}$ switches status from normal to problem is on ({\bf \ref{routingRules2}.28}), this line must have been reached in round $\widehat{\mathtt{t}}$.  In particular, this implies that ({\bf \ref{routingRules2}.25}) was satisfied in round $\widehat{\mathtt{t}}$, which in turn implies that {\it Send Packet} was called in round $\widehat{\mathtt{t}}-1$ (since $d$ is re-set to zero at the end of Stage 1 of every round as on ({\bf \ref{routingRules2}.26})).  But this is a contradiction, since $\widetilde{\mathtt{t}} \leq \widehat{\mathtt{t}}-1 < \mathtt{t}'$, and so $p=\tilde{p}$ was sent in a round before $\mathtt{t}'$, contradicting the choice of $\mathtt{t}'$.

For Statement 2, since $B$ sent $p$ in round $\mathtt{t}'$ and $\mathsf{OUT}$ had normal status at the Start of Stage 2 of this round, we have that $H_{OUT} > H_{IN}$ on line ({\bf \ref{routingRules}.16}) (so that {\it Send Packet} could be called).  Since $sb_{OUT}, H_{OUT}$, and $H_{IN}$ cannot change between ({\bf \ref{routingRules}.15}) and ({\bf \ref{routingRules}.17}) of any round, $FR$ is set to $\mathtt{t}'$ on ({\bf \ref{routingRules2}.38}) of round $\mathtt{t}'$.  Also, $d=1$ after the call to {\it Send Packet} of round $\mathtt{t}'$ ({\bf \ref{routingRules2}.40}).  Notice that neither $FR$ nor $d$ can change value between the call to {\it Create Flagged Packet} in round $\mathtt{t}'$ and the call to {\it Reset Outgoing Variables} in the following round.  Therefore, if $B$ does not receive $RR$ or if $RR < FR=\mathtt{t}'$ when {\it Reset Outgoing Variables} is called in round $\mathtt{t}'+1$, then ({\bf \ref{routingRules2}.25}) and ({\bf \ref{routingRules2}.27}) will be satisfied, and hence $\mathsf{OUT}$ will enter problem status in round $\mathtt{t}'+1$.  That $\mathsf{OUT}$ remains in problem status until the end of the transmission or until the round in which $RR$ is received on ({\bf \ref{routingRules}.06}) with $RR \geq \mathtt{t}'$ now follows from the following subclaim. (Warning: the following subclaim switches notation.  In particular, to apply the subclaim here, replace $(\mathtt{t}, \mathtt{t}_0)$ of the subclaim with $(\mathtt{t}'+1, \mathtt{t}')$.)
	\begin{enumerate}
	\item[] {\bf Subclaim.}  Suppose that at the start of Stage 2 of some round $\mathtt{t}$, an outgoing buffer $\mathsf{OUT}$ has problem status and $\bot \neq FR = \mathtt{t}_0$.  Then $\mathsf{OUT}$ will remain in problem status until the end of the transmission or until the round in which $RR$ is received on ({\bf \ref{routingRules}.06}) with $RR \geq \mathtt{t}_0$.

	\item[] {\it Proof.}  $\mathsf{OUT}$ will certainly return to normal status by the end of the transmission ({\bf \ref{routingRules2}.62}), in which case there is nothing to show.  So suppose that $\mathtt{t}' > \mathtt{t}$ is such that $\mathsf{OUT}$ first returns to normal status (in the same transmission as $\mathtt{t}$) as on ({\bf \ref{routingRules2}.33}) of round $\mathtt{t}'$.  In particular, lines ({\bf \ref{routingRules2}.29}) and ({\bf \ref{routingRules2}.30}) were both satisfied, so $\mathsf{OUT}$ must have received $RR$ on ({\bf \ref{routingRules}.06}) earlier in round $\mathtt{t}'$, with $RR \geq FR$.  If the value of $FR$ on line ({\bf \ref{routingRules2}.30}) equals $\mathtt{t}_0$, then the proof is complete.  So we show by contradiction that this must be the case.

	Assume for the sake of contradiction that $FR \neq \mathtt{t}_0$ on line ({\bf \ref{routingRules2}.30}) of round $\mathtt{t}'$.  Since $FR$ was equal to $\mathtt{t}_0$ at the start of Stage 2 of round $\mathtt{t}$ by hypothesis, $FR$ must have changed at some point between Stage 2 of round $\mathtt{t}$ and round $\mathtt{t}'$.  Notice that between these rounds, $FR$ can only change values on lines ({\bf \ref{routingRules2}.33}) and ({\bf \ref{routingRules2}.38}).  Let $\mathtt{t}''$ denote the first round between $\mathtt{t}$ and $\mathtt{t}'$ such that one of these two lines is reached.  Note that $\mathtt{t}'' > \mathtt{t}$, since ({\bf \ref{routingRules2}.33}) already passed by the start of Stage 2 (which is when the subclaim asserts $FR = \mathtt{t}_0$), and ({\bf \ref{routingRules2}.38}) cannot be reached in round $\mathtt{t}$ since $\mathsf{OUT}$ has problem status when ({\bf \ref{routingRules2}.37}) of round $\mathtt{t}$ is reached (by hypothesis).
		\begin{itemize}
		\item Suppose $FR$ is {\it first} changed from $FR=\mathtt{t}_0$ on ({\bf \ref{routingRules2}.33}) of round $\mathtt{t}''$.  First note that because ({\bf \ref{routingRules2}.33}) is the {\it first} time $FR$ changes its value from $\mathtt{t}_0$, it must be the case that $FR$ was still equal to $\mathtt{t}_0$ on ({\bf \ref{routingRules2}.30}) earlier in round $\mathtt{t}''$.  Also, since ({\bf \ref{routingRules2}.33}) is reached in round $\mathtt{t}''$, $\mathsf{OUT}$ returns to normal status.  Since $\mathtt{t}'$ was defined to be the first round after $\mathtt{t}$ for which this happens, we must have that $\mathtt{t}'' \geq \mathtt{t}'$.  But by construction $\mathtt{t}'' \leq \mathtt{t}'$, so we must have that $\mathtt{t}'' = \mathtt{t}'$.  However, this is a contradiction, because our assumption is that $FR \neq \mathtt{t}_0$ on line ({\bf \ref{routingRules2}.30}) of round $\mathtt{t}'=\mathtt{t}''$, but as noted in the second sentence of this paragraph, we are in the case that $FR = \mathtt{t}_0$ on line ({\bf \ref{routingRules2}.30}) of round $\mathtt{t}''$.

		\item Suppose $FR$ is {\it first} changed from $FR=\mathtt{t}_0$ on ({\bf \ref{routingRules2}.38}) of round $\mathtt{t}''$.  Then ({\bf \ref{routingRules2}.37}) must have been satisfied, and thus $\mathsf{OUT}$ had normal status when {\it Create Flagged Packet} was called in round $\mathtt{t}''$.  Since $\mathsf{OUT}$ had problem status at the start of Stage 2 of round $\mathtt{t}$ (by hypothesis), the status must have switched to normal at some point between $\mathtt{t}$ and $\mathtt{t}''$, which can only happen on ({\bf \ref{routingRules2}.33}).  But if ({\bf \ref{routingRules2}.33}) is reached, then $FR$ will be set to $\bot$ on this line, which contradicts the fact that $FR$ was first changed from $FR=\mathtt{t}_0$ on ({\bf \ref{routingRules2}.38}) of round $\mathtt{t}''$.
		\end{itemize}
	This completes the proof of the subclaim.
	\end{enumerate}
\noindent For the third Statement, first note that $\mathsf{OUT}[H_{FP}] = p$ as of line ({\bf \ref{routingRules2}.38}) of round $\mathtt{t}'$.  This is the case since $sb_{OUT}=0$ on line ({\bf \ref{routingRules}.12}) (by Statement 1 of this claim), and then the fact that {\it Send Packet} is called in round $\mathtt{t}'$ means $H_{OUT} > H_{IN}$ on ({\bf \ref{routingRules}.16}), and therefore since none of these values change between ({\bf \ref{routingRules}.12}) and ({\bf \ref{routingRules}.16}), ({\bf \ref{routingRules2}.37}) will be satisfied in round $\mathtt{t}'$.  Therefore, we will track all changes to $\mathsf{OUT}$ and $H_{FP}$ from Stage 2 of round $\mathtt{t}'$ through the time $p$ is deleted from $\mathsf{OUT}$ as on ({\bf \ref{routingRules2}.32-33}) of some later round\footnote{Or through the end of the transmission, whichever occurs first.}, and show that none of these changes will alter the fact that $\mathsf{OUT}[H_{FP}] = p$.  Notice that (before the end of the transmission) $H_{FP}$ only changes value on lines ({\bf \ref{routingRules2}.33}), ({\bf \ref{routingRules2}.35}), and ({\bf \ref{routingRules2}.38}); while $\mathsf{OUT}$ only changes values on lines ({\bf \ref{routingRules2}.32}), ({\bf \ref{routingRules2}.35}), and ({\bf \ref{reShuffleRules}.89-90}).  Clearly the changes to each value on ({\bf \ref{routingRules2}.35}) will preserve $\mathsf{OUT}[H_{FP}] = p$, so it is enough to check the other changes.  Notice that ({\bf \ref{routingRules2}.32}) is reached if and only if ({\bf \ref{routingRules2}.33}) is reached, which by Statement 2 of this claim does not happen until $\mathsf{OUT}$ gets confirmation of receipt that $p$ was successfully received by $B$'s neighbor, and therefore these changes also do not threaten the validity of Statement 3.  The change to $H_{FP}$ as on ({\bf \ref{routingRules2}.38}) can only occur if ({\bf \ref{routingRules2}.37}) is satisfied, i.e.\ only if $\mathsf{OUT}$ has normal status, and thus again Statement 2 of this claim says this cannot happen until $\mathsf{OUT}$ gets confirmation of receipt that $p$ was successfully received by $B$'s neighbor.  Finally, lines ({\bf \ref{reShuffleRules}.89-90}) will preserve $\mathsf{OUT}[H_{FP}] = p$ by Statement 11 of Lemma \ref{pseudo}.

That $FR = \mathtt{t}'$ from ({\bf \ref{routingRules2}.38}) of $\mathtt{t}'$ through the time $B$ gets confirmation of receipt for $p$ was proven in the subclaim above.  Also, $\tilde{p}$ can only change on ({\bf \ref{routingRules2}.33}) or ({\bf \ref{routingRules2}.38}), which we already proved (in the proof of the subclaim above) are not reached.
\end{proof}
\begin{cor} \label{packetProliferationA2} At any time, an outgoing buffer has at most one flagged packet.\end{cor}
\begin{proof} This follows immediately from Statement 3 of Lemma \ref{obsBelowHere}.
\end{proof}
\begin{claim} \label{last} For any outgoing buffer $\mathsf{OUT}$, if at any time its Flagged Round value $FR$ is equal to $\mathtt{t}$, then $\mathsf{OUT}$ necessarily called {\it Send Packet} on line ({\bf \ref{routingRules}.17}) of round $\mathtt{t}$.
\end{claim}
\begin{proof} Suppose that at some point in time, $FR$ is set to $\mathtt{t}$.  Notice that the only place $FR$ assumes non-$\bot$ values is on ({\bf \ref{routingRules2}.38}), and therefore line ({\bf \ref{routingRules2}.37}) must have been satisfied in round $\mathtt{t}$.  Since the values for $sb_{OUT}, H_{OUT}$, and $H_{IN}$ cannot change between lines ({\bf \ref{routingRules}.15}) and ({\bf \ref{routingRules}.16}), the statement on ({\bf \ref{routingRules}.16}) will also be satisfied in round $\mathtt{t}$, and consequently {\it Send Packet} will be reached in $\mathtt{t}$.
\end{proof}
\begin{lemma} \label{subclaim6} Suppose that $sb_{OUT}=1$ when line ({\bf \ref{routingRules2}.47}) is reached in round $\mathtt{t}$ on an edge linking buffers $\mathsf{OUT}$ and $\mathsf{IN}$.  Further suppose that $\mathsf{IN}$ does receive the communication $(p,FR)$ from $\mathsf{OUT}$ on line ({\bf \ref{routingRules2}.43}) of $\mathtt{t}$.  Also, let $\mathtt{t}_0$ denote the round described by $FR$, let $h$ denote the height of the packet in $\mathsf{OUT}$ in round $\mathtt{t}_0$, and let $h'$ denote the height of $\mathsf{IN}$ at the start of round $\mathtt{t}_0$.  Then the following are true:
    \begin{enumerate}\setlength{\itemsep}{1pt} \setlength{\parskip}{0pt} \setlength{\parsep}{0pt}
    \item $\mathtt{t}_0$ is well-defined (i.e.\ $FR \neq \bot$ and $FR \leq \mathtt{t}$).

    \item $h > h'$.

    \item $\mathsf{OUT}$ sent $p$ to $\mathsf{IN}$ on line ({\bf \ref{routingRules2}.41}) of round $\mathtt{t}_0$.  Furthermore, the height of $p$ in $\mathsf{OUT}$ when it is sent on line ({\bf \ref{routingRules2}.41}) of round $\mathtt{t}$ is greater than or equal to $h$.

    \item If the condition statement on line ({\bf \ref{routingRules2}.51}) of round $\mathtt{t}$ is satisfied, then the value of $H_{GP}$ when this line is entered, which corresponds to the height in $\mathsf{IN}$ that $p$ assumes when it is inserted, satisfies: $\bot \neq H_{GP} \leq h'+1 \leq 2n$.

    \item If the condition statement on line ({\bf \ref{routingRules2}.51}) of round $\mathtt{t}$ is satisfied, then $H_{IN}$ was less than $2n$ at the start of all rounds between $\mathtt{t}_0$ and $\mathtt{t}$.
    \end{enumerate}
\end{lemma}
\noindent {\it Proof of Lemma \ref{subclaim6}.} 
We make a series of Subclaims to prove the 5 statements of the Lemma.
    \begin{itemize}\setlength{\itemsep}{5pt} \setlength{\parskip}{0pt} \setlength{\parsep}{0pt}
    \item[] {\bf Subclaim 1.} The value of $FR$ that is sent on ({\bf \ref{routingRules2}.41}) of round $\mathtt{t}$ is not $\bot$.
    \item[]{\it Proof.} Since ({\bf \ref{routingRules2}.41}) is reached, {\it Send Packet} was called on ({\bf \ref{routingRules}.17}).  By Statement 2 of Claim \ref{stupid}, we have that $H_{FP} \geq 1$ when {\it Send Packet} is called, and in particular $H_{FP}\neq \bot$ on line ({\bf \ref{routingRules}.17}).  Since $H_{FP}$ cannot change between ({\bf \ref{routingRules}.17}) and ({\bf \ref{routingRules2}.41}), we have that $H_{FP} \neq \bot$ on ({\bf \ref{routingRules2}.41}), and hence $FR \neq \bot$ on this line (Claim \ref{HFPFR}).\vspace{.3cm}

    \item[]{\bf Subclaim 2.} $\mathtt{t}_0$ is well-defined (i.e.\ $\bot \neq \mathtt{t}_0 \leq \mathtt{t}$).
    \item[]{\it Proof.} By the definition of $\mathtt{t}_0$ and Subclaim 1, $\mathtt{t}_0 \neq \bot$.  Also, by looking at the three places that $FR$ changes values (({\bf \ref{routingRules2}.33}), ({\bf \ref{routingRules2}.38}), and ({\bf \ref{routingRules2}.62})), it is clear that $FR$ will always be less than or equal to the current round index.\vspace{.3cm}

    \item[]{\bf Subclaim 3.} $\mathtt{t} > \mathtt{t}_0$.
    \item[]{\it Proof.}  That $\mathtt{t} \geq \mathtt{t}_0$ is immediate ($FR$ is reset to $\bot$ at the start of every transmission ({\bf \ref{setupCode}.34}) and ({\bf \ref{routingRules2}.62}), after which time the $FR$ can never attain a value {\it bigger} than the current round ({\bf \ref{routingRules2}.38})).  Therefore, we only have to show $\mathtt{t} \neq \mathtt{t}_0$.  For the sake of contradiction, suppose $\mathtt{t} = \mathtt{t}_0$.  By hypothesis, $sb_{OUT}=1$ when line ({\bf \ref{routingRules2}.47}) of round $\mathtt{t} = \mathtt{t}_0$ is reached.  Notice that $sb_{OUT}$ is reset to 0 on ({\bf \ref{routingRules}.10}) of round $\mathtt{t} = \mathtt{t}_0$, so the only way it can be `1' on ({\bf \ref{routingRules2}.47}) later that round is if it is set to one on ({\bf \ref{routingRules}.11}).  This can only happen if $H_{OUT}=\bot$ or $FR > RR$.  Since ({\bf \ref{routingRules2}.47}) is reached, ({\bf \ref{routingRules2}.44}) must have failed, and since $H_{OUT}$ does not change values between the time it is received on ({\bf \ref{routingRules}.11}) and ({\bf \ref{routingRules2}.44}), we have that $H_{OUT} \neq \bot$ on ({\bf \ref{routingRules}.11}).  Therefore, we must have that $FR > RR$ on ({\bf \ref{routingRules}.11}) of round $\mathtt{t}= \mathtt{t}_0$.

    \vspace{4pt}\hspace{.6cm}Notice the value for $FR$ here comes from the value sent by $\mathsf{OUT}$ on ({\bf \ref{routingRules}.05}), and this happens {\it before} line ({\bf \ref{routingRules2}.38}) has been reached in round $\mathtt{t} = \mathtt{t}_0$.  Therefore, the value of $FR$ received on ({\bf \ref{routingRules}.11}) obeys $FR < \mathtt{t} = \mathtt{t}_0$ (as noted above, $FR$ can never attain a value {\it bigger} than the current round).  Since $RR < FR$, line ({\bf \ref{routingRules2}.30}) {\it cannot} have been satisfied since the time $FR$ was set to its current value (within a transmission, the values $RR$ assumes are strictly {\it increasing}, see ({\bf \ref{setupCode}.34}), ({\bf \ref{routingRules2}.53}), and ({\bf \ref{routingRules2}.64})).  Therefore, we may apply Claim \ref{last} and Claim \ref{obsBelowHere} to argue that $FR$ will not be changed on ({\bf \ref{routingRules2}.38}) of round $\mathtt{t} = \mathtt{t}_0$ (since $\mathsf{OUT}$ will have problem status), and consequently $FR$ will still be strictly smaller than $\mathtt{t} = \mathtt{t}_0$ when line ({\bf \ref{routingRules2}.41}) is reached of round $\mathtt{t}_0$.  This contradicts the definition of $\mathtt{t}_0$ as the value received on line ({\bf \ref{routingRules2}.43}) of round $\mathtt{t}$.\vspace{.3cm}

    \item[]{\bf Subclaim 4.} $\mathsf{OUT}$ had normal status at the start of Stage 2 of round $\mathtt{t}_0$.  For every round between Stage 2 of $\mathtt{t}_0+1$ through $\mathtt{t}-1$, $\mathsf{OUT}$ had problem status and $FR = \mathtt{t}_0$.
    \item[]{\it Proof.} By definition of $\mathtt{t}_0$, it equals the value of $FR$ that was received in round $\mathtt{t}$ on line ({\bf \ref{routingRules2}.43}), which in turn corresponds to the value of $FR$ that was sent on line ({\bf \ref{routingRules2}.41}).  Tracing the values of $FR$ backwards, we see that the only time/place $FR$ is set to a non-$\bot$ value (as we know it has by Subclaim 1) is on line ({\bf \ref{routingRules2}.38}), and this must have happened in round $\mathtt{t}_0$ since $FR = \mathtt{t}_0$ by definition of $\mathtt{t}_0$.  Therefore, in round $\mathtt{t}_0$, line ({\bf \ref{routingRules2}.38}) must have been reached when {\it Create Flagged Packet} was called on line ({\bf \ref{routingRules}.15}); so in particular $sb_{OUT}$ must have been zero on line ({\bf \ref{routingRules2}.37}) to have entered the conditional statement.  Since $sb_{OUT}$ cannot change between the start of Stage 2 and line ({\bf \ref{routingRules}.15}) (where {\it Create Flagged Packet} is called), it must have been zero at the start of Stage 2.  This proves the first part of the subclaim.  Now suppose there is a round $\mathtt{t}'$ between Stage 2 of $\mathtt{t}_0+1$ and $\mathtt{t}-1$ such that $sb_{OUT}=0$ at any time in that round (without loss of generality, let $\mathtt{t}'$ be the first such round).  Since $sb_{OUT}$ can only switch to zero on ({\bf \ref{routingRules2}.33}) inside the call to {\it Reset Outgoing Variables}, it must be that this line is reached in $\mathtt{t}'$, and hence $FR$ is also set to $\bot$ on this line.  Since $FR$ is only assigned non-$\bot$ values on ({\bf \ref{routingRules2}.38}), $FR$ can only assume values at least $\mathtt{t}' > \mathtt{t}_0$ after this point.  Thus, $FR$ will not ever be able to return to the value of $\mathtt{t}_0$, contradicting the fact that $FR = \mathtt{t}_0$ during round $\mathtt{t}$.  By the same reasoning, $FR$ can never change value from $\mathtt{t}_0$ between the rounds $\mathtt{t}_0$ and $\mathtt{t}$.

    \item[]{\bf Subclaim 5.} $\mathsf{OUT}$ attempted to send $p$ in round $\mathtt{t}_0$.
    \item[]{\it Proof.} By definition, $\mathtt{t}_0$ denotes the value of $FR$ during round $\mathtt{t}$.  Since $FR$ can only be set to $\mathtt{t}_0$ on ({\bf \ref{routingRules2}.38}) of round $\mathtt{t}_0$, this line must have been reached in $\mathtt{t}_0$.  In particular, line ({\bf \ref{routingRules2}.37}) was satisfied during the call to {\em Create Flagged Packet} of round $\mathtt{t}_0$, and hence $sb=0$ and $H > H_{IN}$ at that time.  Therefore, ({\bf \ref{routingRules}.16}) will be satisfied when it is reached in round $\mathtt{t}_0$, which implies {\em Send Packet} will be called on the following line.  The fact that it was the same packet $p$ that was sent in $\mathtt{t}_0$ as in $\mathtt{t}$ follows from Statement 3 of Lemma \ref{obsBelowHere}.\vspace{.3cm}

    \item[]{\bf Subclaim 6.} The height of $p$ in $\mathsf{OUT}$ when it is transferred in round $\mathtt{t}$ is greater than or equal to $h$.
    \item[]{\it Proof.} Subclaim 5 stated that $\mathsf{OUT}$ attempted to send $p$ in round $\mathtt{t}_0$, and Subclaim 4 stated that $\mathsf{OUT}$ had normal status at the start of $\mathtt{t}_0$.  Therefore, the packet which was sent in round $\mathtt{t}_0$ (which is $p$) was initialized inside the call to {\it Create Flagged Packet} on line ({\bf \ref{routingRules2}.38}).  By observing the code there, we see that $p$ is set to $\mathsf{OUT}[H]$, i.e.\ $p$ has height $H$ in round $\mathtt{t}_0$, and $H_{FP}$ is set to equal $H$ on this same line.  By Statement 3 of Claim \ref{obsBelowHere}, $p = \tilde{p}$ will remain the flagged packet through the start of round $\mathtt{t}$, and $\mathsf{OUT}[H_{FP}] = p$.  By Statement 11 of Lemma \ref{pseudo}, $H_{FP}$ will not change during any call to re-shuffle.  Indeed, since Subclaim 4 ensures that line ({\bf \ref{routingRules2}.38}) is never reached from $\mathtt{t}_0+1$ through the start of $\mathtt{t}$, the only place $H_{FP}$ can change value is on ({\bf \ref{routingRules2}.33}) or ({\bf \ref{routingRules2}.35}).  We know the former cannot happen between $\mathtt{t}_0+1$ and the start of $\mathtt{t}$, since this would imply $sb_{OUT}$ is re-set to zero on ({\bf \ref{routingRules2}.33}) of that round, contradicting Subclaim 4.  Therefore, $H_{FP}$ can only change values between $\mathtt{t}_0+1$ and the start of $\mathtt{t}$ as on ({\bf \ref{routingRules2}.35}), which can only {\it increase} $H_{FP}$.  Hence, from the time $H_{FP}$ is set to equal the height of $\mathsf{OUT}$ in round $\mathtt{t}_0$ as on ({\bf \ref{routingRules2}.38}) (which by definition is $h$), $H_{FP}$ can only increase through the start of round $\mathtt{t}$.\vspace{.3cm}

    \item[]{\bf Subclaim 7.} $h > h'$.
    \item[]{\it Proof.} This follows immediately from Subclaims 4 and 5 as follows.  Because $\mathsf{OUT}$ tried to send the packet in round $\mathtt{t}_0$ (Subclaim 5) and because $\mathsf{OUT}$ had normal status in this round (Subclaim 4), it must be that the conditional statement on line ({\bf \ref{routingRules}.16}) of round $\mathtt{t}_0$ was satisfied, and in particular that the expression $H > H_{IN}$ was true.  Since $h$ is defined to be the value of $H, H_{FP}$ as of line ({\bf \ref{routingRules2}.38}) of round $\mathtt{t}_0$ (Statement 6 of Lemma \ref{pseudo}), this subclaim will follow if $h'$ equals the value of $H_{IN}$ as of line ({\bf \ref{routingRules2}.38}) of round $\mathtt{t}_0$.  But this is true by Statement 5 of Lemma \ref{pseudo}, since the value of $H_{IN}$ on line ({\bf \ref{routingRules}.16}) comes from the value received on line ({\bf \ref{routingRules}.06}), which in turn corresponds to the value of $H_{IN}$ sent on line ({\bf \ref{routingRules}.09}).\vspace{.3cm}

    \item[]{\bf Subclaim 8.} If the conditional statement on line ({\bf \ref{routingRules2}.51}) is satisfied in round $\mathtt{t}$, then $\mathsf{OUT}$'s attempt to send $p$ in round $\mathtt{t}_0$ failed (i.e.\ $\mathsf{IN}$ did not store $p$ in $\mathtt{t}_0$), and furthermore $\mathsf{IN}$ did not store $p$ in any round between $\mathtt{t}_0$ and $\mathtt{t}$.
    \item[]{\it Proof.} We prove this by contradiction.  Suppose there is some round $\widetilde{\mathtt{t}} \in [\mathtt{t}_0 .. \mathtt{t}-1]$ in which $\mathsf{IN}$ stored $p$.  This would mean that line ({\bf \ref{routingRules2}.51}) was satisfied in round $\widetilde{\mathtt{t}}$, and in particular $RR$ is set to $\widetilde{\mathtt{t}} \geq \mathtt{t}_0$ on ({\bf \ref{routingRules2}.53}).  But as already noted in the proof of Subclaim 2, for the remainder of the transmission, $FR$ can never assume the value of a round {\it before} $\mathtt{t}_0$.  Similarly, once $RR$ changes to $\widetilde{\mathtt{t}} \geq \mathtt{t}_0 \geq FR$ on ({\bf \ref{routingRules2}.53}) of round $\widetilde{\mathtt{t}}$, it can never assume a smaller (non-$\bot$) value for the rest of the transmission ($RR$ can only change to a non-$\bot$ value on line ({\bf \ref{routingRules2}.53})).  But this contradicts the fact that $RR < FR$ on ({\bf \ref{routingRules2}.51}) of round $\mathtt{t}$.\vspace{.3cm}

    \item[]{\bf Subclaim 9.}  If the conditional statement on line ({\bf \ref{routingRules2}.51}) is satisfied in round $\mathtt{t}$, then $RR < \mathtt{t}_0$ between the start of $\mathtt{t}_0$ through line ({\bf \ref{routingRules2}.51}) of round $\mathtt{t}$.  In particular, lines ({\bf \ref{routingRules2}.47}) and ({\bf \ref{routingRules2}.51}) will be satisfied for any round between $\mathtt{t}_0$ and $\mathtt{t}$ for which they are reached.
    \item[]{\it Proof.} $RR$ is set to $-1$ at the start of any transmission (({\bf \ref{setupCode}.31}) and ({\bf \ref{routingRules2}.64})).  Since the only other place $RR$ changes value is ({\bf \ref{routingRules2}.53}), it is always the case that the value of $RR$ is less than or equal to the index of the current round.  Thus, $RR$ can only assume a value greater than (or equal to) $\mathtt{t}_0$ in a round {\it after} (or {\it during}) $\mathtt{t}_0$.  But this would mean there was some round between $\mathtt{t}_0$ and $\mathtt{t}-1$ (inclusive) such that ({\bf \ref{routingRules2}.53}) was reached, which contradicts Subclaim 8.  The fact that ({\bf \ref{routingRules2}.51}) will be satisfied whenever it is reached now follows immediately from Statement 3 of Claim \ref{obsBelowHere}, since in order to reach ({\bf \ref{routingRules2}.51}), line ({\bf \ref{routingRules2}.48}) must have failed, which means the communication on line ({\bf \ref{routingRules2}.43}) was received.  The fact that ({\bf \ref{routingRules2}.47}) will be satisfied whenever it is reached follows from the fact that $sb_{OUT}$ will always be set to one on ({\bf \ref{routingRules}.11}) of each round between $\mathtt{t}_0$ and $\mathtt{t}$ (the first part of this subclaim says $RR < \mathtt{t}_0$, and Subclaim 4 says that if $FR$ is received on ({\bf \ref{routingRules}.11}), then $FR = \mathtt{t}_0$).\vspace{.3cm}

    \item[]{\bf Subclaim 10.}  If the conditional statement on line ({\bf \ref{routingRules2}.51}) is satisfied in round $\mathtt{t}$, then there was no round between $\mathtt{t}_0+1$ and $\mathtt{t}-1$ (inclusive) in which $\mathsf{IN}$ received both $H_{OUT}$ and $p$.
    \item[]{\it Proof.} Suppose for the sake of contradiction that there is such a round, $\widetilde{\mathtt{t}}$.  Notice that line ({\bf \ref{routingRules2}.51}) of round $\widetilde{\mathtt{t}}$ will necessarily be reached (since the conditional statement of line ({\bf \ref{routingRules2}.44}) will fail by assumption, ({\bf \ref{routingRules2}.47}) will be satisfied by Subclaim 4, and ({\bf \ref{routingRules2}.48}) will fail by assumption).  However, line ({\bf \ref{routingRules2}.53}) cannot be reached in round $\widetilde{\mathtt{t}}$ (Subclaim 8 above), and therefore the conditional statement on line ({\bf \ref{routingRules2}.51}) must fail.  This contradicts Subclaim 9.\vspace{.3cm}

    \item[]{\bf Subclaim 11.} If the conditional statement on line ({\bf \ref{routingRules2}.51}) is satisfied in round $\mathtt{t}$, then $\mathsf{IN}$ was in problem status at the end of round $\mathtt{t}_0$, and remained in problem status until line ({\bf \ref{routingRules2}.53}) of round $\mathtt{t}$.
    \item[]{\it Proof.} We first show that $sb_{IN}$ will be set to one on line ({\bf \ref{routingRules2}.45}) or ({\bf \ref{routingRules2}.49}) of round $\mathtt{t}_0$.  To see this, we note that if ({\bf \ref{routingRules2}.44}) fails in round $\mathtt{t}_0$, then necessarily ({\bf \ref{routingRules2}.47}) and ({\bf \ref{routingRules2}.48}) will both be satisfied.  Afterall, ({\bf \ref{routingRules2}.47}) is satisfied (Subclaim 7), and if ({\bf \ref{routingRules2}.48}) {\it failed}, then ({\bf \ref{routingRules2}.51}) would be reached and subsequently satisfied (Subclaim 9), which would contradict Subclaim 8.  For every round between $\mathtt{t}_0 +1$ and $\mathtt{t}$, we will show that either the conditional statement on line ({\bf \ref{routingRules2}.44}) will be satisfied, or the conditional statements on lines ({\bf \ref{routingRules2}.47}) and ({\bf \ref{routingRules2}.48}) will both be satisfied, and hence $sb_{IN}$ can never be reset to zero since lines ({\bf \ref{routingRules2}.53}), ({\bf \ref{routingRules2}.55}), and ({\bf \ref{routingRules2}.57}) will never be reached.  To see this, let $\mathtt{t}' \in [\mathtt{t}_0+1 .. \mathtt{t}-1]$.  If ({\bf \ref{routingRules2}.44}) is satisfied for $\mathtt{t}'$, then we are done.  So assume ({\bf \ref{routingRules2}.44}) is {\it not} satisfied for $\mathtt{t}'$, and hence $\mathsf{IN}$ did {\it not} receive the communication on ({\bf \ref{routingRules2}.43}) (Subclaim 10).  This means ({\bf \ref{routingRules2}.48}) will be satisfied.  The fact that ({\bf \ref{routingRules2}.47}) is also satisfied follows from Subclaim 9.

    \item[]{\bf Subclaim 12.}  If the conditional statement on line ({\bf \ref{routingRules2}.51}) is satisfied in round $\mathtt{t}$, then between the end of round $\mathtt{t}_0$ and the time {\it Receive Packet} is called in round $\mathtt{t}$, we have that $H_{GP} \neq \bot$ and $H_{GP} \leq h'+1 \leq 2n$.
    \item[]{\it Proof.} As in the proof of Subclaim 11, either line ({\bf \ref{routingRules2}.46}) or ({\bf \ref{routingRules2}.50}) will be reached in round $\mathtt{t}_0$ (since either line ({\bf \ref{routingRules2}.45}) or ({\bf \ref{routingRules2}.49}) is reached).  The value of $H_{IN}$ at the start of round $\mathtt{t}_0$ is $h'$ by definition.  Since $h' < h \leq 2n$ (the first inequality is Subclaim 7, the second is Statements 6 and 9 of Lemma \ref{pseudo}), and since $H_{IN}$ cannot change value between the start of $\mathtt{t}_0$ and the time {\it Receive Packet} is called, we have that the value of $H_{IN} < 2n$ when either line ({\bf \ref{routingRules2}.46}) or ({\bf \ref{routingRules2}.50}) is reached.  Therefore, these lines guarantee that $\bot \neq H_{GP} \leq h'+1\leq 2n$ after these lines.  After this, there are five places $H_{GP}$ can change its value: ({\bf \ref{routingRules2}.46}), ({\bf \ref{routingRules2}.50}), ({\bf \ref{routingRules2}.53}), ({\bf \ref{routingRules2}.55}), and ({\bf \ref{routingRules2}.57}).  As in the proof of Subclaim 11, lines ({\bf \ref{routingRules2}.55}) and ({\bf \ref{routingRules2}.57}) will not be reached at any point between $\mathtt{t}_0$ and $\mathtt{t}$, nor will line ({\bf \ref{routingRules2}.53}) by Subclaim 8.  The other two lines that change $H_{GP}$ can only {\it decrease} it (but they cannot set $H_{GP}$ to $\bot$).\vspace{.3cm}

    \item[]{\bf Subclaim 13.} If the condition statement on line ({\bf \ref{routingRules2}.51}) of round $\mathtt{t}$ is satisfied, then the value of $H_{GP}$ when this line is entered, which corresponds to the height in $\mathsf{IN}$ that $p$ assumes when it is inserted, satisfies: $H_{GP} \neq \bot$ and $H_{GP} \leq h'+1 \leq 2n$.
    \item[]{\it Proof.} This follows immediately from Subclaim 12 since $p$ is inserted into $\mathsf{IN}$ at height $H_{GP}$ ({\bf \ref{routingRules2}.53}).\vspace{.3cm}

    \item[]{\bf Subclaim 14.} If the condition statement on line ({\bf \ref{routingRules2}.51}) of round $\mathtt{t}$ is satisfied, then $H_{IN}$ was less than $2n$ at the start of all rounds between $\mathtt{t}_0$ and $\mathtt{t}$.
    \item[]{\it Proof.} Subclaim 12 implies that $h'<2n$ (so $H_{IN}$ had height strictly smaller than $2n$ at the start of round $\mathtt{t}_0$).  Searching through the pseudo-code, we see that $H_{IN}$ is only modified on lines ({\bf \ref{routingRules2}.53}), and during Re-Shuffling ({\bf \ref{reShuffleRules}.91-92}).  Between rounds $\mathtt{t}_0$ and $\mathtt{t}$, line ({\bf \ref{routingRules2}.53}) is never reached (Subclaim 8), and hence all changes to $H_{IN}$ must come from Re-Shuffling.  But because $H_{IN}$ was less than $2n$ when it entered the Re-Shuffle phase in round $\mathtt{t}_0$, Statement 13 of Lemma \ref{pseudo} guarantees that $H_{IN}$ will still be less than $2n$ at the start of round $\mathtt{t}$.\vspace{.3cm}
    \end{itemize}
All Statements of the Lemma have now been proven.\hspace*{\fill} \hspace*{-12pt}$\scriptstyle{\blacksquare}$
\begin{claim} \label{packetHeight2} Every packet is inserted into one of the sender's outgoing buffers at some initial height.  When (a copy of) the packet goes between any two buffers $B_1 \neq B_2$ (either across an edge or locally during re-shuffling), its height in $B_2$ is less than or equal to the height it had in $B_1$.  If $B_1 = B_2$, the statement remains true EXCEPT for on line ({\bf \ref{routingRules2}.35}).\end{claim}
\begin{proof} We separate the proof into cases, based on the nature of the packet movement.  The only times packets are accepted by a new buffer or re-shuffled within the same buffer occurs on lines ({\bf \ref{routingRules2}.32}), ({\bf \ref{routingRules2}.35}), ({\bf \ref{routingRules2}.53}), ({\bf \ref{routingRules2}.55}), ({\bf \ref{routingRules2}.57}), ({\bf \ref{routingRules2}.61}), ({\bf \ref{routingRules2}.64}), ({\bf \ref{reShuffleRules}.89-90}), and ({\bf \ref{reShuffleRules}.101-102}).  Of these, ({\bf \ref{routingRules2}.35}) is excluded from the claim, and the packet movement on lines ({\bf \ref{routingRules2}.32}), ({\bf \ref{routingRules2}.55}), ({\bf \ref{routingRules2}.57}), ({\bf \ref{routingRules2}.61}), ({\bf \ref{routingRules2}.64}), and ({\bf \ref{reShuffleRules}.101-102}) are all clearly strictly downwards.  It remains to consider lines ({\bf \ref{routingRules2}.53}) and ({\bf \ref{reShuffleRules}.89-90}).
    \begin{enumerate}\setlength{\itemsep}{6pt} \setlength{\parskip}{0pt} \setlength{\parsep}{0pt}
    \item[] \textsc{Case 1: The packet moved during Re-Shuffling as on ({\bf \ref{reShuffleRules}.89-90}).}  By investigating the code on these lines, we must show that $m+1 \leq M$.  This was certainly true as of line ({\bf \ref{reShuffleRules}.74}), but we need to make sure this didn't change when {\it Adjust Heights} was called.  The changes made to $M$ and $m$ on ({\bf \ref{reShuffleRules}.83}) and ({\bf \ref{reShuffleRules}.85}) will only serve to help the inequality $m+1 \leq M$, so we need only argue the cases for when ({\bf \ref{reShuffleRules}.81}) and/or ({\bf \ref{reShuffleRules}.87}) is reached.  Notice that if either line is reached, by ({\bf \ref{reShuffleRules}.74}) we must have (before adjusting $M$ and $m$) that $M-m \geq 2$, and therefore modifying {\it only} $M = M-1$ {\it or} $m=m+1$ won't threaten the inequality $m+1 \leq M$.  It remains to argue that both ({\bf \ref{reShuffleRules}.81}) {\bf and} ({\bf \ref{reShuffleRules}.87}) cannot happen simultaneously (i.e.\ cannot both happen within the same call to {\it Re-Shuffle}).  If both of these were to happen, then it must be that during this call to {\it Re-Shuffle}, there was an outgoing buffer $B_1$ that had height 2 or more higher than an incoming buffer $B_2$ (see lines ({\bf \ref{reShuffleRules}.72-74}) and ({\bf \ref{reShuffleRules}.80}) and ({\bf \ref{reShuffleRules}.86})).  We argue that this cannot ever happen.  By Claim \ref{balancing}, at the end of the previous round, we had that the height of $B_1$ was at most one bigger than the height of $B_2$.  During routing, $B_2$ can only get bigger and $B_1$ can only get smaller (({\bf \ref{routingRules2}.53}) and ({\bf \ref{routingRules2}.33}) are the only places these heights change).  Therefore, after Routing but before any Re-Shuffling, we have again that the height of $B_1$ was at most one bigger than the height of $B_2$.  Therefore, in order for $B_1$ to get at least 2 bigger than $B_2$, either a packet must be shuffled {\it into} $B_1$, or a packet must be shuffled {\it out of} $B_2$, and this must happen when $B_1$ is already one bigger than $B_2$.  But analyzing ({\bf \ref{reShuffleRules}.72}) and ({\bf \ref{reShuffleRules}.73}) shows that this can never happen.

    \item[] \textsc{Case 2: The packet moved during Routing as on ({\bf \ref{routingRules2}.53}).}  In order to reach ({\bf \ref{routingRules2}.53}), the conditional statements on lines ({\bf \ref{routingRules2}.47}), ({\bf \ref{routingRules2}.48}), and ({\bf \ref{routingRules2}.51}) all must be satisfied, so $p \neq \bot$, $RR < FR$, and either $sb_{OUT}=1$ or $H_{OUT} > H$ (or both).  We investigate each case separately:
	\begin{enumerate}
	\item[] \textsc{Case A: $sb_{OUT} = 1$ on line ({\bf \ref{routingRules2}.47}).}  Then Statements 2-4 of Lemma \ref{subclaim6} imply that the height of the packet in $B_1$ is greater than or equal to the height it will be stored into in $B_2$, as desired.

	\item[] \textsc{Case B: $sb_{OUT} = 0$ and $H_{OUT}>H_{IN}$ on line ({\bf \ref{routingRules2}.47}).}  For notational convenience, denote the current round (when the hypotheses of Case B hold) by $\mathtt{t}$.  First note that Statements 1 and 2 of Lemma \ref{pseudo} imply that the height the packet assumes in $B_2$ ($H_{GP}$) is less than or equal to $H_{IN}+1$.  Meanwhile, since $sb_{OUT}=0$ (it is set on ({\bf \ref{routingRules}.11}) of round $\mathtt{t}$), the value received for $H_{OUT}$ on ({\bf \ref{routingRules}.11}) is not $\bot$, and the value for $FR$ received on ({\bf \ref{routingRules}.11}) is either $\bot$ or satisfies $FR \leq RR$.  Notice that the case $FR \leq RR$ is not possible, since then ({\bf \ref{routingRules2}.53}) would not be reached (({\bf \ref{routingRules2}.51}) would fail).  Therefore, $FR =\bot$ but $H_{OUT} \neq \bot$, and so $B_2$ received the communication sent by $B_1$ on ({\bf \ref{routingRules}.05}) of round $\mathtt{t}$, which had the first of the two possible forms.  In particular, $H_{FP} = \bot$ at the outset of $\mathtt{t}$, and since $H_{FP}$ cannot change between the start of a round a line ({\bf \ref{routingRules2}.38}) of the previous round, we must have that ({\bf \ref{routingRules2}.37}) failed in round $\mathtt{t}-1$.  By this fact and Claim \ref{HFPFR}, $B_1$ had normal status when ({\bf \ref{routingRules}.16}) was reached in round $\mathtt{t}-1$, and this will not be able to change in the call to {\it Reset Outgoing Variables} of round $\mathtt{t}$ because $d=0$ ({\bf \ref{routingRules2}.25}) (since $d$ is reset to zero every round on ({\bf \ref{routingRules2}.26}), it can only have non-zero values between line ({\bf \ref{routingRules2}.40}) of one round and line ({\bf \ref{routingRules2}.26}) of the following round IF a packet was sent the earlier round.  However, as already noted this did not happen, as the fact that $\mathsf{OUT}$ had normal status and yet ({\bf \ref{routingRules2}.37}) failed in round $\mathtt{t}-1$ implies that ({\bf \ref{routingRules}.16}) will also fail in round $\mathtt{t}-1$).  Therefore, $B_1$ has normal status when {\it Create Flagged Packet} is called in round $\mathtt{t}$, and in particular, $H_{FP}$ is set to $H_{OUT}$ on ({\bf \ref{routingRules2}.38}), i.e.\ the flagged packet to be transferred during $\mathtt{t}$ has height $H_{OUT}$ in $B_1$.  Putting this all together, the packet has height $H_{OUT}$ in $B_1$ and assumes height $H_{GP}$ in $B_2$.  But as argued above, $H_{OUT} \geq H_{IN}+1 \geq H_{GP}$, as desired.
	\end{enumerate}
    \end{enumerate}
\vspace{-22pt}
\end{proof}
\begin{claim} \label{packetProliferation22} Before {\bf \em End of Transmission Adjustments} is called in any transmission $\mathtt{T}$ ({\bf \ref{routingRules2}.61}), any packet that was inserted into the network during transmission $\mathtt{T}$ is either in some buffer (perhaps as a flagged packet) or has been received by $R$.\end{claim}
\begin{proof} As packets travel between nodes, the sending node maintains a copy of the packet until it has obtained verification from the receiving node that the packet was accepted.  This way, packets that are lost due to edge failure are backed-up.  This is the high-level idea of why the claim is true, we now go through rigorous detail.

First notice that the statement only concerns packets corresponding to the current codeword transmission, and packets deleted as on ({\bf \ref{routingRules2}.61}) do not threaten the validity of the Claim.  We consider a specific packet $p$ that has been inserted into the network and show that $p$ is never removed from a buffer $B$ until another buffer $B'$ has taken $p$ from $B$.  We do this by considering every line of code that a buffer could possible remove $p$, and argue that whenever this happens, $p$ has necessarily been accepted from $B$ by some other buffer $B'$.  Notice that the only lines that a buffer could possibly remove $p$ (before line ({\bf \ref{routingRules2}.61}) of $\mathtt{T}$ is reached) are: ({\bf \ref{routingRules2}.32}), ({\bf \ref{routingRules2}.53}), and ({\bf \ref{routingRules2}.89-90}).
    \begin{itemize}\setlength{\itemsep}{6pt} \setlength{\parskip}{0pt} \setlength{\parsep}{0pt}
    \item[]  \underline{Line ({\bf \ref{routingRules2}.53})}.  This line is handled by Lemma \ref{pseudo}, Statements 1 and 2, which say that whenever a slot of an incoming buffer is filled as on line ({\bf \ref{routingRules2}.53}), it fills an empty slot, and therefore cannot correspond to removing (over-writing) $p$.

    \item[]  \underline{Lines ({\bf \ref{reShuffleRules}.89-90})}.  These lines are handled by Lemma \ref{pseudo}, Statement 10.

    \item[]  \underline{Line ({\bf \ref{routingRules2}.32})}.  This is the interesting case, where $p$ is removed from an outgoing buffer after a packet transfer.  We must show that any time $p$ is removed here, it has been accepted by some incoming buffer $B'$.  For notation, we will let $\mathtt{t}$ denote the round that $p$ is deleted from $B$ (i.e.\ when line ({\bf \ref{routingRules2}.32}) is reached), and $\mathtt{t}_0$ denote the round that $B$ first tried to send the packet to $B'$ as on ({\bf \ref{routingRules2}.41}).  By Statement 3 of Claim \ref{obsBelowHere}, $\mathtt{t}_0$ is the round that $\tilde{p}$ was most recently set to $p$ as on line ({\bf \ref{routingRules2}.38}) (note that $\mathtt{t}_0 \leq \mathtt{t}$).  Since line ({\bf \ref{routingRules2}.32}) was reached in round $\mathtt{t}$, the conditional statements on lines ({\bf \ref{routingRules2}.29}) and ({\bf \ref{routingRules2}.30}) were satisfied, and so $\bot \neq RR \geq FR$ when those lines were reached.  By Statement 3 of Claim \ref{obsBelowHere}, $FR$ will equal $\mathtt{t}_0$ when ({\bf \ref{routingRules2}.30}) is satisfied.  Since in {\it any} round $\mathtt{t}'$, the only non-$\bot$ value that $RR$ can ever be set to is $\mathtt{t}'$ ({\bf \ref{routingRules2}.53}), and since $RR \geq \mathtt{t}_0=FR$ ({\bf \ref{routingRules2}.30}), it must be that ({\bf \ref{routingRules2}.53}) was reached in some round $\mathtt{t}' \in [\mathtt{t}_0, \mathtt{t}]$.  In particular, $B'$ stored a packet as on ({\bf \ref{routingRules2}.53}) of round $\mathtt{t}'$, which by Statement 3 of Claim \ref{obsBelowHere} was necessarily $p$.
    \end{itemize}
\vspace{-22pt}
\end{proof}
\begin{claim} \label{packetProliferation32} Not counting flagged packets, there is at most one copy of any packet in the network at any time (not including packets in the sender or receiver's buffers).  Looking at all copies (flagged and un-flagged) of any given packet present in the network at any time, at most one copy of that packet will {\bf ever} be accepted (as in Definition \ref{receive}) by another node.\end{claim}
\begin{proof} For any packet $p$, let $N_p$ denote the copies of $p$ (both flagged and not) present in the network (in an internal node's buffer) at a given time.  We begin the proof via a sequence of observations:
	\begin{enumerate}
	\item[] {\bf Observation 1.}  {\it The only time $N_p$ can ever {\bf \em increase} is on line ({\bf \ref{routingRules2}.53}).}
	\item[] {\it Proof.} The only way for $N_p$ to increase is if (a copy of) $p$ is stored by a new buffer.  Looking at the pseudo-code, the only place a buffer slot can be assigned a new copy of $p$ is on lines ({\bf \ref{routingRules2}.32}), ({\bf \ref{routingRules2}.35}), ({\bf \ref{routingRules2}.53}), ({\bf \ref{routingRules2}.55}), ({\bf \ref{routingRules2}.57}), ({\bf \ref{routingRules2}.61}), ({\bf \ref{routingRules2}.64}), and ({\bf \ref{reShuffleRules}.89}).  Of these, only ({\bf \ref{routingRules2}.53}) and ({\bf \ref{reShuffleRules}.89}) could possibly increase $N_p$, as the others simply shift packets within a buffer and/or delete packets.  In the latter case, $N_p$ does not change by Statement 10 of Lemma \ref{pseudo}.

	\item[] {\bf Observation 2.}  {\it Suppose $A$ (including $A=S$) first sends a (copy of a) packet $p$ to $B$ as on ({\bf \ref{routingRules2}.41}) of round $\mathtt{t}_0$.  Then:
		\begin{enumerate}
		\item The copy of $p$ in $A$'s outgoing buffer along $E(A,B)$ (for which there was a copy made and sent on ({\bf \ref{routingRules2}.41}) of round $\mathtt{t}_0$) will never be transferred to any of $A$'s other buffers.
		\item The copy of $p$ will remain in $A$'s outgoing buffer along $E(A,B)$ as a flagged packet until it is deleted either when $A$ gets confirmation of receipt (see Definition \ref{confRec}) in some round $\mathtt{t}$ ({\bf \ref{routingRules2}.32}), or by the end of the transmission as on ({\bf \ref{routingRules2}.61}).  In the latter case, define $\mathtt{t} := 3D$ (the last round of the transmission) for Statement (c) below.
		\item Between $\mathtt{t}_0$ and line ({\bf \ref{routingRules}.07}) of round $\mathtt{t}$, $B$ will accept (a copy of) $p$ from $A$ as on ({\bf \ref{routingRules2}.53}) at most once.  Furthermore, the copy of $p$ in $A$'s buffer cannot move to any other buffer or generate any other copies other than the one (possibly) received by $B$ as on ({\bf \ref{routingRules2}.53}).
		\end{enumerate}}
	\item[] {\it Proof.} Statement (a) follows from Statement 3 of Claim \ref{obsBelowHere} and Statement 11 of Lemma \ref{pseudo}, together with the fact that lines ({\bf \ref{routingRules2}.32}) and ({\bf \ref{routingRules2}.61}, {\bf \ref{routingRules2}.69}) imply that the relevant copy of $A$ will be deleted when it does get confirmation of receipt as in Definition \ref{confRec} (or the end of the transmission).  By Statement 3 of Claim \ref{obsBelowHere}, this copy of $p$ will be (the unique) flagged packet in $A$'s outgoing buffer to $B$ until confirmation of receipt (or the end of the transmission), which proves Statement (b).  For Statement (c), suppose that $B$ accepts a copy of $p$ as on ({\bf \ref{routingRules2}.53}) during some round $\mathtt{t}' \in [\mathtt{t}_0, \mathtt{t}]$.  Then $RR$ will be set to $\mathtt{t}'$ on ({\bf \ref{routingRules2}.53}) of round $\mathtt{t}'$, and $RR$ cannot obtain a {\it smaller} index until the next transmission ({\bf \ref{routingRules2}.53}).  By Statement 3 of Claim \ref{obsBelowHere}, $FR$ will remain equal to $\mathtt{t}_0$ from line ({\bf \ref{routingRules2}.38}) of round $\mathtt{t}_0$ through the time ({\bf \ref{routingRules2}.33}) of round $\mathtt{t}$ is reached.  Therefore, between $\mathtt{t}' \geq \mathtt{t}_0$ and line ({\bf \ref{routingRules2}.33}) of round $\mathtt{t}$, we have that $FR = \mathtt{t}_0 \leq \mathtt{t}' \leq RR$, and hence line ({\bf \ref{routingRules2}.51}) can never be satisfied during these times, which implies ({\bf \ref{routingRules2}.53}) can never be reached again after $\mathtt{t}'$.  This proves the first part of Statement (c).  The second part follows by looking at all possible places (copies of) packets can move or be created: ({\bf \ref{routingRules2}.32}), ({\bf \ref{routingRules2}.35}), ({\bf \ref{routingRules2}.53}), ({\bf \ref{routingRules2}.55}), ({\bf \ref{routingRules2}.57}), ({\bf \ref{routingRules2}.61}), ({\bf \ref{routingRules2}.64}), and ({\bf \ref{reShuffleRules}.89-90}).  Of these, only ({\bf \ref{routingRules2}.53}) and ({\bf \ref{reShuffleRules}.89-90}) threaten to move $p$ or create a new copy of $p$.  However, the first part of Observation 2(c) says that ({\bf \ref{routingRules2}.53}) can happen at most once (and is accounted for), while Statement 11 of Lemma \ref{pseudo} rules out the case that the packet is re-shuffled as on ({\bf \ref{reShuffleRules}.89-90}).

	\item[] {\bf Observation 3.}  {\it No packet will ever be inserted (see Definition \ref{insert}) into the network more than once.  In particular, for any packet $p$, $N_p=0$ until the sender inserts it (i.e.\ some node accepts the packet from the sender as on ({\bf \ref{routingRules2}.53})), at which point $N_p=1$.  After this point, the only way $N_p$ can become larger than one is if ({\bf \ref{routingRules2}.53}) is reached, where neither the sending node nor the receiving node is $S$ or $R$.}
	\item[] {\it Proof.} Since the packets of $S$ are distributed to his outgoing buffers before being inserted into the network ({\bf \ref{setupCode2}.38}), ({\bf \ref{routingRules2}.65}), and ({\bf \ref{routingRules2}.67-70}), and since $S$ never receives a packet he has already inserted ($S$ has no incoming buffers ({\bf \ref{setupCode}.17})) nor shuffles packets between buffers (({\bf \ref{routingRules}.22}) and ({\bf \ref{reShuffleRules}.95-96})), a given packet $p$ can only be insereted along one edge adjacent to the sender.  The fact the sender can insert at most one (copy of a) packet $p$ along an adjacent edge now follows from Observation 2 above for $A=S$.  This proves the first part of Observation 3.
	
	By Observation 1, the only place $N_p$ can increase is on ({\bf \ref{routingRules2}.53}).  Whenever this line is reached, the copy stored comes from the one received on ({\bf \ref{routingRules2}.43}), which in turn was sent by another node on ({\bf \ref{routingRules2}.41}).  The copy sent on ({\bf \ref{routingRules2}.41}) in turn can only be set on ({\bf \ref{routingRules2}.38}) (perhaps in an earlier round), so in particular a copy of the packet must have already existed in an outgoing buffer of the sending node.  This proves that when $N_p$ goes from zero to one, it can only happen when a packet is inserted for the first time by the sender.  The rest of Observation 3 now follows from Observation 1, the first part of Observation 3, the fact that copies reaching $R$ do not increase $N_p$ (by definition of $N_p$), and the fact $R$ never sends a copy of a packet ({\bf \ref{setupCode}.07}) and $S$ never accepts packets ({\bf \ref{setupCode}.17}).
	\end{enumerate}
Define a copy of a packet $p$ in the network to be {\em dead} if that copy will {\em never} leave the buffer it is currently in, nor will it ever generate any new copies.  A copy of a packet that is not dead will be {\em alive}.
		\begin{enumerate}
		\item[] {\bf Observation 4.}  {\it If a (copy of a) packet is ever flagged and dead, it will forever remain both flagged and dead, until it is deleted.}
		\item[] {\it Proof.}  By definition of being ``dead,'' once a (copy of a) packet becomes dead it can never become alive again.  Also, copies of a packet that are flagged remain flagged until they are deleted by Observation 2(b).
		\end{enumerate}
The Claim now follows immediately from the following subclaim:
	\begin{enumerate}
	\item[] {\bf Subclaim.}  {\it Fix any packet $p$ that is ever inserted into the network.  Then at any time, there is at most one {\bf \em alive} copy of $p$ in the network at any time.  Also at any time, if there is one alive copy of $p$, then all {\bf \em dead} copies of $p$  are flagged packets.  If there are no alive copies, then there is at most one dead copy of $p$ that is not a flagged packet.}

	\item[] {\it Proof.} Before $p$ is inserted into the network, $N_p =0$, and there is nothing to show.  Suppose $p$ is inserted into the network in round $\mathtt{t}_0$, so that $N_p=1$ by the end of the round (Observation 3).  Since $N_p=1$, the validity of the subclaim is not threatened.  Also, if this packet is {\it dead}, then the proof is complete, as by Observation 3 and the definition of {\it deadness}, no other (copies) of $p$ will ever be created, and hence the subclaim will forever be true for $p$.  So suppose $p$ is {\it alive} when it is inserted.  We will show that a (copy of an) alive packet can create at most one new (copy of a) packet, and the instant it does so, the original copy is necessarily both flagged and dead (the new copy may be either alive or dead), from which the subclaim follows from Observation 4.  So suppose an alive copy of $p$ creates a new copy (increasing $N_p$) of itself in round $\mathtt{t}$.  Notice that the only time new copies of any packet can be created is on ({\bf \ref{routingRules2}.53}) (see e.g.\ proof of Observation 2).  Fix notation, so that the alive copy of $p$ was in node $A$'s outgoing buffer to node $B$, and hence it was $B$'s corresponding buffer that entered ({\bf \ref{routingRules2}.53}) in round $\mathtt{t}$.  The fact that the alive copy of $p$ in $A$'s outgoing buffer is flagged and dead the instant $B$ accepts it on ({\bf \ref{routingRules2}.53}) of round $\mathtt{t}$ follows immediately from Observation 2.
	\end{enumerate}
\vspace{-24pt}
\end{proof}
\begin{lemma} \label{packetTransferPotDrop} Suppose that in round $\mathtt{t}$, B {\it accepts} (as in Definition \ref{receive}) a packet from $A$.  Let $O_{A,B}$ denote $A$'s outgoing buffer along $E(A,B)$, and let $O$ denote the height the packet had in $O_{A,B}$ when {\bf \em Send Packet} was called in round $\mathtt{t}$ ({\bf \ref{routingRules}.17}).  Also let $I_{B,A}$ denote $B$'s incoming buffer along $E(A,B)$, and let $I$ denote the height of $I_{B,A}$ at the start of $\mathtt{t}$.  Then the change in non-duplicated potential caused by this packet transfer is less than or equal to:
	\begin{equation}
	-O + I + 1 \qquad OR \qquad -O \thickspace \mbox{ (if $B=R$)}
	\end{equation}
Furthermore, after the packet transfer but before re-shuffling, $I_{B,A}$ will have height $I + 1$.
\end{lemma}
\begin{proof} By definition, $B$ accepts the packet in round $\mathtt{t}$ means that ({\bf \ref{routingRules2}.53}) was reached by $B$'s incoming buffer along $E(A,B)$ in round $\mathtt{t}$.  Since the packet is stored at height $H_{GP}$ ({\bf \ref{routingRules2}.53}), $B$'s non-duplicated potential will increase by $H_{GP}$ due to this packet transfer (if $B=R$, then by definition of non-duplicated potential, packets in $R$ do not contribute anything, so there will be no change).  By Statements 1 and 2 of Lemma \ref{pseudo}, $H_{GP} \leq I + 1$, and hence $B$'s increase in non-duplicated potential caused by the packet transfer is at most $I+1$ (or zero in the case $B=R$).  Also, since $B$ had height $I$ at the start of the round, and $B$ accepts a packet on ({\bf \ref{routingRules2}.53}) of round $\mathtt{t}$, $B$ will have $I+1$ packets in $I$ when the re-shuffling phase of round $\mathtt{t}$ begins, which is the second statement of the lemma.

Meanwhile, the packet transferred along $E(A,B)$ in round $\mathtt{t}$ still has a copy in $O_{A,B}$ (until $A$ receives confirmation of receipt from $B$, see Definition \ref{confRec}), but by definition of non-duplicated potential (see the paragraph between Claim \ref{packetProliferation} and Lemma \ref{item3}), this (flagged) packet will no longer count towards non-duplicated potential the instant $B$ accepts it as on ({\bf \ref{routingRules2}.53}) of round $\mathtt{t}$.  Therefore, $A$'s non-duplicated potential will drop by the value $H_{FP}$ has when $B$ accepts the packet on ({\bf \ref{routingRules2}.53}) (Statement 3 of Claim \ref{obsBelowHere}), which equals $O$ since $H_{FP}$ cannot change between the time {\it Send Packet} is called on ({\bf \ref{routingRules}.17}) and the time the packet is accepted on ({\bf \ref{routingRules2}.53}).  Therefore, counting only changes in non-duplicated potential due to the packet transfer, the change in potential is: $-O + H_{GP} \leq -O + I+1$ (or $-O$ in the case $B=R$), as desired.
\end{proof}
We now re-state and prove Lemma \ref{chainL2}.
\begin{lemma} \label{chainL} Let $\mathcal{C} = N_1 N_2 \dots N_l$ be a path consisting of $l$ nodes, such that $R = N_l$ and $S \notin \mathcal{C}$.  Suppose that in round $\mathtt{t}$, all edges $E(N_i,N_{i+1})$, $1 \leq i <l$ are {\it active} for the entire round.  Let $\phi$ denote the change in the network's non-duplicated potential caused by:
	\begin{enumerate}\setlength{\itemsep}{1pt} \setlength{\parskip}{0pt} \setlength{\parsep}{0pt}
	\item (For $1 \leq i < l$) Packet transfers across $E(N_i, N_{i+1})$ in round $\mathtt{t}$,
	\item (For $1 < i < l$) Re-shuffling packets {\bf \em into} $N_i$'s outgoing buffers during $\mathtt{t}$,
	\end{enumerate}
Then if $O_{N_1, N_2}$ denotes $N_1$'s outgoing buffer along $E(N_1, N_2)$ and $O$ denotes its height at the start of $\mathtt{t}$, we have:
	\begin{enumerate}\setlength{\itemsep}{1pt} \setlength{\parskip}{0pt} \setlength{\parsep}{0pt}
	\item[-] If $O_{N_1, N_2}$ has a flagged packet that has already been accepted by $N_2$ {\bf \em before} round $\mathtt{t}$, then:
		\begin{equation}\label{firstC}
		\phi \leq -O + l-1
		\end{equation}
	\item[-] Otherwise,
		\begin{equation}\label{secondC}
		\phi \leq -O + l-2
		\end{equation}
	\end{enumerate}
\end{lemma}
\begin{proof}(Induction on $l$).\\[.1cm]
\textsc{Base Case: $l=2$}.  So $\mathcal{C} = N_1 R$.
	\begin{enumerate}
	\item[] \textsf{Case 1: $O_{N_1,R}$ had a flagged packet at the start of $\mathtt{t}$ that was already accepted by $N_2$.}  Our aim for this case is to prove \eqref{firstC} for $l=2$.  If $O < 2$, then $-O + l -1 \geq -1 + 2 -1 =0$, and then \eqref{firstC} will be true by Statement 3 of Lemma \ref{item3}.  So assume $O \geq 2$.  Since $E(N_1, R)$ is active during $\mathtt{t}$ and $R$ had already accepted the packet in some previous round $\tilde{\mathtt{t}} < \mathtt{t}$, we have that $RR \geq \tilde{\mathtt{t}}$ ({\bf \ref{routingRules2}.53}), and $N_1$ will receive this value for $RR$ in $R$'s stage one communication ({\bf \ref{routingRules}.06}), ({\bf \ref{routingRules}.09}).  By Statment 3 of Claim \ref{obsBelowHere}, $FR \leq \tilde{\mathtt{t}} \leq RR$, and thus lines ({\bf \ref{routingRules2}.29-30}) will be satisfied in round $\mathtt{t}$, deleting the flagged packet on ({\bf \ref{routingRules2}.32}) and setting $sb=0$.  When {\it Create Flagged Packet} is called on ({\bf \ref{routingRules}.15}), a new packet will be flagged, with $H_{FP}=H_{OUT}=O-1$ and $FR = \mathtt{t}$ (since $O \geq 2$, there will be at least one packet left in $O_{N_1,R}$ of height $O-1>0$ by Lemma \ref{pseudo}).  Letting $I$ denote the height of the receiver's incoming buffer along $E(N_1,R)$, we have that $I=0$ (Claim \ref{hIN}).  Therefore, $H_{OUT} > H_{IN}$, and so the flagged packet will be sent as on ({\bf \ref{routingRules}.17}).  Since $R$ will receive and store this packet (since the edge is active and $RR < \mathtt{t} = FR$, lines ({\bf \ref{routingRules2}.44}) and ({\bf \ref{routingRules2}.48}) will fail, while lines ({\bf \ref{routingRules2}.47}) and ({\bf \ref{routingRules2}.51}) will be satisfied), we apply Lemma \ref{packetTransferPotDrop} to argue there will be a change in non-duplication potential that is less than or equal to $-(O-1)$, which is \eqref{firstC} (for $l=2$).

	\item[] \textsf{Case 2:  Either $O_{N_1,R}$ has no flagged packet at the start of $\mathtt{t}$, or if so, it has not yet been accepted by $R$.}  Our aim for this case is to prove \eqref{firstC} for $l=2$.  If $O=0$, then $-O+l-2 = 0$, and \eqref{secondC} is true by Statement 3 of Lemma \ref{item3}.  So assume $O \geq 1$.  Then necessarily a packet will be sent during round $\mathtt{t}$ (({\bf \ref{routingRules}.16}) is necessarily satisfied since by assumption $E(N_1, R)$ is active during $\mathtt{t}$, $H_{OUT} \geq 1$ by Lemma \ref{pseudo} and $H_{IN}= 0$ by Claim \ref{hIN}).  We first show that the height of the packet in $O_{N_1, R}$ that will be transferred in round $\mathtt{t}$ (which will be the value held by $H_{FP}$ when {\it Send Packet} is called in round $\mathtt{t}$) is greater than or equal to $O$ (whether or not it was flagged before round $\mathtt{t}$):
		\begin{itemize}
		\item If $O_{N_1,R}$ did not have any flagged packets at the outset of $\mathtt{t}$, then $H_{FP} = \bot$ at the start of $\mathtt{t}$, and so $sb=0$ and $FR=\bot$ at the start of $\mathtt{t}$ by Claim \ref{HFPFR}.  Since $H_{FP}$ cannot change between the call to {\it Send Packet} in the previous round and the call to {\it Reset Outgoing Variables} in the current round, Statement 2 of Claim \ref{stupid} implies no packet was sent the previous round, and hence $d=0$ at the start of $\mathtt{t}$ ($d$ was necessarily zero as of ({\bf \ref{routingRules2}.26}) of round $\mathtt{t}-1$, and as argued did not change to `1' on ({\bf \ref{routingRules2}.40}) later that round).  Consequently, $sb$ will remain zero from the start of $\mathtt{t}$ through the time {\it Create Flagged Packet} is called in round $\mathtt{t}$, and because $H_{OUT}=O > 0 = I = H_{IN}$, ({\bf \ref{routingRules2}.38}) will be reached in round $\mathtt{t}$, setting $H_{FP}$ to $O$.

		\item Alternatively, if $O_{N_1,R}$ {\it does} have a flagged packet at the outset of $\mathtt{t}$, we argue that it will have height {\it at least} $O$ when {\bf Send Packet} is called in round $\mathtt{t}$ as follows.  Let $\mathtt{t}_0 < \mathtt{t}$ denote the round $O_{N_1,R}$ first sent (a copy of) the packet to $R$.  We first show that $N_1$ will {\it not} get confirmation of receipt from $R$ (as in Definition \ref{confRec}) for the packet at any point between rounds $\mathtt{t}_0$ and $\mathtt{t}-1$ (inclusive).  To see this, note that since we are \textsf{Case 2}, $R$ has not accepted the flagged packet by the start of $\mathtt{t}$.  This means that at all times between $\mathtt{t}_0$ and the start of $\mathtt{t}$, $RR < \mathtt{t}_0$\footnote{By Statement 3 of Claim \ref{obsBelowHere}, the packet flagged in $\mathtt{t}_0$ is the only packet $O_{N_1,R}$ can send to $R$ between $\mathtt{t}_0+1$ and the time $R$ receives this flagged packet.  Since we know $R$ has still not accepted this flagged packet by the outset of $\mathtt{t}$, this means that between $\mathtt{t}_0$ and $\mathtt{t}-1$, $RR$ cannot be changed as on ({\bf \ref{routingRules2}.53}).  Since $RR$ begins each transmission equal to $-1$ (({\bf \ref{setupCode}.31}) and ({\bf \ref{routingRules2}.64})) and can only be changed after this on ({\bf \ref{routingRules2}.53}), necessarily $RR < \mathtt{t}_0$ through the start of $\mathtt{t}$.}.  Meanwhile, by Statement 3 of Lemma \ref{obsBelowHere}, $FR = \mathtt{t}_0$ and $H_{FP} \neq \bot$ at the start of $\mathtt{t}$.  Since these do not change values before {\it Reset Outgoing Variables} is called in round $\mathtt{t}$, line ({\bf \ref{routingRules2}.34}) guarantees that if $H_{FP}<O$, then line ({\bf \ref{routingRules2}.35}) will be reached, and thus in either case $H_{FP} \geq O$ after the call to {\it Reset Outgoing Variables}.
		\end{itemize}
	Therefore, since $R$ will necessarily receive and accept the flagged packet sent (by the same argument used in \textsf{Case 1}), we may apply Lemma \ref{packetTransferPotDrop} to argue that $\phi \leq -O$, which is \eqref{secondC} (for $l=2$).
	\end{enumerate}
\textsc{Induction Step}.  Assume the lemma is true for any chain of length less that or equal to $l-1$, and let $\mathcal{C}$ be a chain of length $l$ ($l > 2$).   Since we will be applying the induction hypothesis, we extend and change our notation as follows:  Let $O_{N_i, N_j}$ (respectively $I_{N_i, N_j}$) denote the height of $N_i$'s outgoing (respectively incoming) buffer along edge $E(N_i, N_j)$ {\em at the start of round $\mathtt{t}$} (before, the notation referred to the {\it buffer}, now it will refer to the buffer's {\it height}).  Notice that if $O_{N_1, N_2} \leq I_{N_2, N_1}$, then:
	\begin{equation} \label{1164r}
	\phi \leq -O_{N_2,N_3} + (l-1) - 1 \leq -I_{N_2, N_1} + l -2 \leq -O_{N_1, N_2} + l - 2
	\end{equation}
where the first inequality is from the induction hypothesis applied to the chain $N_2 \dots R$, the second follows from Lemma \ref{balancing}, and the third follows from the fact we are assuming $O_{N_1, N_2} \leq I_{N_2, N_1}$.  Therefore, both \eqref{firstC} and \eqref{secondC} are satisfied.  We may therefore assume in both cases below:
	\begin{equation}\label{e1387}
	O_{N_1, N_2} > I_{N_2, N_1}
	\end{equation}
	\begin{enumerate}
	\item[] \textsf{Case 1: $O_{N_1,N_2}$ had a flagged packet at the start of $\mathtt{t}$ that was already accepted by $N_2$.}  If
	$O_{N_1, N_2} = I_{N_2, N_1} +1$, then by the same string of inequalities as in \eqref{1164r}, we would have $\phi \leq -O_{N_1, N_2} + l - 1$, which is \eqref{firstC}.  Therefore, it remains to consider the case:
		\begin{equation}\label{e1394}
		O_{N_1, N_2} \geq I_{N_2, N_1} +2
		\end{equation}
	By an analogous argument to the one made in the \textsc{Base Case}, a packet will be transfered and accepted across $E(N_1,N_2)$ in round $\mathtt{t}$ that will cause the non-duplicated potential to change by an amount less than or equal to:
		\begin{equation}\label{cont1}
		(-O_{N_1,N_2}+1)+ I_{N_2,N_1}+1
		\end{equation}
	Also, when the receiving node $N_2$ accepts this packet as on ({\bf \ref{routingRules2}.53}), the height of the corresponding buffer increases by one on this line.  We emphasize this fact for use below:
		\begin{itemize}
		\item[] {\bf Fact:} After the Routing Phase but before the call to {\it Re-Shuffle} in round $\mathtt{t}$, $N_2$'s incoming buffer along $E(N_1,N_2)$ has height $I_{N_2,N_1}+1$.
		\end{itemize}
	Meanwhile, we may apply the induction hypothesis to the chain $\mathcal{C}' := N_2 \dots R$, so that the change in non-duplicated potential due to contributions 1 and 2 (in the hypothesis of the Lemma) on $\mathcal{C}'$ is less than or equal to:
		\begin{enumerate}
		\item $-O_{N_2,N_3} + (l-1) -1$, if $O_{N_2,N_3}$ had a flagged packet at the start of $\mathtt{t}$ that was already accepted by $N_3$.
		\item $-O_{N_2,N_3} + (l-1) -2$, otherwise.
		\end{enumerate}
	Adding these contributions to \eqref{cont1}, we have that:
		\begin{alignat}{2} \label{1180t}
		\phi &\leq ((-O_{N_1,N_2}+1)+ I_{N_2,N_1}+1) + (-O_{N_2,N_3} + (l-1) -x) \notag \\
		&= (-O_{N_1,N_2} + l -1) +(-O_{N_2,N_3} + I_{N_2,N_1}) + (2-x),
		\end{alignat}
	where $x=1$ or 2, depending on whether we are in case (a) or (b) above.  By Lemma \ref{balancing}, $-O_{N_2,N_3} + I_{N_2,N_1}$ is either 0 or -1.  If $-O_{N_2,N_3} + I_{N_2,N_1} = -1$, then $(-O_{N_2,N_3} + I_{N_2,N_1}) + (2-x) \leq 0$, regardless whether $x=1$ or 2, and hence \eqref{1180t} implies \eqref{firstC}.  Also, if $x=2$, then $(-O_{N_2,N_3} + I_{N_2,N_1}) + (2-x) \leq 0$ (by Lemma \ref{balancing}), and hence hence \eqref{1180t} implies \eqref{firstC}.  It remains to consider the case $x=1$ and $-O_{N_2,N_3} + I_{N_2,N_1} = 0$, in which case \eqref{1180t} becomes:
		\begin{equation}\label{e1408}
		\phi \leq (-O_{N_1,N_2} + l -1) + 1	
		\end{equation}
	In order to obtain \eqref{firstC} from \eqref{e1408}, we therefore need to account for a drop of at least one more to $\phi$.  We will obtain this by the second contribution to $\phi$ (see statement of Lemma) by arguing:
		\begin{enumerate}
		\item After the Routing Phase of round $\mathtt{t}$ but before the call to {\it Re-Shuffling}, the {\it fullest} buffer of $N_2$ has height $O_{N_2,N_3}+1$, and there is at least one {\it incoming} buffer of $N_2$ that has this height.  In particular, during the call to {\it Re-Shuffle} in round $\mathtt{t}$, the first buffer chosen to transfer a packet {\it from} will be an incoming buffer of height $O_{N_2,N_3}+1$.

		\item After the Routing Phase of round $\mathtt{t}$ but before the call to {\it Re-Shuffling}, the {\it emptiest} buffer of $N_2$ has height $O_{N_2,N_3}-1$, and there is at least one {\it outgoing} buffer of $N_2$ that has this height.  In particular, during the call to {\it Re-Shuffle} in round $\mathtt{t}$, the first buffer chosen to transfer a packet {\it to} will be an outgoing buffer of height $O_{N_2,N_3}-1$.
		\end{enumerate}
	Notice that if I can show these two things, this case will be done, as during the first call to {\it Re-Shuffle} in round $\mathtt{t}$, we will have $M-m \geq (O_{N_2,N_3}+1) - (O_{N_2,N_3}-1) \geq 2$ (the call to {\it Adjust Heights} can only help this inequality since the selection process on ({\bf \ref{reShuffleRules}.72-73}) and the two items above guarantee ({\bf \ref{reShuffleRules}.80}) and ({\bf \ref{reShuffleRules}.86}) will both fail if reached), and consequently the re-shuffle on ({\bf \ref{reShuffleRules}.89-90}) will cause a drop of at least one to $\phi$.

	We first argue (a).  As noted at the beginning of \textsf{Case 1} of the \textsc{Induction Step}, Fact 1 implies that there will exist an incoming buffer of the required height (since we are assuming $O_{N_2,N_3} = I_{N_2,N_1}$).  Also, at the start of $\mathtt{t}$, since $N_2$ has an outgoing buffer of height $O_{N_2,N_3}$ (namely, the outgoing buffer along $E(N_2, N_3)$), Lemma \ref{balancing} guarantees that all of $N_2$'s incoming buffers have hieght at most $O_{N_2,N_3}$ at the start of $\mathtt{t}$; and also that all of $N_2$'s outgoing buffers have height at most $O_{N_2,N_3}+1$ at the start of $\mathtt{t}$.  During the Routing Phase but before the Re-Shuffle Phase of $\mathtt{t}$, outgoing buffers cannot {\it increase} in height ({\bf \ref{routingRules2}.33}) and incoming buffers cannot increase in height by more than one ({\bf \ref{routingRules2}.53}).  Therefore, after transferring packets but before Re-Shuffling in round $\mathtt{t}$, the fullest buffer in $N_2$ has height at most $O_{N_2,N_3}+1$, and as already argued, at least one incoming buffer has this height.  The last part of (a) is immediate from the selection rules in ({\bf \ref{reShuffleRules}.72}).

	We now argue (b).  Since $x=1$, we are in the case the outgoing buffer along $E(N_2,N_3)$ had a flagged packet at the start of $\mathtt{t}$ that had already been accepted by $N_3$ in some round $\mathtt{t}_0 < \mathtt{t}$.  By a similar argument that was used in \textsf{Case 1} of the \textsc{Base Case}, the outgoing buffer along $E(N_2, N_3)$ will reach lines ({\bf \ref{routingRules2}.32-33}) in round $\mathtt{t}$.  In particular, the height of the outgoing buffer along $E(N_2,N_3)$ will drop by one on ({\bf \ref{routingRules2}.33}), and thus this buffer has height $O_{N_2, N_3}-1$ after the call to {\it Reset Outgoing Variables}.  Since this height cannot change before the call to {\it Re-Shuffle}, this outgoing buffer has height $O_{N_2, N_3}-1$ after the Routing Phase (but before the call to {\it Re-Shuffle}) in round $\mathtt{t}$.  Also, $O_{N_2, N_3}-1$ is a lower bound for the {\it emptiest} buffer in $N_2$ just before the call to {\it Re-Shuffle} in round $\mathtt{t}$, argued as follows.  At the start of $\mathtt{t}$, since $N_2$ has an incoming buffer of height $I_{N_2,N_1}=O_{N_2,N_3}$ (namely, the incoming buffer along $E(N_1, N_2)$), Lemma \ref{balancing} guarantees that all of $N_2$'s incoming buffers have hieght at least $O_{N_2,N_3}-1$ at the start of $\mathtt{t}$; and also that all of $N_2$'s outgoing buffers have hieght at least $O_{N_2,N_3}$ at the start of $\mathtt{t}$.  During the Routing Phase but before the Re-Shuffle Phase of $\mathtt{t}$, incoming buffers cannot {\it decrease} in height ({\bf \ref{routingRules2}.53}) and outgoing buffers can decrease in height by at most one ({\bf \ref{routingRules2}.33}).  Therefore, after transferring packets but before Re-Shuffling in round $\mathtt{t}$, the emptiest buffer in $N_2$ has height at least $O_{N_2,N_3}-1$, and as already argued, at least one outgoing buffer has this height.  The last part of (b) is immediate from the selection rules in ({\bf \ref{reShuffleRules}.73}).
	
	\item[] \textsf{Case 2:  Either $O_{N_1,N_2}$ has no flagged packet at the start of $\mathtt{t}$, or if so, it has not yet been accepted by $N_2$.}  By the same argument\footnote{For the argument in the \textsc{Base Case}, we used the fact that the receiver's incoming buffer had height zero in order to conclude $H_{OUT}>H_{IN}$ (and thus a packet would be sent).  Here, we use instead \eqref{e1387} to come to the same conclusion.} used in \textsf{Case 2} of the \textsc{Base Case}, there will be a packet transferred across $E(N_1,N_2)$ and accepted by $N_2$ in round $\mathtt{t}$, and this packet will have height at least $O_{N_1,N_2}$ in $N_1$'s outgoing buffer.  Therefore, by Lemma \ref{packetTransferPotDrop}, the change in non-duplicated potential due to this packet transfer is less than or equal to:
		\begin{equation}\label{e1432}
		-O_{N_1,N_2}+I_{N_2, N_1}+1
		\end{equation}
	Meanwhile, we may apply the induction hypothesis to the chain $\mathcal{C}' := N_2 \dots R$, so that the change in non-duplicated potential due to contributions 1 and 2 (in the hypothesis of the Lemma) on $\mathcal{C}'$ is less than or equal to:
		\begin{enumerate}
		\item $-O_{N_2,N_3} + (l-1) -1$, if $O_{N_2,N_3}$ had a flagged packet at the start of $\mathtt{t}$ that was already accepted by $N_3$.
		\item $-O_{N_2,N_3} + (l-1) -2$, otherwise.
		\end{enumerate}
	Adding these contributions to \eqref{e1432}, we have that:
		\begin{alignat}{2} \label{dream}
		\phi &\leq (-O_{N_1,N_2}+ I_{N_2,N_1}+1) + (-O_{N_2,N_3} + (l-1) -x) \notag \\
		&= (-O_{N_1,N_2} + l -2) +(-O_{N_2,N_3} + I_{N_2,N_1}) + (2-x),
		\end{alignat}
	where $x=1$ or 2, depending on whether we are in case (a) or (b) above.  Since the first term of \eqref{dream} matches \eqref{secondC} and the latter two terms match the latter two terms of \eqref{1180t}, we follow the argument of \textsf{Case 1} above to conclude the proof.
	\end{enumerate}
\vspace{-24pt}
\end{proof}
\section{\hspace{-11.46pt}Routing Against a (Node-Controlling\small{+}\Large{Edge-Scheduling) Adversary}} \label{aP}
\subsection{Definitions and High-Level Description of the Protocol} \label{adProtocolyLong}
\indent \indent In this section, we define the variables that appear in the next section and describe how they will be used.

As in the protocol for the edge-scheduling adversary model, the sender first converts the input stream of messages into codewords, and then transmits a single codeword at a time.  The sender will allow (at most) $4D$ rounds for this codeword to reach the receiver (for the edge-scheduling protocol, we only allowed $3D$ rounds; the extra $D$ rounds will be motivated below).  We will call each attempt to transfer a codeword a {\it transmission}, usually denoted by $\mathtt{T}$.  At the end of each transmission, the receiver will broadcast an {\it end of transmission} message, indicating whether it could successfully decode the codeword.  In the case that the receiver cannot decode, we will say that the transmission {\it failed}, and otherwise the transmission was {\it successful}.

As mentioned in Section \ref{highlights}, in the absence of a node-controlling adversary, the only difference between the present protocol and the one presented in Section \ref{nAP} is that digital signatures are used to authenticate the sender's packets and also accompany packet transfers for later use to identify corrupt nodes.  In the case a transmission fails, the sender will determine the reason for failure (cases 2-4 from Section \ref{highlights}, and also F2-F4 below), and request nodes to return {\it status reports} that correspond to a particular piece of signed communication between each node and its neighbors.  We will refer to status report packets as {\it parcels} to clarify discussion in distinguishing them from the codeword {\it packets}.

We give now a brief description of how we handle transmission failures in each of the three cases:
	\begin{enumerate}\setlength{\itemsep}{6pt} \setlength{\parskip}{0pt} \setlength{\parsep}{0pt}
	\item[\textsf{F2.}] The receiver could not decode, and the sender has inserted less than $D$ packets
	\item[\textsf{F3.}] The receiver could not decode, the sender has inserted $D$ packets, and the receiver has {\it not} received any duplicated packets corresponding to the current codeword
	\item[\textsf{F4.}] The receiver could not decode and cases F2 and F3 do not happen
	\end{enumerate}
Below is a short description of the specific kind of information nodes will be required to sign and store when communicating with their neighbors, and how this information will be used to identify corrupt nodes in each case F2-F4.
	\begin{itemize}\setlength{\itemsep}{4pt} \setlength{\parskip}{0pt} \setlength{\parsep}{0pt}
	\item[] \textsf{Case F2.}  Anytime a packet at height $h$ in an outgoing buffer of $A$ is transferred to an incoming buffer of $B$ at height $h'$, $A$'s potential will drop by $h$ and $B$'s potential will increase by $h'$.  So for directed edge $E(A,B)$, $A$ and $B$ will each need to keep two values, the cumulative decrease in $A$'s potential from packets leaving $A$, and the cumulative increase in $B$'s potential from those packets entering $B$.  These quantities are updated every time a packet is transferred across the edge, along with a tag indicating the round index and a signature from the neighbor validating the quantities and round index.  Loosely speaking, case F2 corresponds to {\it packet duplication}.  If a corrupt node attempts to slow transmission by duplicating packets, that node will have introduced extra potential in the network that cannot be accounted for, and the signing of potential changes will allow us to identify such a node.
	
	\item[] \textsf{Case F3.} $A$ and $B$ will keep track of the net number of packets that have travelled across edge $E(A,B)$.  This number is updated anytime a packet is passed across the edge, and the updated quantity, tagged with the round index, is signed by both nodes, who need only store the most recent quantity.  Loosely speaking, case F3 corresponds to {\it packet deletion}.  In particular, the information signed here will be used to find a node who input more packets than it output, and such that the node's capacity to store packets in its buffers cannot account for the difference.

	\item[] \textsf{Case F4.} For each packet $p$ corresponding to the current codeword, $A$ and $B$ will keep track of the net number of times the packet $p$ has travelled across edge $E(A,B)$.  This quantity is updated every time $p$ flows across the edge, and the updated quantity, tagged with the round index, is signed by both nodes, who for each packet $p$ need only store the most recent quantity.  We will show in Section \ref{MProofs} that whenever case F4 occurs, the receiver will have necessarily received a duplicated packet (corresponding to the current codeword).  Therefore, the information signed here will allow the sender to track this duplicated packet, looking for a node that outputted the packet more times than it inputted the packet.
	\end{itemize}
We will prove that whenever cases F2-F4 occur, if the sender has all of the relevant quantities specified above, then he will necessarily be able to identify a corrupt node.  Notice that each case of failure requires each node to transfer back only {\it one} signed quantity for each of its edges, and so the sender only needs $n$ status report packets from each node.

We will show that the maximum number of failed transmissions that can occur before a corrupt node is necessarily identified and eliminated is $n$.\footnote{As mentioned in Section \ref{highlights}, the sender can eliminate a corrupt node as soon as he has received the status reports from every non-blacklisted node.  The reason we require up to $n$ transmissions to guarantee the identification of a corrupt node is that it may take this long for the sender to have the complete information he needs.}  Because there are fewer than $n$ nodes that can be corrupted by the adversary, the cumulative number of failed transmissions is bounded by $n^2$.  This provides us with our main theorem regarding efficiency, stated precisely in Theorem \ref{mainAdThm} in Section \ref{adAnalysis} (note that the additive term that appears there comes from the $n^2$ failed transmissions).
\subsection{Detailed Description of the Node-Controlling $+$ Edge-Scheduling Protocol}\label{adDescLong}
\indent \indent In this section we give a more thorough description of our routing protocol for the node-controlling adversary model.  Formal pseudo-code can be found in Section \ref{MpseudoCode}.\\ \hspace*{\fill} \\
\noindent{\bf Setup.} As in the edge-scheduling protocol, the sender has a sequence of messages $\{m_1, m_2, \dots\}$ that he will divide into {\it codewords} $\{b_1, b_2, \dots\}$.  However, we demand each message has size $M' =\frac{6\sigma (P-2k)n^3}{\lambda}$, so that codewords have size $C' = \frac{M'}{\sigma}$, and these are then divided into packets of size $P' = P-2k$, which will allow packets to have enough room\footnote{The network is equipped with some minimal {\it bandwidth}, by which we mean the the number of bits that can be transferred by an edge in a single round.  We will divide codewords into blocks of this size and denote the size by $P$, which therefore simultaneously denotes the size of any packet and also the network bandwidth.  As in the edge-scheduling protocol, we assume $P > O(k+\log n)$, so that $P'$ is well-defined.} to hold two signatures of size $k$.  Notice that the number of packets per codeword is $D = \frac{C'}{P'} = \frac{6(P-2k)n^3}{(P-2k)\lambda} = \frac{6n^3}{\lambda}$, which matches the value of $D$ for the edge-scheduling protocol.  One of the signatures that packets will carry with them comes from the sender, who will authenticate every packet by signing it, and the packets will carry this signature until they are removed from the network by the receiver.  We re-emphasize Fact 1 from the edge-scheduling protocol, which remains true with these new values:
	\begin{itemize}\setlength{\itemsep}{1pt} \setlength{\parskip}{0pt} \setlength{\parsep}{0pt}
	\item[] {\bf Fact 1$'$.}  If the receiver has obtained $D-6n^3 = (1-\lambda)\left(\frac{6n^3}{\lambda}\right)$ distinct and un-altered packets from any codeword, he will be able to decode the codeword to obtain the corresponding message.
	\end{itemize}
The primary difference between the protocol we present here and that presented in Section \ref{nAP} is the need to maintain and transmit information that will allow the sender to identify corrupt nodes.  To this end, as part of the Setup, each node will have additional buffers:
	\begin{itemize}\setlength{\itemsep}{6pt} \setlength{\parskip}{0pt} \setlength{\parsep}{0pt}
	\item[3.] {\it Signature Buffers.}  Each node has a signature buffer along each edge to keep track of (outgoing/incoming) information exchanged with its neighbor along that edge.  The signature buffers will hold information corresponding to changes in the following values {\it for a single transmission}.  The following considers $A$'s signature buffer along directed edge $E(A,B)$: 1) The net number of packets passed across $E(A,B)$; 2) $B$'s cumulative change in potential due to packet transfers across $E(A,B)$.  Additionally, for codeword packets corresponding to the current transmission only, the signature buffers will hold: 3) For each packet $p$ that $A$ has seen, the net number of times $p$ has passed across $E(A,B)$ during the current transmission.  
	
	\hspace{.5cm}Each of the three items above, together with the current round index and transmission index, ahve been signed by $B$.  Since only a single (value, signature) pair is required for items 1)-2), while item 3) requires a (value, signature) pair for each packet, each signature buffer will need to hold at most $2+D$ (value, signature) pairs.

	\item[4.] {\it Broadcast Buffer.}  This is where nodes will temporarily store their neighbor's (and their own) state information that the sender will need to identify malicious activity.  A node $A$'s broadcast buffer will be able to hold the following (signed) information: 1) One {\it end of transmission} parcel (described below); 2) Up to $n$ different {\it start of transmission} parcels (described below); 3) A list of blacklisted nodes and the transmission they were blacklisted or removed from the blacklist; 4) For each $B \in G$, a list of nodes for which $B$ has claimed knowledge of their status report; and 5) For each $B \in G$ (including $B=A$), $A$'s broadcast buffer can hold the content of up to $n$ slots of $B$'s signature buffer.  Additionally, the broadcast buffer will also keep track of the {\it edges} across which it has already passed broadcasted information.
	
	\item[5.] {\it Data Buffer.}  This keeps a list of 1) All nodes that have been identified as corrupt and eliminated from the network by the sender; 2) Currently blacklisted nodes; and 3) Each node $A \in G$ will keep track of all pairs $(N_1, N_2)$ such that $N_2$ is on the blacklist, and $N_1$ claims to know $N_2$'s complete status report.
	\end{itemize}
The sender's buffers are the same as in the edge-scheduling protocol's Setup, with the addition of four more buffers:
	\begin{itemize}\setlength{\itemsep}{6pt} \setlength{\parskip}{0pt} \setlength{\parsep}{0pt}
	\item[-] {\it Data Buffer.}  Stores all necessary information from the status reports, as well as additional information it will need to identify corrupt nodes.  Specifically, the buffer is able to hold: 1) Up to $n$ status report parcels from each node; 2) For up to $n$ transmissions, the reason for failure, including the label of a duplicated packet (if relevant); 3) For up to $n$ failed transmissions, a {\it participating list} of up to $n$ nodes that were not on the blacklist for at least one round of the failed transmission; 4) A list of eliminated nodes and of blacklisted nodes and the transmission they were blacklisted; 5) The same as items 1)-3) of an internal node's Signature Buffer; and 6) The same as item 3) of an internal node's data buffer.

	\item[-] {\it Broadcast Buffer.} Holds up to $2n$ {\it start of transmission} parcels and the labels of up to $n-1$ nodes that should be removed from the blacklist.

	\item[-] {\it Copy of Current Packets Buffer.}  Maintains a copy of all the packets that are being sent in the current transmission (to be used any time a transmission fails and needs to be repeated).
	\end{itemize}
The receiver's buffers are as in the Edge-Scheduling protocol's Setup, with the addition of a {\it Broadcast Buffer}, {\it  Data Buffer}, and {\it Signature Buffers} which are identical to those of an internal node.

The rest of the Setup is as in the edge-scheduling model, with the added assumption that each node receives a private key from a trusted third party for signing, and each node receives public information that allows them to verify the signature of every other node in the network.\\ \hspace*{\fill} \\
\noindent {\bf Routing Phase.}  As in the edge-scheduling protocol, rounds consist of two stages followed by re-shuffling packets locally.  The main difference between the two protocols will be the addition of signatures to all information, as well as the need to transmit the broadcast information, namely the status reports and {\it start} and {\it end of transmission} broadcasts, which inform the nodes of blacklisted and eliminated nodes and request status reports in the case a transmission fails.  The two stages of a round are divided as they were for the edge-scheduling protocol, with the same treatment of routing codeword packets (with the addition of signatures).  However, we will also require that each round allows all edges the opportunity to transmit broadcast information (e.g.\ status report parcels).  Therefore, for every directed edge and every round of a transmission, there are four main packet transfers, two in each direction: codeword packets, broadcast parcels, and signatures confirming receipt of each of these.  The order of transmitting these is succinctly expressed below in Figure \ref{routingPic}.

At the start of any transmission $\mathtt{T}$, the sender will determine which codeword is to be sent (the next one in the case the previous transmission was successful, or the same one again in the case the previous transmission failed).  He will fill his Copy Of Current Packets Buffer with a copy of all of the codeword packets (to be used in case the transmission fails), and then fill his outgoing buffers with packets corresponding to this codeword.  All codeword packets are signed by the sender, and these signatures will remain with the packets as they travel through the network to the receiver.
\begin{figure}[ht]
\begin{center}
$\begin{array}{c|lcl} \mbox{\underline{\Large Stage}} & \mbox{\hspace{1cm} \underline{\hspace{1.5cm} \Large A \hspace{1.5cm}}} & & \hspace{1cm}\mbox{\underline{\hspace{1.5cm} \Large B \hspace{1.5cm}}} \\  & & & \\ & H_A:=\mbox{Height of buffer along ${\scriptstyle E(A,B)}$} & & \\ \begin{array}{c} \\ 1\end{array} & \hspace{-5pt}\begin{array}{l} \mbox{Height of flagged p$.$ (if there is one)} \\ \mbox{Round prev$.$ packet was sent} \end{array} & \longrightarrow & \\ & \mbox{Confirmation of rec$.$ of broadcast info$.$} & & \\& & \longleftarrow & \begin{array}{l} H_B\mbox{:=Height of buffer along ${\scriptstyle E(A,B)}$} \\ \mbox{Round prev$.$ packet was received} \\ \mbox{Sig's on values for edge ${\scriptstyle E(A,B)}$} \end{array}\\ \hline & \mbox{Send p$.$ and Sig's on values for ${\scriptstyle E(A,B)}$ if:} & & \thickspace \mbox{Receive packet if:}\\ & \hspace{.08cm} {\scriptscriptstyle \bullet} \hspace{.05cm}  \mbox{\small ``Enough'' bdcst info has passed ${\scriptstyle E(A,B)}$, }\mbox{\scriptsize AND}& & \hspace{.2cm} {\scriptscriptstyle \bullet} \hspace{.05cm}  \mbox{\small ``Enough'' bdcst info has passed ${\scriptstyle E(A,B)}$}\\ 2&  \hspace{.08cm} {\scriptscriptstyle \bullet} \hspace{.05cm}  \mbox{\small $B$ is not on $A$'s blacklist/eliminated, }\mbox{\scriptsize AND}&  \longrightarrow & \hspace{2cm} \mbox{\scriptsize AND}\\  & \hspace{1cm} - \hspace{.1cm} H_A > H_B \quad \mbox{\footnotesize OR} & & \hspace{.2cm} {\scriptscriptstyle \bullet} \hspace{.05cm}  \mbox{\small $A$ is not on $B$'s blacklist/eliminated}\\ & \hspace{1cm} - \hspace{.1cm} \mbox{\small $B$ didn't rec$.$ prev$.$ packet sent} & & \\ & & & \\ & & \longleftarrow & \mbox{Broadcast Information}\end{array}$
\end{center}
\vspace{-.6cm}
\caption{\it Description of Communication Exchange Along Directed Edge $E(A,B)$ During the Routing Phase of Some Round.}
\label{routingPic}
\end{figure}

To compliment Figure \ref{routingPic} above, we provide a breif description of the information that should be passed across directed edge $E(A,B)$ ($B \neq S$ and $A \neq R$, and $A, B \in G$, i.e.\ not eliminated) during some transmission $\mathtt{T}$.  The precise and complete description can be found in the pseudo-code of Section \ref{MpseudoCode}.  We state once and for all that if a node ever receives inaccurate or mis-signed information, it will act as if no information was received at all (e.g.\ as if the edge had failed for that stage).
	\begin{itemize}\setlength{\itemsep}{8pt} \setlength{\parskip}{0pt} \setlength{\parsep}{8pt}
	\item[] {\bf Stage 1.}  $A$ will send the same information to $B$ as in the edge-scheduling protocol (height, height of flagged packet, round packet was flagged).  Also, if $A$ received a valid broadcast buffer from $B$ in Stage 2 of the previous round, then $A$ will send $B$ confirmation of this fact.  Also, if $A$ knows that $B$ has crucial broadcast information $A$ needs, $A$ specifies the type of broadcast information he wants from $B$.  Meanwhile, $A$ should receive the seven items that $B$ signed and sent (see below), updating his internal variables as in the edge-scheduling protocol and updating its signature buffer, {\it provided} $B$ has given a valid\footnote{Here, ``valid'' means that $A$ agrees with all the values sent by $B$, and $B$'s signature is verified.} signature.\vspace{.35cm}
	
	At the other end, $B$ will send the following seven items to $A$: 1) the transmission index; 2) index of the current round; 3) current height;  4) index of the round $B$ last received a packet from $A$; 5) the net change in packet transfers so far across $E(A,B)$ for the current transmission; 6) $B$'s cumulative increase in potential due to packet transfers across $E(A,B)$ in the current transmission; and 7) if a packet $p$ was sent and received in Stage 2, the net number of times $p$ has been transferred across $E(A,B)$ for the current transmission.  Meanwhile, $B$ will also receive the information $A$ sent.

	\item[] {\bf Stage 2.}  $A$ will send a packet to $B$ under the same conditions as in the edge-scheduling protocol, with the additional conditions: 1) $A$ has received the sender's {\it start of transmission} broadcast (see below), and this information has passed across $E(A,B)$ or $E(B,A)$; 2) $A$ and $B$ are not on ($A$'s version of) the blacklist; 3) $A$ does not have any {\it end of transmission} information not yet passed across $E(A,B)$ or $E(B,A)$; and 4) $A$ does not have any {\it changes to blacklist} information yet to pass across $E(A,B)$ or $E(B,A)$.  We emphasize these last two points: if $A$ (including $A=S$) has any {\it start of transmission}, {\it end of transmission}, or {\it changes to blacklist} information in its broadcast buffer that it has not yet passed along edge $E(A,B)$, then it will not send a packet along this edge.
	
	\hspace{.5cm}If $A$ does send a packet, the ``packet'' $A$ sends includes a signature on the following seven items\footnote{Recall that packets have room to hold two signatures.  The first will be the sender's signature that accompanies the packet until the packet is removed by the receiver.  The second signature is the one indicated here, and this signature will be replaced/overwritten by the sending node every time the packet is passed across an edge.}: 1) transmission index; 2) index of the current round; 3) the packet itself with sender's signature;  4) index of the round $A$ first tried to send this packet to $B$; 5) {\it One plus} the net change in packet transfers so far across $E(A,B)$ for the current transmission; 6) $A$'s cumulative decrease in potential due to packet transfers across $E(A,B)$ in the current transmission {\it including the potential drop due to the current packet being transferred}; and 7) {\it One plus} the net number of times the packet currently being transferred has been transferred across $E(A,B)$ for the current transmission.
	
	\hspace{.5cm}Also, $A$ should receive broadcast information (if $B$ has something in its broadcast buffer not yet passed along $E(A,B)$) and update its broadcast buffer as described by the {\it Update Broadcast Buffer Rules} below.\vspace{.35cm}

	At the other end, $B$ will receive and store the packet sent by $A$ as in the edge-scheduling protocol, updating his signature buffer appropriately, with the added conditions: 1) $B$ has received the sender's {\it start of transmission} broadcast, and this information has been passed across $E(A,B)$ or $E(B,A)$; 2) The packet has a valid signature from $S$; 3) $A$ and $B$ are not on ($B$'s version of) the blacklist; 4) $B$ does not have any {\it end of transmission} information not yet passed across $E(A,B)$ or $E(B,A)$; 5) $B$ does not have any {\it changes to blacklist} information yet to pass across $E(A,B)$ or $E(B,A)$; and 6) The signatures on the seven items included with the packet from $A$ is ``valid$.$\setcounter{validFootnote}{\value{footnote}}\addtocounter{validFootnote}{-1}\footnotemark[\value{validFootnote}]''  Additionally, if there is anything in $B$'s broadcast buffer that has not been transferred along $E(A,B)$ yet, then $B$ will send one parcel of broadcast information chosen according to the priorities: 1) The receiver's {\it end of transmission} parcel; 2) One of the sender's {\it start of transmission} parcels; 3) Changes to the blacklist or a node to permanently eliminate; 4) The identity of a node $N$ on $B$'s blacklist for which $B$ has complete knowledge of $N$'s status report; 5) The most recent status report parcel $A$ requested in Stage 1 of an earlier round; and 6) Arbitrary status report parcels.
	\end{itemize}
For any edge $E(A,S)$ or $E(R,B)$, only broadcast information is passed along these directed edges, and this is done as in the rules above, with the exceptions mentioned below in the {\it Update Broadcast Buffer Rules} for $S$ and $R$. Additionally, any round in which the sender is unable to insert any packets, he will increase the number of blocked rounds in his data buffer.  The sender will also keep track of the total number of packets inserted in the current transmission in his data buffer.\\ \hspace*{\fill} \\
\noindent{\bf Re-Shuffle Rules.}  The Re-Shuffle rules are exactly as in the edge-scheduling protocol, with the exception that node's record the changes in non-duplicated potential caused by re-shuffling packets locally.\\ \hspace*{\fill} \\
\noindent{\bf Update Broadcast Buffer Rules for Internal Nodes.}
Looking at Stage 2 of the Routing Rules above, we have that in every round and for every edge $E(A,B)$, each node $B$ will have the opportunity to send and receive exactly one parcel of {\it broadcast} information in addition to a single codeword packet.  The order in which a node's broadcast buffer information is transmitted is described above in the Routing Rules.

For any node $A \in \mathcal{P} \setminus \{R,S\}$, we now describe the rules for updating their broadcast buffer when they receive broadcast information.  Assume that $A$ has received broadcast information along one of its edges $E(A,B)$, and has verified that it has a valid signature.  We describe how $A$ will update its Broadcast Buffer, depending on the nature of the new information:
	\begin{itemize}\setlength{\itemsep}{8pt} \setlength{\parskip}{0pt} \setlength{\parsep}{0pt}
	\item[] \textsf{The received information is the receiver's {\it end of transmission} broadcast} (see below).  In this case, $A$ will first make sure the transmission index is for the current transmission, and if so, the information is added to $A$'s broadcast buffer, and edge $E(A,B)$ is marked as having already transmitted this information.

	\item[] \textsf{The received information is the sender's {\it start of transmission} broadcast} (see below).  This broadcast consists of a single parcel containing information about the previous transmission, followed by up to $2n-2$ additional parcels (describing blacklisted/eliminated nodes and labels of up to $n$ previous transmissions that have failed).  When $A$ receives a parcel from the {\it start of transmission} broadcast, if $A$ does not already have it stored in its broadcast buffer, it will add it, and edge $E(A,B)$ is marked as having already transmitted this information.  Additionally, $A$ will handle the parcels concerning blacklisted nodes as described below.  Finally, when $A$ has received every parcel in the {\it start of transmission} broadcast, it will also remove from its blacklist any node {\it not} implicated in this broadcast (i.e.\ this will count as ``$A$ receives information concerning a node to remove from the blacklist,'' see below), as well as clearing its signature buffers for the new transmission.

	\item[] \textsf{The received information indicates a node $N$ to eliminate.}  If the information is current, then $A$ will add the new information to its broadcast buffer and mark edge $E(A,B)$ as already having passed this information.  If $N$ is not already on $A$'s list of eliminated nodes $EN$, then $A$ will add $N$ to $EN$ (in its data buffer), clear all of its incoming and outgoing buffers, its signature buffers, and its broadcast buffer (with the exception of {\it start of transmission} parcels).

	\item[] \textsf{The received information concerns a node $N$ to add to or remove from the blacklist.}  If the received information {\it did not} originate in the current transmission (as signed by the sender) or $A$ has more recent blacklist information regarding $N$, then $A$ ignores the new information.  Otherwise, the information is added to $A$'s broadcast buffer, and edge $E(A,B)$ is marked as having already transmitted this information.  Additionally, parcels that are now outdated in $A$'s broadcast buffer are deleted (such as $N$'s status report parcels or a parcel indicating a node $\widehat{N}$ had $N$'s complete status report).
	
	\hspace{.5cm}If $N=A$, then $A$ adds $n$ of its own status report parcels to its broadcast buffer, choosing these $n$ parcels based on information from the relevant {\it start} and {\it end of transmission} parcels.  $A$ also will add his own signature to each of these parcels, so that they each one will carry two signatures back to the sender ($A$'s signature and the relevant neighbor's signature).
	
	\item[] \textsf{The received information indicates $B$ has some node $N$'s complete status report.}  If $N$ is on $A$'s blacklist for the transmission $B$ claims knowledge for, then $A$ stores the fact that $B$ has complete knowledge of $N$ in its data buffer ($A$ will later use this information when requesting a specific parcel from $B$).

	\item[] \textsf{The received information pertains to some node $N$'s status report the sender has requested.}  $A$ will first make sure that $N$ is on its version of the blacklist.  If so, the information is added to $A$'s broadcast buffer, and edge $E(A,B)$ is marked as having already transmitted this information.  If this completes $A$'s knowledge of $N$'s status report, $A$ will add to its broadcast buffer the fact that it now knows all of $N$'s missing status report.
	\end{itemize}
At the end of a transmission, all nodes will clear their broadcast buffers {\it except} parcels concerning a blacklisted node's status report.  All nodes also clear their version of the blacklist (it will be restored at the beginning of the next transmission).\\ \hspace*{\fill} \\
\noindent{\bf Update Broadcast Buffer Rules for Sender and Receiver.}  The receiver has the same rules as internal nodes for updating its broadcast buffer, with the addition that when there are exactly $n$ rounds left in any transmission\footnote{We note that because there is always an active honest path between sender and receiver, and the receiver's final broadcast has top priority in terms of broadcast order, the sender will necessarily receive this broadcast by the end of the transmission.  Alternatively, we could modify our protocol to add an extra $n$ rounds to allow this broadcast to reach the sender.  However, the exposition is easier without adding an extra $n$ rounds, and we will show that wasting the final $n$ rounds of a transmission by having the receiver determine if it can decode with $n$ rounds still left is not important, as the $n$ wasted rounds is insignificant compared to the $O(n^3)$ rounds per transmission.}, $R$ will add to its broadcast buffer a single (signed) parcel that indicates the transmission index, whether or not he could decode the current codeword, and the label of a duplicate packet he received (if there was one).  We will refer to this as the receiver's {\it end of transmission} parcel.

The rules for the sender updating its broadcast buffer are slightly more involved, as the sender will be the one determining which information it requires of each node, as well as managing the blacklisted and eliminated nodes.  The below rules dictate how $S$ will update his broadcast buffer and status buffer at the end of a transmission, or when the sender receives new (appropriately signed) broadcast information along $E(S,B)$.
	\begin{itemize}\setlength{\itemsep}{8pt} \setlength{\parskip}{0pt} \setlength{\parsep}{0pt}
	\item[] \textsf{A transmission $\mathtt{T}$ has just ended.}  Note that the sender will have necessarily received and stored the receiver's {\it end of transmission} broadcast by the end of the transmission (Lemma \ref{finalBC}).  In the case that the transmission was successful, $S$ will clear his outgoing buffers and Copy of Old Packets Buffer, then re-fill them with codewords corresponding to the next message.  If $EN$ denotes the eliminated nodes and $\mathcal{B}_{\mathtt{T}}$ denotes the nodes on the sender's blacklist at the end of this transmission, then the sender will set $\Omega_{\mathtt{T}+1} = (|EN|, |\mathcal{B}_{\mathtt{T}}|, F, 0)$, where $F$ denotes the number of failed transmissions that have taken place since the last corrupt node was eliminated.  Finally, the sender will delete the information in his data buffer concerning the number of packets inserted and the number of blocked rounds for the just completed transmission.\vspace{.2cm}
	
	In the case that the transmission failed, $S$ will clear his outgoing buffers and then re-fill them using the Copy of Old Packets Buffer, while leaving the latter buffer unchanged.  The sender will determine the reason the transmission failed (F2-F4), and add this fact along with the relevant information from his own signature buffers to the data buffer.  For any node (not including $S$ or $R$ or eliminated nodes) not on the sender's blacklist, the sender will add the node to the {\it participating list} $\mathcal{P}_{\mathtt{T}}$ in his data buffer, and then add each of these nodes to the blacklist, recording that $\mathtt{T}$ was the most recent transmission that the node was blacklisted.  Also, the sender will set:
		\begin{equation*}
		\Omega_{\mathtt{T}+1} = \left\{ \begin{array}{cl} (|EN|, |\mathcal{B}_{\mathtt{T}}|, F, p) & \mbox{if the transmission failed and $p$ was included in} \\ & \quad \mbox{the receiver's {\it end of transmission} parcel} \\ (|EN|, |\mathcal{B}_{\mathtt{T}}|, F, 1) & \mbox{if the transmission failed and $S$ inserted $D$ packets} \\ (|EN|, |\mathcal{B}_{\mathtt{T}}|, F, 2) & \mbox{otherwise}\end{array} \right.
		\end{equation*}
	For both a failed transmission and a successful transmission, the sender will sign and add to his broadcast buffer the following parcels, which will comprise the sender's {\em Start of Transmission} (SOT) broadcast: $\Omega_{\mathtt{T}+1}$, a list of eliminated and blacklisted nodes, and the reason for failure of each of the last (up to $n-1$) failed transmissions since the last node was eliminated.  Note that each parcel added to the broadcast buffer regarding a blacklisted node includes the index of the transmission for which the node was blacklisted, and all parcels added to the broadcast buffer include the index of the transmission about to start (as a timestamp) and are signed by $S$.  Notice that the rules regarding priority of transferring broadcast information guarantee that $\Omega_{\mathtt{T}+1}$ will be the first parcel of the {\it SOT} that is received, and because it reveals the number of blacklisted and eliminated nodes and the number of failed transmissions to expect, as soon as each node receives $\Omega_{\mathtt{T}+1}$, they will know exactly how many more parcels remain in the {\it SOT} broadcast.  Nodes will not be allowed to transfer any (codeword) packets until the {\it SOT} broadcast for the current transmission is received in its entirety.
	
	\item[] \textsf{The sender receives the receiver's {\it end of transmission} parcel for the current transmission}.  The sender will store this parcel in its data buffer.

	\item[] \textsf{The sender receives information along $E(S,B)$ indicating $B$ has some node $N$'s complete status report.}  If $N$ is on the sender's blacklist for the transmission $B$ claims knowledge for, then the sender stores the fact that $B$ has complete knowledge of $N$ in its data buffer (the sender will later use this information when requesting a specific parcel from $B$).
	
	\item[] \textsf{The sender receives information along $E(S,B)$ that pertains to some node $N$'s status report that the sender has requested.}  The sender will first make sure that $N$ is on its blacklist and that the parcel received contains the appropriate information, i.e.\ the sender checks its data buffer to see which transmission $N$ was added to the blacklist and the reason this transmission failed, and makes sure the status report parcel is from this transmission and contains the information corresponding to this reason for failure.  If the parcel has faulty information that has been signed by $N$, i.e.\ $N$ sent back information that was not requested by the sender, then $N$ is eliminated from the network.  Otherwise, the sender will add the information to its data buffer.  If the information completes $N$'s missing status report, the sender updates his broadcast buffer indicating $N$'s removal from the blacklist, including the index of the current transmission and his signature.
	
	\hspace{.5cm}The sender will then determine if he has enough information to eliminate a corrupt node $N'$.  If so, $N'$ will be added to his list of {\it Eliminate Nodes}, and his broadcast buffer and data buffer will be wiped completely ({\it except} for the list of eliminated nodes).  Also, he will abandon the current transmission and begin a new one corresponding to the same codeword.  In particular, he will clear his outgoing buffers and re-fill them with the codewords in his Copy of Old Packets Buffer, leaving the latter unchanged, and he will skip to the ``A transmission has just ended'' case above, setting the {\it start of transmission} parcel $\Omega_{\mathtt{T}+1} = (|EN|, 0, 0, 0)$.
	\end{itemize}
\subsection{Analysis of Our Node-Controlling $+$ Edge-Scheduling Protocol} \label{adAnalysis}
\indent \indent We state our results concerning the correctness, throughput, and memory of our adversarial routing protocol, leaving the analysis and proofs to Section \ref{MProofs}.
\begin{thm} \label{mainAdThm} Except for the at most $n^2$ transmissions that may fail due to malicious activity, our Routing Protocol enjoys linear throughput.  More precisely, after $x$ transmissions, the receiver has correctly outputted at least $x-n^2$ messages.  If the number of transmissions $x$ is quadratic in $n$ or greater, than the failed transmissions due to adversarial behavior become asymptotically negligible.  Since a transmission lasts $O(n^3)$ rounds and messages contain $O(n^3)$ bits, information is transferred through the network at a linear rate.\end{thm}
%
\vspace{-.1cm}
\begin{thm} \label{memoryThm} The memory required of each node is at most $O(n^4(k+\log n))$.\end{thm}
\begin{proof}[{\it Proofs.}] See Section \ref{MProofs}.\end{proof}
\section{Pseudo-Code for Node-Controlling $+$ Edge-Scheduling Protocol}\label{MpseudoCode}
\indent \indent We now modify the pseudo-code from our edge-scheduling adversarial protocol to pseudo-code for the (node-controlling $+$ edge-scheduling) adversarial model.  The two codes will be very similar, with differences emphasized by marking the line number in {\bf bold}.  The Re-Shuffle Rules will remain the same as in the edge-scheduling protocol, with the addition of line ({\bf \ref{reShuffleRules}.76}) (see Figure \ref{reShuffleRules}).\newpage 
\begin{figure}[h!t]
\begin{center}
\framebox{\footnotesize
\begin{minipage}[b]{6.0in}
\begin{footnotesize}
\begin{tabbing}
aaa\=aa\=aa\=aa\=aa\=aaaaaaaaaa\=aaaaaaaaaaa\=aaaaaaa\=aaaaaaa\=aaaaaaaaaaaaaaaaaa\=aaaaaa\=aaaaaaaa \kill
{\bf Setup} \\
\>{\bf DEFINITION OF VARIABLES}:\\
01\>\>$n :=$ Number of nodes in $G$;\\
02\>\>$D :=\frac{3n^3}{\lambda}$;\\
03\>\>$\mathtt{T} :=$ Transmission index;\\
04\>\>$\mathtt{t} :=$ Stage/Round index;\\
{\bf 05}\>\>$k :=$ Security Parameter;\\
{\bf 06}\>\>$P :=$ Capacity of edge= $O(k+\log n)$;\\
{\bf 07}\>\>\For $N \in \mathcal{P} \setminus S$\\
{\bf 08}\>\>\>$BB \in [n^2+5n] \times \{0,1\}^{P+n}$; \>\>\>\>\>{\scriptsize \#\# Broadcast Buffer}\\
{\bf 09}\>\>\>$DB \in [1..n^2] \times \{0,1\}^P$; \>\>\>\>\>{\scriptsize \#\# Data Buffer.  Holds $BL$ and $EN$ below, and info$.$ as on line 151}\\
{\bf 10}\>\>\>\>$BL \in [1..n-1] \times \{0,1\}^{P}$;\>\>\>\>{\scriptsize \#\# Blacklist}\\
{\bf 11}\>\>\>\>$EN \in [1..n-1] \times \{0,1\}^{P}$;\>\>\>\>{\scriptsize \#\# List of Eliminated Nodes}\\
{\bf 12}\>\>\>$SIG_{N,N} \in \{0,1\}^{O(\log n)}$;\>\>\>\>\>{\scriptsize \#\# Holds change in potential due to local re-shuffling of packets}\\
13\>\>\For $N \in G$\\
{\bf 14}\>\>\>$SK, \{PK \}^n_i$ \>\>\>\>\>{\scriptsize \#\# Secret Key for signing, Public Keys to verify sig's of all nodes}\\
15\>\>\>\For {\it outgoing} {\bf edge} $E(N,B) \in G, B \neq S$ and $N \neq R$\\
16\>\>\>\>$\mathsf{OUT} \in [2n] \times \{0,1\}^{P}$; \>\>\>\>{\scriptsize \#\# Outgoing Buffer able to hold $2n$ packets}\\
{\bf 17}\>\>\>\>$SIG_{N,B} \in [D+3] \times \{0,1\}^{O(\log n)}$;\>\>\>\>{\scriptsize \#\# Signature Buffer for current trans., indexed as follows:}\\
\>\>\>\>\>\>\>\>{\scriptsize \#\# $SIG[1]$= net no$.$ of {\it current codeword} p's transferred across $E(N,B)$}\\
\>\>\>\>\>\>\>\>{\scriptsize \#\# $SIG[2]$= net change in $B$'s pot$.$ due to p$.$ transfers across $E(N,B)$}\\
\>\>\>\>\>\>\>\>{\scriptsize \#\# $SIG[3]$= net change in $N$'s pot$.$ due to p$.$ transfers across $E(N,B)$}\\
\>\>\>\>\>\>\>\>{\scriptsize \#\# $SIG[p]$= net no. of times packet $p$ transferred across $E(N,B)$}\\
18\>\>\>\>$\tilde{p} \in \{0,1\}^P \cup \bot$; \>\>\>\>{\scriptsize \#\# Copy of packet to be sent}\\
19\>\>\>\>$sb \in \{0,1\}$; \>\>\>\>{\scriptsize \#\# Status bit}\\
20\>\>\>\>$d \in \{0,1\}$; \>\>\>\>{\scriptsize \#\# Bit indicating if a packet was sent in prev$.$ round}\\
21\>\>\>\>$FR \in [0..8D] \cup \bot$; \>\>\>\>{\scriptsize \#\# Flagged Round (index of round $N$ first tried to send $\tilde{p}$ to $B$)}\\
22\>\>\>\>$RR \in [-1..8D] \cup \bot$; \>\>\>\>{\scriptsize \#\# Round Received (index of round that $N$ last rec$.$ a p$.$ from $A$)}\\
23\>\>\>\>$H \in [0..2n]$; \>\>\>\>{\scriptsize \#\# Height of $\mathsf{OUT}$. Also denoted $H_{OUT}$ when there's ambiguity}\\
24\>\>\>\>$H_{FP} \in [1..2n] \cup \bot$; \>\>\>\>{\scriptsize \#\# Height of Flagged Packet}\\
25\>\>\>\>$H_{IN}\in [0..2n] \cup \bot$; \>\>\>\>{\scriptsize \#\# Height of incoming buffer of $B$}\\
{\bf 26}\>\>\>\For {\it outgoing} {\bf edge} $E(N,B) \in G$, {\it including} $B = S$ and $N=R$\\
{\bf 27}\>\>\>\>$bp \in \{0,1\}^P$; \>\>\>\>{\scriptsize \#\# Broadcast Parcel received along this edge}\\
{\bf 28}\>\>\>\>$\alpha \in \{0,1\}^P$; \>\>\>\>{\scriptsize \#\# Broadcast Parcel request}\\
{\bf 29}\>\>\>\>$c_{bp} \in \{0,1\}$;\>\>\>\>{\scriptsize \#\# Verification bit of broadcast parcel receipt}\\
30\>\>\>\For {\it incoming} {\bf edge} $E(A,N) \in G, A \neq R$ and $N \neq S$\\
31\>\>\>\>$\mathsf{IN} \in [2n] \times \{0,1\}^{P}$; \>\>\>\>{\scriptsize \#\# Incoming Buffer able to hold $2n$ packets}\\
{\bf 32}\>\>\>\>$SIG_{A,N} \in [D+3] \times \{0,1\}^{O(\log n)}$;\>\>\>\>{\scriptsize \#\# Signature Buffer, indexed as on line 17}\\
33\>\>\>\>$p \in \{0,1\}^P \cup \bot$; \>\>\>\>{\scriptsize \#\# Packet just received}\\
34\>\>\>\>$sb \in \{0,1\}$; \>\>\>\>{\scriptsize\#\# Status bit}\\
35\>\>\>\>$RR \in \{0,1\}^{8D}$; \>\>\>\>{\scriptsize \#\# Round Received index}\\
36\>\>\>\>$H \in [0..2n]$; \>\>\>\>{\scriptsize \#\# Height of $\mathsf{IN}$. Also denoted $H_{IN}$ when there's ambiguity}\\
37\>\>\>\>$H_{GP} \in [1..2n] \cup \bot$; \>\>\>\>{\scriptsize \#\# Height of Ghost Packet}\\
38\>\>\>\>$H_{OUT}\in [0..2n] \cup \bot$; \>\>\>\>{\scriptsize \#\# Height of outgoing buffer or height of Flagged Packet of $A$}\\
39\>\>\>\>$sb_{OUT} \in \{0,1\}$; \>\>\>\>{\scriptsize \#\# Status Bit of outgoing buffer of $A$}\\
40\>\>\>\>$FR \in \{0,1\}^{8D} \cup \bot$; \>\>\>\>{\scriptsize \#\# Flagged Round index (from adjacent outgoing buffer $A$)}\\
{\bf 41}\>\>\>\For {\it incoming} {\bf edge} $E(A,N) \in G$, {\it including} $A = R$ and $N =S$\\
{\bf 42}\>\>\>\>$bp \in \{0,1\}^P$;\>\>\>\>{\scriptsize \#\# Broadcast Parcel to send along this edge}\\
{\bf 43}\>\>\>\>$c_{bp} \in \{0,1\}$;\>\>\>\>{\scriptsize \#\# Verification bit of packet broadcast parcel receipt}
\end{tabbing}
\end{footnotesize}
\end{minipage} }
\end{center}
\caption{Pseudo-Code for Internal Nodes' Setup for the (Node-Controlling $+$ Edge-Scheduling) Protocol}
\label{setupCodeM}
\end{figure}
\newpage
\begin{figure}[h!t]
\begin{center}
\framebox{\footnotesize
\begin{minipage}[b]{6.0in}
\begin{footnotesize}
\begin{tabbing}
aaa\=aa\=aa\=aa\=aa\=aaaaaaaaaaaaaaaaaaaaaaaaaaaaaaaa\=aaaaaaaaaaaaaaaaaa\=a \kill
\>{\bf INITIALIZATION OF VARIABLES}:\\
44\>\>\For $N \in G$\\
{\bf 45}\>\>\>{\it Receive Keys};\>\>\>{\scriptsize \#\# Receive $\{PK \}^n_i$ and $SK$ from KEYGEN}\\
{\bf 46} \>\>\>{\it Initialize } $BB$, $DB$, $BL$, $EN$, $SIG_{N,N}$;\>\>\>{\scriptsize \#\# Set $SIG_{N,N}=0$, set each entry of $DB$ and $BB$ to $\bot$}\\
47\>\>\>\For {\it incoming} {\bf edge} $E(A,N) \in G, A \neq R$ and $N \neq S$\\
{\bf 48}\>\>\>\>{\it Initialize} $\mathsf{IN}$, $SIG$;\>\>{\scriptsize \#\# Set each entry in $\mathsf{IN}$ to $\bot$ and each entry of $SIG$ to zero}\\
49\>\>\>\>$p, H_{GP}, FR= \bot$;\\
50\>\>\>\>$sb, sb_{OUT}, c_p, H, H_{OUT} = 0$; $RR = -1$;\\
{\bf 51}\>\>\>\For {\it incoming} {\bf edge} $E(A,N) \in G$, {\it including} $A = R$ and $N =S$\\
{\bf 52}\>\>\>\>$bp = \bot$; $c_{bp} = 0$;\\
53\>\>\>\For {\it outgoing} {\bf edge} $E(N,B) \in G, B \neq S$ and $N \neq R$\\
{\bf 54}\>\>\>\>{\it Initialize } $\mathsf{OUT}$, $SIG$;\>\>{\scriptsize \#\# Set each entry in $\mathsf{OUT}$ to $\bot$ and each entry of $SIG$ to zero}\\
54\>\>\>\>$\tilde{p}, H_{FP},FR, RR = \bot$;\\
55\>\>\>\>$sb, d, H, H_{IN}, 0$;\\
{\bf 57}\>\>\>\For {\it outgoing} {\bf edge} $E(N,B) \in G$, {\it including} $B = S$ and $N=R$\\
{\bf 58}\>\>\>\>$bp, \alpha = \bot$; $c_{bp} = 0$; \\
\\
{\bf Sender's Additional Setup} \\
\>{\bf DEFINITION OF ADDITIONAL VARIABLES FOR SENDER}:\\
59\>\>$\mathcal{M} := \{m_1, m_2, \dots \} = $ Input Stream of Messages; \\
{\bf 60}\>\>$COPY \in [D] \times \{0,1\}^P:=$ Copy of Packets for Current Codeword;\\
{\bf 61}\>\>$BB \in [3n] \times \{0,1\}^P:=$ Broadcast Buffer;\\
{\bf 62}\>\>$DB \in [1..n^3+n^2+n] \times \{0,1\}^P:=$ Data Buffer, which includes:\\
{\bf 63}\>\>\>$BL \in [1..n] \times \{0,1\}^{P}:=$ Blacklist;\\
{\bf 64}\>\>\>$EN \in [1..n] \times \{0,1\}^{P}:=$ List of Eliminated Nodes;\\
65\>\>$\kappa \in [0..D]$ := Number of packets corresponding to current codeword the sender has {\it knowingly} inserted;\\
{\bf 66}\>\>$\Omega_{\mathtt{T}} \in \{0,1\}^{O(\log n)}$ := First parcel of Start of Transmission broadcast for transmission $\mathtt{T}$;\\
{\bf 67}\>\>$\beta_{\mathtt{T}} \in [0..4D]$ := Number of rounds blocked in current transmission;\\
{\bf 68}\>\>$F \in [0..n-1]$ := Number of failed transmissions since the last corrupt node was eliminated;\\
{\bf 69}\>\>$\mathcal{P}_{\mathtt{T}} \in \{0,1\}^n$ := Participating List for current transmission;\\[.3cm]
\>{\bf INITIALIZATION OF SENDER'S VARIABLES}:\\
70\>\>$\kappa = 0$;\\
{\bf 71}\>\>$\beta_{1}, F = 0$;\\
{\bf 72}\>\>$\Omega_{1} = (0,0,0,0)$;\\
{\bf 73}\>\>{\it Initialize }$BB$, $DB$, $\mathcal{P}_{1}$;\>\>\>\>{\scriptsize \#\# Set each entry of $DB$ to $\bot$, add $\Omega_1$ to $BB$, and set $\mathcal{P}_{1}=G$}\\
74\>\>{\bf \em Distribute Packets};\\
\\
{\bf Receiver's Additional Setup} \\
\>{\bf DEFINITION OF ADDITIONAL VARIABLES FOR RECEIVER}:\\
75\>\>$I_R \in [D] \times (\{0,1\}^P \cup \bot)$ := Storage Buffer to hold packets corresponding to current codeword;\\
76\>\>$\kappa \in [0..D]$ := Number of packets received corresponding to current codeword;\\
{\bf 77}\>\>$\Theta_{\mathtt{T}} \in \{0,1\}^{O(k+\log n)}$ := End of Transmission broadcast for transmission $\mathtt{T}$;\\[.3cm]
\>{\bf INITIALIZATION OF RECEIVER'S VARIABLES}:\\
78\>\>$\kappa = 0$;\\
{\bf 79}\>\>$\Theta_{1} = \bot$;\\
{\bf 80}\>\>\For {\it outgoing} {\bf edge} $E(R,B) \in G$:\\
{\bf 81}\>\>\>$bp, \alpha = \bot$;\\
82\>\>{\it Initialize }$I_R$;\>\>\>\>{\scriptsize \#\# Sets each element of $I_R$ to $\bot$}\\
{\bf End Setup}
\end{tabbing}
\end{footnotesize}
\end{minipage} }
\end{center}
\caption{Additional Setup Code for (Node-Controlling $+$ Edge-Scheduling) Protocol}
\label{setupCode2M}
\end{figure}
\newpage
\begin{figure}[h!t]
\begin{center}
\leavevmode
\framebox{\footnotesize
\begin{minipage}[b]{6.0in}
\begin{footnotesize}
\begin{tabbing}
aaa\=aa\=a\=a\=aa\=aa\=aa\=aaa\=aaa\=aaa\=aaaaaaaaaaaaaaaaa\=aaaaaaaaaaaaaaaa\=aaaa \kill
{\bf Transmission $\mathtt{T}$}\\[-1.300pt]
{\bf 01}\>\For $N \in G$, $N \notin EN$:\\[-1.300pt]
{\bf 02}\>\>\For $\mathtt{t} < 2*(4D)$\>\>\>\>\>\>\>\>\>\>{\scriptsize \#\# The factor of 2 is for the 2 stages per round}\\[-1.300pt]
03\>\>\>\>\If $\mathtt{t}$ (mod 2) = 0  \Then \>\>\>\>\>\>\>\>{\scriptsize \#\# \bf STAGE 1}\\[-1.300pt]
{\bf 04}\>\>\>\>\>{\bf \em Update Broadcast Buffer One};\\[-1.300pt]
05\>\>\>\>\>\For {\it outgoing} {\bf edge} $E(N,B) \in G, N \neq R, B \neq S$\\[-1.300pt]
06\>\>\>\>\>\>\If $H_{FP} \neq \bot$: \Send $(H, \bot, \bot)$; \quad \Else \Send $(H-1, H_{FP}, FR)$;\\[-1.300pt]
{\bf 07}\>\>\>\>\>\>\Receive {\scriptsize \it Signed($\mathtt{T}, \mathtt{t}, H_{IN}, RR, SIG[1], SIG[2], SIG[p]$)};\>\>\>\>\>\>{\scriptsize \#\# ${\scriptstyle SIG[3]}$, 6$^{th}$ coord sent on line 11, is kept as ${\scriptstyle SIG[2]}$}\\[-1.300pt]
{\bf 08}\>\>\>\>\>\>{\bf \em Verify Signature Two}:\\[-1.300pt]
09\>\>\>\>\>\>{\bf \em Reset Outgoing Variables};\\[-1.300pt]
10\>\>\>\>\>\For {\it incoming} {\bf edge} {\scriptsize $E(A,N) \in G, N \neq S, A \neq R$}\>\>\>\>\>\>\>{\scriptsize \#\# ``$p$'' on line 11 refers to last p. rec'd on ${\scriptstyle E(A,N)}$}\\[-1.300pt]
{\bf 11}\>\>\>\>\>\>\Send {\it Sign}$(\mathtt{T}, \mathtt{t}, H, RR, SIG[1], SIG[3], SIG[p])$;\>\>\>\>\>\>{\scriptsize \#\# If $p$ was from an {\it old} codeword, send instead:}\\[-1.300pt]
\>\>\>\>\>\>\>\>\>\>\>\>{\scriptsize \#\# {\it Sign}$(\mathtt{T}, \mathtt{t}, H, RR, SIG[1], SIG[3], \bot)$}\\[-1.300pt]
12\>\>\>\>\>\>$sb_{OUT}=0$; $FR = \bot$;\\[-1.300pt]
13\>\>\>\>\>\>\Receive $(H, \bot, \bot)$ \boldmath$\mathsf{or}$\unboldmath $\thickspace (H, H_{FP}, FR)$;\>\>\>\>\>\>{\scriptsize \#\# If $H= \bot$ or ${\scriptstyle FR > RR}$, set $sb_{\scriptscriptstyle OUT}$=$1$; and}\\[-1.300pt]
\>\>\>\>\>\>\>\>\>\>\>\>{\scriptsize \#\# $H_{\scriptscriptstyle OUT}$=$H_{FP}$; O.W$.$ set $H_{\scriptscriptstyle OUT}$=$H$; $sb_{\scriptscriptstyle OUT}$=$0$;}\\[-1.300pt]
14\>\>\>\>\ElseIf $\mathtt{t}$ (mod 2) = 1 \Then \>\>\>\>\>\>\>\>{\scriptsize \#\# \bf STAGE 2}\\[-1.300pt]
{\bf 15}\>\>\>\>\>{\bf \em Send/Receive Broadcast Parcels};\\[-1.300pt]
16\>\>\>\>\>\For {\it outgoing} {\bf edge} $E(N,B) \in G, N \neq R, B \neq S$\\[-1.300pt]
17\>\>\>\>\>\>\If $H_{IN}\neq \bot$ \Then \\[-1.300pt]
18\>\>\>\>\>\>\>{\bf \em Create Flagged Packet};\\[-1.300pt]
19\>\>\>\>\>\>\>\If $sb$=$1$ \boldmath$\mathsf{or}$\unboldmath $\thickspace (sb$=$0$ \boldmath$\mathsf{and}$\unboldmath $\thickspace H > H_{IN})$ \Then \\[-1.300pt]
20\>\>\>\>\>\>\>\>{\bf \em Send Packet};\\[-1.300pt]
21\>\>\>\>\>\For {\it incoming} {\bf edge} $E(A,N) \in G, N \neq S, A \neq R$\\[-1.300pt]
22\>\>\>\>\>\>{\bf \em Receive Packet};\\[-1.300pt]
\\[-1.300pt]
{\bf 23}\>\>\>\>\>\If $N \notin \{S,R\} \thickspace $\boldmath$\mathsf{and}\thickspace$\unboldmath $N$ has rec'd {\it SOT} broadcast for $\mathtt{T}$ \Then {\bf \em Re-Shuffle};\\[-1.300pt]
{\bf 24}\>\>\>\>\>\ElseIf $N = R \thickspace $\boldmath$\mathsf{and}\thickspace$\unboldmath $N$ has rec'd {\it SOT} broadcast for $\mathtt{T}$ \Then {\bf \em Receiver Re-Shuffle};\\[-1.300pt]
25\>\>\>\>\>\ElseIf $N = S$ \Then \\[-1.300pt]
26\>\>\>\>\>\>{\bf \em Sender Re-Shuffle};\\[-1.300pt]
{\bf 27}\>\>\>\>\>\>\If {\it All (non-$\bot$) values $S$ received on line 07 had $H_{IN}=2n$} \Then $\beta_{\mathtt{T}} = \beta_{\mathtt{T}}+1$;\\[-1.300pt]
\\[-1.300pt]
{\bf 28}\>\>\>\>\>\If $\mathtt{t}=2(4D-n) \thickspace $\boldmath$\mathsf{and}\thickspace$\unboldmath $N=R$ \Then {\bf \em Send End of Transmission Parcel};\\[-1.300pt]
{\bf 29}\>\>\>\>\>\If $\mathtt{t}=2(4D) \thickspace $\boldmath$\mathsf{and}\thickspace$\unboldmath $N=S$ \Then {\bf \em Prepare Start of Transmission Broadcast};\\[-1.300pt]
30\>\>\>\>\>\If $\mathtt{t}=2(4D)$ \Then {\bf \em End of Transmission Adjustments};\\[-1.300pt]
{\bf End Transmission $\mathtt{T}$}\\[-1.300pt]
\\[-1.300pt]
{\bf 31}\>{\bf \em Okay to Send Packet}\\[-1.300pt]
{\bf 32}\>\>\If $\left\{ \begin{array}{lr} \mbox{$N$ does not have $(\Omega_{\mathtt{T}}, \mathtt{T})$ in $BB$} & {\mathsf{OR}} \\ \mbox{$N$ has $(\Omega_{\mathtt{T}}, \mathtt{T})$ with $\Omega_{\mathtt{T}} = (|EN|, |\mathcal{B}_{\mathtt{T}}|, F, *)$, but has not yet rec'd $|EN|$ parcels as in line 200b,}&\\ \mbox{$\qquad F$ parcels as in line 200c, or $|\mathcal{B}_{\mathtt{T}}|$ parcels as in line 200d} & {\mathsf{OR}} \\ \mbox{$N$ has rec'd the complete {\it SOT} broadcast, but every parcel hasn't yet passed across $E(N,B)$} & {\mathsf{OR}} \\ N \mbox{ or } B \in BL & {\mathsf{OR}} \\ \mbox{$N$ has $\Theta_{\mathtt{T}} \in BB$, but this has not passed across $E(N,B)$ yet} & {\mathsf{OR}} \\ \mbox{$N$ has $BL$ info. in $BB$ (as on line 115, items 3 or 4) not yet passed across $E(N,B)$} &\end{array} \right.$ \\
{\bf 33}\>\>\>\> Return \textsf{False}; \\[-1.300pt]
{\bf 34}\>\> \Else Return \textsf{True};\\[-1.300pt]
\\[-1.300pt]
{\bf 35}\>{\bf \em Okay to Receive Packet}\\[-1.300pt]
{\bf 36}\>\>\If $\left\{ \begin{array}{lr} \mbox{$N$ does not have $(\Omega_{\mathtt{T}}, \mathtt{T})$ in $BB$} & {\mathsf{OR}} \\ \mbox{$N$ has $(\Omega_{\mathtt{T}}, \mathtt{T})$ with $\Omega_{\mathtt{T}} = (|EN|, |\mathcal{B}_{\mathtt{T}}|, F, *)$, but has not yet rec'd $|EN|$ parcels as in line 200b,}&\\ \mbox{$\qquad F$ parcels as in line 200c, or $|\mathcal{B}_{\mathtt{T}}|$ parcels as in line 200d} & {\mathsf{OR}} \\ \mbox{$N$ has rec'd the complete {\it SOT} broadcast, but every parcel hasn't yet passed across $E(A,N)$} & {\mathsf{OR}} \\ N \mbox{ or } A \in BL & {\mathsf{OR}} \\ \mbox{$N$ has $\Theta_{\mathtt{T}} \in BB$, but this has not passed across $E(A,N)$ yet} & {\mathsf{OR}} \\ \mbox{$N$ has $BL$ info. in $BB$ (as on line 115, items 3 or 4) not yet passed across $E(A,N)$} & {\mathsf{OR}} \end{array} \right.$ \\[-1.300pt]
{\bf 37}\>\>\>\> Return \textsf{False}; \\[-1.300pt]
{\bf 38}\>\> \Else Return \textsf{True};
\end{tabbing}
\end{footnotesize}
\end{minipage} }
\end{center}
\vspace{-20.7pt}
\caption{Routing Rules for Transmission $\mathtt{T}$, (Node-Controlling $+$ Edge-Scheduling) Protocol}
\label{routingRulesM}
\end{figure}
\newpage
\begin{figure}[h!t]
\begin{center}
\leavevmode
\framebox{\footnotesize
\begin{minipage}[b]{6.0in}
\begin{footnotesize}
\begin{tabbing}
aaa\=aa\=aa\=aa\=aa\=aa\=aa\=aaaaaaaaaaaaaaaaaaaaaaaaaaaaaaaaaaaaa\=aaaa\=aaaa \kill
39\>{\bf \em Reset Outgoing Variables}\\[-1.4pt]
{\bf 40}\>\>$c_{bp}=0$; \\[-1.4pt]
41\>\>\If $d = 1$:\>\>\>\>\>\>{\scriptsize \#\# $N$ sent a packet previous round}\\[-1.4pt]
42\>\>\>$d=0$;\\[-1.4pt]
43\>\>\>\If $RR = \bot$ \boldmath$\mathsf{or}$\unboldmath $\thickspace \bot \neq FR > RR$\>\>\>\>\>{\scriptsize \#\# Didn't receive conf$.$ of packet receipt}\\[-1.4pt]
44\>\>\>\>$sb=1$;\\[-1.4pt]
45\>\>\If $RR \neq \bot$:\\[-1.4pt]
46\>\>\>\If $\bot \neq FR \leq RR$:\>\>\>\>\>{\scriptsize \#\# $B$ rec'd most recently sent packet}\\[-1.4pt]
47\>\>\>\>\If $N=S$ \Then $\kappa = \kappa +1$;\\[-1.4pt]
{\bf 48}\>\>\>\>For $i = 1, 2, p$: $SIG[i] = $ value rec'd on line 07;\\[-1.4pt]
{\bf 49}\>\>\>\>$SIG[3] = SIG[3] + H_{FP}$;\>\>\>\>{\scriptsize \#\# If $N=S$, skip this line}\\[-1.4pt]
50\>\>\>\>$\mathsf{OUT}[H_{FP}] = \bot$; {\it Fill Gap};\>\>\>\>{\scriptsize \#\# Remove $\tilde{p}$ from $\mathsf{OUT}$, shifting down packets on top}\\[-1.4pt]
\>\>\>\>\>\>\>\>{\scriptsize \#\# of $\tilde{p}$ (if necessary) and adjusting $SIG_{N,N}$ accordingly}\\[-1.4pt]
51\>\>\>\>$FR,\tilde{p},H_{FP}=\bot$; $sb=0$; $H=H-1$;\\[-1.4pt]
52\>\>\If $\bot \neq RR < FR$ \boldmath$\mathsf{and}$\unboldmath $\thickspace \bot \neq H_{FP} < H$:\>\>\>\>\>\>{\scriptsize \#\# $B$ did {\em not} receive most recently sent packet}\\[-1.4pt]
53\>\>\>{\it Elevate Flagged Packet};\>\>\>\>\>{\scriptsize \#\# Swap packets in ${\scriptscriptstyle \mathsf{OUT}[H]}$ and ${\scriptscriptstyle \mathsf{OUT}[H_{FP}]}$; Set $H_{FP}$=$H$;}\\[-1.4pt]
\\[-1.4pt]
54\>{\bf \em Create Flagged Packet}\\[-1.4pt]
55\>\>\If $sb=0$ \boldmath$\mathsf{and}$\unboldmath $\thickspace H > H_{IN}$:\>\>\>\>\>\>{\scriptsize \#\# Normal Status, will send top packet}\\[-1.4pt]
56\>\>\>$\tilde{p} = \mathsf{OUT}[H]$; $H_{FP} = H$; $FR = \mathtt{t}$;\\[-1.4pt]
\\[-1.4pt]
57\>{\bf \em Send Packet}\\[-1.4pt]
58\>\>$d=1$;\\[-1.4pt]
{\bf 59}\>\>\If {\bf \em Okay to Send Packet} \Then \>\>\>\>\>\>{\scriptsize \#\# If $\tilde{p}$ is from an {\it old} codeword, send instead:}\\[-1.4pt]
{\bf 60}\>\>\>\Send {\it Sign}$(\mathtt{T}, \mathtt{t}, \tilde{p}, FR, SIG[1]$+$1, SIG[3] +H_{FP}, SIG[\tilde{p}]$+$1)$;\>\>\>\>\>$\qquad \qquad ${\scriptsize \#\# {\it Sign}$(\mathtt{T}, \mathtt{t}, \tilde{p}, FR, SIG[1], SIG[3]+H_{FP}, \bot)$}\\[-1.4pt]
\\[-1.4pt]
61\>{\bf \em Receive Packet}\\[-1.4pt]
{\bf 62}\>\>\Receive {\it Sign}$(\mathtt{T}, \mathtt{t}-2, p, FR, SIG[1], SIG[2], SIG[p])$;\>\>\>\>\>\>{\scriptsize \#\# $SIG[3]$, 6$^{th}$ coord$.$ sent on line 60, is kept as $SIG[2]$}\\[-1.4pt]
{\bf 63}\>\>\If $H_{OUT} = \bot$ \boldmath$\mathsf{or}$\unboldmath $\space$ {\bf \em Okay to Receive Packet} is {\it false}:\>\>\>\>\>\>{\scriptsize \#\#Didn't rec$.$ $A$'s ht$.$ info, or ${\scriptstyle BB}$ info prevents p$.$ transfer}\\[-1.4pt]
64\>\>\>$sb = 1$;\\[-1.4pt]
65\>\>\>\If $H_{GP} > H$ \boldmath$\mathsf{or}$\unboldmath $\thickspace (H_{GP} = \bot$ \boldmath$\mathsf{and}$\unboldmath $\thickspace H <2n)$:\\[-1.4pt]
66\>\>\>\>$H_{GP} = H+1$;\\[-1.4pt]
67\>\>\ElseIf $sb_{OUT}=1$ \boldmath$\mathsf{or}$\unboldmath $\thickspace H_{OUT} > H$:\>\>\>\>\>\>{\scriptsize \#\# A packet should've been sent}\\[-1.4pt]
{\bf 68}\>\>\>{\bf \em Verify Signature One};\\[-1.4pt]
{\bf 69}\>\>\>\If {\large (}{\bf \em Verify Signature One} returns false \boldmath$\mathsf{or} \thickspace$\unboldmath\>\>\>\>\>{\scriptsize \#\# Signature from $A$ was not valid, or}\\[-1.4pt]
\>\>\>\>$\quad p =\bot$ \boldmath$\mathsf{or} \thickspace$\unboldmath $p$ {\it not properly signed by $S$}{\large )} \Then \>\>\>\>{\scriptsize \#\# Packet wasn't rec'd. or wasn't signed by $S$}\\[-1.4pt]
70\>\>\>\>$sb = 1$;\\[-1.4pt]
71\>\>\>\>\If $H_{GP} > H$ \boldmath$\mathsf{or}$\unboldmath $\thickspace (H_{GP} = \bot$ \boldmath$\mathsf{and}$\unboldmath $\thickspace H <2n)$:\\[-1.4pt]
72\>\>\>\>\>$H_{GP} = H+1$;\\[-1.4pt]
73\>\>\>\ElseIf $RR<FR$:\>\>\>\>\>{\scriptsize \#\# Packet was rec'd and should keep it}\\[-1.4pt]
{\bf 74}\>\>\>\>For $i = 1, 2, p$: $SIG[i] = $ value rec'd on line 62;\\[-1.4pt]
{\bf 75}\>\>\>\>$SIG[3] = SIG[3] + H_{GP}$;\>\>\>\>{\scriptsize \#\# If $N=R$, skip this line}\\[-1.4pt]
76\>\>\>\>\If $H_{GP} = \bot$: $\quad H_{GP} = H+1$;\>\>\>\>{\scriptsize \#\# If no slot is saved for $p$, put it on top}\\[-1.4pt]
77\>\>\>\>$\mathsf{IN}[H_{GP}]=p$;\\[-1.4pt]
78\>\>\>\>$sb=0$; $H = H+1$; $H_{GP} = \bot$; $RR = \mathtt{t}$;\\[-1.4pt]
79\>\>\>\Else\>\>\>\>\>{\scriptsize \#\# Packet was rec'd, but already had it}\\[-1.4pt]
80\>\>\>\>$sb=0$; {\it Fill Gap}; $H_{GP}=\bot$;\>\>\>\>{\scriptsize \#\# See comment about {\it Fill Gap} on line 82 below}\\[-1.4pt]
81\>\>\Else\>\>\>\>\>\>{\scriptsize \#\# A packet should NOT have been sent}\\[-1.4pt]
82\>\>\>$sb=0$; {\it Fill Gap}; $H_{GP} = \bot$;\>\>\>\>\>{\scriptsize \#\# If packets occupied slots {\it above} the Ghost }\\[-1.4pt]
\>\>\>\>\>\>\>\>{\scriptsize \#\# Packet, then {\it Fill Gap} will Slide them {\it down} one slot,}\\[-1.4pt]
\>\>\>\>\>\>\>\>{\scriptsize \#\# updating $SIG_{N,N}$ to reflect this shift, if necessary}\\[-1.4pt]
{\bf 83}\>{\bf \em Verify Signature One}\\[-1.4pt]
{\bf 84}\>\>\If {\it Signature is Valid} \boldmath$\mathsf{and} \thickspace$\unboldmath {\it Values are correct}\>\>\>\>\>\>{\scriptsize \#\# $N$ verifies the values $A$ sent on line 60 are consistent:}\\[-1.4pt]
{\bf 85}\>\>\>Return {\it true};\>\>\>\>\>{\scriptsize \#\# Change in ${\scriptstyle SIG[1]}$ and ${\scriptstyle SIG[p]}$ is `1', change in ${\scriptstyle SIG[2]}$ is}\\[-1.4pt]
{\bf 86}\>\>\Else\>\>\>\>\>\>{\scriptsize \#\# at least $H_{GP}$, $(\mathtt{T}, \mathtt{t})$ is correct and p$.$ has sender's sig}\\[-1.4pt]
{\bf 87}\>\>\>Return {\it false};\\[-1.4pt]
\\[-1.4pt]
{\bf 88}\>{\bf \em Verify Signature Two}\>\>\>\>\>\>\>{\scriptsize \#\# $N$ verifies the values $B$ sent on line 11 are consistent:}\\[-1.4pt]
{\bf 89}\>\>\If {\it Signature is {\scriptsize NOT} Valid} \boldmath$\mathsf{or} \thickspace$\unboldmath {\it Values are {\scriptsize NOT} Correct}:\>\>\>\>\>\>{\scriptsize \#\# Change in ${\scriptstyle SIG[1]}$ and ${\scriptstyle SIG[p]}$ is `1', change in ${\scriptstyle SIG[2]}$}\\[-1.4pt]
{\bf 90}\>\>\>$RR, H_{IN} = \bot$;\>\>\>\>\>{\scriptsize \#\# is at most $H_{FP}$, and $\mathtt{T}$ and $\mathtt{t}$ are correct}
\end{tabbing}
\end{footnotesize}
\end{minipage} }
\end{center}
\vspace{-20pt}
\caption{Routing Rules for Transmission $\mathtt{T}$, (Node-Controlling $+$ Edge-Scheduling) Protocol (cont)}
\label{routingRules2M}
\end{figure}
\newpage
\begin{figure}[h!t]
\begin{center}
\leavevmode
\framebox{\footnotesize
\begin{minipage}[b]{6.0in}
\begin{footnotesize}
\begin{tabbing}
aaaa\=aa\=aa\=aa\=aa\=aa\=aa\=aaaaaaaaaaaaaaaaaaaaaaaaaaaaaaaaaaaaaaaaaaaaaaaaaaaaaaa\=aaaa\=aaaa \kill
{\bf 91}\>{\bf \em Send/Receive Broadcast Parcels}\\
{\bf 92}\>\>\For {\it outgoing} {\bf edge} $E(N,B) \in G$, {\it including} $N=R, B = S$\\
{\bf 93}\>\>\>\Receive $bp$;\\
{\bf 94}\>\>\>{\bf \em Update Broadcast Buffer Two};\\
{\bf 95}\>\>\For {\it incoming} {\bf edge} $E(A,N) \in G$, {\it including} $N=S, A = R$\\
{\bf 96}\>\>\>{\bf \em Determine Broadcast Parcel to Send};\\
{\bf 97}\>\>\>\Send $bp$;\\
\\
{\bf 98}\>{\bf \em Update Broadcast Buffer One}\\
{\bf 99}\>\>\For {\it outgoing} {\bf edge} $E(N,B) \in G$, {\it including} $N=R, B = S$\\
{\bf 100}\>\>\>\If $bp \neq \bot$ \Then \\
{\bf 101}\>\>\>\>\Send $c_{bp}$;\\
{\bf 102}\>\>\>{\bf \em Broadcast Parcel to Request};\\
{\bf 103}\>\>\>\Send $\alpha$;\\
{\bf 104}\>\>\For {\it incoming} {\bf edge} $E(A,N) \in G$, {\it including} $N=S, A = R$\\
{\bf 105}\>\>\>\Receive $c_{bp}$; \Receive $\alpha$;\\
{\bf 106}\>\>\>\If $\alpha \neq \bot$ \Then Update Broadcast Buffer;$\thickspace \qquad \qquad${\scriptsize \#\# Update $BB$ to preferentially send $\alpha$}\\
{\bf 107}\>\>\>\If $c_{bp}=1$ \Then Update Broadcast Buffer;$\qquad \qquad${\scriptsize \#\# Update $BB$ that $bp$ crossed $E(A,N)$}\\
{\bf 108}\>\>\>$c_{bp}=0$;\\
\\
{\bf 109}\>{\bf \em Update Broadcast Buffer Two}\\
{\bf 110}\>\>\If $\bot \neq bp$ has valid sig$.$ \boldmath$\mathsf{and}$\unboldmath $\left\{ \begin{array}{l} N \mbox{ has received full {\it SOT} broadcast for } \mathtt{T} \qquad \mathsf{OR}\\ bp \mbox{ is a valid {\it SOT} broadcast parcel rec'd in correct order (see 115 and 200)} \end{array} \right.$\\
\>\>\>\>{\scriptsize \#\# Here, a ``valid'' signature means both from $B$ and the from node $bp$ originated from, and}\\
\>\>\>\>{\scriptsize \#\# a ``valid'' {\it SOT} parcel means that $N$ has already received all {\it SOT} parcels that}\\
\>\>\>\>{\scriptsize \#\# should have arrived before $bp$, as indicated by the ordering of line 115, items 2a-2d}\\
{\bf 111}\>\>\>$c_{bp}=1$;\\
{\bf 112}\>\>\>\If $N=S$: {\bf \em Sender Update Broadcast Buffer};\\
{\bf 113}\>\>\>\Else {\bf \em Internal Node and Receiver Update Broadcast Buffer};\\
\\
{\bf 114}\>{\bf \em Determine Broadcast Parcel to Send}\\
{\bf 115}\>\>Among all information in $BB$, choose some $bp \in BB$ that has not passed along $E(A,N)$ by priority:\\
\>\>\>1) The receiver's {\it end of transmission} parcel $\Theta_{\mathtt{T}}$\\
\>\>\>2) The sender's {\it start of transmission} ({\it SOT}) broadcast, in the order indicated on line 200:\\
\>\>\>\> a) $(\Omega_{\mathtt{T}}, \mathtt{T}) \qquad$ b) $(\widehat{N} \in EN, \mathtt{T}) \qquad$ c) $(\mathtt{T}', Fi, \mathtt{T}) \qquad$ d) $(\widehat{N} \in BL, \mathtt{T}', \mathtt{T})$\\
\>\>\>3) $(\widehat{N}, 0, \mathtt{T}) =$ label of a node to {\it remove} from the blacklist, see line 165\\
\>\>\>4) $(N, \widehat{N}, \mathtt{T}') =$ label of a node $\widehat{N}$ on $BL$ for which $N$ has the complete status report for $\mathtt{T}'$, see line 155\\
\>\>\>5) A status report parcel requested by $A$ as indicated by $\alpha$ (received on line 105)\\
\>\>\>6) An arbitrary status report parcel of a node on $N$'s blacklist\\
\\
{\bf 116}\>{\bf \em Broadcast Parcel to Request}\\
{\bf 117}\>\>$\alpha = \bot$;\\
{\bf 118}\>\>\If $B$ is on $N$'s blacklist and $N$ is missing a status report from $B$:\\
{\bf 119}\>\>\>Set $\alpha$ to indicate $B$'s label and an index of the parcel $N$ is missing from $B$;\\
{\bf 120}\>\>\ElseIf $DB$ indicates that $B$ has complete status report for some node $\widehat{N}$ on $BL$ (see lines 150-151, 155):\\
{\bf 121}\>\>\>\If $N$ is missing a status report of node $\widehat{N}$:\\
{\bf 122}\>\>\>\>Set $\alpha$ to the label of the node $\widehat{N}$ and the index of a status report parcel from $\widehat{N}$ that $N$ is missing;
\end{tabbing}
\end{footnotesize}
\end{minipage} }
\end{center}
\caption{Routing Rules for Transmission $\mathtt{T}$, (Node-Controlling $+$ Edge-Scheduling) Protocol (cont)}
\label{routingRules3M}
\end{figure}
\newpage
\begin{figure}[h!t]
\begin{center}
\leavevmode
\framebox{\footnotesize
\begin{minipage}[b]{6.0in}
\begin{footnotesize}
\begin{tabbing}
aaaa\=a\=aa\=aa\=aa\=aa\=aa\=aaaaaaaaaaaaaaaaaaaaaaaaaaaaaaaaaaaaa\=aaaa\=aaaa \kill
{\bf 123}\>{\bf \em Internal Node and Receiver Update Broadcast Buffer}\\[-1.3pt]%
\>\>\>{\scriptsize \#\# Below, a broadcast parcel $bp$ is ``{\it Added}'' only if it is not already in $BB$. Also, view $BB$ as being}\\[-1.3pt]%
\>\>\>{\scriptsize \#\# indexed by each $bp$ with $n-1$ slots for each parcel to indicate which edges $bp$ has already traversed.}\\[-1.3pt]%
\>\>\>{\scriptsize \#\# Then when $bp$ is removed from $BB$, the edge ``markings'' are removed as well.}\\[-1.3pt]%
{\bf 124}\>\>\If $bp = \Theta_{\mathtt{T}}$ is receiver's {\it end of transmission} parcel (for current transmission $\mathtt{T}$, see line 179):\\[-1.3pt]%
{\bf 125}\>\>\>Add $bp$ to $BB$ and mark edge $E(N,B)$ as having passed this info.;\\[-1.3pt]%
{\bf 126}\>\>\ElseIf $bp = (\Omega_{\mathtt{T}}, \mathtt{T})$ is the first parcel of the sender's {\it start of transmission} ({\it SOT}) broadcast (see line 200a):\\[-1.3pt]%
{\bf 127}\>\>\>Add $bp$ to $BB$, and mark edge $E(N,B)$ as having passed this information;\\[-1.3pt]%
{\bf 128}\>\>\>\If $\Omega_{\mathtt{T}} = (*, 0, *, *): \thickspace$ Clear all entries of $SIG$, and set $SIG_{N,N}=0$;\\[-1.3pt]%
{\bf 129}\>\>\ElseIf $bp$= {\scriptsize $(\widehat{N}, \mathtt{T})$} is from the {\it SOT} broadcast indicating a node to {\it eliminate}, as on line 200b:\\[-1.3pt]%
{\bf 130}\>\>\>Add $bp$ to $BB$ and mark edge $E(N,B)$ as having passed this info.;\\[-1.3pt]%
{\bf 131}\>\>\>\If $\widehat{N} \notin EN$:\>\>\>\>\>{\scriptsize \#\# $N$ is just learning $\widehat{N}$ is to be eliminated}\\[-1.3pt]%
{\bf 132}\>\>\>\>Add $\widehat{N}$ to $EN$;\\[-1.3pt]%
{\bf 133}\>\>\>\>Clear all incoming and outgoing buffers, clear all entries of $SIG$, and set $SIG_{N,N} = 0$;\\[-1.3pt]%
{\bf 134}\>\>\>\>Clear $BB$, EXCEPT for parcels from current {\it SOT} broadcast; Clear $DB$, EXCEPT for $EN$;\\[-1.3pt]%
{\bf 135}\>\>\ElseIf $bp$= {\scriptsize $(\mathtt{T}', Fi, \mathtt{T})$} is from the {\it SOT} broadcast indicating why a previous trans$.$ failed, as on line 200c:\\[-1.3pt]%
{\bf 136}\>\>\>Add $bp$ to $BB$ and mark edge $E(N,B)$ as having passed this information;\\[-1.3pt]%
{\bf 137}\>\>\ElseIf $bp$= {\scriptsize $(\widehat{N}, \mathtt{T}', \mathtt{T})$} is from the {\it SOT} broadcast indicating a node to blacklist, as on line 200d:\\[-1.3pt]%
{\bf 138}\>\>\>Add $\widehat{N}$ to $BL$; Add $bp$ to $BB$ and mark edge $E(N,B)$ as having passed this information;\\[-1.3pt]%
{\bf 139}\>\>\>Remove outdated info$.$ from $BB$ and $DB$;\\[-1.3pt]%
\>\>\>\>{\scriptsize \#\# This includes for any trans$.$ $\mathtt{T}'' \neq \mathtt{T}'$ removing from $DB$ all entries of form $(\widetilde{B}, \widehat{N}, \mathtt{T}'')$, see line 115, item 4;}\\[-1.3pt]%
\>\>\>\>{\scriptsize \#\#  and removing from $BB$: 1) $(N, \widehat{N}, \mathtt{T}'')$, see line 115 item 4, and 2) Any status report parcel of $\widehat{N}$ for $\mathtt{T}''$}\\[-1.3pt]%
{\bf 140}\>\>\>\If $\widehat{N}=N$ has not already added its own status report info$.$ corresponding to $\mathtt{T}'$ to $BB$:\\[-1.3pt]%
\>\>\>\>{\scriptsize \#\# The following reasons for failure come from {\it SOT}.  See lines 190, 193, and 196-197}\\[-1.3pt]%
\>\>\>\>{\scriptsize \#\# The information added in each case will be referred to as the node's {\it status report} for transmission $\mathtt{T}'$}\\[-1.3pt]%
{\bf 141}\>\>\>\>\If entries of $SIG_{N,N}$ and $SIG$ correspond to a transmission $\mathtt{T}'' \neq \mathtt{T}'$: $\space$ {\it Clear} $SIG$ and set $SIG_{N,N}=0$;\\[-1.3pt]%
{\bf 142}\>\>\>\>\If $\mathtt{T}'$ failed as in F2:  For each incoming and outgoing edge, sign and add to $BB$: $(SIG[2], SIG[3], \mathtt{T}')$;\\[-1.3pt]%
{\bf 143}\>\>\>\>\>Also sign and add $(SIG_{N,N}, \mathtt{T}')$ to $BB$ (see line 12 of Figure \ref{setupCodeM});\\[-1.3pt]%
{\bf 144}\>\>\>\>\ElseIf $\mathtt{T}'$ failed as in F3:  For each incoming and outgoing edge, sign and add $(SIG[1], \mathtt{T}')$ to $BB$;\\[-1.3pt]%
{\bf 145}\>\>\>\>\ElseIf $\mathtt{T}'$ failed as in F4:  For each incoming and outgoing edge, sign and add $(SIG[p], \mathtt{T}')$ to $BB$;\\[-1.3pt]%
{\bf 146}\>\>\>\If $N$ has received {\scriptsize $|BL_{\mathtt{T}}|$ \it SOT} parcels of form {\scriptsize $(\widehat{N}, \mathtt{T}', \mathtt{T})$} $:\thickspace$ Clear all entries of $SIG$ and set $SIG_{N,N}=0$;\\[-1.3pt]%
{\bf 147}\>\>\ElseIf $bp$= {\scriptsize $(\widehat{N}, 0, \mathtt{T})$} is from sender, indicating a node to remove from $BL$, as on line 165:\\[-1.3pt]%
{\bf 148}\>\>\>Remove $\widehat{N}$ from $BL$; Add $bp$ to $BB$ and mark edge $E(N,B)$ as having passed this information;\\[-1.3pt]%
{\bf 149}\>\>\>Remove outdated info$.$ from $BB$ and $DB$ as on line 139 above;\\[-1.3pt]%
{\bf 150}\>\>\ElseIf $bp$= {\scriptsize $(B, \widehat{N}, \mathtt{T}')$} indicates $B$ has a blacklisted node $\widehat{N}$'s complete status report for trans$.$ $\mathtt{T}'$:\\[-1.3pt]%
{\bf 151}\>\>\>\If $(\widehat{N}, \mathtt{T}', \mathtt{T})$ is on $N$'s blacklist: $\quad$ Add fact that $B$ has $\widehat{N}$'s complete status report to $DB$;\\[-1.3pt]%
{\bf 152}\>\>\ElseIf $bp$ is a status report parcel for trans$.$ $\mathtt{T}'$ of some node {\scriptsize $(\widehat{N}, \mathtt{T}', \mathtt{T})$} on $BL$, see lines 140-145 and 200d:\\[-1.3pt]%
{\bf 153}\>\>\>\If $bp$ has valid sig. from $\widehat{N}$ and concerns correct info.:\\[-1.3pt]%
\>\>\>\>\>$\thickspace${\scriptsize \#\# $N$ finds $(\widehat{N}, \mathtt{T}', \mathtt{T})$ and $(\mathtt{T}', Fi, \mathtt{T})$ in $BB$ (from {\it SOT} broadcast) and checks that $bp$ concerns correct info$.$}\\[-1.3pt]%
{\bf 154}\>\>\>\>Add $bp$ to $BB$, and mark edge $E(N,B)$ as having passed this information;\\[-1.3pt]%
{\bf 155}\>\>\>\>\If $bp$ completes $N$'s knowledge of $\widehat{N}$'s missing status report for transmission $\mathtt{T}'$: $\space$ Add {\scriptsize $(N, \widehat{N}, \mathtt{T}')$} to $BB$;\\[.345cm]
{\bf 156}\>{\bf \em Sender Update Broadcast Buffer} \>\>\>\>\>\>\>{\scriptsize \#\# Below, a parcel $bp$ is ``{\it Added}'' only if it is not in $DB$}\\[-1.3pt]%
{\bf 157}\>\>\If $bp = \Theta_{\mathtt{T}}$ is receiver's {\it end of transmission} parcel (for current transmission $\mathtt{T}$):\\[-1.3pt]%
{\bf 158}\>\>\>Add $bp$ to $DB$;\\[-1.3pt]%
{\bf 159}\>\>\ElseIf $bp$ indicates $B$ has a blacklisted node $\widehat{N}$'s complete status report for trans$.$ $\mathtt{T}'$:\\[-1.3pt]%
{\bf 160}\>\>\>\If $(\widehat{N}, \mathtt{T}', \mathtt{T})$ is on $S$'s blacklist: $\quad$ Add $(B, \widehat{N}, \mathtt{T}')$ to $DB$;\\[-1.3pt]%
{\bf 161}\>\>\ElseIf $bp$ is a status report parcel of some node $\widehat{N}$ on the sender's blacklist (see lines 140-145):\\[-1.3pt]%
{\bf 162}\>\>\>Add $bp$ to $DB$;\\[-1.3pt]%
{\bf 163}\>\>\>\If $bp$ contains faulty info$.$ but has a valid sig. from $\widehat{N}: \quad$ {\bf \em Eliminate }$\widehat{N}$;\\[-1.3pt]%
\>\>\>\>\>{\scriptsize \#\# $S$ checks $DB$ for reason of failure and makes sure $\widehat{N}$ has returned an appropriate value}\\[-1.3pt]%
{\bf 164}\>\>\>\If $bp$ completes the sender's knowledge of $\widehat{N}$'s missing status report from transmission $\mathtt{T}'$:\\[-1.3pt]%
{\bf 165}\>\>\>\>Sign $(\widehat{N}, 0, \mathtt{T})$ and add to $BB$; \>\>\>\>{\scriptsize \#\# Indicates that $\widehat{N}$ should be removed from blacklist}\\[-1.3pt]%
{\bf 166}\>\>\>\>Remove outdated info$.$ from $DB$; Remove $(\widehat{N}, \mathtt{T}')$ from $BL$;\\[-1.3pt]%
\>\>\>\>\>{\scriptsize \#\# ``Outdated'' refers to parcels as on 159-160 whose second entry is $\widehat{N}$}\\[-1.3pt]%
{\bf 167}\>\>\>\>\If $bp$ completes sender's knowledge of all relevant status reports from some transmission:\\[-1.3pt]%
{\bf 168}\>\>\>\>\>{\bf \em Eliminate} $\widehat{N}$;\>\>\>{\scriptsize \#\# $S$ can eliminate a node. See pf$.$ of Thm \ref{mainAdThm} for details}
\end{tabbing}
\end{footnotesize}
\end{minipage} }
\end{center}
\vspace{-.6cm}
\caption{Routing Rules for Transmission $\mathtt{T}$, (Node-Controlling $+$ Edge-Scheduling) Protocol (cont)}
\label{routingRules4M}
\end{figure}
\newpage
\begin{figure}[h!t]
\begin{center}
\leavevmode
\framebox{\footnotesize
\begin{minipage}[b]{6.0in}
\begin{footnotesize}
\begin{tabbing}
aaaa\=aa\=aa\=aa\=aa\=aa\=aa\=aaaaaaaaaaaaaaaaaaaaaaaaaaaaaaaaaaaaaaa\=aaaa\=aaaa \kill
{\bf 169}\>{\bf \em Eliminate $\widehat{N}$}\\[-1.3pt]%
{\bf 170}\>\>{\it Add} $(\widehat{N}, \mathtt{T})$ to $EN$;\\[-1.3pt]%
{\bf 171}\>\>{\it Clear} $BB$, $DB$ (except for $EN$), and signature buffers;\\[-1.3pt]%
{\bf 172}\>\>$\beta_{\mathtt{T}}, F = 0$;\\[-1.3pt]%
{\bf 173}\>\>$\mathcal{P}_{\mathtt{T}+1} = \mathcal{P} \setminus EN$;\\[-1.3pt]%
{\bf 174}\>\>$\Omega_{\mathtt{T}+1} = (|EN|, 0, 0, 0)$;\\[-1.3pt]%
{\bf 175}\>\>{\it Sign} and {\it Add} $\Omega_{\mathtt{T}+1}$ to $BB$;\\[-1.3pt]%
{\bf 176}\>\>\For $N \in EN$, {\it Sign} and {\it Add} $(N,\mathtt{T}+1)$ to $BB$;\\[-1.3pt]%
{\bf 177}\>\>{\it Halt} until {\scriptsize \bf \em End of Transmission Adjustments} is called;\>\>\>\>\>\>{\scriptsize \#\# $S$ does not begin inserting p's until next trans.,}\\[-1.3pt]%
\>\>\>\>\>\>\>\>{\scriptsize \#\# and $S$ ignores all instructions for $\mathtt{T}$ until line 30}\\[-1.3pt]%
{\bf 178}\>{\bf \em Send End of Transmission Parcel}\\[-1.3pt]%
{\bf 179}\>\>{\it Add signed $\Theta_{\mathtt{T}} = (b, p', \mathtt{T})$ to $BB$}\>\>\>\>\>\>{\scriptsize \#\# $b$ is a bit indicating if $R$ could decode, $p'$ is}\\[-1.3pt]%
\>\>\>\>\>\>\>\>{\scriptsize \#\# the label of a packet $R$ rec'd twice, or else $\bot$}\\[-1.3pt]%
{\bf 180}\>{\bf \em Prepare Start of Transmission Broadcast}\\[-1.3pt]%
{\bf 181}\>\>{\scriptsize \#\# Let $\Theta_{\mathtt{T}} = (b, p', \mathtt{T})$ denote Sender's value obtained from Receiver's transmission above (as stored in $DB$)}\\[-1.3pt]%
{\bf 182}\>\>\If $b=1$ \Then \>\>\>\>\>\>{\scriptsize \#\# $R$ was able to decode}\\[-1.3pt]%
{\bf 183}\>\>\>{\it Clear} each entry of signature buffers holding data corresponding to $\mathtt{T}$;\\[-1.3pt]%
{\bf 184}\>\>\>$\Omega_{\mathtt{T}+1} = (|EN|, |BL|, F, 0)$;\\[-1.3pt]%
{\bf 185}\>\>\ElseIf $b=0$ \Then \>\>\>\>\>\>{\scriptsize \#\# $R$ was {\it not} able to decode: a {\it failed} transmission}\\[-1.3pt]%
{\bf 186}\>\>\>$F = F+1$;\\[-1.3pt]%
{\bf 187}\>\>\>Set $\mathcal{P}_{\mathtt{T}} = \mathcal{P} \setminus (EN \cup BL)$ and add $(\mathcal{P}_{\mathtt{T}}, \mathtt{T})$ to $DB$;\\[-1.3pt]%
{\bf 188}\>\>\>For each $N \in \mathcal{P}_{\mathtt{T}} \setminus S$: Add {\scriptsize $(N, \mathtt{T})$} to $BL$;\>\>\>\>\>{\scriptsize \#\# $(N, \mathtt{T})$ records the trans. $N$ was added to $BL$}\\[-1.3pt]%
{\bf 189}\>\>\>{\it Clear} outgoing buffers;\\[-1.3pt]%
{\bf 190}\>\>\>\If $p' \neq \bot$:\>\>\>\>\>{\scriptsize \#\# $R$ rec'd a duplicate packet}\\[-1.3pt]%
{\bf 191}\>\>\>\>Add $(p', \mathtt{T})$ to $DB$;  Add $SIG[p']$ to $DB$;\>\>\>\>{\scriptsize \#\# Record that reason $\mathtt{T}$ failed was F4}\\[-1.3pt]%
{\bf 192}\>\>\>\>$\Omega_{\mathtt{T}+1} = (|EN|, |BL|, F, p')$;\\[-1.3pt]%
{\bf 193}\>\>\>\ElseIf $\kappa < D$:\>\>\>\>\>{\scriptsize \#\# $S$ did not insert at least $D$ packets}\\[-1.3pt]%
{\bf 194}\>\>\>\>Add $(1, \mathtt{T})$ to $DB$; Add $SIG[2]$ and $SIG[3]$ to $DB$;\>\>\>\>{\scriptsize \#\# Record that reason $\mathtt{T}$ failed was F2}\\[-1.3pt]%
{\bf 195}\>\>\>\>$\Omega_{\mathtt{T}+1} = (|EN|, |BL|, F, 1)$;\\[-1.3pt]%
{\bf 196}\>\>\>\Else \\[-1.3pt]%
{\bf 197}\>\>\>\>Add $(2, \mathtt{T})$ to $DB$;  Add $SIG[1]$ to $DB$;\>\>\>\>{\scriptsize \#\# Record that reason $\mathtt{T}$ failed was F3}\\[-1.3pt]%
{\bf 198}\>\>\>\>$\Omega_{\mathtt{T}+1} = (|EN|, |BL|, F, 2)$;\\[-1.3pt]%
{\bf 199}\>\>{\it Clear} $BB$ and $SIG[i]$ for each $i=1,2,p$; {\it Remove} $\Theta_{\mathtt{T}}$ from $DB$;\\[-1.3pt]%
{\bf 200}\>\>{\it Sign} and {\it Add} to $BB$:\>\>\>\>\>\>{\scriptsize \#\# The {\it Start of Transmission (SOT)} broadcast}\\[-1.3pt]%
\>\>\>a) $(\Omega_{\mathtt{T}+1}, \mathtt{T}$+$1)$\\[-1.3pt]%
\>\>\>b) For each $N \in EN$, add the parcel $(N, \mathtt{T}$+$1)$\\[-1.3pt]%
\>\>\>c) For each failed transmission $\mathtt{T}'$ since the last node was eliminated, add the parcel $(\mathtt{T}', Fi, \mathtt{T}$+$1)$\\[-1.3pt]%
\>\>\>\>{\scriptsize \#\# Here, $Fi$ is the reason trans. $\mathtt{T}'$ failed (F2, F3, or F4).  See pf. of Thm. \ref{mainAdThm} for details}\\[-1.3pt]%
\>\>\>d) For each $N \in BL$, add the parcel $(N, \mathtt{T}', \mathtt{T}$+$1)$, where $\mathtt{T}'$ indicates the trans. $N$ was last added to $BL$\\[-1.3pt]%
{\bf 201}\>\>$\beta_{\mathtt{T}} = 0$;\\[-1.3pt]%
\\[-1.3pt]%
202\>{\bf \em End of Transmission Adjustments}\\[-1.3pt]%
{\bf 203}\>\>\If $N \neq S: \quad$ Clear $\Theta_{\mathtt{T}}$, $BL$, all parcels from {\it SOT} broadcast, and info$.$ of form $(\widehat{N}, 0, \mathtt{T})$ from $BB$;\\[-1.3pt]%
204\>\>\For {\it outgoing} edge $E(N,B)$, $B \in G, N \neq R, B \neq S$:\\[-1.3pt]%
205\>\>\>\If $H_{FP} \neq \bot$: \\[-1.3pt]%
206\>\>\>\>$OUT[H_{FP}] = \bot$; {\it Fill Gap};\>\>\>\>{\scriptsize \#\# Remove any flagged packet $\tilde{p}$ from $\mathsf{OUT}$, shifting}\\[-1.3pt]%
\>\>\>\>\>\>\>\>{\scriptsize \#\#  down packets on top of $\tilde{p}$ if necessary}\\[-1.3pt]%
207\>\>\>\>$sb = 0;$ $FR,H_{FP}, \tilde{p}=\bot$; $H=H-1$;\\[-1.3pt]%
208\>\>\For {\it incoming} edge $E(A,N)$, $A \in G$, $A \neq R$:\\[-1.3pt]%
209\>\>\>$H_{GP} = \bot$; $sb =0$; $RR=-1$; {\it Fill Gap};\\[-1.3pt]%
210\>\>\If $N=S$ \Then {\bf \em Distribute Packets};\\[-1.3pt]%
211\>\>\If $N=R$ \Then $\kappa = 0$; {\it Clear $I_R$};\>\>\>\>\>\>{\scriptsize \#\# Set each entry of $I_R$ to $\bot$}\\[-1.3pt]%
\\[-1.3pt]%
212\>{\bf \em Distribute Packets}\\[-1.3pt]%
213\>\>$\kappa = 0; H_{OUT} = 2n$;\>\>\>\>\>\>{\scriptsize \#\# Set height of each outgoing buffer to $2n$}\\[-1.3pt]%
{\bf 214}\>\>Fill each outgoing buffer with codeword packets;\\[-1.3pt]%
\>\>$\qquad${\scriptsize \#\# If $\mathtt{T}$ was successful, make {\it new} codeword p's, and fill out. buffers and {\tiny COPY} with these.}\\[-1.3pt]%
\>\>$\qquad${\scriptsize \#\# If $\mathtt{T}$ failed or a node was just eliminated, use codeword packets in {\tiny COPY} to fill out$.$ buffers.}
\end{tabbing}
\end{footnotesize}
\end{minipage} }
\end{center}
\vspace{-20pt}
\caption{Routing Rules for Transmission $\mathtt{T}$, (Node-Controlling $+$ Edge-Scheduling) Protocol (cont)}
\label{routingRules5M}
\end{figure}
\newpage
%
\section{Node-Controlling $+$ Edge-Scheduling Protocol:\\ Proofs of Theorems} \label{MProofs}
\indent \indent We restate and prove our two main theorems for the node-controlling adversary Routing Protocol:\\[.2cm]
\noindent{\bf Theorem \ref{memoryThm}.} {\it The memory required of each node is at most $O(n^4(k+\log n))$.}
\begin{proof} By looking at the domains of the variables in Figures \ref{setupCodeM} and \ref{setupCode2M}, it is clear that the Broadcast Buffer and Signature Buffer maintained by all nodes, and the Data Buffer of the sender and Storage Buffer of the receiver require the most memory.  Each of these require $O(n^3(k+ \log n))$ bits of memory, but each node must maintain $O(n)$ signature buffers, which yields the memory bound of $O(n^4(k+\log n))$.  It remains to show that the domains described are accurate, i.e.\ that the protocol never calls for the nodes to store in more (or different) information than their domains allow.  The proof of this fact walks through the pseudo-code and analyses every time information is added or deleted from each buffer, and it can be found in Section \ref{mBareBones} (see Lemma \ref{mProperDomains}).
\end{proof}
\noindent{\bf Theorem \ref{mainAdThm}.} {\it Except for the at most $n^2$ transmissions that may fail due to malicious activity, the routing protocol  presented described in Sections \ref{aP} and \ref{MpseudoCode} enjoys linear throughput.  More precisely, after $x$ transmissions, the receiver has correctly outputted at least $x-n^2$ messages.  If the number of transmissions $x$ is quadratic in $n$ or greater, than the failed transmissions due to adversarial behavior become asymptotically negligible.  Since a transmission lasts $O(n^3)$ rounds and messages contain $O(n^3)$ bits, information is transferred through the network at a linear rate.}\\[.2cm]
Theorem \ref{mainAdThm} will follow immediately from the following three theorems.  As with the proofs for the edge-scheduling protocol, line numbers for the pseudo-code have form ({\bf X.YY}), where {\bf X} refers to the Figure number, and {\bf YY} refers to the line number.  It will be convenient to introduce new terminology:
\begin{defn} We will say a node $N \in G$ {\em participated} in transmission $\mathtt{T}$ if there was at least one round in the transmission for which $N$ was not on the (sender's) blacklist.  The sender's variable that keeps track of nodes participating in transmission $\mathtt{T}$ ({\bf \ref{setupCode2M}.69}) will be called the {\em participating list} for transmission $\mathtt{T}$, denoted by $\mathcal{P}_{\mathtt{T}}$ (it is updated at the end of every transmission on {\bf \ref{routingRules5M}.187}).\end{defn}
\begin{thm} \label{numRounds} Every transmission (regardless of its success/failure) lasts $O(n^3)$ rounds. \end{thm}
\begin{proof} Line ({\bf \ref{routingRulesM}.02}) shows that each transmission, regardless of success or failure, lasts $4D = O(n^3)$ rounds.\end{proof}
\begin{thm} \label{completeInfo} Suppose transmission $\mathtt{T}$ failed and at some later time (after transmission $\mathtt{T}$ but before any additional nodes have been eliminated) the sender has received all of the status report parcels from all nodes on $\mathcal{P}_{\mathtt{T}}$.  Then the sender can eliminate a corrupt node.
\end{thm}
\begin{proof} The proof of this theorem is rather involved, as it needs to address the three possible reasons (F2-F4) that a transmission can fail.  It can be found in Section \ref{pfCompleteInfo} below.
\end{proof}
\begin{thm} \label{maxBadness} After a corrupt node has been eliminated (or at the outset of the protocol) and before the next corrupt node is eliminated, there can be at most $n-1$ failed transmissions $\{\mathtt{T}_1, \dots, \mathtt{T}_{n-1}\}$ before there is necessarily some index $1 \leq i \leq n-1$ such that the sender has the complete status report from every node on $\mathcal{P}_{\mathtt{T}_i}$. \end{thm}
\begin{proof} The theorem will follow from a simple observation:
	\begin{enumerate}\setlength{\itemsep}{4pt} \setlength{\parskip}{0pt} \setlength{\parsep}{0pt}
	\item[] {\bf Observation.}  If $N \in \mathcal{P}_{\mathtt{T}}$, then the sender is not missing any status report parcel for $N$ for any transmission prior to transmission $\mathtt{T}$.  In other words, there is no transmission $\mathtt{T}' < \mathtt{T}$ such that $N$ was blacklisted at the end of $\mathtt{T}'$ (as on {\bf \ref{routingRules5M}.188}) and the sender is still missing status report information from $N$ at the end of $\mathtt{T}$.
	\item[] {\bf Proof.} Nodes are added to the blacklist whenever they were participating in a transmission that failed ({\bf \ref{routingRules5M}.187-88}).  Nodes are removed from the blacklist whenever the sender receives all of the status report information he requested of them ({\bf \ref{routingRules4M}.164-166}), or when he has just eliminated a node ({\bf \ref{routingRules5M}.171}), in which case the sender no longer needs status reports from nodes for old failed transmissions\footnote{The sender already received enough information to eliminate a node.  Even though it is possible that other nodes acted maliciously and caused one of the failed transmissions, it is also possible that the node just eliminated caused all of the failed transmissions.  Therefore, the protocol does not spend further resources attempting to detect another corrupt node, but rather starts anew with a reduced network (the eliminated node no longer legally participates), and will address future failed transmissions as they arise.} (and in particular, this case falls outside the hypotheses of the Theorem).  Since $\mathcal{P}_{\mathtt{T}}$ is defined as non-blacklisted nodes ({\bf \ref{routingRules5M}.187}), the fact that $N \in \mathcal{P}_{\mathtt{T}}$ implies that $N$ was not on the sender's blacklist at the end of $\mathtt{T}$.  Also, notice the next line guarantees that {\it all} nodes not already on the sender's blacklist will be put on the blacklist if the transmission fails.  Therefore, if $N$ has not been blacklisted since the last node was eliminated ({\bf \ref{routingRules5M}.169-177}), then there have not been any failed transmissions, and hence the sender is not missing any status reports.  Otherwise, let $\mathtt{T}' < \mathtt{T}$ denote the last time $N$ was put on the blacklist, as on ({\bf \ref{routingRules5M}.188}).  In order for $N$ to be put on $\mathcal{P}_{\mathtt{T}}$ on line ({\bf \ref{routingRules5M}.187}) of transmission $\mathtt{T}$, it must have been removed from the blacklist at some point between $\mathtt{T}'$ and the end of $\mathtt{T}$.  In this case, the remarks at the start of the proof of this observation indicate the sender is not missing any status reports from $N$.\hspace*{\fill} \hspace*{-12pt}$\scriptstyle{\blacksquare}$
	\end{enumerate}
Suppose now for the sake of contradiction that we have reached the end of transmission $\mathtt{T}_{n-1}$, which marks the $(n-1)^{st}$ transmission $\{\mathtt{T}_1, \dots, \mathtt{T}_{n-1}\}$ such that for each of these $n-1$ failed transmissions, the sender does not have the complete status report from at least one of the nodes that participated in the transmission.  Define the set $\mathcal{S}$ to be the set of nodes that were necessarily {\em not} on $\mathcal{P}_{\mathtt{T}_{n-1}}$, and initialize this set to be empty.

Since the sender is missing some node's complete status report that participated in $\mathtt{T}_1$, there is some node $N_1 \in \mathcal{P}_{\mathtt{T}_1}$ from which the sender is still missing a status report parcel corresponding to $\mathtt{T}_1$ by the end of transmission $\mathtt{T}_{n-1}$.   Notice by the observation above that $N_1$ will not be on $\mathcal{P}_{\mathtt{T}'}$ for any $\mathtt{T}_2 \leq \mathtt{T}' \leq \mathtt{T}_{n-1}$, so put $N_1$ into the set $\mathcal{S}$.  Now looking at $\mathtt{T}_2$, there must be some node $N_2 \in \mathcal{P}_{\mathtt{T}_2}$ from which the sender is still missing a status report parcel from $\mathtt{T}_2$ by the end of transmission $\mathtt{T}_{n-1}$.  Notice that $N_2 \neq N_1$ since $N_1 \notin \mathcal{P}_{\mathtt{T}_2}$, and also that $N_2 \notin \mathcal{P}_{\mathtt{T}_{n-1}}$ (both facts follow from the above observation), so put $N_2$ into $\mathcal{S}$.  Continue in this manner, until we have found the $(n-2)^{th}$ distinct node that was put into $\mathcal{S}$ due to information the sender was still missing by the end of $\mathtt{T}_{n-2}$.  But then $|\mathcal{S}| = n-2$, which implies that all nodes, except for the sender and the receiver, are not on $\mathcal{P}_{\mathtt{T}_{n-1}}$ (the sender and receiver participate in every transmission by Lemma \ref{rParticipates}).  But now we have a contradiction, since Lemma \ref{partListNotEmpty} says that transmission $\mathtt{T}_{n-1}$ will not fail.\end{proof}
We are now ready to prove Theorem \ref{mainAdThm}, reserving the proof of Theorem \ref{completeInfo} to the next section.
\begin{proof}[{\it Proof of Theorem \ref{mainAdThm}}.]  By Theorem \ref{numRounds}, every transmission lasts at most $O(n^3)$ rounds, so it remains to show that there are at most $n^2$ failed transmissions.  By Theorem \ref{maxBadness}, by the end of at most $n-1$ failed transmissions, there will be at least one failed transmission $\mathtt{T}$ such that the sender will have all status report parcels from every node on $\mathcal{P}_{\mathtt{T}}$.  Then by Theorem \ref{completeInfo}, the sender can eliminate a corrupt node.  At this point, lines ({\bf \ref{routingRules5M}.169-177}) essentially call for the protocol to start over, wiping clear all buffers except for the eliminated nodes buffer, which will now contain the identity of a newly eliminated node.  The transmission of the latest codeword not yet transmitted then resumes (see comments on ({\bf \ref{routingRules5M}.214})), and the argument can be applied to the new network, consisting of $n-1$ nodes.  Since the node-controlling adversary can corrupt at most $n-2$ nodes (the sender and receiver are incorruptible by the conforming assumption), this can happen at most $n-2$ times, yielding the bound of $n^2$ for the maximum number of failed transmissions.
\end{proof}
\subsection{Main Technical Proof of the Node-Controlling $+$ Edge-Scheduling Protocol} \label{pfCompleteInfo}
\indent \indent In this section, we aim to prove Lemma \ref{completeInfo} which states that the sender will be able to eliminate a corrupt node if he has the complete status reports from every node that participated in some failed transmission $\mathtt{T}$.  We begin by formally defining the three reasons a transmission may fail, and prove that every failed transmission falls under one of these three cases.
\begin{thm} \label{cases} At the end of any transmission $\mathtt{T}$, (at least) one of the following necessarily happens:
	\begin{enumerate}\setlength{\itemsep}{1pt} \setlength{\parskip}{0pt} \setlength{\parsep}{0pt}
	\item[\textsf{S1.}] The receiver has received at least $D-6n^3$ {\it distinct} packets corresponding to the current codeword
	\item[\textsf{F2.}] S1 does not happen, and the sender has {\em knowingly}\footnote{Recall that by the definition of ``inserted'' (see \ref{insert}), the sender may not have received confirmation (as in Definition \ref{confRec}) that a packet he outputted along some edge was received by the adjacent node.  Case F3 requires that the sender has {\it received confirmation} for at least $D$ packets.} inserted less than $D$ packets
	\item[\textsf{F3.}] S1 does not happen, the sender has {\em knowingly}\lastfootnote \space inserted at least $D$ packets, and the receiver has {\it not} received any duplicated packets corresponding to the current codeword
	\item[\textsf{F4.}] S1, F2, and F3 all do not happen
	\end{enumerate}
\end{thm}
\begin{proof} That the four cases cover all possibilities (and are disjoint) is immediate.  Also, in the case of S1, the receiver can necessarily decode by Lemma \ref{decoding}, and hence that case corresponds to a {\it successful} transmission.  Therefore, all failed transmissions must fall under one of the other three cases. \end{proof}
Note that case F2 roughly corresponds to {\it packet duplication}, since the sender is blocked from inserting packets in at least $4D-D$ rounds, indicating jamming that cannot be accounted for by edge failures alone.  Case F3 roughly corresponds to {\it packet deletion}, since the $D$ packets the sender inserted do not reach the receiver (otherwise the receiver could have decoded as by Lemma \ref{decoding}), and case F4 corresponds to a mixed adversarial strategy of {\it packet deletions and duplications}.  We treat each case separately in Theorems \ref{F2}, \ref{F3} and \ref{F4} below, thus proving Theorem \ref{completeInfo}:
\begin{proof}[{\it Proof of Theorem \ref{completeInfo}.}]  Theorem \ref{cases} guarantees that each failed transmission falls under F2, F3, or F4, and the theorem is proven for each case below in Theorems \ref{F2}, \ref{F3} and \ref{F4}.
\end{proof}
We declare once-and-for-all that at any time, $G$ will refer to nodes still a part of the network, i.e.\ nodes that have not been eliminated by the sender.\\[.5cm]
%
{\Large \bf Handling Failures as in F2: Packet Duplication}\\[.2cm]
\indent The goal of this section will be to prove the following theorem.
\begin{thm} \label{F2} Suppose transmission $\mathtt{T}$ failed and falls under case F2, and at some later time (after transmission $\mathtt{T}$ but before any additional nodes have been eliminated) the sender has received all of the status report parcels from all nodes on $\mathcal{P}_{\mathtt{T}}$.  Then the sender can eliminate a corrupt node.
\end{thm}
The idea of the proof is as follows.  Case F2 of transmission failure roughly corresponds to {\it packet duplication}: there is a node $N \in G$ who is jamming the network by outputting duplicate packets.  Notice that in terms of network potential (see Definition \ref{potDef}), the fact that $N$ is outputting more packets than he is inputting means that $N$ will be responsible for illegal increases in network potential.  Using the status reports for case F2, which include nodes' signatures on changes of network potential due to packet transfers and re-shuffling, we will catch $N$ by looking for a node who caused a greater increase in potential than is possible if it had been acting honestly.  The formal proof of this fact will require some work.  We begin with the following definitions:
\begin{defn} The {\it conforming} assumption on the node-controlling and edge-scheduling adversaries demand that for every round there is a path connecting the sender and receiver consisting of edges that are ``up'' and through uncorrupted nodes.  We will refer to this path as the \textsf{active honest path} for round $\mathtt{t}$ and denote it by $\mathsf{P}_{\mathtt{t}}$, noting that the path may not be the same for all rounds.\end{defn}
\begin{defn} \label{wasted} We will say that some round $\mathtt{t}$ (of transmission $\mathtt{T}$) is {\it wasted} if there is an edge $E(A,B)$ on that round's active honest path such that either {\bf \em Okay To Send Packet} ({\bf \ref{routingRulesM}.31}) or {\bf \em Okay To Receive Packet} ({\bf \ref{routingRulesM}.35}) returned false.\end{defn}
Intuitively, a round is wasted if an edge on the active honest path was prevented from passing a packet either because one of the nodes was blacklisted or because there was important broadcast information that had to be communicated before packets could be transferred.
\begin{lemma} \label{maxWasted} There are at most $4n^3$ wasted rounds in any transmission $\mathtt{T}$.\end{lemma}
\begin{proof}
We will prove this lemma via two claims.
    \begin{itemize}\setlength{\itemsep}{1pt} \setlength{\parskip}{0pt} \setlength{\parsep}{0pt}
    \item[] {\bf Claim 1.} Every wasted round $\mathtt{t}$ falls under (at least) one of the following cases:
        \begin{enumerate}\setlength{\itemsep}{1pt} \setlength{\parskip}{0pt} \setlength{\parsep}{0pt}
        \item An edge on $\mathsf{P}_{\mathtt{t}}$ transfers $\Theta_{\mathtt{T}}$ or a parcel of the sender's {\it Start of Transmission} (SOT) broadcast

	\item An edge on $\mathsf{P}_{\mathtt{t}}$ transfers the label of a node to remove from the blacklist

	\item An edge on $\mathsf{P}_{\mathtt{t}}$ transfers the information that one of the terminal nodes (on that edge) has the complete status report for a blacklisted node

	\item A node on $\mathsf{P}_{\mathtt{t}}$ {\it learns} a status report parcel for a blacklisted node.  More specifically, there is some node $(\widehat{N}, \mathtt{T}', \mathtt{T})$ that was part of the {\it SOT} broadcast (i.e.\ the node began the transmission on the sender's blacklist) and some other honest node $N \in G$ such that $N$ learns a new status report parcel from $\widehat{N}$ corresponding to transmission $\mathtt{T}'$.
	\end{enumerate}
	\item[] {\it Proof.} Let $\mathtt{t}$ be a wasted round.  Denote the active honest path for round $\mathtt{t}$ by $\mathsf{P}_{\mathtt{t}}= N_0 N_1 \dots N_l$.  By looking at {\bf \em Okay To Send Packet} and {\bf \em Okay To Receive Packet} ({\bf \ref{routingRulesM}.31} and {\bf \ref{routingRulesM}.35}), we first argue that that cases 1-3 cover all possible reasons for a wasted round, {\it except} the possibility that one node is on the other's blacklist.  To see this, we go through each line of {\bf \em Okay To Send Packet} and {\bf \em Okay To Receive Packet} and consider what happens along a specified edge on $\mathsf{P}_{\mathtt{t}}$, noting that by assumption this edge is {\it active} and the neighboring nodes are {\it honest}, so the appropriate broadcast parcel will be successfully transferred ({\bf \ref{routingRulesM}.15}).  In particular, it will be enough to show that for every reason a round may be wasted, there is a node on $\mathsf{P}_{\mathtt{t}}$ that has broadcast information of type 1-4 (see line ({\bf \ref{routingRules3M}.115})) that it has yet to transfer across an adjacent edge on $\mathsf{P}_{\mathtt{t}}$, as then we will fall under cases 1-3 of the Claim.
		\begin{itemize}\setlength{\itemsep}{1pt} \setlength{\parskip}{0pt} \setlength{\parsep}{0pt}
		\item If there is a node $N_i$ on $\mathsf{P}_{\mathtt{t}}$ that does not know all parcels of the {\it SOT} broadcast ({\bf \ref{routingRules5M}.200}), then find the last index $0 \leq j < i$ such that $N_j$ knows all of {\it SOT} but $N_{j+1}$ does not ($j$ is guaranteed to exist since $S=N_0$ knows all of {\it SOT} and $N_i$ does not).  Then $N_j$ has broadcast information of type 2 ({\bf \ref{routingRules3M}.115}) it has not yet sent along its edge to $N_{j+1}$.

		\item If there is a node $N_i$ on $\mathsf{P}_{\mathtt{t}}$ that knows $\Theta_{\mathtt{T}}$ or all of {\it SOT} but has not yet transferred one of these parcels across an edge of $\mathsf{P}_{\mathtt{t}}$, or $N_i$ knows the complete status report for some blacklisted node $\widehat{N}$ and $N_i$ has not yet passed this fact along an edge on $\mathsf{P}_{\mathtt{t}}$, then $N_i$ has broadcast information of type 1, 2, or 4 ({\bf \ref{routingRules3M}.115}).

		\item If there is a node $N_i$ on $\mathsf{P}_{\mathtt{t}}$ that knows of a node $\widehat{N}$ that should be removed from the blacklist, but it has yet to transfer this information across an edge of $\mathsf{P}_{\mathtt{t}}$, then $N_i$ has broadcast information of type 3 ({\bf \ref{routingRules3M}.115}).
		\end{itemize}
	It remains to consider the final reason one of these two functions may return false, namely when there is some $N_i$ on $\mathsf{P}_{\mathtt{t}}$ that is on the blacklist of either $N_{i-1}$ or $N_{i+1}$.  Let $BL_S$ denote the sender's blacklist at the start of round $\mathtt{t}$.
		\begin{itemize}\setlength{\itemsep}{1pt} \setlength{\parskip}{0pt} \setlength{\parsep}{0pt}
		\item If $N_i \notin BL_S$, then there will be some index $0 \leq j <i+1$ such that at the start of round $\mathtt{t}$, $N_i$ is not on $N_j$'s blacklist but $N_i$ {\em is} on $N_{j+1}$'s blacklist.  We may assume that both $N_j$ and $N_{j+1}$ have received the full {\it start of transmission} broadcast, else we would be in one of the above covered cases.  Since $N_i$ is on $N_{j+1}$'s blacklist, $N_i$ must have begun the transmission on the sender's blacklist (all internal nodes' blacklists are cleared at the end of each transmission ({\bf \ref{routingRules5M}.203}) and restored when they receive the {\it SOT} broadcast ({\bf \ref{routingRules5M}.200}), ({\bf \ref{routingRules4M}.137-138})).  However, since $N_i$ is not on $N_j$'s blacklist as of round $\mathtt{t}$ and $N_j$ has received the full {\it SOT} broadcast, at some point in $\mathtt{T}$, $N_j$ must have received a parcel from the sender indicating $N_i$ should be removed from the blacklist, as on ({\bf \ref{routingRules4M}.147-149}).  Since $N_j$ and $N_{j+1}$ are both honest and $N_j$ has received the information that $N_i$ should be removed from the blacklist (but $N_{j+1}$ has {\it not} received this information yet), it must be that this broadcast information of type 3 ({\bf \ref{routingRules3M}.115}) has not yet been successfully passed along $E(N_j, N_{j+1})$ yet.  In particular, $N_j$ has broadcast information of priority at least 3 that he has yet to successfully send to $N_{j+1}$, so he will send a parcel of priority 1, 2, or 3 in round $\mathtt{t}$, which are in turn covered by Statements 1 and 2 of the Lemma.

		\item If $N_i \in BL_S$, then there exist some $0 \leq j < i$ such that $N_j$ does {\it not} have $N_i$'s complete status report, but $N_{j+1}$ does (since $N_i \in BL_S$ implies $S$ does not have the complete status report, but $N_i$ has its own complete status report in its broadcast buffer, see Statement 2 of Lemma \ref{sigBuffersAreCorrect}).  Then if $N_{j+1}$ has not yet passed the fact that it has such knowledge along $E(N_{j+1}, N_j)$, then $N_{j+1}$ had broadcast information of type 4, in which case we fall under case 3 of the Claim.  On the other hand, if this information has already been passed along $E(N_{j+1}, N_j)$, then Statement 4 of Lemma \ref{sigBuffersAreCorrect} implies that $N_j$ is aware that $N_{j+1}$ knows the complete status report of $N_i$ (who by choice of $j$ is on $N_j$'s blacklist), and hence $\alpha$ will necessarily be set as on ({\bf \ref{routingRules3M}.119} or {\bf \ref{routingRules3M}.122}) and sent to $N_j$ on ({\bf \ref{routingRules3M}.103})).  Consequently, $N_{j+1}$ will receive $\alpha$ ({\bf \ref{routingRules3M}.105}) during Stage 1 communication of round $\mathtt{t}$, and will have broadcast information of type 5 ({\bf \ref{routingRules3M}.115}) it has not sent along $E(N_j, N_{j+1})$ yet.  This broadcast parcel can then be sent in Stage 2 communication of round $\mathtt{t}$ ({\bf \ref{routingRulesM}.15}), and this is covered by case 4 of the Claim.\hspace*{\fill} \hspace*{-12pt}$\scriptstyle{\blacksquare}$
		\end{itemize}

    \item[] {\bf Claim 2.} The maximum number of wasted rounds due to Case 1 of Claim 1 is $n^3$, the maximum number of wasted rounds due to Case 2 of Claim 1 is $n^3/2$, the maximum number of wasted rounds due to Case 3 of Claim 1 is $n^3$, and the maximum number of wasted rounds due to Case 4 of Claim 1 is $n^3$.

    \item[] {\it Proof.}
        \begin{enumerate}\setlength{\itemsep}{1pt} \setlength{\parskip}{0pt} \setlength{\parsep}{0pt}
        \item $\Theta_{\mathtt{T}}$ is one parcel ({\bf \ref{routingRules5M}.179}), and the {\it SOT} is at most $2n-1$ parcels ({\bf \ref{routingRules5M}.200}), so together they are at most $2n$ parcels.  Since each honest node will only broadcast each of these parcels at most once across any edge (as long as the broadcast is successful, which it will be if the round is wasted due to Case 1) and there are at most $n^2/2$ such edges, we have that Case 1 can happen at most $n^3$ times.

        \item Lemma \ref{removeOnce} says that no honest node $N$ will accept more than one distinct parcel (per transmission) that indicates some node $\widehat{N}$ should be removed from the blacklist.  Therefore, in terms of broadcasting this information, $N$ will have at most one broadcast parcel per transmission per node $\widehat{N}$ indicating $\widehat{N}$ should be removed from the blacklist.  Therefore, it can happen at most $n$ times that an {\it edge} adjacent to an honest node will need to broadcast a parcel indicating a node to remove.  Again since the number of edges is bounded by $n^2/2$, Case 2 can be responsible for a wasted round at most $n^3/2$ times.

	\item Lemma \ref{completeKnowledgeOnce} says that for any node $N \in G$ that has received the full {\it SOT} broadcast for transmission $\mathtt{T}$, if $N$ is honest then it will transmit along each edge at most once (per transmission) the fact that it knows some $\widehat{N}$'s complete status report.  Since each node has at most $n-1$ adjacent edges and there are at most $n$ nodes in $G$, Case 3 can be responsible for a wasted round at most $n^3$ times.

	\item Notice that Case 4 emphasizes the fact that a node on $\mathsf{P}_{\mathtt{t}}$ {\it learned} a blacklisted node's status report parcel.  Since there are at most $n-1$ blacklisted nodes at any time (see ({\bf \ref{routingRules5M}.187-188}) and Claim \ref{blacklistStuff}), and at most $n$ status report parcels per blacklisted node (see ({\bf \ref{routingRules4M}.142-45}) and Lemma \ref{blacklistStuff2}), an honest node can {\it learn a new} status report parcel at most $n(n-1) < n^2$ times per transmission (see Statement 3 of Lemma \ref{sigBuffersAreCorrect} which says honest nodes will not ever ``unlearn'' relevant status report parcels).  Since there are at most $n$ nodes, Case 4 can be responsible for a wasted round at most $n^3$ times.\hspace*{\fill} \hspace*{-12pt}$\scriptstyle{\blacksquare}$
        \end{enumerate}
    \end{itemize}
Claim 1 guarantees every wasted round falls under Case 1-4, and Claim 2 says these can happen at most $4n^3$ rounds, which proves the lemma.
\end{proof}
We now define the notation we will use to describe the specific information the status reports contain in the case of F2 (see ({\bf \ref{setupCodeM}.12}), ({\bf \ref{setupCodeM}.17}), ({\bf \ref{setupCodeM}.32}), and ({\bf \ref{routingRules4M}.142-145}))\footnote{On a technical point, since our protocol calls for internal nodes to keep {\it old} codeword packets in their buffers from one transmission to the next, packets being transferred during some transmission may correspond to old codewords.  We emphasize that the quantities in $SIG_{A,A}$, $SIG[2]$, and $SIG[3]$ include old codeword packets, while $SIG[1]$ and $SIG[p]$ do {\it not} count old codeword packets (see {\bf \ref{routingRulesM}.11} and {\bf \ref{routingRules2M}.59-60}).}:
	\begin{itemize}\setlength{\itemsep}{1pt} \setlength{\parskip}{0pt} \setlength{\parsep}{0pt}
	\item[-] $SIG_{A,A}$ denotes the net decrease in $A$'s potential due to re-shuffling packets in the current transmission.
	\item[-] $SIG^A[2]_{A,B}$ denotes the net increase in $B$'s potential due to packet transfers across directed edge $E(A,B)$, as signed by $B$ and stored in $A$'s signature buffer (({\bf \ref{routingRules2M}.75}), ({\bf \ref{routingRulesM}.11}), and ({\bf \ref{routingRulesM}.07})).
	\item[-] $SIG^A[2]_{B,A}$ denotes the net decrease in $B$'s potential due to packet transfers across directed edge $E(B,A)$, as signed by $B$ and stored in $A$'s signature buffer.  Notice that $SIG^A[2]_{B,A}$ is measured as a {\it positive quantity}, see lines ({\bf \ref{routingRules2M}.60}), ({\bf \ref{routingRules2M}.62}), and ({\bf \ref{routingRules2M}.74}).
	\item[-] $SIG^A[3]_{A,B}$ denotes the net decrease in $A$'s potential due to packet transfers across directed edge $E(A,B)$, which is signed by $A$ and stored its own signature buffer.  Notice that $SIG^A[3]_{A,B}$ is measured as a {\it positive quantity}, see line ({\bf \ref{routingRules2M}.49}).
	\item[-] $SIG^A[3]_{B,A}$ denotes the net increase in $A$'s potential due to packet transfers across directed edge $E(B,A)$, which is signed by $A$ and stored its own signature buffer ({\bf \ref{routingRules2M}.75}).
	\end{itemize}
\begin{lemma} \label{mainF2Lemma} Suppose transmission $\mathtt{T}$ failed and falls under case F2, and at some later time (after transmission $\mathtt{T}$ but before any additional nodes have been eliminated) the sender has received all of the status reports from every node on $\mathcal{P}_{\mathtt{T}}$.  Then one of the following two things happens:
	\begin{enumerate}\setlength{\itemsep}{1pt} \setlength{\parskip}{0pt} \setlength{\parsep}{0pt}
	\item There is some node $A \in G$ whose status report indicates that $A$ is corrupt\footnote{This includes, but is not limited to: 1) The node has returned a (value, signature) pair, where the value is not in an appropriate domain; 2) The node has returned non-zero values indicating interaction with blacklisted or eliminated nodes; 3)The node has reported values for $SIG^A[3]_{S,A}$ that are inconsistent with the sender's quantity $SIG^S[2]_{S,A}$; or 4) The node has returned outdated information in their status report.  By ``outdated'' information, we mean that as part of its status report, $A$ returned a (value, signature) pair using a signature he received in round $\mathtt{t}$ from one of $A$'s neighbors $N$, but in $N$'s status report, $N$ provided a (value, signature) pair from $A$ indicating they communicated {\it after} round $\mathtt{t}$ and that $A$ was necessarily using an outdated signature from $N$.}.

	\item There is some $A \in G$ whose potential at the start of $\mathtt{T}$ plus the net increase in potential during $\mathtt{T}$ is smaller than its net decrease in potential during $\mathtt{T}$.  More specifically, note that $A$'s net {\bf \em increase} in potential, as claimed by itself, is given by:
		\begin{equation*}
		\sum_{B \in \mathcal{P} \setminus A} SIG^A[3]_{B,A}
		\end{equation*}
	Also, $A$'s net {\bf \em decrease} in potential, as documented by all of its neighbors and its own loss due to re-shuffling, is given by:
		\begin{equation*}
		SIG_{A,A} + \sum_{B \in \mathcal{P} \setminus A} SIG^B[2]_{A,B}
		\end{equation*}
	Then case 2) says there exists some $A \in G$ such that:
		\begin{equation} \label{F2Incrimination}
		4n^3 -4n^2 + \sum_{B \in \mathcal{P} \setminus A} SIG^A[3]_{B,A} \quad < \quad SIG_{A,A} + \sum_{B \in \mathcal{P} \setminus A} SIG^B[2]_{A,B},
		\end{equation}
	where the $4n^3-4n^2$ term on the LHS is an upper bound for the maximum potential a node should have at the outset of a transmission (see proof of Claim \ref{capacity1}).
	\end{enumerate}
\end{lemma}
\begin{proof} The idea of the proof is to use Lemma \ref{potentialDrop2}, which argues that in the absence of malicious activity, the network potential should drop by at least $n$ every (non-wasted) round in which the sender is unable to insert a packet.  Then since the sender could not insert a packet in at least $3D$ rounds (case F2 states the sender inserted fewer than $D$ packets in the $4D$ rounds of the transmission) and since there are at most $4n^3$ wasted rounds per transmission, the network potential should have dropped by at least $(n)(3D-4n^3) > 2nD + 8n^4$ (since $D = \frac{6n^3}{\lambda} > 12n^3$ as $\lambda < 1/2$).  However, this is impossible, since the maximum network potential in the network at the start of the transmission (which is upper bounded from the capacity of the network) is $4n^4$ (Lemma \ref{capacity1}) plus the maximum amount of network potential increase during transmission $\mathtt{T}$ is $2nD$ (since the sender inserted fewer than $D$ packets at maximum height $2n$), and hence the sum of these is less than $2nD +8n^4$, resulting in a negative network potential.  Since network potential can never be negative, there must be illegal increases to network potential not accounted for above, and the node responsible for these increases is necessarily corrupt.  We now formalize this argument, showing how to find such an offending node and prove it is corrupt.

Let $\beta$ denote the number of rounds in transmission $\mathtt{T}$ that the sender was blocked from inserting any packets, and $\mathcal{P}$ denote the participating list for $\mathtt{T}$.
	\begin{enumerate}\setlength{\itemsep}{4pt} \setlength{\parskip}{0pt} \setlength{\parsep}{0pt}	
	\item[] \textsf{Obs.\ 1:  If there exists $A \in \mathcal{P}$ such that one of the following inequalities is not true, then $A$ is corrupt.}
		\begin{alignat}{2}
		0 \thickspace &\leq \thickspace SIG_{A,A} \qquad 0 \thickspace \leq \hspace{-.4cm} \sum_{B \in \mathcal{P} \setminus \{A,S\}} \hspace{-.2cm} (SIG^A[2]_{B,A} - SIG^A[3]_{B,A})  \notag
		\end{alignat}
	
	\item[] \textsf{Proof.} The above inequalities state that for honest nodes, the potential changes due to re-shuffling and packet transfers are strictly non-positive (this was the content of Lemma \ref{item3}).  This observation is proved rigorously as Statements 4 and 5 of Lemma \ref{item3M} in Section \ref{mBareBones}.\hspace*{\fill} \hspace*{-12pt}$\scriptstyle{\square}$\vspace{.4cm}
	
	\item[] \textsf{Obs.\ 2:  The increase in network potential due to packet insertions is at most $2nD +2n^2$.  More precisely, either there exists a node $A \in G$ such that the sender can eliminate $A$, or the following inequality is true:}
		\begin{equation} \label{Obs3}
		\sum_{A \in \mathcal{P} \setminus S} SIG^A[3]_{S,A} < 2nD + 2n^2
		\end{equation}
	\item[] \textsf{Proof.}  By hypothesis, the sender {\it knowingly} inserted less than $D$ packets in transmission $\mathtt{T}$, and each packet can increase network potential by at most $2n$.  The sum on the LHS of \eqref{Obs3} represents the increase in potential claimed by nodes participating in $\mathtt{T}$ caused by packet insertions.  This quantity should match the sender's perspective of the potential increase (which is at most $2nD$), with the exception of potential increases caused by packets that were inserted but $S$ did not received confirmation of receipt (see Definition \ref{confRec}).  There can be at most one such packet per edge, causing an additional potential increase of at most $2n$ per edge.  Adding this additional potential increase to the maximum increase of $2nD$ of the sender's perspective is the RHS of \eqref{Obs3}.  The formal proof can be found in Lemma \ref{obsProofs} in Section \ref{mBareBones}.\hspace*{\fill} \hspace*{-12pt}$\scriptstyle{\square}$\vspace{.4cm}
	
	\item[] \textsf{Obs.\ 3:  $\beta \geq  3D-n$}.  (Recall that $\beta$ denotes the number of blocked rounds in $\mathtt{T}$.)
	
	\item[] \textsf{Proof.}   Since the sender knowingly inserted fewer than $D$ packets, there could be at most $n$ packets (one packet per edge) that was inserted unbeknownst to $S$, and hence the sender must have been blocked for (at least) all but $D+n$ of the rounds of the transmission.  Since the number of rounds in a transmission is $4D$ ({\bf \ref{routingRulesM}.02}), we have that $\beta \geq 3D-n$.\hspace*{\fill} \hspace*{-12pt}$\scriptstyle{\square}$\vspace{.4cm}
	\end{enumerate}
Let $\mathcal{H}_{\mathtt{T}} \subseteq \mathcal{P}_{\mathtt{T}}$ denote the subset of participating nodes that are honest (the sender is of course oblivious as to which nodes are honest, but we will nevertheless make use of $\mathcal{H}_{\mathtt{T}}$ in the following argument).  For notational convenience, since transmission $\mathtt{T}$ is fixed, we suppress the subscript and write simply $\mathcal{H}$ and $\mathcal{P}$.  We make the following simple observations:
	\begin{enumerate}\setlength{\itemsep}{4pt} \setlength{\parskip}{0pt} \setlength{\parsep}{0pt}
	\item[] \textsf{Obs. 4:  The following inequality is true:}
		\begin{equation} \label{1297}
		2nD+4n^4-4n^3+2n^2 < \sum_{A \in \mathcal{H} \setminus S} \hspace{-.2cm} SIG_{A,A} + \sum \hspace{-1cm} \sum_{A \in \mathcal{H} \setminus S \space \thickspace B \in \mathcal{P} \setminus \{A,S\} \quad} \hspace{-1cm}(SIG^A[2]_{B,A} - SIG^A[3]_{B,A})
		\end{equation}
	\item[] \textsf{Proof.}  This follows immediately from Observation 3 and Lemma \ref{potentialDrop2}, since:
		\begin{alignat}{2}
		n(\beta_{\mathtt{T}} -4n^3) &\geq n(3D-n-4n^3) \notag \\
		&\geq 2nD + 4n^4-4n^3+2n^2, \notag
		\end{alignat}
	where the first inequality is Observation 3, and the second follows because $D = 6n^3/\lambda \geq 8n^3 \geq 8n^3-4n^2+3n$.\hspace*{\fill} \hspace*{-12pt}$\scriptstyle{\square}$\vspace{.4cm}
	
	\item[] \textsf{Obs. 5:  Either a corrupt node can be identified as in Obs.\ 1 or 2, or there is some $A \in \mathcal{P}$ such that:}
		\begin{equation}\label{yesT}
		4n^3 - 4n^2< SIG_{A,A} + \sum_{B \in \mathcal{P} \setminus A} SIG^B[2]_{A,B} - SIG^A[3]_{B,A}
		\end{equation}
	\item[] \textsf{Proof.} Consider the following inequalities:
		\begin{alignat}{2}
		2nD + 4n^4 -  4n^3 +2n^2 &< \sum_{A \in \mathcal{H} \setminus S} \hspace{-.2cm} SIG_{A,A} + \thickspace \sum \hspace{-1cm} \sum_{A \in \mathcal{H} \setminus S \space \thickspace B \in \mathcal{P} \setminus \{A,S\} \quad} \hspace{-1cm}(SIG^A[2]_{B,A} - SIG^A[3]_{B,A}) \notag \\
		&\leq \sum_{A \in \mathcal{P} \setminus S} \hspace{-.2cm} SIG_{A,A} + \thickspace \sum \hspace{-1cm} \sum_{A \in \mathcal{P} \setminus S \space \thickspace B \in \mathcal{P} \setminus \{A,S\} \quad} \hspace{-1cm}(SIG^A[2]_{B,A} - SIG^A[3]_{B,A}) \notag \\
		&= \sum_{A \in \mathcal{P} \setminus S} \hspace{-.2cm} SIG_{A,A} + \thickspace \sum \hspace{-1cm} \sum_{A \in \mathcal{P} \setminus S \space \thickspace B \in \mathcal{P} \setminus \{A,S\} \quad} \hspace{-1cm}(SIG^B[2]_{A,B} - SIG^A[3]_{B,A}) \label{1338}
		\end{alignat}
	Above, the top inequality follows from Obs.\ 4, the second inequality follows from Obs.\ 1, and the third line is a re-arranging and re-labelling of terms.  Subtracting $2nD+2n^2$ from both sides:
		\begin{alignat}{2}
		4n^4 -4n^3&< \sum_{A \in \mathcal{P} \setminus S} \hspace{-.2cm} SIG_{A,A} + \thickspace \sum \hspace{-1cm} \sum_{A \in \mathcal{P} \setminus S \space \thickspace B \in \mathcal{P} \setminus \{A,S\} \quad} \hspace{-1cm}(SIG^B[2]_{A,B} - SIG^A[3]_{B,A}) \thickspace - 2nD-2n^2 \notag \\
		&< \sum_{A \in \mathcal{P} \setminus S} \hspace{-.2cm} SIG_{A,A} + \thickspace \sum \hspace{-1cm} \sum_{A \in \mathcal{P} \setminus S \space \thickspace B \in \mathcal{P} \setminus \{A,S\} \quad} \hspace{-1cm}(SIG^B[2]_{A,B} - SIG^A[3]_{B,A}) + \sum_{A \in \mathcal{P} \setminus S} \hspace{-.2cm} -SIG^A[3]_{S,A}\notag \\
		&= \sum_{A \in \mathcal{P} \setminus S} \hspace{-.2cm} SIG_{A,A} + \thickspace \sum \hspace{-1cm} \sum_{A \in \mathcal{P} \setminus S \space \thickspace B \in \mathcal{P} \setminus \{A,S\} \quad} \hspace{-1cm}(SIG^B[2]_{A,B} - SIG^A[3]_{B,A}) \thickspace +\notag \\
		& \quad \hspace{1cm} \sum_{A \in \mathcal{P} \setminus S} (SIG^S[2]_{A,S}-SIG^A[3]_{S,A}) \notag \\
		&= \sum_{A \in \mathcal{P} \setminus S} \hspace{-.2cm} SIG_{A,A} + \hspace{-.2cm} \sum_{A \in \mathcal{P} \setminus S \thickspace} \hspace{-.4cm} \sum_{\quad B \in \mathcal{P} \setminus A} \hspace{-.2cm}(SIG^B[2]_{A,B} - SIG^A[3]_{B,A}) \label{1349}
		\end{alignat}
	Above, the top inequality is from \eqref{1338}, the second follows from Obs.\ 2, the third line is because $SIG^S[2]_{A,S} = 0$ for all $A \in G$ ($S$ never {\it receives} a packet from anyone, see ({\bf \ref{routingRulesM}.21-22})), and the final line comes from combining sums.
	Using an averaging argument, this implies there is some $A \in \mathcal{P} \setminus S$ such that:
		\begin{equation} \label{1351}
		4n^3-4n^2 < SIG_{A,A} + \sum_{B \in \mathcal{P} \setminus A} (SIG^B[2]_{A,B} - SIG^A[3]_{B,A}),
		\end{equation}
	which is \eqref{yesT}.\hspace*{\fill} \hspace*{-12pt}$\scriptstyle{\square}$
	\end{enumerate}
Therefore, if a node cannot be eliminated as in Obs.\ 1 or 2 (which are covered by Case 1 of Lemma \ref{mainF2Lemma}), then Obs.\ 5 implies that Case 2 of Lemma \ref{mainF2Lemma} is true.
\end{proof}
{\it Proof of Theorem \ref{F2}.}  This Theorem now follows immediately from Lemma \ref{mainF2Lemma} and the fact that a node $A \in G$ for which \eqref{F2Incrimination} is true is necessarily corrupt.  Intuitively, such a node $A \in G$ is corrupt since the potential decrease at $A$ is higher than can be accounted for by $A$'s potential at the outset of $\mathtt{T}$ plus the potential increase due to packet insertions from the sender.  The formal statement and proof of this fact is the content of Corollary \ref{F2Cor}.\hspace*{\fill} \hspace*{-12pt}$\scriptstyle{\blacksquare}$\\[.5cm]
{\Large \bf Handling Failures as in F3: Packet Deletion}\\[.2cm]
\indent The goal of this section will be to prove the following theorem.
\begin{thm} \label{F3} Suppose transmission $\mathtt{T}$ failed and falls under case F3, and at some later time (after transmission $\mathtt{T}$ but before any additional nodes have been eliminated) the sender has received all of the status report parcels from all nodes on $\mathcal{P}_{\mathtt{T}}$.  Then the sender can eliminate a corrupt node.
\end{thm}
The idea of the proof is as follows.  Case F3 of transmission failure roughly corresponds to {\it packet deletion}: there is a node $N \in G$ who is deleting some packets transferred to it instead of forwarding them on.  Using the status reports for case F3, which include nodes' signatures on the net number of packets that have passed across each of their edges, we will catch $N$ by looking for a node who input more packets than it output, and this difference is greater than the buffer capacity of the node.
\begin{proof}  We first define the notation we will use to describe the specific information the status reports contain in the case of F3 (({\bf \ref{setupCodeM}.17}), ({\bf \ref{setupCodeM}.32}), and ({\bf \ref{routingRules4M}.144})):
	\begin{itemize}\setlength{\itemsep}{1pt} \setlength{\parskip}{0pt} \setlength{\parsep}{0pt}
	\item[-] $SIG^A[1]_{A,B}$ denotes the net number of packets that have travelled across directed edge $E(A,B)$, as signed by $B$ and stored in $A$'s (outgoing) signature buffer.
	\item[-] $SIG^A[1]_{B,A}$ denotes the net number of packets that have travelled across directed edge $E(B,A)$, as signed by $B$ and stored in $A$'s (incoming) signature buffer.
	\end{itemize}
By the third Statement of Lemma \ref{secondSigFacts}, either a corrupt node can be eliminated, or the following is true for all $A,B \in G$:
	\begin{equation*}
	|SIG^A[1]_{B,A}-SIG^B[1]_{B,A}| \leq 1 \qquad \mbox{and} \qquad |SIG^A[1]_{A,B}-SIG^B[1]_{A,B}| \leq 1
	\end{equation*}
Then summing over all $A,B \in \mathcal{P}$:
	\begin{equation} \label{1385}
	\sum_{A,B \in \mathcal{P}, A \neq B} \hspace*{-.5cm} |SIG^A[1]_{B,A}-SIG^B[1]_{B,A}| \thickspace \leq \thickspace n^2
	\end{equation}
This in turn implies that:
	\begin{alignat}{2}
	-n^2 \thickspace &\leq \sum_{A,B \in \mathcal{P}, A \neq B} (SIG^A[1]_{B,A}-SIG^B[1]_{B,A}) \notag \\
	&= \sum_{A \in \mathcal{P}} \sum_{B \in \mathcal{P} \setminus A} (SIG^A[1]_{B,A} - SIG^A[1]_{A,B}) \notag \\
	&= \sum_{B \in \mathcal{P} \setminus R} SIG^R[1]_{B,R} \thickspace - \sum_{B \in \mathcal{P} \setminus S} SIG^S[1]_{S,B} \thickspace + \hspace*{-1cm} \sum_{\qquad A \in \mathcal{P} \setminus \{R,S\} \space B \in \mathcal{P} \setminus A} \hspace*{-1cm} \sum \left( SIG^A[1]_{B,A} - SIG^A[1]_{A,B} \right) \notag \\
	&\leq -6n^3 \thickspace + \hspace*{-1cm} \sum_{\qquad A \in \mathcal{P} \setminus \{R,S\} \space B \in \mathcal{P} \setminus A} \hspace*{-1cm} \sum (SIG^A[1]_{B,A} - SIG^A[1]_{A,B}) \notag
	\end{alignat}
The first inequality is from \eqref{1385}, the second line is from re-labelling and re-arranging terms, the third line comes from separating out the terms $A=S$ and $A=R$ and noting that $SIG^R[1]_{R,B} = SIG^S[1]_{B,S} = 0$ (since the receiver will never output packets to other nodes and the sender will never input packets, see ({\bf \ref{routingRulesM}.16-20}) and ({\bf \ref{routingRulesM}.21-22})), and the final inequality is due to the fact that we are in case F3, so the sender knowingly inserted $D$ packets, but the receiver received fewer than $D-6n^3$ packets corresponding to the current codeword\footnote{More precisely, F3 states that the sender knowingly inserted at least $D$ packets and the receiver did not receive any packet (from the current codeword) more than once.  By Fact 1', since we are not in case S1, the receiver got fewer than $D - 6n^3$ distinct packets corresponding to the current codeword.}.  Using an averaging argument, we can find some $A \in G$ such that:
	\begin{equation}\label{line1401}
	4n^2 - 8n < 6n^2-n < \sum_{B \in \mathcal{P} \setminus A} (SIG^A[1]_{B,A}-SIG^A[1]_{A,B}),
	\end{equation}
where the first inequality is obvious.  Statement 7 of Lemma \ref{item3M} now guarantees that $A$ is corrupt\footnote{Intuitively, $A$ must be corrupt since the sum on the RHS of \eqref{line1401} represents the net number of packets $A$ input minus the number of packets $A$ output.  Since this difference is larger than the capacity of $A$'s internal buffers, $A$ must have deleted at least one packet and is necessarily corrupt.}.
\end{proof}
\hspace*{\fill}\\[.3cm]
{\Large \bf Handling Failures as in F4: Packet Duplication $+$ Deletion}\\[.2cm]
\indent The goal of this section will be to prove the following theorem.
\begin{thm} \label{F4} Suppose transmission $\mathtt{T}$ failed and falls under case F4, and at some later time (after transmission $\mathtt{T}$ but before any additional nodes have been eliminated) the sender has received all of the status report parcels from all nodes on $\mathcal{P}_{\mathtt{T}}$.  Then the sender can eliminate a corrupt node.
\end{thm}
The idea of the proof is as follows.  Case F4 of transmission failure roughly corresponds to packet duplication {\it and} packet deletion: there is a node $N \in G$ who is replacing valid packets with copies of old packets it has already passed on.  Therefore, simply tracking potential changes and net packets into and out of $N$ will not help us to locate $N$, as both of these quantities will be consistent with honest behavior.  Instead, we use the fact that case F4 implies that the receiver will have received some packet $p$ (from the current codeword) twice.  We will then use the status reports for case F4, which include nodes' signatures on the net number of times $p$ has crossed each of their edges, to find a corrupt node $N$ by looking for a node who output $p$ more times than it input $p$.
\begin{proof} By definition of F4, the receiver received some packet $p$ (corresponding to the current codeword) at least twice.  Therefore, when ({\bf \ref{routingRules5M}.178-179}) is reached, the receiver will send the label of $p$ back to the sender (which reaches $S$ by the end of the transmission by Lemma \ref{finalBC}), and this is in turn broadcasted as part of the sender's {\it start of transmission} broadcast in the following transmission (({\bf \ref{routingRules5M}.190-192}) and ({\bf \ref{routingRules5M}.200})).  We will use the following notation to describe the specific information the status reports contain in the case of F4 (see ({\bf \ref{setupCodeM}.17}) and ({\bf \ref{setupCodeM}.32})):
	\begin{itemize}\setlength{\itemsep}{1pt} \setlength{\parskip}{0pt} \setlength{\parsep}{0pt}
	\item[-] $SIG^A[p]_{A,B}$ denotes the net number of times $p$ has travelled across directed edge $E(A,B)$, as signed by $B$ and stored in $A$'s (outgoing) signature buffer.
	\item[-] $SIG^A[p]_{B,A}$ denotes the net number of times $p$ has travelled across directed edge $E(B,A)$, as signed by $B$ and stored in $A$'s (incoming) signature buffer.
	\end{itemize}
Consider the following string of equalities:
	\begin{alignat}{2} \label{line1440}
	0 \thickspace &= \sum_{A \in \mathcal{P}_{\mathtt{T}}} \sum_{B \in \mathcal{P}_{\mathtt{T}}} (SIG^A[p]_{B,A} - SIG^A[p]_{B,A}) \notag \\
	&= \sum_{A \in \mathcal{P}_{\mathtt{T}}} \sum_{B \in \mathcal{P}_{\mathtt{T}}} (SIG^B[p]_{A,B} - SIG^A[p]_{B,A}) \notag \\
	&= \hspace*{-.6cm} \sum_{\qquad A \in \mathcal{P}_{\mathtt{T}} \setminus \{R,S\} \thickspace B \in \mathcal{P}_{\mathtt{T}}} \hspace*{-.75cm} \sum (SIG^B[p]_{A,B} - SIG^A[p]_{B,A}) + \sum_{B \in \mathcal{P}_{\mathtt{T}}} (SIG^B[p]_{R,B} - SIG^R[p]_{B,R}) \thickspace + \notag \\ & \hspace{3cm} \sum_{B \in \mathcal{P}_{\mathtt{T}}} (SIG^B[p]_{S,B} - SIG^S[p]_{B,S}) \notag \\
	&= \hspace*{-.6cm} \sum_{\qquad A \in \mathcal{P}_{\mathtt{T}} \setminus \{R,S\} \thickspace B \in \mathcal{P}_{\mathtt{T}}} \hspace*{-.75cm} \sum (SIG^B[p]_{A,B} - SIG^A[p]_{B,A}) + \sum_{B \in \mathcal{P}_{\mathtt{T}}} (SIG^B[p]_{S,B} - SIG^R[p]_{B,R})
	\end{alignat}
The first equality is trivial, the second equality comes from re-labelling and rearranging the terms of the sum, the third comes from separating out the $A=S$ and $A=R$ terms, and the final equality results from the fact that $R$ never outputs packets and $S$ never inputs packets, and hence they will never sign non-zero values for $SIG[p]_{R,B}$ or $SIG[p]_{B,S}$, respectively (see ({\bf \ref{routingRulesM}.16-20}) and ({\bf \ref{routingRulesM}.21-22})).  Because $p$ was received by $R$ at least twice (by choice of $p$) and $S$ will never send any packet to more than one node\footnote{This was proven in Observations 2-3 of Lemma \ref{packetProliferation32} for the edge-scheduling protocol.  However, the proofs of these observations remain valid in the (node-controlling$+$edge-scheduling) model because the sender is honest (by the conforming adversary assumption).}, we have that:
	\begin{equation}
	\sum_{B \in \mathcal{P}_{\mathtt{T}}} (SIG^B[p]_{S,B} - SIG^R[p]_{B,R}) \leq -1
	\end{equation}
Plugging this into \eqref{line1440} and rearranging:
	\begin{equation}
	1 \thickspace \leq \hspace*{-.6cm} \sum_{\qquad A \in \mathcal{P}_{\mathtt{T}} \setminus \{R,S\} \thickspace B \in \mathcal{P}_{\mathtt{T}}} \hspace*{-.75cm} \sum (SIG^B[p]_{A,B} - SIG^A[p]_{B,A})
	\end{equation}
By an averaging argument, there must be some $A \in \mathcal{P}_{\mathtt{T}} \setminus \{R,S\}$ such that:
	\begin{equation}\label{line1455}
	1 \thickspace \leq \sum_{B \in \mathcal{P}_{\mathtt{T}}} (SIG^B[p]_{A,B} - SIG^A[p]_{B,A})
	\end{equation}
Now Statement 8 of Lemma \ref{item3M} says that $A$ is necessarily corrupt\footnote{Intuitively, $A$ is corrupt since \eqref{line1455} says that it has output $p$ more times than it input $p$.}.
\end{proof}
\section{Node-Controlling $+$ Edge-Scheduling Protocol: \hspace*{\fill}\\ Pseudo-Code Intensive Proofs} \label{mBareBones}
\indent \indent In this section, we give detailed proofs that walk through the pseudo-code of Figures \ref{setupCodeM} - \ref{routingRules5M} to argue very basic properties the protocol satisfies.  The following lemma will relieve the need to re-prove many of the lemmas of Sections \ref{proofs} and \ref{bareBones}.
\begin{lemma} \label{similar} Differences between the edge-scheduling adversary protocol and the (node-controlling $+$ edge-scheduling) adversary protocol all fall under one of the following cases:
	\begin{enumerate}
	\item Extra variables in the Setup Phase

	\item Length of transmission and codeword being transmitted for the current transmission

	\item Need to authenticate signatures on packets, as on ({\bf \ref{routingRulesM}.08}) and ({\bf \ref{routingRules2M}.68})

	\item Need to check if it is okay to send/receive packets, as on ({\bf \ref{routingRules2M}.59}) and ({\bf \ref{routingRules2M}.63})

	\item Broadcasting information, i.e.\ transmission of broadcast parcels and modifications of Broadcast Buffer, Data Buffer, and Signature Buffer
	\end{enumerate}
Furthermore, differences as in Cases 3 and 4 are identical to having an edge-fail in the edge-scheduling adversary protocol.  Also, differences as in Case 5 affect the routing protocol only insofar as their affect on Cases 3 and 4 above.  Furthermore, between any two honest nodes, the authentications of Case 3 never fail, and Case 4 failures correspond to ``wasted'' rounds (see Definition \ref{wasted}).
\end{lemma}
\begin{proof} Comparing the pseudo-code of Figures {\bf \ref{routingRules}}, {\bf \ref{routingRules2}}, and {\bf \ref{reShuffleRules}} to Figures {\bf \ref{routingRulesM}}, {\bf \ref{routingRules2M}}, and {\bf \ref{routingRules5M}}, as emphasized by line numbers in {\bf bold face} in the latter three, it is clear that all differences fall under Cases 1-5 of the lemma.  Also, all of the other methods in Figures {\bf \ref{routingRules3M}}-{\bf \ref{routingRules5M}} fall under Cases 4 and 5.

As for the differences as in Cases 3 and 4, it is clear that failing {\bf \em Verify Signature One} on ({\bf \ref{routingRules2M}.86-87}) is equivalent to the edge failing during Stage 2 (i.e.\ as if $p = \bot$ on ({\bf \ref{routingRules2M}.62}) causing ({\bf \ref{routingRules2M}.69}) to fail); failing {\bf \em Verify Signature Two} on ({\bf \ref{routingRules2M}.89-90}) is equivalent to the edge failing during Stage 1 (since this sets $H_{IN}$ and $RR$ to $\bot$ on ({\bf \ref{routingRules2M}.90}), which is equivalent to the communication on ({\bf \ref{routingRules}.06}) not being received); failing {\bf \em Okay to Send Packet} on ({\bf \ref{routingRules2M}.59}) is equivalent to the edge failing during Stage 2 (so that nothing is received on lines ({\bf \ref{routingRulesM}.22}/{\bf \ref{routingRules2M}.62})); and failing {\bf \em Okay to Receive Packet} on ({\bf \ref{routingRules2M}.63}) is equivalent to the edge failing during Stage 1 (i.e.\ as if nothing is received on ({\bf \ref{routingRulesM}.13}), so that $H_{OUT} = \bot$ on ({\bf \ref{routingRules2M}.63})).  Finally, differences as in Case 5 do not directly affect routing (except their affects captured by Cases 3 and 4) since the transfer of broadcast parcels and maintenance of the related buffers (signature, broadcast, and data buffers) happen independently of the routing of codeword packets.  This is evident by investigating the relevant {\bf bold} lines in Figures {\bf \ref{routingRulesM}}, {\bf \ref{routingRules2M}}, and {\bf \ref{routingRules5M}}.

The second part of the last sentence is true by definition of wasted (see Definition \ref{wasted}), and the first part follows from lines ({\bf \ref{routingRulesM}.11}), ({\bf \ref{routingRules2M}.49}), ({\bf \ref{routingRules2M}.60}), ({\bf \ref{routingRules2M}.75}), and Lemma \ref{sigRelationships}.
\end{proof}
\begin{lemma} \label{mProperDomains} The domains of all of the variables in Figures \ref{setupCodeM} and \ref{setupCode2M} are appropriate.  In other words, the protocol never calls for more information to be stored in an honest node's variable (buffer, packet, etc$.$) than the variable has room for.\end{lemma}
\begin{proof} The proof for variables and buffers that also appear in the edge-scheduling protocol follows from Lemmas \ref{pseudo} and \ref{variables}, since all differences between the edge-scheduling protocol and the (node-controlling $+$ edge-scheduling) protocol are equivalent to an edge-failure (Lemma \ref{similar}).  So it remains to prove the lemma for the new variables appearing in Figures \ref{setupCodeM} and \ref{setupCode2M} (i.e.\ the {\bf bold} line numbers).  The distribution of public and private keys ({\bf \ref{setupCodeM}.14}) is performed by a trusted third party, so these variables are as specified.  Below, when we refer to a specific node's variable, we implicitly assume the node is honest, as the lemma is only concerned about honest nodes.
	\begin{itemize}
	\item[] \textsf{Bandwidth $P$} ({\bf \ref{setupCodeM}.06}).  We look at all transfers along each directed edge in each stage of any round.  In Stage 1, this includes the transfer of $H_{OUT}$, $H_{FP}$, $FR$ ({\bf \ref{routingRulesM}.06}), $c_{bp}$, $\alpha$ ({\bf \ref{routingRulesM}.04}), and the seven signed items on ({\bf \ref{routingRulesM}.11}).  All of these have collective size $O(k + \log n)$ (({\bf \ref{setupCodeM}.03-04}), ({\bf \ref{setupCodeM}.21}), ({\bf \ref{setupCodeM}.23}), ({\bf \ref{setupCodeM}.24}), ({\bf \ref{setupCodeM}.27}), ({\bf \ref{setupCodeM}.29}), ({\bf \ref{setupCodeM}.32}), ({\bf \ref{setupCodeM}.35}), and ({\bf \ref{setupCodeM}.36})).  In Stage 2, this includes the transfer of the seven items on ({\bf \ref{routingRules2M}.60}) and $bp$ ({\bf \ref{routingRulesM}.15}).  Collectively, these have size $O(k + \log n)$ (({\bf \ref{setupCodeM}.03-04}), ({\bf \ref{setupCodeM}.17}), ({\bf \ref{setupCodeM}.18}), ({\bf \ref{setupCodeM}.21}), and ({\bf \ref{setupCodeM}.42})).
	
	\item[] \textsf{Potential Lost Due to Re-Shuffling $SIG_{N,N}$} ({\bf \ref{setupCodeM}.12}).  This is initialized to zero on ({\bf \ref{setupCode2M}.46}), after which it is only updated on ({\bf \ref{reShuffleRules}.76}), ({\bf \ref{routingRules2M}.50}), ({\bf \ref{routingRules2M}.80}), ({\bf \ref{routingRules2M}.82}), ({\bf \ref{routingRules4M}.128}), ({\bf \ref{routingRules4M}.133}), ({\bf \ref{routingRules4M}.141}), and ({\bf \ref{routingRules4M}.146}).  The first four of these increment $SIG_{N,N}$ by at most $2n$, and the latter four all reset $SIG_{N,N}$ to zero.  We will see in Lemma \ref{sigBuffersAreCorrect} below that $SIG_{N,N}$ will always represent the potential lost due to re-shuffling in at most one failed transmission, and consequently $SIG_{N,N}$ is polynomial in $n$, as required.
	
	\item[] \textsf{Broadcast Parcel $bp$ to Receive} ({\bf \ref{setupCodeM}.27}).  This is initialized to $\bot$ on ({\bf \ref{setupCode2M}.58}), after which it is only updated on ({\bf \ref{routingRules3M}.93}).  Either no value was received on ({\bf \ref{routingRules3M}.93}) (in which case $bp = \bot$), or it corresponds to the value sent on ({\bf \ref{routingRules3M}.97}).  As discussed below, the value of $bp$ sent on ({\bf \ref{routingRules3M}.97}) lies in the appropriate domain, and hence so does $bp$.

	\item[] \textsf{Broadcast Buffer Request $\alpha$} ({\bf \ref{setupCodeM}.28}).  This is initialized to $\bot$ on ({\bf \ref{setupCode2M}.58}), after which it is only updated as in {\bf \em Broadcast Parcel to Request} ({\bf \ref{routingRules3M}.117-122}).  On ({\bf \ref{routingRules3M}.117}), $\alpha$ is set to $\bot$, and on ({\bf \ref{routingRules3M}.119}) and ({\bf \ref{routingRules3M}.122}), $\alpha$ includes the label of a node and a status report parcel (see {\bf \ref{routingRules4M}.142-145}), and so $\alpha$ is bounded by $O(k+\log n)=P$ as required.

	\item[] \textsf{Outgoing Verification of Broadcast Parcel Bit $c_{bp}$} ({\bf \ref{setupCodeM}.29}).  This is initialized to zero on ({\bf \ref{setupCode2M}.58}), after which it is only updated as on ({\bf \ref{routingRules2M}.40}) and ({\bf \ref{routingRules3M}.111}), where it clearly lies in the appropriate domain.

	\item[] \textsf{Broadcast Parcel $bp$ to Send} ({\bf \ref{setupCodeM}.42}).  This is initialized to $\bot$ on ({\bf \ref{setupCode2M}.52}), after which it is only updated as in {\bf \em Determine Broadcast Parcel to Send} ({\bf \ref{routingRules3M}.115}).  Looking at the six types of broadcast parcels on line ({\bf \ref{routingRules3M}.115}) and comparing the corresponding domains of these variables in Figures {\bf \ref{setupCodeM}} and {\bf \ref{setupCode2M}}, we see that in each case, $bp$ can be expressed in $O(k+\log n) = P$ bits.

	\item[] \textsf{Incoming Verification of Broadcast Parcel Bit $c_{bp}$} ({\bf \ref{setupCodeM}.43}).  This is initialized to zero on ({\bf \ref{setupCode2M}.52}), after which it is only updated as on ({\bf \ref{routingRules3M}.105}) and ({\bf \ref{routingRules3M}.108}).  The value it takes on ({\bf \ref{routingRules3M}.105}) will either be set to zero (if no value was received), or it will equal the value of $c_{bp}$ sent on ({\bf \ref{routingRules3M}.101}), which as shown above is either a one or zero.  Meanwhile, the value it takes on ({\bf \ref{routingRules3M}.108}) is zero, so at all times $c_{bp}$ equals one or zero, as required.

	\item[] \textsf{First Parcel of Start of Transmission Broadcast $\Omega_{\mathtt{T}}$} ({\bf \ref{setupCode2M}.66}).  This is initialized to $(0,0,0,0)$ on ({\bf \ref{setupCode2M}.72}) and is only changed on ({\bf \ref{routingRules5M}.174}), ({\bf \ref{routingRules5M}.184}), ({\bf \ref{routingRules5M}.192}), ({\bf \ref{routingRules5M}.195}), and ({\bf \ref{routingRules5M}.198}).  In all of these cases, it is clear that $\Omega_{\mathtt{T}}$ can be expressed in $O(\log n)$ bits, as required.

	\item[] \textsf{Number of Rounds Blocked $\beta_{\mathtt{T}}$} ({\bf \ref{setupCode2M}.67}).  This is initialized to zero on ({\bf \ref{setupCode2M}.71}) and is only changed on ({\bf \ref{routingRulesM}.27}), ({\bf \ref{routingRules5M}.172}), and ({\bf \ref{routingRules5M}.201}).  Notice that in the latter two cases, $\beta_{\mathtt{T}}$ is reset to zero, while $\beta_{\mathtt{T}}$ can only be incremented by one on ({\bf \ref{routingRulesM}.27}) at most $4D$ times per transmission by ({\bf \ref{routingRulesM}.02}).  Since either line ({\bf \ref{routingRules5M}.172}) or line ({\bf \ref{routingRules5M}.201}) is reached at the end of every transmission (in the case a node is not eliminated as on line ({\bf \ref{routingRules4M}.163}) or ({\bf \ref{routingRules4M}.168}), line ({\bf \ref{routingRules5M}.201}) will be reached by the call on ({\bf \ref{routingRulesM}.29})), $\beta_{\mathtt{T}} \in [0..4D]$ at all times, as required.

	\item[] \textsf{Number of Failed Transmissions $F$} ({\bf \ref{setupCode2M}.68}).  This is initialized to zero on ({\bf \ref{setupCode2M}.71}) and is only changed on ({\bf \ref{routingRules5M}.172}) and ({\bf \ref{routingRules5M}.186}).  Notice that $F$ is only incremented by one as on line ({\bf \ref{routingRules5M}.186}) when a transmission fails.  As was shown in Theorem \ref{maxBadness}, there can be at most $n-1$ failed transmissions before a node can necessarily be eliminated, in which case $F$ is reset to zero on ({\bf \ref{routingRules5M}.172}).

	\item[] \textsf{Participating List $\mathcal{P}_{\mathtt{T}}$} ({\bf \ref{setupCode2M}.69}).  This is initialized to $G$ on ({\bf \ref{setupCode2M}.73}) and is only changed on ({\bf \ref{routingRules5M}.173}) and ({\bf \ref{routingRules5M}.187}), where it is clear each time that $\mathcal{P}_{\mathtt{T}} \subseteq G$ in both places.
	
	\item[] \textsf{End of Transmission Parcel $\Theta_{\mathtt{T}}$} ({\bf \ref{setupCode2M}.77}).  This is initialized to $\bot$ on ({\bf \ref{setupCode2M}.79}) and is only changed on ({\bf \ref{routingRules5M}.179}), where it is clear that $\Omega_{\mathtt{T}}$ can be expressed in $O(k+\log n)$ bits as required (packets have size $O(k+\log n)$, and the index of a transmission requires $O(\log n)$ bits).

	\item[] \textsf{Broadcast Buffer $BB$} ({\bf \ref{setupCodeM}.08}).  We treat the sender's broadcast buffer separately below, and consider now only the broadcast buffer of any internal node or the receiver.  Notice that the broadcast buffer is initially empty ({\bf \ref{setupCode2M}.46}).  Looking at all places information is {\it added} to $BB$ (lines ({\bf \ref{routingRules3M}.106-107}), ({\bf \ref{routingRules4M}.125}), ({\bf \ref{routingRules4M}.127}), ({\bf \ref{routingRules4M}.130}), ({\bf \ref{routingRules4M}.136}), ({\bf \ref{routingRules4M}.138}), ({\bf \ref{routingRules4M}.142-145}), ({\bf \ref{routingRules4M}.148-149}), ({\bf \ref{routingRules4M}.154}), and ({\bf \ref{routingRules4M}.155})), we see that there are 7 kinds of parcels stored in the broadcast buffer, as listed on ({\bf \ref{routingRules3M}.115}) (the $7^{th}$ type is to indicate which parcel to send across each edge, as on ({\bf \ref{routingRules3M}.106})).  We look at each one separately, stating the maximum number of bits it requires in any broadcast buffer.  For all of the items below, the comments on ({\bf \ref{routingRules4M}.123}) ensure that there are never duplicates of the same parcel in $BB$ at the same time, and also that every parcel in $BB$ has associated with it $n-1$ bits to indicate which edges the parcel has travelled across (see e.g.\ ({\bf \ref{routingRules3M}.107}), ({\bf \ref{routingRules4M}.125}), ({\bf \ref{routingRules4M}.127}), ({\bf \ref{routingRules4M}.130}), ({\bf \ref{routingRules4M}.136}), ({\bf \ref{routingRules4M}.138}), ({\bf \ref{routingRules4M}.148}), and ({\bf \ref{routingRules4M}.154})).  Totaling all numbers below, we see that the $BB$ needs to hold at most $n^2+5n$ broadcast parcels, with each parcel needing to record which of the $n-1$ edges it has traversed, which proves the domain on ({\bf \ref{setupCodeM}.08}) is correct.
		\begin{enumerate}
		\item \textsc{Receiver's End of Transmission Parcel $\Theta_{\mathtt{T}}$}.  This is added to a node's broadcast buffer on ({\bf \ref{routingRules4M}.125}), and removed on ({\bf \ref{routingRules5M}.203}).  Since every internal node and the receiver will reach ({\bf \ref{routingRules5M}.203}) at the end of every transmission (({\bf \ref{routingRulesM}.30}) and ({\bf \ref{routingRules5M}.203})), and by the inforgibility of the signature scheme, there is only one valid $\Theta_{\mathtt{T}}$ per transmission $\mathtt{T}$.  Therefore, each node will have at most one broadcast parcel of this type in $BB$ at any time.

		\item \textsc{Sender's Start of Transmission Parcels}.  These are added to a node's broadcast buffer on ({\bf \ref{routingRules4M}.127}), ({\bf \ref{routingRules4M}.130}), ({\bf \ref{routingRules4M}.136}), and ({\bf \ref{routingRules4M}.138}), and they are removed on ({\bf \ref{routingRules5M}.203}).  Since every internal node and the receiver will reach ({\bf \ref{routingRules5M}.203}) at the end of every transmission (({\bf \ref{routingRulesM}.30}) and ({\bf \ref{routingRules5M}.203})), by the inforgibility of the signature scheme, for every transmission $\mathtt{T}$, there is only one valid $\Omega_{\mathtt{T}}$ in $BB$ at any time.  Notice that $\Omega_{\mathtt{T}}$ can hold up to 1 parcel for 200a, $n-1$ valid parcels for 200b and 200d together, and up to $n-1$ valid parcels for 200c (see ({\bf \ref{routingRules5M}.200}), and use the fact that $S \notin EN, BL$ and Theorem \ref{maxBadness}).  Therefore, each node will have at most $2n$ broadcast parcels of this type in $BB$ at any time.

		\item \textsc{Label of a Node to Remove from the Blacklist}.  Parcels of this nature are added to a node's broadcast buffer on ({\bf \ref{routingRules4M}.148}) and removed on ({\bf \ref{routingRules5M}.203}).  Since every internal node and the receiver will reach ({\bf \ref{routingRules5M}.203}) at the end of every transmission (({\bf \ref{routingRulesM}.30}) and ({\bf \ref{routingRules5M}.203})), we argue that in any transmission, every node will have at most $n-1$ parcels in their broadcast buffer corresponding to the label of a node to remove from the blacklist.  To see this, we argue that the sender will add $(\widehat{N}, 0 ,\mathtt{T})$ to his broadcast buffer as on ({\bf \ref{routingRules4M}.165}) at most once for each node $\widehat{N} \in \mathcal{P} \setminus S$ per transmission, and then use the inforgibility of the signature scheme to argue each node will add a corresponding broadcast parcel to their broadcast buffer as on ({\bf \ref{routingRules4M}.148}) at most $n-1$ times.  That the sender will enter line ({\bf \ref{routingRules4M}.165}) at most once per node per transmission is clear since once the sender has reached ({\bf \ref{routingRules4M}.165}) for some node $\widehat{N}$, the node will be removed from his blacklist on ({\bf \ref{routingRules4M}.166}), and nodes are not re-added to the blacklist until the end of any transmission, as on ({\bf \ref{routingRules5M}.188}).  Therefore, once the sender has received some node $\widehat{N}$'s complete status report as on ({\bf \ref{routingRules4M}.164}), that same line cannot be entered again by the same node $\widehat{N}$ in the same transmission.  In summary, there are at most $n-1$ broadcast parcels of this type in any node's broadcast buffer at any time.

		\item \textsc{The Label	of a Node $\widehat{N}$ Whose Status Report is Known to $N$.}  We show that for any node $N \in \mathcal{P} \setminus S$, there are at most $(n-1)$ broadcast parcels of type 4 ({\bf \ref{routingRules3M}.115}) in $BB$ at any time\footnote{The $(n-1)$ comes from the fact that there are no status reports for the sender.}.  This follows from the same argument as above, where it was shown that ({\bf \ref{routingRules4M}.164}) can be true at most once per node per transmission.  The inforgibility of the signature scheme ensures that the same will be true for internal nodes regarding line ({\bf \ref{routingRules4M}.155}), and since this is the only line on which broadcast parcels of this kind are added to $BB$, this can happen at most $n-1$ times per transmission.  However, we are not yet done with this case, because broadcast information of this type is {\it not} removed from $BB$ at the end of each transmission like the above forms of broadcast information.  Therefore, we fix $\widehat{N} \in G$, and show that if $N$ adds a broadcast parcel to $BB$ of form $(N, \widehat{N}, \mathtt{T}')$ as on ({\bf \ref{routingRules4M}.155}) of transmission $\mathtt{T}$, then necessarily $BB$ was {\it not} already storing a broadcast parcel of form $(N, \widehat{N}, \mathtt{T}'')$ for some other $\mathtt{T}'' \neq \mathtt{T}'$ (if $\mathtt{T}'' = \mathtt{T}'$, then there is nothing to show, as nothing new will be added to $BB$ by the comments on {\bf \ref{routingRules4M}.123}).
		
		\hspace{.5cm}For the sake of contradiction, suppose that $BB$ is already storing a parcel of form $(N, \widehat{N}, \mathtt{T}'')$ when ({\bf \ref{routingRules4M}.155}) of transmission $\mathtt{T}$ is entered and $N$ is called to add $(N, \widehat{N}, \mathtt{T}')$ to $BB$ for some $\mathtt{T}' \neq \mathtt{T}''$.  Since ({\bf \ref{routingRules4M}.155}) is reached, we must have that ({\bf \ref{routingRules4M}.152}) was satisfied for the $bp$ appearing there.  In particular, $\widehat{N}$ is on $N$'s version of the blacklist.  Since the blacklist is cleared at the end of every transmission ({\bf \ref{routingRules5M}.203}), it must be that $(\widehat{N}, \mathtt{T}', \mathtt{T})$ was added to $N$'s version of the blacklist during the {\it SOT} broadcast for the current transmission $\mathtt{T}$, as on ({\bf \ref{routingRules4M}.137-138}).  Therefore, all parcels in $BB$ of form $(N, \widehat{N}, \mathtt{T}'')$ for $\mathtt{T}'' \neq \mathtt{T}'$ should have been removed from $BB$ on line ({\bf \ref{routingRules4M}.139}), yielding the desired contradiction.

		\item[5-6.] \textsc{Status Report Parcels.}  We fix $N \in G$ and show that for every $\widehat{N} \in \mathcal{P} \setminus \{S, N \}$, there are at most $n$ status report parcels corresponding to $\widehat{N}$ in $N$'s broadcast buffer, and hence $N$'s broadcast buffer will hold at most $n(n-1)$ status report parcels at any time.  Since a single node's status report for a single transmission consists of at most $n$ parcels (see lines ({\bf \ref{routingRules4M}.142-145})\footnote{We assume that the signature buffer information for two directed edges $E(A,B)$ and $E(B,A)$ are combined into one status report parcel.}), it will be enough to show that for every $\widehat{N} \in \mathcal{P} \setminus S$, at all times $N$'s broadcast buffer only holds status report parcels for $\widehat{N}$ corresponding to a {\it single} failed transmission $\mathtt{T}'$.

		\hspace{.5cm}For the sake of contradiction, suppose that during some transmission $\mathtt{T}$, there is some node $\widehat{N} \in \mathcal{P} \setminus S$ and two transmissions $\mathtt{T}'$ and $\mathtt{T}''$ such that $N$'s broadcast buffer holds at least one status report parcel for $\widehat{N}$ from both $\mathtt{T}'$ and $\mathtt{T}''$.  Notice that status report parcels are only added to $BB$ as on ({\bf \ref{routingRules4M}.154}), and without loss of generality suppose that the status report parcel of $\widehat{N}$ corresponding to $\mathtt{T}''$ was already in $BB$ when one corresponding to $\mathtt{T}'$ is added to $BB$ as on ({\bf \ref{routingRules4M}.154}) of transmission $\mathtt{T}$.  As was argued above, since ({\bf \ref{routingRules4M}.154}) is reached in $\mathtt{T}$, ({\bf \ref{routingRules4M}.152}) must have been satisfied, and since $N$'s blacklist is cleared at the end of every transmission ({\bf \ref{routingRules5M}.203}), it must be that a broadcast parcel of form $(\widehat{N}, \widehat{\mathtt{T}}, \mathtt{T})$, adding $\widehat{N}$ to $N$'s version of the blacklist, was received earlier in transmission $\mathtt{T}$.  Notice that necessarily $\widehat{\mathtt{T}} = \mathtt{T}'$, since otherwise line ({\bf \ref{routingRules4M}.153}) will not be satisfied.  But then since $\mathtt{T}'' \neq \mathtt{T}'$, all status report parcels of $\widehat{N}$ corresponding to transmission $\mathtt{T}''$ should have been removed from $BB$ on ({\bf \ref{routingRules4M}.139}), yielding the desired contradiction.

		\hspace{.5cm}Now for a $N$'s {\it own} status report parcels, these are added to $BB$ on ({\bf \ref{routingRules4M}.142-145}).  Investigating lines ({\bf \ref{routingRules4M}.137}), ({\bf \ref{routingRules4M}.139}), and ({\bf \ref{routingRules4M}.140}), we see that status reports of $N$ can occupy $BB$ for at most one failed transmission.

		\item[7.] \textsc{Requested Parcel for Each Edge.}  For any edge $E(A,N)$, $N$ will have at most one copy of a parcel like $\alpha$ as on ({\bf \ref{routingRules3M}.106}) at any time, since the old version of $\alpha$ is simultaneously deleted when the new one is added on ({\bf \ref{routingRules3M}.106}).  Since each node has $(n-1)$ incoming edges, $BB$ need hold at most $n-1$ parcels of this form at any time.
		\end{enumerate}

	\item[] \textsf{Data Buffer $DB$, Eliminated List $EN$, and Blacklist $BL$} ({\bf \ref{setupCodeM}.09-11}).  We treat the sender's broadcast buffer separately below, and consider now only the data buffer of any internal node or the receiver.  The data buffer (which includes the blacklist and list of eliminated nodes) is initially empty ({\bf \ref{setupCode2M}.46}).  A node $N$'s data buffer holds three different kinds of information: blacklist, list of eliminated nodes, and for each neighbor $\widehat{B} \in G$, a list of nodes $\widehat{N} \in G$ for which $\widehat{B}$ knows the complete status report (see item 4 on line ({\bf \ref{routingRules3M}.115})).  Below, we show that these contribute at most $n-1$, $n-1$, and $(n-1)^2$ parcels (respectively), so that $DB$ requires at most $n^2$ parcels at any time.
		\begin{itemize}
		\item[] \textsc{Blacklist $BL$}.  Each entry of $BL$ is initialized to $\bot$ on ({\bf \ref{setupCode2M}.46}), and $BL$ is only modified on lines ({\bf \ref{routingRules4M}.134}), ({\bf \ref{routingRules4M}.138}), ({\bf \ref{routingRules4M}.148}), and ({\bf \ref{routingRules5M}.203}).  $BL$ is an array with $n-1$ entries, indexed by the nodes in $\mathcal{P} \setminus S$.  When a node $(\widehat{N}, \mathtt{T})$ is added to $BL$ as on ({\bf \ref{routingRules4M}.138}), this means that the entry of $BL$ corresponding to $\widehat{N}$ is switched to be $\mathtt{T}$.  When a node $(\widehat{N},\mathtt{T})$ is removed from $BL$ as on ({\bf \ref{routingRules4M}.148}), this means that the entry of $BL$ corresponding to $\widehat{N}$ is switched to $\bot$.  Finally, when $BL$ is to be cleared as on ({\bf \ref{routingRules4M}.134}) and ({\bf \ref{routingRules5M}.203}), this means that $BL$ each entry of $BL$ is set to $\bot$.  Thus, in all cases, $BL \in [1..n-1] \times \{0,l\}^P$ as required.\vspace{.25cm}

		\item[] \textsc{List of Eliminated Nodes $EN$}.  Each entry of $EN$ is initialized to $\bot$ on ({\bf \ref{setupCode2M}.46}), and is only modified on line ({\bf \ref{routingRules4M}.132}).  $EN$ is an array with $n-1$ entries, indexed by the nodes in $\mathcal{P} \setminus S$.  Here, when a node $(\widehat{N},\mathtt{T})$ is added to $EN$, this means the entry of $EN$ corresponding to $\widehat{N}$ is switched to $\mathtt{T}$.  Thus, at all times $EN \in [1..n-1] \times \{0,l\}^P$ as required.\vspace{.25cm}

		\item[] \textsc{Which Neighbor's Know Another Node's Status Report}.  Parcels of this kind are only added to or removed from $DB$ on lines ({\bf \ref{routingRules4M}.134}), ({\bf \ref{routingRules4M}.139}), ({\bf \ref{routingRules4M}.149}), and ({\bf \ref{routingRules4M}.151}).  We will now show that for any pair of nodes $\widehat{N}, \widehat{B} \in \mathcal{P} \setminus S$, the data buffer of any node $N \in G$ will have at most one parcel of the form $(\widehat{B}, \widehat{N}, \mathtt{T}')$, from which we conclude that this portion of $N$'s data buffer need hold at most $(n-1)^2$ parcels.  To see this, we fix $\widehat{B}$ and $\widehat{N}$ in $G$ and suppose for the sake of contradiction that $N$'s data buffer holds two different parcels $(\widehat{B}, \widehat{N}, \mathtt{T}')$ and $(\widehat{B}, \widehat{N}, \mathtt{T}'')$, for $\mathtt{T}' \neq \mathtt{T}''$.  We consider the transmission $\mathtt{T}$ for which this first happens, i.e.\ without loss of generality, $(\widehat{B}, \widehat{N}, \mathtt{T}')$ is added to $DB$ as on ({\bf \ref{routingRules4M}.151}) of $\mathtt{T}$.  Since the second part of ({\bf \ref{routingRules4M}.151}) is reached, the first part of ({\bf \ref{routingRules4M}.151}) must have been satisfied, and since the blacklist is cleared at the end of every transmission ({\bf \ref{routingRules5M}.203}), it must be that a broadcast parcel of form $(\widehat{N}, \widehat{\mathtt{T}}, \mathtt{T})$ adding $\widehat{N}$ to $N$'s version of the blacklist was received earlier in transmission $\mathtt{T}$.  Notice that necessarily $\widehat{\mathtt{T}} = \mathtt{T}'$, since otherwise line ({\bf \ref{routingRules4M}.151}) will not be satisfied.  But then since $\mathtt{T}'' \neq \mathtt{T}'$, $(\widehat{B}, \widehat{N}, \mathtt{T}'')$ should have been removed from $DB$ as on ({\bf \ref{routingRules4M}.139}) of transmission $\mathtt{T}$, yielding the desired contradiction.
		\end{itemize}
	Adding these three contributions, we see that that $DB$ requires at most $n^2$ parcels, as required.

	\item[] \textsf{Outgoing Signature Buffers $SIG$} ({\bf \ref{setupCodeM}.17}).  Each outgoing signature buffer is initially empty ({\bf \ref{setupCode2M}.54}), and they are only modified on ({\bf \ref{routingRules2M}.48-49}), ({\bf \ref{routingRules4M}.128}), ({\bf \ref{routingRules4M}.133}), ({\bf \ref{routingRules4M}.141}), and ({\bf \ref{routingRules4M}.146}).  The first of these increments $SIG[3]$ by at most $2n$, increments $SIG[1]$, and $SIG[p]$ by at most 1, and increments $SIG[2]$ by at most $2n$, and the latter four lines all reset all entries of $SIG$ to $\bot$.  Since our protocol is only intended to run polynomially-long (in $n$), each entry of $SIG$ is polynomial in $n$, as required.

	\item[] \textsf{Incoming Signature Buffers $SIG$} ({\bf \ref{setupCodeM}.32}).  Each incoming signature buffer is initially empty ({\bf \ref{setupCode2M}.48}), and they are only modified on ({\bf \ref{routingRules2M}.74-75}), ({\bf \ref{routingRules4M}.128}), ({\bf \ref{routingRules4M}.133}), ({\bf \ref{routingRules4M}.141}), and ({\bf \ref{routingRules4M}.146}).  The first of these increments $SIG[3]$ by at most $2n$, $SIG[1]$ and $SIG[p]$ by at most 1, and $SIG[2]$ by at most $2n$, and the latter four lines all reset all entries of $SIG$ to $\bot$.  Since our protocol is only intended to run polynomially-long (in $n$), each entry of $SIG$ is polynomial in $n$, as required.

	\item[] \textsf{Copy of Packets Buffer $COPY$} ({\bf \ref{setupCode2M}.60}).  $COPY$ is first filled on ({\bf \ref{setupCode2M}.74})/({\bf \ref{routingRules5M}.214}), with a copy of every packet corresponding to the first codeword.  The only place it is modified after this is on ({\bf \ref{routingRules5M}.214}), where the old copies are first deleted and then replaced with new ones.

	\item[] \textsf{Sender's Broadcast Buffer $BB$}.  In contrast to an internal node's broadcast buffer, the only thing the sender's broadcast buffer holds is the {\it Start of Transmission} broadcast ({\bf \ref{routingRules5M}.200}) and the information that a node should be {\it removed} from the blacklist, see ({\bf \ref{routingRules4M}.165}).  Notice that at the outset of the protocol, $BB$ only holds the Start of Transmission broadcast, which is comprised by only $\Omega_1 = (0,0,0,0)$ ({\bf \ref{setupCode2M}.72-73}).  After this, the only changes made to $BB$ appear on lines ({\bf \ref{routingRules4M}.165}), ({\bf \ref{routingRules5M}.171}), ({\bf \ref{routingRules5M}.199}), and ({\bf \ref{routingRules5M}.200}).  Notice that for every transmission, necessarily either ({\bf \ref{routingRules5M}.171}) or ({\bf \ref{routingRules5M}.199}) will be reached, and hence at any time of any transmission $\mathtt{T}$, $BB$ contains parcels corresponding to at most one {\it Start of Transmission} broadcast, and whatever parcels were added to $BB$ so far in $\mathtt{T}$.  By investigating line ({\bf \ref{routingRules5M}.200}) and using Lemma \ref{maxBadness}, the former requires at most $2n$ parcels, and by the comment on ({\bf \ref{routingRules4M}.156}), the latter requires at most $n$ parcels ({\bf \ref{routingRules4M}.165}).  Therefore, the sender's broadcast buffer requires at most $3n$ parcels, as required.
	
	\item[] \textsf{Sender's Data Buffer $DB$, Eliminated List $EN$, and Blacklist $BL$} ({\bf \ref{setupCode2M}.62-64}).  We will show that the sender's $DB$ needs to hold at most $n^3 +n^2 +n$ parcels at any time, and that the blacklist and list of eliminated nodes need at most $n$ parcels each.  Notice that every entry of $DB$ is initialized to $\bot$ on ({\bf \ref{setupCode2M}.73}), after which modifications to $DB$ occur only on lines ({\bf \ref{routingRules4M}.158}), ({\bf \ref{routingRules4M}.160}), ({\bf \ref{routingRules4M}.162}), ({\bf \ref{routingRules4M}.166}), ({\bf \ref{routingRules5M}.170}), ({\bf \ref{routingRules5M}.171}), ({\bf \ref{routingRules5M}.187}), ({\bf \ref{routingRules5M}.188}), ({\bf \ref{routingRules5M}.191}), ({\bf \ref{routingRules5M}.194}), ({\bf \ref{routingRules5M}.197}), and ({\bf \ref{routingRules5M}.199}).  The sender's data buffer holds eight different kinds of information: end of transmission parcel $\Theta_{\mathtt{T}}$, status report parcels, the participating list for up to $n-1$ failed transmissions, the reason for failure for up to $n-1$ failed transmissions, its own status reports for up to $n-1$ failed transmissions, the blacklist, list of eliminated nodes, and for each neighbor $B \in G$, a list of nodes $\widehat{N} \in G$ for which $\widehat{B}$ knows the complete status report (see item 4 on line ({\bf \ref{routingRules3M}.115})).
		\begin{enumerate}
		\item \textsc{End of Transmission Parcel $\Theta_{\mathtt{T}}$}.  Modifications to this occur only on lines ({\bf \ref{routingRules4M}.158}), ({\bf \ref{routingRules5M}.171}), and ({\bf \ref{routingRules5M}.199}).  Every transmission, the inforgibility of the signature scheme and the comment on line ({\bf \ref{routingRules4M}.156}) guarantee that the sender will add $\Theta_{\mathtt{T}}$ to $DB$ as on ({\bf \ref{routingRules4M}.158}) at most once.  Meanwhile, for every transmission, either ({\bf \ref{routingRules5M}.171}) or ({\bf \ref{routingRules5M}.199}) will be reached exactly once.  Therefore, there is at most one End of Transmission parcel in $DB$ at any time.

		\item \textsc{Blacklist $BL$}.  We show that $BL$ consists of at most $n$ parcels at any time.  More specifically, we will show that $BL$ lives in the domain $[1..n] \times \{0,1\}^{O(\log n)}$, i.e.\ an array with $n$ slots indexed by each $N \in G$, with each slot holding $\bot$ (if the corresponding node is not on the blacklist) or the index of the transmission in which the corresponding node was most recently added to the blacklist.  To see this, notice that modifications to the blacklist occur only on lines ({\bf \ref{routingRules4M}.166}), ({\bf \ref{routingRules5M}.171}), and ({\bf \ref{routingRules5M}.188}).  ``Removing'' a node $\widehat{N}$ from $BL$ as on ({\bf \ref{routingRules4M}.166}) means changing the entry indexed by $\widehat{N}$ to $\bot$.  ``Clearing'' the blacklist as on ({\bf \ref{routingRules5M}.171}) means making every entry of the array equal to $\bot$.  Finally, ``adding'' a node to the blacklist as on ({\bf \ref{routingRules5M}.188}) means switching the entry indexed by $\widehat{N}$ to be the index of the current transmission.

		\item \textsc{Status Report Parcels.} Modifications to this occur only on lines ({\bf \ref{routingRules4M}.162}) and ({\bf \ref{routingRules5M}.171}).  We show in Lemma \ref{numStatusReports} below that for any node $\widehat{N} \in \mathcal{P} \setminus S$, $DB$ will hold at most $n(n-1)$ status report parcels from $\widehat{N}$ at any time, from which we conclude that $DB$ need hold at most $n(n-1)^2$ status report parcels.
		
		\item \textsc{Participating Lists}.  We will view the participating list corresponding to transmission $\mathtt{T}$ as an array $[1..n] \times \{0, 1\}^{P}$, where the array is indexed by the nodes, and an entry corresponding to node $N \in G$ is either the index of the transmission $\mathtt{T}$ (if $N$ participated in $\mathtt{T}$) or $\bot$ otherwise.  Therefore, since each participating list consists of $n$ parcels, we can argue that participating lists require at most $n(n-1)$ parcels if we can show that $DB$ need hold at most $n-1$ participating lists at any time.  To see this, notice that ({\bf \ref{routingRules5M}.187}) is reached only in the case the transmission {\it failed} ({\bf \ref{routingRules5M}.185}), and we showed in Lemma \ref{maxBadness} that there can be at most $n-1$ failed transmissions before a node is necessarily eliminated and $DB$ is cleared as on ({\bf \ref{routingRules5M}.171}).

		\item \textsc{Reason Transmissions Failed}.  Modifications to this occur only on lines ({\bf \ref{routingRules5M}.171}), ({\bf \ref{routingRules5M}.191}), ({\bf \ref{routingRules5M}.194}), and ({\bf \ref{routingRules5M}.197}).  Notice that of the latter three, exactly one will be reached if and only if the transmission failed.  Also, each one of the three will add at most one parcel to $DB$.  Since $DB$ is cleared any time {\bf \em Eliminate Node} is called as on ({\bf \ref{routingRules5M}.171}), we again use Lemma \ref{maxBadness} to conclude that Reason for Transmission Failures require at most $n-1$ parcels of $DB$.

		\item \textsc{Sender's Own Status Reports}.  These are added to $DB$ on lines ({\bf \ref{routingRules5M}.191}), ({\bf \ref{routingRules5M}.194}), and ({\bf \ref{routingRules5M}.197}), and removed from $DB$ on ({\bf \ref{routingRules5M}.171}).  Notice that of the former three lines, exactly one will be reached if and only if the transmission failed.  Also, each one of the three will add at most $n$ parcels to $DB$.  Since $DB$ is cleared any time {\bf \em Eliminate Node} is called as on ({\bf \ref{routingRules5M}.171}), we again use Lemma \ref{maxBadness} to conclude that Reason for Transmission Failures require at most $n(n-1)$ parcels of $DB$.

		\item \textsc{List of Eliminated Nodes $EN$}.  Modifications to this occur only on line ({\bf \ref{routingRules5M}.170}).  Since $EN$ is viewed as living in $[1..n] \times \{0,1\}^{O(\log n)}$, ``adding'' a node $\widehat{N}$ to $EN$ means changing the entry indexed by $\widehat{N}$ from $\bot$ to the index of the current transmission.  Notice that $EN \in [1..n] \times \{0,1\}^{O(\log n)}$ can be expressed using $n$ parcels.

		\item \textsc{The Label of a Node $\widehat{N}$ Whose Status Report is Known to $B$.}  Modifications to this occur only on lines ({\bf \ref{routingRules4M}.160}), ({\bf \ref{routingRules4M}.166}), and ({\bf \ref{routingRules5M}.171}).  We show in Lemma \ref{numStatusReportsKnown} below that for any pair of nodes $B, \widehat{N} \in \mathcal{P} \setminus S$, $DB$ will hold at most one parcel of the form $(B, \widehat{N}, \mathtt{T}')$ at any time (see e.g.\ ({\bf \ref{routingRules4M}.160})), from which we conclude that $DB$ need hold at most $(n-1)^2$ parcels of this type.
		\end{enumerate}
	Adding together these changes, the sender's $DB$ needs to hold at most $n^3+n^2 + n$ parcels, as required.
	\end{itemize}
We have now shown each of the variables of Figures \ref{setupCodeM} and \ref{setupCode2M} have domains as indicated.
\end{proof}
\begin{lemma} \label{numStatusReports} For any node $\widehat{N} \in \mathcal{P} \setminus S$, the sender's data buffer will hold at most $n(n-1)$ status report parcels from $\widehat{N}$ at any time.  More specifically, let $\{\mathtt{T}_1, \dots, \mathtt{T}_j\}$ denote the set of transmissions for which the sender has at least one status report parcel from $\widehat{N}$.  Then $j \leq n-1$ and for every $i < j$, the sender has  $\widehat{N}$'s {\it complete} status report for transmission $\mathtt{T}_i$.
\end{lemma}
\begin{proof} We first note that the first sentence follows immediately from the latter two since each status report consists of at most $n$ status report parcels ({\bf \ref{routingRules4M}.142-145}).  Fix $\widehat{N} \in \mathcal{P} \setminus S$ and let $\{\mathtt{T}_1, \dots, \mathtt{T}_j\}$ be as in the lemma, ordered chronologically.  We first show that $j \leq n-1$.  For the sake of contradiction, suppose $j \geq n$.  We first argue that for all $1 \leq i \leq j$, transmission $\mathtt{T}_i$ necessarily failed.  Fix $1 \leq i \leq j$.  Since $DB$ contains a status report parcel from $\widehat{N}$ for transmission $\mathtt{T}_i$, it must have been added on ({\bf \ref{routingRules4M}.162}) of some transmission $\widehat{\mathtt{T}}$.  Therefore, line ({\bf \ref{routingRules4M}.161}) must have been satisfied, and in particular, $(\widehat{N}, \mathtt{T}_i)$ must have been on $BL$ during $\widehat{\mathtt{T}}$.  Therefore, $(\widehat{N}, \mathtt{T}_i)$ must have been added to $BL$ as on ({\bf \ref{routingRules5M}.188}) of transmission $\mathtt{T}_i$, which in turn implies transmission $\mathtt{T}_i$ failed ({\bf \ref{routingRules5M}.185}).  Therefore, transmission $\mathtt{T}_i$ failed for each $1 \leq i \leq j$.

By Lemma \ref{maxBadness}, there can be at most $n-1$ failed transmissions before a node is eliminated as on ({\bf \ref{routingRules5M}.169-177}).  Since $j \geq n$, considering failed transmissions $\{\mathtt{T}_2, \dots, \mathtt{T}_j\}$, there must have been a transmission $\mathtt{T}_2 \leq \mathtt{T} \leq \mathtt{T}_j$ such that {\bf \em Eliminate N} ({\bf \ref{routingRules5M}.169}) was entered in transmission $\mathtt{T}$.  We first argue that $\mathtt{T} < \mathtt{T}_j$ as follows.  If $\mathtt{T} = \mathtt{T}_j$, then $(\widehat{N}, \mathtt{T}_j)$ would {\it not} be added to $BL$ as on ({\bf \ref{routingRules5M}.188}) (once the protocol enters ({\bf \ref{routingRules5M}.169}), it halts until the end of the transmission ({\bf \ref{routingRules5M}.177}), thus skipping ({\bf \ref{routingRules5M}.188})).  But then ({\bf \ref{routingRules4M}.161}) of any transmission after $\mathtt{T}_j$ cannot be satisfied for any status report parcel corresponding to $\mathtt{T}_j$, and hence none of $\widehat{N}$'s status report parcels corresponding to $\mathtt{T}_j$ could be added to $DB$ after transmission $\mathtt{T}_j$.  Similarly, none of $\widehat{N}$'s status report parcels corresponding to $\mathtt{T}_j$ can be added to $DB$ before or during transmission $\mathtt{T}_j$ by Claim \ref{whenDB} below.  This then contradicts the fact that at some point in time, $DB$ contains one of $\widehat{N}$'s status report parcels corresponding to $\mathtt{T}_j$.

We now have that for some transmission $\mathtt{T}_1 < \mathtt{T} < \mathtt{T}_j$, {\bf \em Eliminate N} is entered during $\mathtt{T}$.  Therefore, all of $\widehat{N}$'s status report parcels for $\mathtt{T}_1$ are removed from $DB$ on ({\bf \ref{routingRules5M}.171}) and $(\widehat{N}, \mathtt{T}_1)$ is removed from $BL$ on ({\bf \ref{routingRules5M}.171}) of transmission $\mathtt{T} \leq \mathtt{T}_j$.  Since $\mathtt{T}_1 < \mathtt{T}$, $(\widehat{N}, \mathtt{T}_1)$ will never be put on $BL$ as on ({\bf \ref{routingRules5M}.190}) for any transmission after $\mathtt{T}$, and consequently, ({\bf \ref{routingRules4M}.161}) will never be satisfied after $\mathtt{T}$ for any of $\widehat{N}$'s status report parcels from $\mathtt{T}_1$.  Therefore, none of $\widehat{N}$'s status report parcels will be put into $DB$ after they are removed on ({\bf \ref{routingRules5M}.171}) of $\mathtt{T}$.  Meanwhile, by the end of transmission $\mathtt{T} < \mathtt{T}_j$, $DB$ cannot have any of $\widehat{N}$'s status report parcels corresponding to $\mathtt{T}_j$ by Claim \ref{whenDB} below.  We have now contradicted the assumption that $DB$ simultaneously holds some of $\widehat{N}$'s status report parcels from $\mathtt{T}_1$ and $\mathtt{T}_j$.  Thus, $j \leq n-1$, as desired.

We now show that for every $i < j$, the sender has $\widehat{N}$'s {\it complete} status report for transmission $\mathtt{T}_i$.  If $j=1$, there is nothing to prove.  So let $1 < j \leq n-1$, and for the sake of contradiction suppose there is some $i < j$ such that the sender has at least one of $\widehat{N}$'s status report parcels for $\mathtt{T}_i$, but not the entire report.  Let $\widehat{\mathtt{T}}_i$ denote the transmission that the status report parcel corresponding to $\mathtt{T}_i$ was added to $DB$ as on ({\bf \ref{routingRules4M}.162}), and let $\widehat{\mathtt{T}}_{i+1}$ denote the transmission that the parcel corresponding to $\mathtt{T}_{i+1}$ was added to $DB$ as on ({\bf \ref{routingRules4M}.162}).  Without loss of generality, we suppose that $\widehat{\mathtt{T}}_i \leq \widehat{\mathtt{T}}_{i+1}$.  Since ({\bf \ref{routingRules4M}.162}) is entered during transmission $\widehat{\mathtt{T}}_i$, it must be that ({\bf \ref{routingRules4M}.161}) was satisfied, and in particular $(\widehat{N}, \mathtt{T}_i)$ was on $BL$ during $\widehat{\mathtt{T}}_i$.  Similarly, $(\widehat{N}, \mathtt{T}_{i+1})$ was on $BL$ during $\widehat{\mathtt{T}}_{i+1}$.  Lemma \ref{blacklistStuff} below states that for each $N \in G$, $N$ is on $BL$ at most once, i.e.\ there is at most one entry of the form $(N, \widehat{\mathtt{T}})$ on $BL$ at any time.  Since nodes are only added to $BL$ at the very end each transmission ({\bf \ref{routingRules5M}.188}), we may conclude that we have strict inequality: $\widehat{\mathtt{T}}_i < \widehat{\mathtt{T}}_{i+1}$.  In particular, $(\widehat{N}, \mathtt{T}_{i+1})$ was {\it not} on $BL$ at the start of $\widehat{\mathtt{T}}_i$, but $(\widehat{N}, \mathtt{T}_i)$ was.  Therefore, $(\widehat{N}, \mathtt{T}_{i+1})$ was added to $BL$ as on ({\bf \ref{routingRules5M}.188}) of transmission $\mathtt{T}_{i+1}$, and so $\widehat{N} \in \mathcal{P}_{\mathtt{T}}$ ({\bf \ref{routingRules5M}.187-188}).  In particular, $\widehat{N}$ is {\it not} blacklisted by the end of $\mathtt{T}_{i+1}$ ({\bf \ref{routingRules5M}.187}).  Therefore, there must be some transmission $\mathtt{T} \in [\widehat{\mathtt{T}}_i, \mathtt{T}_{i+1}]$ such that $(\widehat{N}, \mathtt{T}_i)$ is removed from $BL$ as on ({\bf \ref{routingRules4M}.166}) or ({\bf \ref{routingRules5M}.171}).  Both of these lead to a contradiction, as the first implies the sender has $\widehat{N}$'s complete status report for $\mathtt{T}_i$, while the latter implies that all status reports corresponding to $\mathtt{T}_i$ should have been removed from $DB$.
\end{proof}
\begin{claim} \label{whenDB} For any $\widehat{N} \in G$ and for any transmission $\mathtt{T}$, the sender's data buffer $DB$ will never hold any of $\widehat{N}$'s status report parcels corresponding to $\mathtt{T}$ {\it before} or {\it during} transmission $\mathtt{T}$.
\end{claim}
\begin{proof}  Let $\widehat{N} \in G$, and for the sake of contradiction, let $\mathtt{T}$ be a transmission such that $DB$ has one of $\widehat{N}$'s status report parcels from $\mathtt{T}$ before or during $\mathtt{T}$.  Since status reports are only added to $DB$ on ({\bf \ref{routingRules4M}.162}), this implies that there is some transmission $\mathtt{T}' \leq \mathtt{T}$ such that ({\bf \ref{routingRules4M}.161}) is satisfied at some point of $\mathtt{T}'$ {\it before} the {\bf \em Prepare Start of Transmission Broadcast} of transmission $\mathtt{T}'$ is called.  This in turn implies that $(\widehat{N}, \mathtt{T})$ was on $BL$ {\it before} the {\bf \em Prepare Start of Transmission Broadcast} of transmission $\mathtt{T}' \leq \mathtt{T}$ was called.  However, this contradicts the fact that the only time $(\widehat{N}, \mathtt{T})$ can be added to $BL$ is during the {\bf \em Prepare Start of Transmission Broadcast} of transmission $\mathtt{T}$ on ({\bf \ref{routingRules5M}.188}).
\end{proof}
\begin{claim} \label{numStatusReportsKnown} For any pair of nodes $B, \widehat{N} \in G \setminus S$, the sender's data buffer will hold at most one status report parcels of the form $(B, \widehat{N}, \mathtt{T}')$ at any time.
\end{claim}
\begin{proof} Fix $B, \widehat{N} \in G \setminus S$, and suppose for the sake of contradiction that there are two transmissions $\mathtt{T}'$ and $\mathtt{T}''$ such that both $(B, \widehat{N}, \mathtt{T}')$ and $(B, \widehat{N}, \mathtt{T}'')$ are in $DB$ at the same time (note that $\mathtt{T}' \neq \mathtt{T}''$ by the comment on {\bf \ref{routingRules4M}.156}).  Since parcels of this form are only added to $DB$ on ({\bf \ref{routingRules4M}.160}), we suppose without loss of generality that $\mathtt{T}$ is a transmission and $\mathtt{t}$ is a round in $\mathtt{T}$ such that $(B, \widehat{N}, \mathtt{T}'')$ is already in $DB$ when $(B, \widehat{N}, \mathtt{T}')$ is added to $DB$ as on ({\bf \ref{routingRules4M}.160}) of round $\mathtt{t}$.  Since ({\bf \ref{routingRules4M}.160}) is reached, ({\bf \ref{routingRules4M}.159}) was satisfied, so in particular $(\widehat{N}, \mathtt{T}')$ is on the sender's (current version of the) blacklist.  Similarly, since $(B, \widehat{N}, \mathtt{T}'')$ was (most recently) added to $DB$ as on ({\bf \ref{routingRules4M}.160}) of some round $\widehat{\mathtt{t}}$ of some transmission $\widehat{\mathtt{T}} \leq \mathtt{T}$, it must have been that $(\widehat{N}, \mathtt{T}'')$ was on the sender's (version of the) blacklist during round $\widehat{\mathtt{t}}$ of $\widehat{\mathtt{T}}$.  By Lemma \ref{blacklistStuff} below, since $(\widehat{N}, \mathtt{T}')$ is on the sender's current blacklist (as of round $\mathtt{t}$ of transmission $\mathtt{T}$), and $(\widehat{N}, \mathtt{T}'')$ was on an earlier version of the sender's blacklist, it must be that $(\widehat{N}, \mathtt{T}'')$ was removed from the blacklist at some point between round $\widehat{\mathtt{t}}$ of $\widehat{\mathtt{T}}$ and round $\mathtt{t}$ of $\mathtt{T}$.  Notice that nodes are removed from the blacklist only on ({\bf \ref{routingRules4M}.166}) and ({\bf \ref{routingRules5M}.171}).  However, in both of these cases, $(B, \widehat{N}, \mathtt{T}'')$ should have been removed from $DB$ (see ({\bf \ref{routingRules4M}.166}) and ({\bf \ref{routingRules5M}.171})), contradicting the fact that it is still in $DB$ when $(B, \widehat{N}, \mathtt{T}')$ is added to $DB$ in round $\mathtt{t}$ of transmission $\mathtt{T}$.
\end{proof}
\begin{lemma} \label{blacklistStuff} A node is on at most one blacklist at a time.  In other words, whenever a node $(N, \mathtt{T})$ is added to the sender's blacklist as on ({\bf \ref{routingRules5M}.188}), we have that $(N, \mathtt{T}') \notin BL$ for any other (earlier) transmission $\mathtt{T}'$.  Additionally, if $(N, \mathtt{T}') \in BL$ at any time, then:
	\begin{enumerate}
	\item Transmission $\mathtt{T}'$ failed
	\item No node has been eliminated since $\mathtt{T}'$ to the current time
	\item The sender has not received $N$'s complete status report corresponding to $\mathtt{T}'$
	\end{enumerate}
\end{lemma}
\begin{proof} The first statement of the lemma is immediate, since the only place a (node, transmission) pair is added to $BL$ is on ({\bf \ref{routingRules5M}.188}), and by the previous line, necessarily any such node is not already on the blacklist.  Also, Statement 1) is immediate since ({\bf \ref{routingRules5M}.188}) is only reached if the transmission fails ({\bf \ref{routingRules4M}.185}).  To prove Statements 2) and 3), notice that $(N, \mathtt{T}')$ is only added to the blacklist at the very end of transmission $\mathtt{T}'$ ({\bf \ref{routingRules5M}.188}).  In particular, if $(N, \mathtt{T}')$ is ever removed from the blacklist during some transmission $\mathtt{T}$\footnote{If $(N, \mathtt{T}')$ is removed from the blacklist as on ({\bf \ref{routingRules4M}.166}) or ({\bf \ref{routingRules5M}.171}) of transmission $\mathtt{T}$, then necessarily $\mathtt{T} > \mathtt{T}'$, since $(N, \mathtt{T}')$ can only be added to $BL$ at the very end of a transmission ({\bf \ref{routingRules5M}.188}), i.e.\ lines ({\bf \ref{routingRules4M}.166}) and ({\bf \ref{routingRules4M}.171}) cannot be reached after line ({\bf \ref{routingRules5M}.188}) in the same transmission.} as on ({\bf \ref{routingRules4M}.166}) or ({\bf \ref{routingRules5M}.171}), then $(N, \mathtt{T}')$ can never again appear on the blacklist (as remarked in the footnote, $\mathtt{T} > \mathtt{T}'$, and so at any point during or after transmission $\mathtt{T}$, $(N, \mathtt{T}')$ can never again be added to $BL$ as on ({\bf \ref{routingRules5M}.188}) since $\mathtt{T}'$ has already passed).  Therefore, if during transmission $\mathtt{T}$ a node is eliminated as on ({\bf \ref{routingRules5M}.169-177}) or the sender receives $N$'s complete status report of transmission $\mathtt{T}'$ as on ({\bf \ref{routingRules4M}.164}), then $N$ will be removed from the blacklist as on ({\bf \ref{routingRules4M}.166}) or ({\bf \ref{routingRules5M}.171}), at which point $(N, \mathtt{T}')$ can never be added to $BL$ again (since necessarily $\mathtt{T} > \mathtt{T}'$ as remarked in the footnote).  This proves Statements 2) and 3).
\end{proof}
\begin{lemma} \label{blacklistStuff2} For any $(N, \mathtt{T})$ on the sender's blacklist, the sender needs at most $n$ parcels from $N$ in order to have $N$'s complete status report, and subsequently remove $N$ from the blacklist ({\bf \ref{routingRules4M}.164-165}).
\end{lemma}
\begin{proof} This was proven when discussing the appropriateness of variable domains in Lemma \ref{mProperDomains}.
\end{proof}

We set the following notation for the remainder of the section.  $\mathtt{T}$ will denote a transmission, $G_{\mathtt{T}}$ will denote the set of non-eliminated nodes at the start of $\mathtt{T}$, $\mathcal{P}_{\mathtt{T}}$ will denote the participating list for $\mathtt{T}$, and $\mathcal{H}_{\mathtt{T}}$ will denote the uncorrupted nodes in the network.  If the transmission is clear or unimportant, we suppress the subscripts for convenience, writing instead $G, \mathcal{P}$, and $\mathcal{H}$.
\begin{lemma}\label{sigBuffsInit} For any honest node $A \in G$ and any transmission $\mathtt{T}$, $A$ must receive the complete {\it Start of Transmission} (SOT) broadcast before it transfers or re-shuffles any packets.  Additionally, the signature buffers $SIG_{A,A}$ and $SIG^A$ of any honest node $A \in G$ are always cleared upon receipt of the complete {\it SOT} broadcast (and hence before any packets are transferred to/from/within $A$).
\end{lemma}
\begin{proof} Fix an honest node $A \in G$ and a transmission $\mathtt{T}$.  If $A$ has not received the full {\it Start of Transmission} (SOT) broadcast for $\mathtt{T}$ yet, then $A$ will not transfer any packets ({\bf \ref{routingRules2M}.59}), ({\bf \ref{routingRulesM}.31-33}), ({\bf \ref{routingRules2M}.63}) and ({\bf \ref{routingRulesM}.35-37}).  This means that ({\bf \ref{routingRules2M}.63}) will always be satisfied, and hence ({\bf \ref{routingRules2M}.78}) can never be reached, and so $RR$ will remain equal to $-1$ (({\bf \ref{setupCode2M}.50}), and ({\bf \ref{routingRules5M}.209})) so long as no codeword packets have been transferred.  This in turn implies ({\bf \ref{routingRules2M}.46}) cannot be satisfied before any codeword packets have been transferred.  Putting these facts together, the signature buffers cannot change as on ({\bf \ref{routingRules2M}.48-50}), ({\bf \ref{routingRules2M}.74-75}), ({\bf \ref{routingRules2M}.80}), or ({\bf \ref{routingRules2M}.82}) before $A$ receives the complete {\it SOT} broadcast.  Also, no packets will be re-shuffled during the call to Re-Shuffle if no packets have moved during the Routing Phase, as the condition statement on ({\bf \ref{reShuffleRules}.74}) was eventually false in the last round of the previous transmission, and the state of the buffers will not have changed if no packets have been transferred in the current transmission.  Therefore, before $A$ has received the complete {\it SOT} broadcast, no packet movement to/from/within $A$ is possible, and changes to the signature buffers are confined to the ones appearing on lines ({\bf \ref{routingRules4M}.128}), ({\bf \ref{routingRules4M}.133}), ({\bf \ref{routingRules4M}.141}), and ({\bf \ref{routingRules4M}.146}), all of which clear the signature buffers.

Suppose now that $A$ has received the full {\it SOT} broadcast for $\mathtt{T}$.  Recall that part of the {\it SOT} broadcast contains $\Omega_{\mathtt{T}} = (|EN|, |BL|, F, *)$, where $EN$ refers to the eliminated nodes, $BL$ is the sender's current blacklist, $F$ is the number of failed transmissions since the last node was eliminated, and the last coordinate denotes the reason for failure of the previous transmission (in the case it failed) ({\bf \ref{routingRules5M}.200}).  If $|BL| =0$, then $A$ will clear all its entries of $SIG^A$ and $SIG_{A,A}$ on ({\bf \ref{routingRules4M}.128}).  Otherwise, $|BL| > 0$, and $N$ will clear all its entries of $SIG^A$ and $SIG_{A,A}$ when it learns the last blacklisted node on ({\bf \ref{routingRules4M}.146}).  Therefore, in all cases $A$'s signature buffers are cleared by the time it receives the full {\it SOT} broadcast, and in particular before it transfers any packets in transmission $\mathtt{T}$.
\end{proof}
In order to prove a variant of Lemma \ref{potentialDrop} in terms of the variables used in the (node-controlling $+$ edge-scheduling) adversary protocol, we will need to first re-state and prove variants of Lemmas \ref{item3}, \ref{packetTransferPotDrop}, and \ref{chainL}.  We begin with a variant of Lemma \ref{item3} (the first 5 Statements correspond directly with Lemma \ref{item3}, the others do not, but will be needed later):
\begin{lemma} \label{item3M} For any {\it honest} node $A \in G$ and at all times of any transmission:
	\begin{enumerate}
	\item For {\it incoming} edge $E(S,A)$, all changes to $SIG^A[3]_{S,A}$ are strictly non-negative.  In particular, at all times:
		\begin{equation} \label{eqM1}
		0 \leq SIG^A[3]_{S,A}
		\end{equation}
	\item For {\it outgoing} edge $E(A,R)$, all changes to $SIG^A[3]_{A,R}$ are strictly non-negative\footnote{$SIG^A[3]$ along outgoing edges measures the {\it decrease} in potential as a {\it positive} quantity.  Thus, a positive value for $SIG^A[3]$ along an outgoing edge corresponds to a decrease in non-duplicated potential.}.  In particular, at all times:
		\begin{equation} \label{eqM2}
		0 \leq SIG^A[3]_{A,R}
		\end{equation}
	\item For {\it outgoing} edges $E(A, B)$, $B \neq R$, all changes to the quantity $(SIG^A[3]_{A,B} - SIG^A[2]_{A,B})$ are strictly non-negative.  This remains true even if $B$ is corrupt.  In particular, at all times:
		\begin{equation} \label{eqM3}
		0 \thickspace \leq \hspace{-.4cm} \sum_{B \in \mathcal{P} \setminus \{A,S\}} \hspace{-.2cm} (SIG^A[3]_{A,B} - SIG^A[2]_{A,B})
		\end{equation}
	\item For {\it incoming} edges $E(B, A)$, $B \neq S$, all changes to the quantity $(SIG^A[2]_{B,A} - SIG^A[3]_{B,A})$ are strictly non-negative.  This remains true even if $B$ is corrupt.  In particular, at all times:
		\begin{equation} \label{eqM4}
		0 \thickspace \leq \hspace{-.4cm} \sum_{B \in \mathcal{P} \setminus \{A,S\}} \hspace{-.2cm} (SIG^A[2]_{B,A} - SIG^A[3]_{B,A})
		\end{equation}
	\item All changes to $SIG_{A,A}$ are strictly non-negative.  In particular, at all times:
		\begin{equation} \label{eqM5}
		0 \leq SIG_{A,A}
		\end{equation}
	\item The net {\it decrease} in potential at $A$ (due to transferring packets {\it out} of $A$ and re-shuffling packets within $A$'s buffers) in any transmission is bounded by $A$'s potential at the start of the transmission, plus $A$'s increase in potential caused by packets transferred {\it into} $A$.  In particular:
		\begin{equation} \label{eqM6}
		SIG_{A,A} \thickspace + \sum_{B \in \mathcal{P} \setminus A} SIG^A[3]_{A,B} \quad \leq \quad (4n^3-6n^2) + \sum_{B \in \mathcal{P} \setminus A} SIG^A[3]_{B,A}
		\end{equation}
	\item The number of packets transferred {\bf \em out} of $A$ in any transmission must be at least as much as the number of packets transferred {\bf \em into} $A$ during the transmission minus the capacity of $A$'s buffers.  In particular:
		\begin{equation} \label{eqM7}
		4n^2-8n \geq \sum_{B \in \mathcal{P} \setminus A} (SIG^A[1]_{B,A}-SIG^A[1]_{A,B})
		\end{equation}
	\item The number of times a packet $p$ corresponding to the current codeword has been transferred {\bf \em out} of $A$ during any transmission is bounded by the number of times that packet has been transferred {\bf \em into} $A$.  In particular\footnote{Notice that \eqref{eqM8} is the only statement of the Lemma that involves quantities in the {\it neighbors'} signature buffers (in addition to $A$'s buffers).  Since there is no assumption made about the honesty of the neighbor's of $A$, this may seem problematic.  However, we show in the proof that regardless of the honesty of $A$'s neighbors $B \in G$, \eqref{eqM8} will be satisfied if $A$ in honest.}:
		\begin{equation} \label{eqM8}
		0 \thickspace \geq \sum_{B \in \mathcal{P}} (SIG^B[p]_{A,B} - SIG^A[p]_{B,A})
		\end{equation}
	\end{enumerate}
\end{lemma}
\begin{proof} We prove each inequality separately, using an inductive type argument on a node $A$'s signature buffers.  First, note that all signature buffers are cleared at the outset of the protocol ({\bf \ref{setupCode2M}.46}), ({\bf \ref{setupCode2M}.48}), and ({\bf \ref{setupCode2M}.54}).  Also, anytime the signature buffers are cleared as on ({\bf \ref{routingRules4M}.128}), ({\bf \ref{routingRules4M}.133}), ({\bf \ref{routingRules4M}.141}), and ({\bf \ref{routingRules4M}.146}), then all of the statements (except possible Statement 8, which depends on values from potentially corrupt nodes $B \in G$) will be true.  So it remains to check the other places signature buffers can change values (({\bf \ref{routingRules2M}.48-50}), ({\bf \ref{routingRules2M}.74-75}), ({\bf \ref{routingRules2M}.80}), ({\bf \ref{routingRules2M}.82}), and ({\bf \ref{reShuffleRules}.76})), and argue inductively that all such changes will preserve the inequalities of Statements 1-7 (Statement 8 will be proven separately).  Since all of these lines represent packet movement, they can only be reached if $A$ has received the complete {\it SOT} broadcast for the current transmission (Lemma \ref{sigBuffsInit}), and so we may (and do) assume this is the case in each item below.  In particular, Lemma \ref{sigBuffsInit} states that because we are assuming $A$ has received the complete {\it SOT} broadcast for transmission $\mathtt{T}$, all of $A$'s signature buffers will be cleared before any changes are made to them.
	\begin{enumerate}
	\item Aside from being cleared, in which case \eqref{eqM1} is trivially true, the only changes made to $SIG^A[3]_{S,A}$ occur on ({\bf \ref{routingRules2M}.75}), where it is clear that all changes are non-negative since $H_{GP}$ is non-negative (Statement 9 of Lemma \ref{pseudo} together with Lemma \ref{similar}).

	\item Aside from being cleared, in which case \eqref{eqM2} is trivially true, the only changes made to $SIG^A[3]_{A,R}$ occur on ({\bf \ref{routingRules2M}.49}), where it is clear that all changes are non-negative since $H_{FP}$ is non-negative (Statement 9 of Lemma \ref{pseudo} together with Lemma \ref{similar}).

	\item Fix $B \in \mathcal{P} \setminus {S,A}$.  Intuitively, this inequality means that considering directed edge $E(A,B)$, the {\it decrease} in $A$'s potential caused by packet transfers must be greater than or equal to $B$'s {\it increase}, which is a consequence of Lemma \ref{item3}.  Formally, we will track all changes to the relevant values in the pseudo-code and argue that at all times and for any fixed $B \in G$ (honest or corrupt), if $A$ is honest, then $0 \leq SIG^A[3]_{A,B} - SIG^A[2]_{A,B}$.  All changes to these values (aside from being cleared) occur only on ({\bf \ref{routingRules2M}.48-49}) since here we are considering $A$'s values along {\it outgoing} edge $E(A,B)$.  Notice that $H_{FP}$ cannot change between ({\bf \ref{routingRulesM}.08}) of some round and ({\bf \ref{routingRules2M}.49}) of the same round.  Since lines ({\bf \ref{routingRules2M}.48-49}) are only reached if {\it Verify Signature Two} accepts the signature (otherwise $RR$ is set to $\bot$ on ({\bf \ref{routingRules2M}.90}) and hence ({\bf \ref{routingRules2M}.45}) will fail), we have that $SIG^A[2]_{A,B}$ changes by at most the value that $H_{FP}$ had on ({\bf \ref{routingRules2M}.89}) (see comments on line ({\bf \ref{routingRules2M}.88-90})), and this is the value sent/received on lines ({\bf \ref{routingRulesM}.07}) and ({\bf \ref{routingRulesM}.11}) and eventually stored on ({\bf \ref{routingRules2M}.48}).  Meanwhile, when $SIG^A[3]_{A,B}$ changes, for honest nodes it will always be an increase of $H_{FP}$ ({\bf \ref{routingRules2M}.49}), and as noted above, this value of $H_{FP}$ is the same as it had on ({\bf \ref{routingRules2M}.89}).  Therefore, for honest nodes, whenever the relevant values change on ({\bf \ref{routingRules2M}.48-49}), the change will respect the inequality $SIG^A[3]_{A,B} - SIG^A[2]_{A,B} \geq H_{FP} - H_{FP} = 0$.
	
	\item Fix $B \in \mathcal{P} \setminus {S,A}$.  Intuitively, this inequality means that considering directed edge $E(B,A)$, the {\it decrease} in $B$'s potential caused by packet transfers must be greater than or equal to $A$'s {\it increase}, which is a consequence of Lemma \ref{item3}.  Formally, we will track all changes to the relevant values in the pseudo-code and argue that at all times and for any fixed $B \in G$ (honest or corrupt), if $A$ is honest, then $0 \leq SIG^A[2]_{B,A} - SIG^A[3]_{B,A}$.  All changes to these values (aside from being cleared) occur only on ({\bf \ref{routingRules2M}.74-75}) since here we are considering $A$'s values along {\it incoming} edge $E(B,A)$.  When $SIG^A[2]_{B,A}$ changes on ({\bf \ref{routingRules2M}.74}), they take on the values sent by $B$ on ({\bf \ref{routingRules2M}.60}) and received by $A$ on ({\bf \ref{routingRules2M}.62}).  However, in order to reach ({\bf \ref{routingRules2M}.74}), the call to {\bf \em Verify Signature One} on ({\bf \ref{routingRules2M}.69}) must have returned true.  In particular, the comments on ({\bf \ref{routingRules2M}.84-86}) require that $A$ verify that the change in $SIG^A[2]_{B,A}$ that $B$ sent to $A$ is at least $H_{GP}$ {\it bigger} than the previous value $A$ had from $B$.  Meanwhile, when $SIG^A[3]_{B,A}$ changes, for honest nodes it will always be an increase of $H_{GP}$ ({\bf \ref{routingRules2M}.75}).  Therefore, since $H_{GP}$ cannot change between ({\bf \ref{routingRules2M}.84}) of some round and ({\bf \ref{routingRules2M}.75}) later in the same round for honest nodes, whenever the relevant values change on ({\bf \ref{routingRules2M}.74-75}), the change will respect the inequality $SIG^A[2]_{B,A} - SIG^A[3]_{B,A} \geq H_{GP} - H_{GP} = 0$.

	\item Intuitively, this inequality says that all changes in potential due to packet re-shuffling should be strictly non-positive ($SIG_{A,A}$ measures potential {\it drop} as a {\it positive} quantity), which is a consequence of Lemma \ref{item3}.  Formally, all changes made to $SIG_{A,A}$ (aside from being cleared) occur on ({\bf \ref{reShuffleRules}.76}), where the change is $M+m-1$.  The fact that this quantity is strictly non-negative for honest nodes follows from Claim \ref{packetHeight}.

	\item Since the inequality concerns $SIG_{A,A}$ and $SIG[3]$ (along both incoming and outgoing edges), we will focus on changes to these values when a packet is transferred (or re-shuffled).  More specifically, we will look at a specific packet $p$ and consider $p$'s affect on $A$'s potential during each of $p$'s stays in $A$, where a ``stay'' refers to the time $A$ receives (an instance of) $p$ as on ({\bf \ref{routingRules2M}.77}) to the time it sends and gets confirmation of receipt (as in Definition \ref{confRec}) for (that instance of) $p$\footnote{A given packet $p$ may have multiple stays in $A$ during a single transmission, one for each time $A$ sees $p$.}.  We fix $p$ and distinguish between the four possible ways $p$ can ``stay'' in $A$:
		\begin{enumerate}
		\item \textsf{The stay is initiated by $A$ receiving $p$ during $\mathtt{T}$ and then sending $p$ at some later round of $\mathtt{T}$, and getting confirmation of $p$'s receipt as in Definition \ref{confRec}}.  More specifically, the stay includes an increase to some incoming signature buffer $SIG^A[3]$ as on ({\bf \ref{routingRules2M}.75}) and then an increase to some outgoing signature buffer $SIG^A[3]$ as on ({\bf \ref{routingRules2M}.49}).  Let $B$ denote the edge along which $A$ received $p$ in this stay, and $B'$ denote the edge along which $A$ sent $p$.  Then $SIG^A[3]_{B,A}$ will increase by $H_{GP}$ on ({\bf \ref{routingRules2M}.75}) when $p$ is accepted.  Let $M$ denote the value of $H_{GP}$ when $p$ is received.  The packet $p$ is eventually re-shuffled to the outgoing buffer along $E(A, B')$.  Let $m$ denote the value of $H_{FP}$ when ({\bf \ref{routingRules2M}.49}) is reached, so that the change to $SIG^A[3]_{A,B'}$ due to sending $p$ is $m$.  By Statement 3 of Claim \ref{obsBelowHere} (which remains valid by Lemma \ref{similar}), any packet that is eventually deleted as on ({\bf \ref{routingRules2M}.50-51}) will be the flagged packet, and so the packet that is deleted did actually have height $m$ in $A$'s outgoing buffer.  In particular, the packet began its stay in an incoming buffer at height $M$, and was eventually deleted when it had height $m$ in some outgoing buffer.  In particular, since $SIG_{A,A}$ accurately tracks changes in potential due to re-shuffling (Statement 1 of Lemma \ref{sigBuffersAreCorrect}), we have that during this stay of $p$, $SIG_{A,A}$ changed by $M-m$.  Therefore, considering only $p$'s affect on the following terms, we have that:
			\begin{equation} \label{firstE}
			SIG_{A,A} + SIG^A[3]_{A,B'} - SIG^A[3]_{B,A} = (M-m) + m - M = 0
			\end{equation}

		\item \textsf{The stay begins at the outset of the protocol, i.e.\ $p$ started the transmission in one of $A$'s buffers, and the stay ends when $p$ is deleted (after having been sent across an edge) in some round of $\mathtt{T}$}.  More specifically, there is no incoming signature buffer $SIG^A[3]$ that changes value as on ({\bf \ref{routingRules2M}.75}) due to this stay of $p$, but there is an increase to some outgoing signature buffer $SIG^A[3]$ as on ({\bf \ref{routingRules2M}.49}).  Using the notation from (a) above with the exception that $M$ denotes the initial height of $p$ in one of $A$'s buffers at the start of $\mathtt{T}$, then considering only $p$'s affect on the following terms, we have that:
			\begin{equation} \label{secE}
			SIG_{A,A} + SIG^A[3]_{A,B'} = (M-m) + m = M
			\end{equation}
				
		\item \textsf{The stay is initiated by $A$ receiving $p$ during $\mathtt{T}$, but $p$ then remains in $A$ through the end of the transmission (either as a normal or a flagged packet)}.  More specifically, the stay includes an increase to some incoming signature buffer $SIG^A[3]$ as on ({\bf \ref{routingRules2M}.75}), but there is no outgoing signature buffer $SIG^A[3]$ that changes value as on ({\bf \ref{routingRules2M}.49}) due to this stay of $p$.  Using the notation from (a) above with the exception that $m$ denotes the final height of $p$ in one of $A$'s buffers at the end of $\mathtt{T}$, then considering only $p$'s affect on the following terms, we have that:
			\begin{equation}
			SIG_{A,A} - SIG^A[3]_{B,A} = (M-m) - M = -m \leq 0
			\end{equation}

		\item \textsf{The stay begins at the outset of the protocol, i.e.\ $p$ started the transmission in one of $A$'s buffers, and $p$ remains in $A$'s buffers through the end of the transmission (either as a normal or a flagged packet)}.  More specifically, there is no incoming signature buffer $SIG^A[3]$ that changes value as on ({\bf \ref{routingRules2M}.75}) due to this stay of $p$, and there is no outgoing signature buffer $SIG^A[3]$ that changes value as on ({\bf \ref{routingRules2M}.49}) due to this stay of $p$.  Letting $M$ denote the initial height of $p$ in one of $A$'s buffers at the start of $\mathtt{T}$ and $m$ the final height of $p$ in one of $A$'s buffers at the end of $\mathtt{T}$, then considering only $p$'s affect on the following terms, we have that:
			\begin{equation} \label{lastE}
			SIG_{A,A} = M-m \leq M
			\end{equation}
		\end{enumerate}
	We note that the above four cases cover all possibilities by Claim \ref{packetProliferation3} (which remains valid since $A$ is honest, and Lemma \ref{similar}).  We will now bound $SIG_{A,A} + \sum_{B \in \mathcal{P} \setminus A} SIG^A[3]_{A,B} - SIG^A[3]_{B,A}$ by adding all contributions to $SIG_{A,A}$ and $SIG^A[3]_{A,B'}$ and $SIG^A[3]_{B,A}$ from all stays of all packets and for all adjacent nodes $B,B'$.  Notice that ignoring contributions as in Case (c) will only help our desired equality, and contributions as in Case (a) are zero, so we consider only packet stays as in \eqref{secE} and \eqref{lastE}.  Since these contributions to potential correspond to the initial height the packet had in one of $A$'s buffers at the outset of $\mathtt{T}$, the sum over all such contributions cannot exceed $A$'s potential at the outset of $\mathtt{T}$, which for an honest node $A$ is bounded by $2(n-2)2n(2n+1)/2 < 4n^3-6n^2$ (see e.g.\ proof of Claim \ref{capacity1}).

	\item Intuitively, this inequality means that because a node can hold at most $2(n-2)(2n)$ packets at any time, the difference between the number of packets received and the number of packets sent by an honest node will be bounded by $4n^2-8n$.  Formally, during a transmission $\mathtt{T}$, the only places the quantities $SIG[1]$ change are on ({\bf \ref{routingRules2M}.74}) and ({\bf \ref{routingRules2M}.48}).  As with the proof of Statement 6 above, we consider the contribution of each packet $p$'s stay in $A$\footnote{Note that necessarily $p$ is a packet corresponding to the current codeword, since packets corresponding to old codewords do not increment $SIG[1]$, see comments on ({\bf \ref{routingRulesM}.59-60}) and ({\bf \ref{routingRulesM}.11}).  Therefore, there are only two cases to consider.}:
		\begin{enumerate}
		\item \textsf{The stay is initiated by $A$ receiving $p$ during $\mathtt{T}$ and then sending $p$ at some later round of $\mathtt{T}$, and getting confirmation of $p$'s receipt as in Definition \ref{confRec}}.  More specifically, the stay includes an increase to some incoming signature buffer $SIG^A[1]$ as on ({\bf \ref{routingRules2M}.74}) and then an increase to some outgoing signature buffer $SIG^A[1]$ as on ({\bf \ref{routingRules2M}.48}).  Let $B$ denote the edge along which $A$ received $p$ in this stay, and $B'$ denote the edge along which $A$ sent $p$.  Since $A$ will be verifying that $B$ (respectively $B'$) signed the correct values (see comments on ({\bf \ref{routingRules2M}.84-86}) and  ({\bf \ref{routingRules2M}.88-90})), we have that $SIG^A[1]_{B,A}$ will increase by 1 on ({\bf \ref{routingRules2M}.74}) due to receiving $p$ for the first time, and $SIG^A[1]_{A,B'}$ will increase by 1 when it receives confirmation of receipt for sending $p$ as on ({\bf \ref{routingRules2M}.48}).  Therefore, considering only $p$'s affect on the following terms, we have that:
			\begin{equation} \label{firstE2}
			SIG^A[1]_{B,A} - SIG^A[1]_{A,B'} = 1-1= 0
			\end{equation}
				
		\item \textsf{The stay is initiated by $A$ receiving $p$ during $\mathtt{T}$, but $p$ then remains in $A$ through the end of the transmission (either as a normal or a flagged packet)}.  More specifically, the stay includes an increase to some incoming signature buffer $SIG^A[1]$ as on ({\bf \ref{routingRules2M}.74}), but there is no outgoing signature buffer $SIG^A[1]$ that changes value as on ({\bf \ref{routingRules2M}.48}) due to this stay of $p$.  Using the notation from (a) above, then considering only $p$'s affect on the following terms, we have that:
			\begin{equation} \label{thirdE2}
			SIG^A[1]_{B,A} = 1
			\end{equation}
		\end{enumerate}
	We note that the above two cases cover all possibilities by Claim \ref{packetProliferation3} (which remains valid since $A$ is honest, see Lemma \ref{similar}).  We now add all contributions to $SIG^A[1]_{A,B'}$ and $SIG^A[1]_{B,A}$ from all stays of all packets from all neighbors.  Notice that the only non-zero contributions come from packets stays as in \eqref{thirdE2}, and these contributions will correspond to packets that are still in $A$'s buffers at the end of the transmission.  Since an honest node $A$ can end the transmission with at most $2(n-2)(2n)$ packets, summing over all such contributions results cannot exceed $4n^2-8n$, as required.

	\item Intuitively, this is saying that an honest node cannot output a packet more times than it inputs the packet (see Claim \ref{packetProliferation3}).  Note that this is the only place in the theorem that depends on status reports {\it not} originating from $A$ ($SIG^B[p]$ is a status report parcel from $B$).  A priori, there is the danger that a corrupt $B$ can return a faulty status report, thereby framing $A$.  However, because $SIG^B[p]_{A,B}$ includes a valid signature from $A$, the inforgibility of the signature scheme guarantees that the only way a corrupt node $B$ can frame $A$ in this manner is by reporting {\it out-dated} signatures.  But if $A$ is honest, then $SIG^B[p]_{A,B}$ is strictly increasing in value as the transmission progresses (the only place it changes is ({\bf \ref{routingRules2M}.74}), which comes from the value received on ({\bf \ref{routingRules2M}.62}), corresponding to the value sent on ({\bf \ref{routingRules2M}.60})), and hence a corrupt $B$ cannot ``frame'' $A$ by reporting outdated signatures for $SIG^B[p]_{A,B}$; indeed such a course of action only helps the inequality stated in the theorem.  Also notice that (other than out-dated signatures) the only place $B$ gets valid signatures from $A$ is on ({\bf \ref{routingRules2M}.62}), and this value is one higher than the value that $A$ itself is recording ({\bf \ref{routingRules2M}.60}) until $A$ updates $SIG^A[p]_{A,B}$ on ({\bf \ref{routingRules2M}.48}).  We argue in case (b) below, that whenever $B$ has received an updated $SIG^B[p]_{A,B}$ as on ({\bf \ref{routingRules2M}.74}) but $A$ has not yet updated $SIG^A[p]_{A,B}$ as on ({\bf \ref{routingRules2M}.48}) (and so these two values differ by one), then Case (b) will contribute -1 to the sum in \eqref{eqM8}, and therefore the difference of +1 between $SIG^B[p]_{A,B}$ and $SIG^A[p]_{A,B}$ will exactly cancel.  These two facts allow us to argue \eqref{eqM8} by using $SIG^A[p]_{A,B}$ instead of $SIG^B[p]_{A,B}$.\vspace{.2cm}

	Formally, during a transmission $\mathtt{T}$, the only places the quantities $SIG[p]$ change are on ({\bf \ref{routingRules2M}.74}) and ({\bf \ref{routingRules2M}.48}).  As with the proof of Statement 3 above, we consider the contribution of each packet $p$'s stay in $A$\footnote{Note that necessarily $p$ is a packet corresponding to the current codeword, since packets corresponding to old codewords do not increment $SIG[p]$, see comments on ({\bf \ref{routingRules2M}.59-60}) and ({\bf \ref{routingRulesM}.11}).  Therefore, there are only two cases to consider.}:
		\begin{enumerate}
		\item \textsf{The stay is initiated by $A$ receiving $p$ during $\mathtt{T}$ and then sending $p$ at some later round of $\mathtt{T}$, and getting confirmation of $p$'s receipt as in Definition \ref{confRec}}.  More specifically, the stay includes an increase to some incoming signature buffer $SIG^A[p]$ as on ({\bf \ref{routingRules2M}.74}) and then an increase to some outgoing signature buffer $SIG^A[p]$ as on ({\bf \ref{routingRules2M}.48}).  Let $B$ denote the edge along which $A$ received $p$ in this stay, and $B'$ denote the edge along which $A$ sent $p$.  Since $A$ will be verifying that $B$ (respectively $B'$) signed the correct values (see comments on ({\bf \ref{routingRules2M}.84-86}) and ({\bf \ref{routingRules2M}.88-90})), we have that $SIG^A[p]_{B,A}$ will increase by 1 on ({\bf \ref{routingRules2M}.74}) due to receiving $p$ for the first time, and $SIG^A[p]_{A,B'}$ will increase by 1 when it receives confirmation of receipt for sending $p$ as on ({\bf \ref{routingRules2M}.48}).  Therefore, considering only $p$'s affect on the following terms, we have that:
			\begin{equation} \label{firstE22}
			SIG^A[p]_{A,B'} - SIG^A[p]_{B,A} = 1-1= 0
			\end{equation}
				
		\item \textsf{The stay is initiated by $A$ receiving $p$ during $\mathtt{T}$, but $p$ then remains in $A$ through the end of the transmission (either as a normal or a flagged packet)}.  More specifically, the stay includes an increase to some incoming signature buffer $SIG^A[p]$ as on ({\bf \ref{routingRules2M}.74}), but there is no outgoing signature buffer $SIG^A[p]$ that changes value as on ({\bf \ref{routingRules2M}.48}) due to this stay of $p$.  Using the notation from (a) above, then considering only $p$'s affect on the following terms, we have that:
			\begin{equation} \label{thirdE22}
			-SIG^A[1]_{B,A} = -1
			\end{equation}
		\end{enumerate}
	We note that the above two cases cover all possibilities by Claim \ref{packetProliferation3} (which remains valid since $A$ is honest, see Lemma \ref{similar}).  We now add all contributions to $SIG^A[p]_{A,B'}$ and $SIG^A[p]_{B,A}$ from all stays of $p$ from all neighbors on $\mathcal{P}$ (note that it is enough to consider only neighbors on $\mathcal{P}$ by Claim \ref{honestP}).  Notice that \eqref{firstE22} does not contribute anything, so we have that:
		\begin{equation}
		\sum_{B \in \mathcal{P}} (SIG^A[p]_{A,B} - SIG^A[p]_{B,A}) = -x,
		\end{equation}
	where $x$ is the number of times Case (b) occurs.  Notice that \eqref{eqM8} is interested in $SIG^B[p]_{A,B}$ (as opposed to $SIG^A[p]_{A,B}$).  However, since $B$ cannot report values of $SIG^B[p]_{A,B}$ from previous transmissions\footnote{We are only interested in packets $p$ corresponding to the current codeword, and all signatures that $A$ provides for $SIG^B[p]_{A,B}$ include the transmission index, so $A$'s honesty plus the inforgibility of the signature scheme imply that $B$ cannot have any valid signatures from $A$ contributing to $SIG^B[p]_{A,B}$ before the current transmission $\mathtt{T}$.}, the only inaccurate value that $B$ can report in its status report parcel concerning $SIG^B[p]_{A,B}$ is by using an older value from $\mathtt{T}$.  As discussed above, cheating in this manner only serves to help \eqref{eqM8}.  On the other hand, if $B$ does report the valid value for $SIG^B[p]_{A,B}$ (i.e.\ not outdated), then Lemma \ref{sigRelationships} guarantees that $SIG^B[p]_{A,B} -SIG^A[p]_{A,B} \leq 1$, with equality if $SIG^B[p]_{A,B}$ has been updated as on ({\bf \ref{routingRules2M}.74}) and $SIG^A[p]_{A,B}$ has not yet been updated after this point as on ({\bf \ref{routingRules2M}.48}).  Notice that every time this happens, we fall under Case (b) above, and in particular it can happen at most $x$ times (see definition of $x$ above).  Therefore: 
		\begin{equation}
		\sum_{B \in \mathcal{P}} (SIG^B[p]_{A,B} - SIG^A[p]_{B,A}) \leq x + \sum_{B \in \mathcal{P}} (SIG^A[p]_{A,B} - SIG^A[p]_{B,A}) = x-x = 0,
		\end{equation}
	which is \eqref{eqM8}.
	\end{enumerate}
All Statements of the Theorem have now been proven.
\end{proof}
%
We now prove a variant of Lemma \ref{packetTransferPotDrop}.
\begin{lemma} \label{packetTransferPotDropM} Suppose that $A,B \in G$ are both {\it honest} nodes, and that in round $\mathtt{t}$, B {\it accepts} (as in Definition \ref{receive}) a packet from $A$.  Let $O_{A,B}$ denote $A$'s outgoing buffer along $E(A,B)$, and let $H$ denote the height the packet had in $O_{A,B}$ when {\bf \em Send Packet} was called in round $\mathtt{t}$ ({\bf \ref{routingRulesM}.20}).  Also let $I_{B,A}$ denote $B$'s incoming buffer along $E(A,B)$, and let $I$ denote the height of $I_{B,A}$ at the start of $\mathtt{t}$.  Let $\Delta \varphi_B$ denote the {\bf \em change} in potential caused by this packet transfer, from $B$'s perspective.  More specifically, define:
	\begin{equation}
	\varphi_B := SIG^B[2]_{A,B} -SIG^B[3]_{A,B}
	\end{equation}
and then $\Delta \varphi_B$ measures the difference between the value of $\varphi_B$ at the {\bf \em end} of $\mathtt{t}$ and the {\bf \em start} of $\mathtt{t}$.  Then:
	\begin{equation}
	\Delta \varphi_B \geq H -I -1 \qquad OR \qquad \Delta \varphi_B \geq H \thickspace \mbox{ (if $B=R$)}
	\end{equation}
Furthermore, after the packet transfer but before re-shuffling, $I_{B,A}$ will have height $I + 1$.
\end{lemma}
\begin{proof} By definition, $B$ accepts the packet in round $\mathtt{t}$ means that ({\bf \ref{routingRules2M}.77}) was reached in round $\mathtt{t}$, and hence so was ({\bf \ref{routingRules2M}.74-75}).  In particular, $SIG^B[3]_{A,B}$ will increase by $H_{GP}$ on ({\bf \ref{routingRules2M}.75}) (if $B=R$, then $SIG^B[3]_{A,B}$ will not change on this line- see comment there).  By Statements 1 and 2 of Lemma \ref{pseudo} (which remain valid since $B$ is honest by Lemma \ref{similar}), $H_{GP} \leq I + 1$, and hence $SIG^B[3]_{A,B}$ will increase by at most $I+1$.  Also, since $B$ had height $I$ at the start of the round, and $B$ accepts a packet on ({\bf \ref{routingRules2M}.77}) of round $\mathtt{t}$, $B$ will have $I+1$ packets in $I$ when the re-shuffling phase of round $\mathtt{t}$ begins, which is the second statement of the lemma.

Meanwhile, $SIG^B[2]_{A,B}$ will change on ({\bf \ref{routingRules2M}.74}) to whatever value $B$ received on ({\bf \ref{routingRules2M}.62}) (as sent by $A$ on ({\bf \ref{routingRules2M}.60}) earlier in the round).  Since $A$ is honest, this value is $H_{FP}$ larger than $A$'s current value in $SIG^A[3]_{A,B}$ ({\bf \ref{routingRules2M}.60}).  By Lemma \ref{sigRelationships}, the value of $SIG^A[3]_{A,B}$ at the start of $\mathtt{t}$ equals the value of $SIG^B[2]_{A,B}$ at the {\it start} ({\it before} $B$ has accepted the packet) of $\mathtt{t}$.  Therefore, the change in $SIG^B[2]_{A,B}$ from the start of the round to the end of the round will be the value of $H_{FP}=H$ when $A$ reached ({\bf \ref{routingRules2M}.60}) in round $\mathtt{t}$ (by definition of $H$ and Statement 3 of Claim \ref{obsBelowHere}).  Since these are the only places $SIG^B[3]_{A,B}$ and $SIG^B[2]_{A,B}$ change, we have that $\Delta \varphi_B = H - H_{GP} \geq H -I -1$, as desired (if $B=R$, then $\Delta \varphi_B = H$).
\end{proof}
The following is a variant of Lemma \ref{chainL}.
\begin{lemma} \label{chainLM} Let $\mathcal{C} = N_1 N_2 \dots N_l$ be a path consisting of $l$ {\bf \em honest} nodes, such that $R = N_l$ and $S \notin \mathcal{C}$.  Suppose that in some {\bf \em non-wasted} round $\mathtt{t}$, all edges $E(N_i,N_{i+1})$, $1 \leq i <l$ are {\it active} for the entire round.  For $1 \leq i < l$, let $\Delta \phi$ denote the following changes to $SIG_{N_i, N_i}$ and $SIG^{N_i}$ during round $\mathtt{t}$:
	\begin{enumerate}\setlength{\itemsep}{1pt} \setlength{\parskip}{0pt} \setlength{\parsep}{0pt}
	\item Changes to $\varphi_{N_i}$ (see notation of Lemma \ref{packetTransferPotDropM}),
	\item Changes to $SIG_{N_i, N_i}$
	\end{enumerate}
Then if $O_{N_1, N_2}$ denotes $N_1$'s outgoing buffer along $E(N_1, N_2)$, we have:
	\begin{enumerate}\setlength{\itemsep}{1pt} \setlength{\parskip}{0pt} \setlength{\parsep}{0pt}
	\item[-] If $O_{N_1, N_2}$ has a flagged packet that has already been accepted by $N_2$ {\bf \em before} round $\mathtt{t}$, then:
		\begin{equation} \label{firstCM}
		\Delta \phi \geq O - l+1
		\end{equation}
	\item[-] Otherwise,
		\begin{equation}\label{secondCM}
		\Delta \phi \geq O - l+2
		\end{equation}
	\end{enumerate}
where $O$ denotes its height at the outset of $\mathtt{t}$.
\end{lemma}
\begin{proof} Since $A$ and $B$ are honest, we use Lemma \ref{similar} and then follow exactly the proof of the analogous claim for the edge-scheduling model (Lemma \ref{chainL}).  In particular, the exact proof can be followed, using the fact that signature buffers record accurate changes in non-duplicated potential (Statement 1 of Lemma \ref{sigBuffersAreCorrect}), and using Lemma \ref{item3M} in place of Lemma \ref{item3}, and Lemma \ref{packetTransferPotDropM} in place of Lemma \ref{packetTransferPotDrop}.
\end{proof}
\begin{lemma} \label{potentialDrop2} If at any point in any transmission $\mathtt{T}$, the number of blocked rounds is $\beta_{\mathtt{T}}$, then the participating honest nodes of $G$ will have recorded a drop in non-duplicated potential of at least $n(\beta_{\mathtt{T}}-4n^3)$.  More specifically, the following inequality is true:
	\begin{equation}
	n(\beta_{\mathtt{T}}-4n^3) < \sum_{A \in \mathcal{H} \setminus S} \hspace{-.2cm} SIG_{A,A} + \sum \hspace{-1cm} \sum_{A \in \mathcal{H} \setminus S \space \thickspace B \in \mathcal{P} \setminus \{A,S\} \quad} \hspace{-1cm}(SIG^A[2]_{B,A} - SIG^A[3]_{B,A})
	\end{equation}
\end{lemma}
\begin{proof}  For every blocked, non-wasted round $\mathtt{t}$, by the {\em conforming} assumption there exists a chain $\mathcal{C}_{\mathtt{t}}$ connecting the sender and receiver that satisfies the hypothesis of Lemma \ref{chainLM}.  Letting $N_1$ denote the first node on this chain (not including the sender), the fact that the round was blocked (and not wasted) means that $N_1$'s incoming buffer was full (see Lemma \ref{similar}), and then by Lemma \ref{balancing}, so was $N_1$'s outgoing buffer along $E(N_1,N_2)$.  Since the length of the chain $l$ is necessarily less than or equal to $n$, Lemma \ref{chainLM} says that the change of $\Delta \phi$ (see notation there) in round $\mathtt{t}$ satisfies:
	\begin{equation}
	\Delta \phi \geq O_{N_1, N_2} -l + 1  \geq 2n - n +1 > n
	\end{equation}
Since $\Delta \phi$ only records some of the changes to the signature buffers, we use Lemma \ref{item3M} to argue that the contributions not counted will only help the bound since they are strictly non-negative.  Since we are not double counting anywhere, each non-wasted, blocked round will correspond to an increase in $\Delta \phi$ of at least $n$, which then yields the lemma since the number of wasted rounds is bounded by $4n^3$ (Lemma \ref{maxWasted}).
\end{proof}
\begin{lemma} \label{secondSigFacts} If there exists $A,B \in G$ such that one of the following inequalities is {\it not} true, then either $A$ or $B$ is necessarily corrupt, and furthermore the sender can identify conclusively\footnote{As long as the adversary does not break the signature scheme, which will happen with all but negligible probability, the sender will never falsely identify an honest node.} which is corrupt\footnote{The values of the quantities $SIG^B$ and $SIG^A$ all correspond to a common transmission $\mathtt{T}$ and refer to values the sender has received in the form of status reports for $\mathtt{T}$ as on ({\bf \ref{routingRules4M}.161}).}:
	\begin{alignat}{2}
	&1. \quad SIG^B[2]_{A,B} \leq SIG^A[3]_{A,B} + 2n \notag \\
	&2. \quad SIG^A[3]_{S,A} - SIG^S[2]_{S,A} \leq 2n \notag \\
	&3. \quad |SIG^A[1]_{B,A}-SIG^B[1]_{B,A}| \leq 1 \qquad \mbox{and} \qquad |SIG^A[1]_{A,B}-SIG^B[1]_{A,B}| \leq 1 \notag \\
	\end{alignat}
\end{lemma}
\begin{proof} As in the first paragraph of the proof of Lemma \ref{item3M}, we may assume that both $A$ and $B$ have received the full {\it Start of Transmission} broadcast for $\mathtt{T}$, so $SIG^A$ and $SIG^B$ should both be cleared (if $A$ and $B$ are both honest) of its values from the previous transmission before being updated with values corresponding to the current transmission $\mathtt{T}$.  We prove each Statement separately:
	\begin{enumerate}
	\item That either $A$ or $B$ is necessarily corrupt follows from Lemma \ref{sigRelationships}.  It remains to show that the sender can identify a node that is necessarily corrupt.  We begin by assuming that $SIG^B[2]_{A,B}$ and $SIG^A[3]_{A,B}$ have appropriate signatures corresponding to $\mathtt{T}$ (otherwise, they either would not have been accepted as a valid status report parcel on ({\bf \ref{routingRules4M}.161}), or a node will be eliminated as on {\bf \ref{routingRules4M}.163}).  We now show that if the inequality in Statement 1 is {\it not} true for some $A,B \in G$, then $A$ is necessarily corrupt.  Notice that if $A$ is honest, then $SIG^A[3]_{A,B}$ is monotone increasing (other than being cleared upon receipt of the {\it SOT} broadcast, $SIG^A[3]_{A,B}$ is only updated on {\bf \ref{routingRules2M}.49}).  Similarly, other than being cleared upon receipt of the {\it SOT} broadcast, $SIG^B[2]_{A,B}$ is only updated on ({\bf \ref{routingRules2M}.74}), and tracing this backwards, this comes from the value received on ({\bf \ref{routingRules2M}.62}) which in turn was sent on ({\bf \ref{routingRules2M}.60}).  Therefore, since $B$ cannot forge $A$'s signature (except with negligible probability or in the case $A$ and $B$ are both corrupt and colluding), $SIG^B[2]_{A,B}$ can only take on values $A$ sent $B$ as on ({\bf \ref{routingRules2M}.60}).  Meanwhile, as mentioned, if $A$ is honest, $SIG^A[3]_{A,B}$ is monotone increasing, and thus an honest $A$ will never send a value for $SIG^A[3]_{A,B}$ on ({\bf \ref{routingRules2M}.60}) of some round that is {\it smaller} than a value it sent for $SIG^A[3]_{A,B}$ on ({\bf \ref{routingRules2M}.60}) of some earlier round.  Therefore, since the value $A$ is supposed to send $B$ is $H_{FP} \leq 2n$ (the inequality follows from Statement 9 of Lemma \ref{pseudo} and Lemma \ref{similar}), unless $A$ is corrupt or $B$ has broken the signature scheme, $B$ will never have a signed value from $A$ such that $SIG^B[2]_{A,B} > 2n + SIG^A[3]_{A,B}$.  Therefore, if the inequality in the first statement is not satisfied, $A$ is necessarily corrupt (except with negligible probability).
	
	\item That $A$ is necessarily corrupt follows from Lemma \ref{sigRelationships} and the fact that the sender cannot be corrupted by the conforming restriction placed on the adversary.  

	\item Note that the two statements are redundant, since the second is identical to the first after swapping the terms on the LHS and re-labelling.  We therefore only consider the second inequality of Statement 3.  That either $A$ or $B$ is necessarily corrupt follows from Lemma \ref{sigRelationships}.  It remains to show that the sender can identify a node that is necessarily corrupt.  As in the proof of Statement 1 above, we begin by assuming that $SIG^B[1]_{A,B}$ and $SIG^A[1]_{A,B}$ have appropriate signatures corresponding to $\mathtt{T}$ (otherwise, they either would not have been accepted as a valid status report parcel on ({\bf \ref{routingRules4M}.161}), or a node will be eliminated as on {\bf \ref{routingRules4M}.163}).  We now show that if $|SIG^A[1]_{A,B}-SIG^B[1]_{A,B}| > 1$ for some $A,B \in G$, then either $A$ or $B$ is necessarily corrupt, and the sender can identify which one is corrupt.\vspace{.2cm}
	
	\hspace{.5cm}Notice that the quantities $SIG^B[1]_{A,B}$ and $SIG^A[1]_{A,B}$ include the {\it round} in which the quantity last changed (({\bf \ref{routingRulesM}.11}) and ({\bf \ref{routingRules2M}.60})).  Let $\mathtt{t}_B$ denote the round $SIG^B[1]_{A,B}$ indicates it was last updated (which has been signed by $A$), and $\mathtt{t}_A$ denote the round $SIG^A[1]_{A,B}$ indicates it was last updated (which has been signed by $B$.  Note that these quantities refer to the values returned to the sender in the form of status report parcels, and node $A$ (respectively $B$) has signed the entire parcel $SIG^A[1]_{A,B}$ (respectively $SIG^B[1]_{A,B}$), indicating this is indeed the parcel he wishes to commit to as his status report.  We assume $|SIG^A[1]_{A,B}-SIG^B[1]_{A,B}| > 1$, and break the proof into the following two cases:
		\begin{enumerate}
		\item[] \textsf{Case 1: $\mathtt{t}_A > \mathtt{t}_B$}.  We will show that $B$ is corrupt.  Notice that the fact that $A$ has a valid signature on $SIG^A[1]_{A,B}$ from $B$ for round $\mathtt{t}_A$ means that (with all but negligible probability that $A$ could forge $B$'s signature, or if $A$ and $B$ are both corrupt, allowing $A$ to forge $B$'s signature) $B$ sent communication as on ({\bf \ref{routingRulesM}.11}) of $\mathtt{t}_A$ with the fifth coordinate equal to the value $A$ used for $SIG^A[1]_{A,B}$.  In particular, this fifth coordinate represents the value $B$ has stored for $SIG^B[1]_{A,B}$ during $\mathtt{t}_A$.  Since $\mathtt{t}_B < \mathtt{t}_A$, $B$ does not update $SIG^B[1]_{A,B}$ from $\mathtt{t}_B$ through the end of $\mathtt{T}$, and hence the value for $SIG^B[1]_{A,B}$ that $B$ returns the sender in its status report should be the same as the value $B$ sent to $A$ on ({\bf \ref{routingRulesM}.11}) of round $\mathtt{t}_A$, which as noted above equals the value of $SIG^A[1]_{A,B}$ that $A$ returned in its status report.  However, since this is not the case ($SIG^B[1]_{A,B} \neq SIG^A[1]_{A,B}$), $B$ has returned an outdated signature and must be corrupt.

		\item[] \textsf{Case 2: $\mathtt{t}_A \leq \mathtt{t}_B$}.  If $\mathtt{t}_A = \mathtt{t}_B = 0$, i.e.\ both nodes agree that they did not update their signature buffers along $E(A,B)$ in the entire transmission (except to clear them when they received the {\it SOT} broadcast), then necessarily both $SIG^A[1]_{A,B}$ and $SIG^B[1]_{A,B}$ should be set to $\bot$, so if one of them is {\it not} $\bot$, the node signing the non-$\bot$ value can be eliminated.  So assume that one of the nodes has a valid signature from the other for some round in $\mathtt{T}$ (i.e.\ that $\mathtt{t}_B >0$).  We will show that $A$ is corrupt in a manner similar to showing $B$ was corrupt above.  Indeed, since $B$ has a valid signature from $A$ on $SIG^B[1]_{A,B}$ from round $\mathtt{t}_B$, unless $A$ and $B$ are colluding or $B$ has managed to forge $A$'s signature, this value for $SIG^B[1]_{A,B}$ comes from the communication sent by $A$ on ({\bf \ref{routingRules2M}.60}).  In particular, since $\mathtt{t}_A \leq \mathtt{t}_B$ and $A$ claims he was not able to update $SIG^A[1]_{A,B}$ after round $\mathtt{t}_A$, the value $A$ signed and sent on ({\bf \ref{routingRules2M}.60}) should be exactly one l more than the value stored in $SIG^A[1]_{A,B}$ as of line ({\bf \ref{routingRulesM}.07}) of round $\mathtt{t}_A$, the latter of which was returned by $A$ in its status report (by definition of $\mathtt{t}_A$ and the inforgibility of the signature scheme).  But since $|SIG^A[1]_{A,B}-SIG^B[1]_{A,B}| > 1$, this must not be the case, and hence $A$ is corrupt.
		\end{enumerate}
	\end{enumerate}
\vspace{-22pt}
\end{proof}
\begin{cor} \label{F2Cor} If there exists a node $A \in G$ such that:
	\begin{equation} \label{2013}
	4n^3 -4n^2 < SIG_{A,A} \thickspace + \sum_{B \in \mathcal{P} \setminus A} SIG^B[2]_{A,B} - SIG^A[3]_{B,A},
	\end{equation}
then either a node can be eliminated as in Statement 1 of Lemma \ref{secondSigFacts} or as in Statement 6 of Lemma \ref{item3M}.
\end{cor}
\begin{proof} Suppose no node can be eliminated because of Statement 1 of Lemma \ref{secondSigFacts}, so that for all $B \in G$:
	\begin{equation}\label{2018}
	SIG^B[2]_{A,B} \leq SIG^A[3]_{A,B} + 2n.
	\end{equation}
Then if \eqref{2013} is true, we have that:
	\begin{alignat}{2}
	4n^3 -4n^2 &< SIG_{A,A} \thickspace + \sum_{B \in \mathcal{P} \setminus A} SIG^B[2]_{A,B} - SIG^A[3]_{B,A} \notag \\
	&\leq SIG_{A,A} \thickspace + 2n^2 + \sum_{B \in \mathcal{P} \setminus A} SIG^A[3]_{A,B} - SIG^A[3]_{B,A}
	\end{alignat}
where the second inequality follows from applying \eqref{2018} to each term of the sum.  Therefore, $A$ can be eliminated by Statement 6 of Lemma \ref{item3M}.
\end{proof}
\begin{cor} \label{obsProofs} In the case a transmission fails as in F2, the increase in network potential due to packet insertions is at most $2nD+2n^2$.  In other words, either there exists a node $A \in G$ such that the sender can eliminate $A$, or the following inequality is true\footnote{The values of the quantities $SIG^A$ correspond to some transmission $\mathtt{T}$ and refer to values the sender has received in the form of status reports for $\mathtt{T}$ as on ({\bf \ref{routingRules4M}.161}).}:
	\begin{equation}
	\sum_{A \in \mathcal{P} \setminus S} SIG^A[3]_{S,A} < 2nD + 2n^2
	\end{equation}
\end{cor}
\begin{proof} If the inequality in Statement 2 of Lemma \ref{secondSigFacts} fails for any node $A \in \mathcal{P} \setminus S$, the sender can immediately eliminate $A$.  So assume that the inequality in Statement 2 of Lemma \ref{secondSigFacts} holds for every $A \in \mathcal{P} \setminus S$.  The corollary will be a consequence of the following observation:
	\begin{itemize}
	\item[] {\bf Observation.} {If a transmission $\mathtt{T}$ fails as in F2, then}:
		\begin{equation}
		\sum_{A \in \mathcal{P} \setminus S} SIG^S[2]_{S,A} < 2nD
		\end{equation}

	\item[] {\it Proof}.  Let $\kappa_{\mathtt{T}}$ denote the value that $\kappa$ had at the end of $\mathtt{T}$.  Then formally, a transmission falling under F2 means that $\kappa_{\mathtt{T}}$ is less than $D$.  The structure of this proof will be to first show that for any $A \in \mathcal{P} \setminus S$, anytime $SIG^S[2]_{S,A}$ is updated as on ({\bf \ref{routingRules2M}.48}), it will always be the case that $2n*SIG^S[1]_{S,A} \geq SIG^S[2]_{S,A}$ (so that in particular that the final value for $SIG^S[2]_{S,A}$ at the end of $\mathtt{T}$ is less than or equal to $2n$ times the final value for $SIG^S[1]_{S,A}$).  We will then show that at the end of $\mathtt{T}$: $\sum_{A \in \mathcal{P} \setminus S} SIG^S[1]_{S,A} = \kappa_{\mathtt{T}}$.  From these two facts, we will have shown:
		\begin{equation}
		\sum_{A \in \mathcal{P} \setminus S} SIG^S[2]_{S,A} \leq \sum_{A \in \mathcal{P} \setminus S} 2n*SIG^S[1]_{S,A} = 2n \kappa_{\mathtt{T}} < 2nD
		\end{equation}
	as required.\vspace{.2cm}

	The first fact is immediate, since for any $A \in \mathcal{P} \setminus S$, whenever $SIG^S[2]_{S,A}$ is updated as on ({\bf \ref{routingRules2M}.48}), the statement on ({\bf \ref{routingRules2M}.45}) must have been satisfied, and so the statement on ({\bf \ref{routingRules2M}.89}) must have been false.  In particular, the change in $SIG^S[1]_{S,A}$ was exactly one, and the change in $SIG^S[2]_{S,A}$ was at most $H_{FP} \leq 2n$, where the inequality comes from Statement 9 of Lemma \ref{pseudo} and Lemma \ref{similar} (see comments on lines ({\bf \ref{routingRules2M}.88-90})).  The second fact is also immediate, as $\kappa$ and $SIG^S[1]_{S,A}$ all start the transmission with value zero (or $\bot$) by lines ({\bf \ref{setupCode2M}.54}), ({\bf \ref{setupCode2M}.70}), ({\bf \ref{routingRules5M}.199}), and ({\bf \ref{routingRules5M}.213}), and then $\kappa$ is incremented by one on line ({\bf \ref{routingRules2M}.47}) of the outgoing buffer along some edge $E(S,N)$ if and only if $SIG^S[1]_{S,N}$ is incremented by one as on ({\bf \ref{routingRules2M}.48}) (as already argued, changes to $SIG^S[1]_{S,A}$ as on ({\bf \ref{routingRules2M}.48}) are always increments of one, see e.g.\ the comments on lines ({\bf \ref{routingRules2M}.88-90})).\hspace*{\fill} \hspace*{-12pt}$\scriptstyle{\square}$
	\end{itemize}
The corollary now follows immediately from the following string of inequalities:
	\begin{alignat}{2}
	2nD &> \sum_{A \in \mathcal{P} \setminus S} SIG^S[2]_{S,A} \notag \\
	&\geq -2n^2 + \sum_{A \in \mathcal{P} \setminus S} SIG^A[3]_{S,A} \notag
	\end{alignat}
where the top inequality is the statement of the Observation and the second inequality comes from applying the inequality in Statement 2 of Lemma \ref{secondSigFacts} to each term of the sum.
\end{proof}
\begin{lemma} \label{sigBuffersAreCorrect} For any honest node $N \in G$ and for any transmission $\mathtt{T}$:
	\begin{enumerate}
	\item Upon receipt of the complete {\it Start of Transmission} (SOT) broadcast for transmission $\mathtt{T}$, $SIG_{N,N}$ will be cleared.  After this point through the end of transmission $\mathtt{T}$, $SIG_{N,N}$ stores the correct value corresponding to the current transmission $\mathtt{T}$ (as listed on {\bf \ref{setupCodeM}.12}).

	\item Suppose that $N$ transfers at least one packet during $\mathtt{T}$ (i.e.\ $N$ sends or receives at least one packet, as on ({\bf \ref{routingRules2M}.60}) or ({\bf \ref{routingRules2M}.74-78})).  Then through all transmissions after $\mathtt{T}$ until the transmission and round $(\mathtt{T}', \mathtt{t}' \in \mathtt{T}')$ that $N$ next receives the complete {\it SOT} transmission for $\mathtt{T}'$, one of the following must happen:
		\begin{enumerate}
		\item All of $N$'s signature buffers contain information (i.e.\ signatures from neighbors) pertaining to $\mathtt{T}$, or
		\item All of $N$'s signature buffers are clear and $N$'s broadcast buffer\footnote{Or the Data Buffer in the case $N=S$.} contains all of the information that was in the signature buffers at the end of $\mathtt{T}$, or
		\item $(N, \mathtt{T}, \mathtt{T}')$ is not on the blacklist for transmission $\mathtt{T}'$
		\end{enumerate}
	
	\item If $N$ has received the full {\it SOT} broadcast for $\mathtt{T}$, then all parcels in $N$'s broadcast buffer\lastfootnote $BB$ corresponding to some node $\widehat{N}$'s status report are current and correct.  More precisely:
		\begin{enumerate}
		\item If $(\widehat{N}, \widehat{\mathtt{T}})$ is on the sender's blacklist, and at {\it any} time $N$ has stored a parcel of $\widehat{N}$'s corresponding status report in its broadcast buffer $BB$, then this parcel will not be deleted until $(\widehat{N}, \widehat{\mathtt{T}})$ is removed from the sender's blacklist.

		\item If $(\widehat{N}, \widehat{\mathtt{T}}, \mathtt{T}')$ is a part of the {\it SOT} broadcast of transmission $\mathtt{T}'$, then upon receipt of this parcel, all of $\widehat{N}$'s status report parcels in $N$'s broadcast buffer correspond to transmission $\mathtt{T}'$ and are of the form as indicated on ({\bf \ref{routingRules4M}.141-144}), where the reason for failure of transmission $\mathtt{T}'$ was determined as on ({\bf \ref{routingRules5M}.190}), ({\bf \ref{routingRules5M}.193}), or ({\bf \ref{routingRules5M}.196}).
		\end{enumerate}
	\item If at any time $N$ is storing a parcel of the form $(B, \widehat{N}, \widehat{\mathtt{T}})$ in its broadcast buffer (indicating $B$ knows $\widehat{N}$'s complete status report for transmission $\widehat{\mathtt{T}}$), then this will not be deleted until $(\widehat{N}, \widehat{\mathtt{T}})$ has been removed from the blacklist.
	\end{enumerate}
\end{lemma}
\begin{proof} Fix an honest $N \in G$ and a transmission $\mathtt{T}$.  We prove each Statement separately:
	\begin{enumerate}
	\item The first part of statement 1 is Lemma \ref{sigBuffsInit}.  To prove the second part, we track all changes to $SIG_{N,N}$ and show that each change accurately records the value $SIG_{N,N}$ is supposed to hold.  The only changes made to $SIG_{N,N}$ after receiving the full {\it SOT} broadcast occur on lines ({\bf \ref{reShuffleRules}.76}), ({\bf \ref{routingRules2M}.50}), ({\bf \ref{routingRules2M}.80}), and ({\bf \ref{routingRules2M}.82}).  Meanwhile, $SIG_{N,N}$ is {\it supposed} to track all packet movement that occurs within $N$'s own buffers (i.e.\ all packet movement except packet transfers).  The only places packets move within buffers of $N$ are on lines ({\bf \ref{reShuffleRules}.89-90}), ({\bf \ref{routingRules2M}.50}), ({\bf \ref{routingRules2M}.80}), and ({\bf \ref{routingRules2M}.82}).  By the comments on lines ({\bf \ref{routingRules2M}.50}), ({\bf \ref{routingRules2M}.53}), ({\bf \ref{routingRules2M}.80}), and ({\bf \ref{routingRules2M}.82}), it is clear that $SIG_{N,N}$ appropriately tracks changes in potential due to the call to {\it Fill Gap}, while packet movement as on ({\bf \ref{routingRules2M}.53}) does not need to change $SIG_{N,N}$ as packets are swapped, and so there is no net change in potential.  In terms of re-shuffling ({\bf \ref{reShuffleRules}.89-90}), we see that every packet that is re-shuffled causes a change in $SIG_{N,N}$ of $M-m-1$ ({\bf \ref{reShuffleRules}.76}).  Notice the actual change in potential matches this amount, since a packet is removed from a buffer at height $M$ ({\bf \ref{reShuffleRules}.90}), reducing the height of that buffer from $M$ to $M-1$ (a drop in potential of $M$),  and put into a buffer at height $m+1$, increasing the height of the buffer from $m$ to $m+1$ (an increase of $m+1$ to potential).

	\item If $N=S$, there is nothing to show, since the sender's signature buffers' information is stored as needed on ({\bf \ref{routingRules5M}.191}), ({\bf \ref{routingRules5M}.194}), and ({\bf \ref{routingRules5M}.197}), and they are then cleared at the end of every transmission on ({\bf \ref{routingRules5M}.171}) or ({\bf \ref{routingRules5M}.199}).  For any $N \neq S$, we show that from the time $N$ receives the full {\it SOT} broadcast in a transmission $\mathtt{T}$ through the next transmission $\mathtt{T}'$ in which $N$ next hears the full {\it SOT} broadcast, either all of $N$'s signature buffers contain information from the last time they were updated in some round of $\mathtt{T}$, or they are empty and either this information has already been transferred to $N$'s broadcast buffer or $N$ is not on the blacklist for transmission $\mathtt{T}'$ (this will prove Statement 2).  During transmission $\mathtt{T}$, there is nothing to show, as all changes made to any signature buffer over-write earlier changes, so throughout $\mathtt{T}$, the signature buffers will always contain the most current information.  It remains to show that between the end of $\mathtt{T}$ and the time $N$ receives the full {\it SOT} broadcast of transmission $\mathtt{T}'$, the only change that $N$'s signature buffers can make is to be cleared, and this can happen only if either the information contained in them is first transferred to $N$'s broadcast buffer, or if $(N,\mathtt{T}, \mathtt{T}')$ does not appear in the {\it SOT} broadcast of transmission $\mathtt{T}'$ (and hence the signature information will not be needed anyway).  To do this, we list all places in the pseudo-code that call for a change to one of the signature buffers or removing data from the broadcast buffer, and argue that one of these two things must happen.  In particular, the only places the signature buffers of $N$ change (after initialization) are: ({\bf \ref{routingRules2M}.48-49}), ({\bf \ref{routingRules2M}.50}), ({\bf \ref{routingRules2M}.74-75}), ({\bf \ref{routingRules2M}.80}), ({\bf \ref{routingRules2M}.82}), ({\bf \ref{routingRules4M}.128}), ({\bf \ref{routingRules4M}.133}), ({\bf \ref{routingRules4M}.141}), ({\bf \ref{routingRules4M}.146}), and ({\bf \ref{reShuffleRules}.76}).  The only place that information that was once in one of $N$'s signature buffers is removed from the broadcast buffer is ({\bf \ref{routingRules4M}.134}).
	
	\hspace{.5cm}First notice that because $N$ transfers a packet in transmission $\mathtt{T}$, $N$ must have received the complete {\it SOT} broadcast for transmission $\mathtt{T}$ (Lemma \ref{sigBuffsInit}).  For all rounds of all transmissions between $\mathtt{T}+1$ and the time $N$ receives the full {\it SOT} broadcast for transmission $\mathtt{T}'$, lines ({\bf \ref{routingRules2M}.48-51}), ({\bf \ref{routingRules2M}.74-78}), and ({\bf \ref{reShuffleRules}.76}) will never be reached by $N$ (see Lemma \ref{sigBuffsInit} and its proof).  Similarly, line ({\bf \ref{routingRules2M}.80}) will never be reached since ({\bf \ref{routingRules2M}.63}) will always be satisfied.  Although ({\bf \ref{routingRules2M}.82}) may be reached, we argue that it will not change $SIG_{N,N}$ by arguing that for all rounds between $\mathtt{T}+1$ and the time $N$ receives the full {\it SOT} broadcast for transmission $\mathtt{T}'$, there will never be codeword packets occupying a higher slot than the ghost packet.  More precisely, we will show that for all rounds between $\mathtt{T}+1$ and the time $N$ receives the full {\it SOT} broadcast for transmission $\mathtt{T}'$, either $H_{GP} = \bot$ or $H_{GP} = H_{IN}+1$, and then by Statements 1 and 2 of Lemma \ref{pseudo} (together with the fact that $N$ is honest and so we may apply Lemma \ref{similar}), {\it Fill Gap} on ({\bf \ref{routingRules2M}.82}) will not be performed (see comments on that line).  That $H_{GP}=\bot$ or $H_{GP}=H_{IN}+1$ for all of these rounds follows from the fact that $H_{GP}$ will be set to $\bot$ at the end of $\mathtt{T}$ ({\bf \ref{routingRules5M}.209}), after which it can only be modified on ({\bf \ref{routingRules2M}.66}), ({\bf \ref{routingRules2M}.72}), ({\bf \ref{routingRules2M}.76}), ({\bf \ref{routingRules2M}.78}), ({\bf \ref{routingRules2M}.80}), or ({\bf \ref{routingRules2M}.82}).  Notice that all of these set $H_{GP}$ to $\bot$ or $H_{IN}+1$ and that $H_{IN}+1$ cannot change for all rounds between $\mathtt{T}+1$ and the time $N$ receives the full {\it SOT} broadcast for transmission $\mathtt{T}'$ by Lemma \ref{sigBuffsInit}.

	\hspace{.5cm}It remains to consider lines ({\bf \ref{routingRules4M}.128}), ({\bf \ref{routingRules4M}.133}), ({\bf \ref{routingRules4M}.141}), ({\bf \ref{routingRules4M}.146}), and ({\bf \ref{routingRules4M}.134}); the first four clear the signature buffers, and the last clears the broadcast buffer.  So it remains to argue that if any of these lines are reached, either the broadcast buffer is storing all of the information that the signature buffers held at the end of $\mathtt{T}$, or $(N, \mathtt{T},\mathtt{T}')$ cannot appear as part of the {\it SOT} broadcast of transmission $\mathtt{T}'$.  Line ({\bf \ref{routingRules4M}.128}) is clearly covered by the latter case, since if a parcel of this form is received in some transmission $\widehat{\mathtt{T}} \in [\mathtt{T}+1 .. \mathtt{T}']$, then $(N,\mathtt{T})$ is not on the sender's blacklist as of $\widehat{\mathtt{T}} > \mathtt{T}$, and hence $(N, \mathtt{T})$ will never be able to be re-added to the blacklist after this point (see ({\bf \ref{routingRules5M}.188})).  Similar reasoning shows that line ({\bf \ref{routingRules4M}.146}) is covered by one of these two cases.  In particular, if $N$ reaches line ({\bf \ref{routingRules4M}.146}) in some transmission $\widehat{\mathtt{T}} \in [\mathtt{T}+1..\mathtt{T}']$, then either $N$ will add the information in its signature buffers into its broadcast buffers as on ({\bf \ref{routingRules4M}.142-145}) before reaching ({\bf \ref{routingRules4M}.146}), or else $N$ was not on the blacklist as of $\widehat{\mathtt{T}}$, and hence it is impossible for $(N, \mathtt{T}, \mathtt{T}')$ to be a part of the {\it SOT} broadcast for transmission $\mathtt{T}'$.  Now suppose $N$ reaches ({\bf \ref{routingRules4M}.133-134}) in some round of a transmission $\widehat{\mathtt{T}} > \mathtt{T}$ indicating that a node $\widehat{N}$ is to be eliminated.  In order to reach ({\bf \ref{routingRules4M}.133-134}) in transmission $\widehat{\mathtt{T}}$, $N$ must not have known that $\widehat{N}$ was to be eliminated before that point ({\bf \ref{routingRules4M}.131}), and since $N$ received the complete {\it SOT} broadcast of transmission $\mathtt{T}$ (by Lemma \ref{sigBuffsInit} together with the hypotheses that $N$ is honest and transferred a packet in $\mathtt{T}$), $\widehat{N}$ must have been eliminated in some transmission $\widetilde{\mathtt{T}} \geq \mathtt{T}$.  In particular, if $\widetilde{\mathtt{T}} = \mathtt{T}$, then $(N,\mathtt{T})$ can never be added to the blacklist (since ({\bf \ref{routingRules4M}.188}) cannot be reached in transmission $\mathtt{T}$ if {\it Eliminate Node} is reached in that transmission); while if $\widetilde{\mathtt{T}} > \mathtt{T}$, then $(N, \mathtt{T})$ will be cleared from the blacklist as on ({\bf \ref{routingRules4M}.171}) (if it was on the blacklist), and as already remarked, $(N, \mathtt{T})$ can never again appear on the blacklist after this.  
	
	\hspace{.5cm} Now suppose ({\bf \ref{routingRules4M}.141}) is reached in some transmission $\widehat{\mathtt{T}} > \mathtt{T}$ and the signature buffers are cleared on this line.  Now before line ({\bf \ref{routingRules4M}.141}) was reached, by induction, one of the three statements (a), (b), or (c) was true.  If (b) or (c) was true, then changes made on ({\bf \ref{routingRules4M}.141}) will not affect the fact that (b) or (c) will remain true.  Therefore, assume that we are in case (a) before reaching ({\bf \ref{routingRules4M}.141}), i.e.\ that when ({\bf \ref{routingRules4M}.141}) is reached in transmission $\widehat{\mathtt{T}}$, $N$'s signature buffers contain the information that they had at the end of $\mathtt{T}$.  Since ({\bf \ref{routingRules4M}.141}) was reached, it must have been that $(N, \widetilde{\mathtt{T}}, \widehat{\mathtt{T}})$ was received on ({\bf \ref{routingRules4M}.137}) as part of the {\it SOT} broadcast for transmission $\widehat{\mathtt{T}}$, for some $\widetilde{\mathtt{T}}$.  We first argue $\widetilde{\mathtt{T}} \geq \mathtt{T}$.  To see this, since $N$ is honest, it will not transfer any packets in $\mathtt{T}$ if it is on its own version of the blacklist (({\bf \ref{routingRulesM}.31-33}) and ({\bf \ref{routingRulesM}.35-37})).  Since we know that $N$ {\it did} transfer packets in transmission $\mathtt{T}$ (by hypothesis), and also $N$ received the full {\it SOT} broadcast of that same transmission (Lemma \ref{sigBuffsInit}), either $N$ was not on the blacklist as of the start of transmission $\mathtt{T}$, or $N$ received information as on ({\bf \ref{routingRules4M}.147}) indicating $N$ could be removed from the blacklist.  Both of these cases imply that by the end of $\mathtt{T}$, $(N, \widetilde{\mathtt{T}})$ can never be on the blacklist for any $\widetilde{\mathtt{T}} < \mathtt{T}$.  Thus, $\widetilde{\mathtt{T}} \geq \mathtt{T}$, as claimed.  Since we are assuming case (a), if $\widetilde{\mathtt{T}} = \mathtt{T}$, then ({\bf \ref{routingRules4M}.141}) will {\it not} be satisfied.  On the other hand, if $\widetilde{\mathtt{T}} > \mathtt{T}$, then $N$ has appeared on the blacklist for some transmission {\it after} $\mathtt{T}$, and then Lemma \ref{blacklistStuff} guarantees that $(N, \mathtt{T})$ is not on the blacklist as of $\widehat{\mathtt{T}} > \mathtt{T}$, which as noted above implies $(N, \mathtt{T}, \mathtt{T}')$ cannot be part of the {\it SOT} broadcast of transmission $\mathtt{T}'$.

	\item For Statement (a), we track all the times parcels are removed from $N$'s broadcast buffer $BB$, and ensure that if ever $N$ removes a status report parcel belonging to $\widehat{N}$ for some transmission $\widehat{\mathtt{T}}$, then $(\widehat{N}, \widehat{\mathtt{T}})$ is no longer on the sender's blacklist.  If $N=S$, notice the only place that information concerning other nodes' status report parcels is removed from the sender's data buffer is ({\bf \ref{routingRules5M}.171}), and at this point $\widehat{N}$ is not on the blacklist since the blacklist is cleared on this same line.
	
	\hspace{.5cm}If $N \neq S$, changes to $BB$ occur only on lines ({\bf \ref{routingRules4M}.134}), ({\bf \ref{routingRules4M}.139}), ({\bf \ref{routingRules4M}.149}), ({\bf \ref{routingRules4M}.142-145}), and ({\bf \ref{routingRules4M}.154}).  The former three lines {\it remove} things from $BB$, while the latter lines {\it add} things to $BB$.  In terms of statement (a), we must ensure whenever one of the former three lines is reached, there will never be a status report parcel from $\widehat{N}$ and corresponding to transmission $\widehat{\mathtt{T}}$ that is removed from $BB$ if $(\widehat{N},\widehat{\mathtt{T}})$ is on the blacklist.  Looking first at line ({\bf \ref{routingRules4M}.134}), suppose that $N$ reaches line ({\bf \ref{routingRules4M}.134}) in some transmission $\widetilde{\mathtt{T}} \geq \mathtt{T}$.  If $(\widehat{N}, \widehat{\mathtt{T}}, \widetilde{\mathtt{T}})$ was {\it not} a part of the {\it SOT} broadcast of transmission $\widetilde{\mathtt{T}}$, then there is nothing to show (since $\widehat{N}$ is not on the blacklist as of the outset of $\widetilde{\mathtt{T}}$).  So suppose that $(\widehat{N}, \widehat{\mathtt{T}}, \widetilde{\mathtt{T}})$ was a part of the {\it SOT} broadcast of transmission $\widetilde{\mathtt{T}}$.  Since reaching line ({\bf \ref{routingRules4M}.134}) requires that $N$ has newly learned that a node has been added to $EN$ ({\bf \ref{routingRules4M}.131}), let $N'$ denote this node, and let $\mathtt{T}'$ denote the round that $N'$ was eliminated from the network as on ({\bf \ref{routingRules5M}.170}).  First note that necessarily $\mathtt{T}' < \widehat{\mathtt{T}}$.  After all, the blacklist will be cleared on line ({\bf \ref{routingRules5M}.171}) of round $\mathtt{T}'$, and hence if $(\widehat{N}, \widehat{\mathtt{T}})$ is still on the blacklist as of the outset of $\widetilde{\mathtt{T}}$, it must have been added afterwards.  We now argue that because $\mathtt{T}' < \widehat{\mathtt{T}}$, the priority rules of transferring broadcast information will dictate that all honest nodes will necessarily learn $N'$ has been eliminated {\it before} they learn that $(\widehat{N}, \widehat{\mathtt{T}})$ is on the blacklist.  From this, we will conclude that when $N$ reaches ({\bf \ref{routingRules4M}.134}) in transmission $\widetilde{\mathtt{T}}$ and learns that $N'$ should be eliminated, that $N$ has not yet learned that $(\widehat{N}, \widehat{\mathtt{T}})$ is on the blacklist, and hence $N$'s broadcast buffer will not be storing any of $\widehat{N}$'s status report parcels for $\widehat{\mathtt{T}}$ ({\bf \ref{routingRules4M}.152}).  
	
	\hspace{.5cm}It remains to show that any honest node $A \in G$ will learn that $N'$ has been eliminated {\it before} they learn $(\widehat{N},\widehat{\mathtt{T}})$ is on the blacklist.  So fix an honest node $A \in G$.  Suppose $A$ first learns $(\widehat{N}, \widehat{\mathtt{T}})$ is on the blacklist via a parcel of the form $(\widehat{N}, \widehat{\mathtt{T}}, \mathtt{X})$ that it received as on ({\bf \ref{routingRules4M}.137}) of transmission $\mathtt{X}$.  Clearly, $\mathtt{X} > \widehat{\mathtt{T}}$, since $(\widehat{N}, \widehat{\mathtt{T}})$ can only be put on the blacklist at the very end of transmission $\widehat{\mathtt{T}}$.  Therefore, since $\mathtt{T}' < \widehat{\mathtt{T}} < \mathtt{X}$, we have that $(N', \mathtt{X})$ will be a part of the {\it SOT} broadcast for transmission $\mathtt{X}$, indicating that $N'$ has been eliminated ({\bf \ref{routingRules5M}.200}).  Since $A$ is honest, it will therefore receive $(N', \mathtt{X})$ {\it before} it receives $(\widehat{N}, \widehat{\mathtt{T}}, \mathtt{X})$ (see priority rules for receiving broadcast parcels, ({\bf \ref{routingRules3M}.110}) and ({\bf \ref{routingRules3M}.115}))

	\hspace{.5cm}We next consider when status report parcels are removed from $BB$ as on ({\bf \ref{routingRules4M}.139}).  In this case, $N$ has received a {\it SOT} broadcast parcel of form $(\widehat{N}, \widehat{\mathtt{T}}, \mathtt{T}')$ ({\bf \ref{routingRules4M}.137}), and $N$ is removing from $BB$ all of $\widehat{N}$'s status report parcels corresponding to transmissions {\it other} than $\widehat{\mathtt{T}}$.  First note that Lemma \ref{blacklistStuff} guarantees that $\widehat{N}$ is on at most one blacklist at any time.  Since $N$ received a {\it SOT} parcel of the form  $(\widehat{N}, \widehat{\mathtt{T}}, \mathtt{T}')$ during transmission $\mathtt{T}'$, it must be that $(\widehat{N}, \widehat{\mathtt{T}})$ was on the sender's blacklist at the outset of $\mathtt{T}'$, and since nothing can be added to the blacklist until the very end of a transmission ({\bf \ref{routingRules5M}.188}), only $(\widehat{N}, \widehat{\mathtt{T}})$ can be on the sender's blacklist at the outset of $\mathtt{T}'$.  This case is now settled, as we have shown that $N$ does not remove any of the status report parcels from $\widehat{N}$ corresponding to $\widehat{\mathtt{T}}$ on ({\bf \ref{routingRules4M}.139}), and this is the only transmission for which $\widehat{N}$ can be on the blacklist (at least through $\mathtt{T}'$).
	
	\hspace{.5cm}To complete Statement (a), it remains to consider line ({\bf \ref{routingRules4M}.149}).  But this is immediate, as if the sender at any time removes $(\widehat{N}, \widehat{\mathtt{T}})$ from the blacklist, then it can never again be re-added (since nodes are added to the blacklist at the very {\it end} of a transmission ({\bf \ref{routingRules5M}.188}), they are not removed as on ({\bf \ref{routingRules4M}.166}) or ({\bf \ref{routingRules5M}.171}) until at least the next transmission, at which point the same $(node, transmission)$ pair $(\widehat{N}, \widehat{\mathtt{T}})$ can never again be added to the blacklist as on ({\bf \ref{routingRules5M}.188}) since $\widehat{\mathtt{T}}$ has already passed).  Therefore, when $N$ reaches ({\bf \ref{routingRules4M}.149}), if the items deleted from $BB$ correspond to $\widehat{N}$, then $N$ must have received a broadcast parcel of form $(\widehat{N},0,\mathtt{T})$ as on ({\bf \ref{routingRules4M}.147}), indicating that $\widehat{N}$ was no longer on the blacklist.  Consequently, the status parcels deleted will never again be needed since $(\widehat{N}, \widehat{\mathtt{T}})$ can never again be on the blacklist.

	\hspace{.5cm}Part (a) of Statement 3 of the lemma (now proven) states that no status report parcel still needed by the sender will ever be deleted from a node's broadcast buffer.  Part (b) states that a node's broadcast buffer will not hold extraneous status report parcels, i.e.\ status reports corresponding to multiple transmissions for the same node.  This is immediate, since whenever a node $N$ learns a node $(\widehat{N}, \mathtt{T}')$ is on the blacklist as on ({\bf \ref{routingRules4M}.137}), then $N$ will immediately delete all of its status report parcels from $\widehat{N}$ corresponding to transmissions other than $\mathtt{T}'$ ({\bf \ref{routingRules4M}.139}).  The fact that the stored parcels have the correct information (i.e.\ that they address the appropriate reason for failure as on ({\bf \ref{routingRules4M}.142-145})) follows from the fact that $N$ will only initially store a status report parcel if it contains the correct information ({\bf \ref{routingRules4M}.153}).

	\item There are three lines on which the broadcast parcels of the kind relevant to Statement 4) are removed from $N$'s broadcast buffer: ({\bf \ref{routingRules4M}.134}), ({\bf \ref{routingRules4M}.139}), and ({\bf \ref{routingRules4M}.149}).  We consider each of these three lines.  Suppose first that the parcel $(B, \widehat{N}, \widehat{\mathtt{T}})$ is removed from $N$'s broadcast buffer as on line ({\bf \ref{routingRules4M}.134}) of some transmission $\mathtt{T}'$.  In particular, $N$ learns for the first time in the {\it SOT} broadcast of transmission $\mathtt{T}'$ that some node $\widetilde{N}$ has been eliminated.  Let $\widetilde{\mathtt{T}}$ denote the transmission that the sender eliminated this node (as on ({\bf \ref{routingRules5M}.169-177})).  If $\widetilde{\mathtt{T}} > \widehat{\mathtt{T}}$, then $(\widehat{N}, \widehat{\mathtt{T}})$ will be cleared from the blacklist on line ({\bf \ref{routingRules5M}.171}) of $\widetilde{\mathtt{T}}$, and hence when $(B, \widehat{N}, \widehat{\mathtt{T}})$ is removed from $N$'s broadcast buffer in transmission $\mathtt{T}' > \widetilde{\mathtt{T}}$, $(\widehat{N}, \widehat{\mathtt{T}})$ will no longer be on the blacklist, as required.  Therefore, assume $\widetilde{\mathtt{T}} < \widehat{\mathtt{T}}$ (equality here is impossible since lines ({\bf \ref{routingRules5M}.169-177}) and ({\bf \ref{routingRules5M}.188}) can never both be reached in a single transmission, see e.g.\ ({\bf \ref{routingRules5M}.177})).  Let $\mathtt{X}$ denote the transmission in which $N$ first learned that $(\widehat{N}, \widehat{\mathtt{T}})$ was on the blacklist, i.e.\ $N$ received a parcel of the form $(\widehat{N}, \widehat{\mathtt{T}}, \mathtt{X})$ on ({\bf \ref{routingRules4M}.137}) of transmission $\mathtt{X}$.  Clearly, $\mathtt{X} > \widehat{\mathtt{T}}$, since $(\widehat{N}, \widehat{\mathtt{T}})$ can only be added to the blacklist at the end of $\widehat{\mathtt{T}}$ ({\bf \ref{routingRules5M}.188}).  Also, $\mathtt{X} \leq \mathtt{T}'$, since by hypothesis a parcel of the form $(B, \widehat{N}, \widehat{\mathtt{T}})$ is removed from $N$'s broadcast buffer on line ({\bf \ref{routingRules4M}.134}) of $\mathtt{T}'$, and this parcel can only have been added to $N$'s broadcast buffer in the first place if $N$ already knew that $(\widehat{N}, \widehat{\mathtt{T}})$ was blacklisted ({\bf \ref{routingRules4M}.151}).  Lastly, $\mathtt{X} \geq \mathtt{T}'$, since $\widetilde{\mathtt{T}} < \widehat{\mathtt{T}}$ implies that $\widetilde{N}$ was eliminated {\it before} $(\widehat{N}, \widehat{\mathtt{T}})$ was added to the blacklist, and therefore by the priorities of sending/receiving broadcast parcels (({\bf \ref{routingRules3M}.110}) and ({\bf \ref{routingRules3M}.115})), we have that an honest $N$ will learn that $\widetilde{N}$ has been eliminated {\it before} it will learn that $(\widehat{N}, \widehat{\mathtt{T}})$ is on the blacklist.  Combining these inequalities shows that $\mathtt{X} \geq \mathtt{T}'$ and $\mathtt{X} \leq \mathtt{T}'$, so $\mathtt{X} = \mathtt{T}'$.  But this implies that when ({\bf \ref{routingRules4M}.134}) is reached in $\mathtt{T}'$, $N$ does not yet know that $(\widehat{N}, \widehat{\mathtt{T}})$ is on the blacklist, and consequently the parcel $(B, \widehat{N}, \widehat{\mathtt{T}})$ cannot yet be stored in $N$'s broadcast buffer, which contradicts the fact that it was removed on ({\bf \ref{routingRules4M}.134}) of $\mathtt{T}'$.  Therefore, whenever ({\bf \ref{routingRules4M}.134}) is reached, either $(\widehat{N}, \widehat{\mathtt{T}})$ will no longer be on the blacklist, or there will be no parcels of the form $(B, \widehat{N}, \widehat{\mathtt{T}})$ that are removed.

\hspace{.5cm}Suppose now that the parcel $(B, \widehat{N}, \widehat{\mathtt{T}})$ is removed from $N$'s broadcast buffer as on line ({\bf \ref{routingRules4M}.139}) or ({\bf \ref{routingRules4M}.149}) of some transmission $\mathtt{T}'$.  In either case, by looking at the comments on these lines together with Lemma \ref{blacklistStuff}, $(\widehat{N}, \widehat{\mathtt{T}})$ has already been removed from the blacklist if a parcel of the form $(B, \widehat{N}, \widehat{\mathtt{T}})$ is removed on either of these lines.
	\end{enumerate}
\vspace{-22pt}
\end{proof}
\begin{lemma} \label{sigRelationships} If $A,B \in G$ are honest (not corrupt), in any transmission $\mathtt{T}$ for which both $A$ and $B$ have received the full {\it SOT} broadcast:
	\begin{enumerate}
	\item Between the time $B$ accepts a packet from $A$ on line ({\bf \ref{routingRules2M}.77}) through the time $A$ gets confirmation of receipt (see Definition \ref{confRec}) for it as on ({\bf \ref{routingRules2M}.50}), we have:
		\begin{itemize}
		\item $SIG^B[1]_{A,B} = 1 + SIG^A[1]_{A,B}$\footnote{If the packet accepted corresponds to an old codeword, then $SIG^B[1]_{A,B} = SIG^A[1]_{A,B}$ and $SIG^B[p]_{A,B} = SIG^A[1]_{A,B} = \bot$.}
		\item $SIG^B[p]_{A,B} = 1 + SIG^A[p]_{A,B}$\lastfootnote
		\item $SIG^B[2]_{A,B} = M + SIG^A[3]_{A,B}$, where $M$ is the value of $H_{FP}$ on ({\bf \ref{routingRules2M}.60}) (according to $A$'s view) in the same round in which ({\bf \ref{routingRules2M}.77}) was reached by $B$
		\item $SIG^B[3]_{A,B} = m + SIG^A[2]_{A,B}$, where $m$ is the value of $H_{GP}$ on ({\bf \ref{routingRules2M}.75}) (according to $B$'s view) in the same round in which ({\bf \ref{routingRules2M}.77}) was reached by $B$
		\end{itemize}
	\item At all other times, we have that $SIG^B[1]_{A,B} = SIG^A[1]_{A,B}$, $SIG^B[2]_{A,B} = SIG^A[3]_{A,B}$, $SIG^B[3]_{A,B} = SIG^A[2]_{A,B}$, and $SIG^B[p]_{A,B} = SIG^A[p]_{A,B}$ for each packet $p$ that is part of the current codeword.
	\end{enumerate}
\end{lemma}
\begin{proof} The structure of the proof will be as follows.  We begin by observing all signature buffers are initially empty ({\bf \ref{setupCode2M}.48}) and ({\bf \ref{setupCode2M}.54}), and that for any transmission $\mathtt{T}$, both $SIG^A$ and $SIG^B$ are cleared before any packets are transferred (Lemma \ref{sigBuffsInit}).  We will then focus on a single transmission for which $A$ and $B$ have both received the full {\it SOT} broadcast, and prove that all changes made to $SIG^A$ and $SIG^B$ during this transmission (after the buffers are cleared upon receipt of the {\it SOT} broadcast) respect the relationships in the lemma.  Since the only changes occur on lines ({\bf \ref{routingRules2M}.48-49}) and ({\bf \ref{routingRules2M}.74-75}), it will be enough to consider only these 4 lines.  Furthermore, if lines ({\bf \ref{routingRules2M}.48-49}) were reached $x$ times by $A$ in the transmission, and lines ({\bf \ref{routingRules2M}.74-75}) were reached $y$ times by $B$, then:
	\begin{enumerate}
	\item[(a)] Either $y=x$ or $y=x+1$,
	\item[(b)] Neither set of lines can be reached twice consecutively (without the other set being reached in between)
	\item[(c)] Lines ({\bf \ref{routingRules2M}.74-75}) are necessarily reached {\it before} lines ({\bf \ref{routingRules2M}.48-49}) (i.e.\ in any transmission, necessarily $y$ will change from zero to 1 {\it before} $x$ does).
	\end{enumerate}
Notice that the the top statement follows from the second two statements, so we will only prove them below.  

We first prove the three statements above.  We first define $x$ more precisely: $x$ begins each transmission set to zero, and increments by one every time line 50 is reached (just {\it after} $A$'s signature buffers are updated on lines ({\bf \ref{routingRules2M}.48-49})).  Also, define $y$ to begin each transmission equal to zero, and to increment by one when line ({\bf \ref{routingRules2M}.74}) is reached (just {\it before} $B$'s signature buffers are updated on lines ({\bf \ref{routingRules2M}.74-75})).  Statement (c) is immediate, since $RR$ begins every round equal to $-1$ (lines ({\bf \ref{setupCode2M}.50}) and ({\bf \ref{routingRules5M}.209})), and can only be changed to a higher index on ({\bf \ref{routingRules2M}.78}).  Therefore, ({\bf \ref{routingRules2M}.46}) can never be satisfied before ({\bf \ref{routingRules2M}.78}) is reached, which implies ({\bf \ref{routingRules2M}.48}) is never reached before ({\bf \ref{routingRules2M}.74}) is.  We now prove Statement (b).  Suppose lines ({\bf \ref{routingRules2M}.48-49}) are reached in some round $\mathtt{t}$.  Notice since we are in round $\mathtt{t}$ when this happens, and because $RR$ can never have a higher index than the current round index, and the most recent round $RR$ could have been set is the previous round, we have that $B$'s value for $RR$ (and the one $A$ is using on the comparison on ({\bf \ref{routingRules2M}.45-46})) is at most $\mathtt{t}-1$.  Also, $H_{FP}$ and $FR$ will be set to $\bot$ on ({\bf \ref{routingRules2M}.51}) of $\mathtt{t}$.  If $FR$ ever changes to a non-$\bot$ value after this, it can only happen on ({\bf \ref{routingRules2M}.56}), and so the value it takes must be at least $\mathtt{t}$.  Therefore, if at any time after $\mathtt{t}$ we have that $FR \neq \bot$, then if $RR$ has not changed since $\mathtt{t}-1$, then ({\bf \ref{routingRules2M}.46}) can never pass, since $RR \leq \mathtt{t}-1 < \mathtt{t} \leq FR$.  Consequently, ({\bf \ref{routingRules2M}.78}) must be reached before ({\bf \ref{routingRules2M}.48-49}) can be reached again after round $\mathtt{t}$, and hence so must ({\bf \ref{routingRules2M}.74-75}).  This shows that ({\bf \ref{routingRules2M}.48-49}) can never be reached twice, without ({\bf \ref{routingRules2M}.74-75}) being reached in between.

Conversely, suppose lines ({\bf \ref{routingRules2M}.74-75}) are reached in some round $\mathtt{t}$.  Notice since we are in round $\mathtt{t}$ when this happens, and because $FR$ can never have a higher index than the current round index, we have that $A$'s value for $FR$ (and the one $B$ is using on the comparison on ({\bf \ref{routingRules2M}.73})) is at most $\mathtt{t}$.  Also, $RR$ will be set to $\mathtt{t}$ on ({\bf \ref{routingRules2M}.78}) of round $\mathtt{t}$, and $RR$ cannot change again until (at some later round) ({\bf \ref{routingRules2M}.73}) is satisfied again (or the end of the transmission, in which case their is nothing to show).  If line ({\bf \ref{routingRules2M}.56}) is NOT reached after ({\bf \ref{routingRules2M}.74-75}) of round $\mathtt{t}$, then $FR$ can never increase to a larger round index, so $FR$ will remain at most $\mathtt{t}$.  Consequently, line ({\bf \ref{routingRules2M}.73}) can never pass, since if $B$ receives the communication from $A$ on line ({\bf \ref{routingRules2M}.62}), then by the above comments $RR \geq \mathtt{t} \geq FR$.  Consequently, ({\bf \ref{routingRules2M}.56}) must be reached before ({\bf \ref{routingRules2M}.73}) can be reached again after round $\mathtt{t}$.  However, by Statement 3 of Lemma \ref{obsBelowHere}, ({\bf \ref{routingRules2M}.56}) cannot be reached until $A$ receives confirmation of receipt from $B$ (see Definition \ref{confRec}), i.e.\ ({\bf \ref{routingRules2M}.56}) can be reached after ({\bf \ref{routingRules2M}.74-75}) of round $\mathtt{t}$ only if lines ({\bf \ref{routingRules2M}.48-49}) are reached.

We now prove the lemma by using an inductive argument on the following claim:\\
\noindent {\bf Claim.} {\em Every time line ({\bf \ref{routingRules2M}.74}) is reached (and $y$ is incremented), we have that equalities of Statement 2 of the lemma are true, and between this time and the time line ({\bf \ref{routingRules2M}.48}) is reached (or the end of the transmission, whichever comes first), we have that the equalities of the first statement of the lemma are true.}\\[.2cm]
To prove the base case, notice that before lines ({\bf \ref{routingRules2M}.74-75}) are reached for the first time, but after both nodes have received the transmission's {\it SOT} broadcast, all entries to both signature buffers are $\bot$, and so the induction hypothesis is true.  Now consider any time in the transmission for which $y$ is incremented by one in some round $\mathtt{t}$ (i.e.\ line ({\bf \ref{routingRules2M}.74}) is reached).  Since neither $x$ nor $y$ can change between lines ({\bf \ref{routingRulesM}.20}) and ({\bf \ref{routingRulesM}.22}), by the induction hypothesis we have that the equalities of the second statement of the lemma are true when $A$ sends the communication as on ({\bf \ref{routingRules2M}.60}) of round $\mathtt{t}$.  Since $A$ has actually sent $(SIG^A[1]+1, SIG^A[p]+1, SIG^A[3]+H_{FP})$, and these are the quantities that $B$ stores on lines ({\bf \ref{routingRules2M}.74-75}), we have that the first statement of the lemma will be true after leaving line ({\bf \ref{routingRules2M}.75}) (and in particular the claim remains true).  More specifically, letting $M$ denote the value of $H_{FP}$ (respectively letting $m$ denote the value of $H_{GP}$) when ({\bf \ref{routingRules2M}.60}) (respectively ({\bf \ref{routingRules2M}.74})) is reached in round $\mathtt{t}$, we will have that immediately after leaving ({\bf \ref{routingRules2M}.75}):
	\begin{enumerate}
	\item $SIG^B[1]_{A,B} = 1 + SIG^A[1]_{A,B}$
	\item $SIG^B[p]_{A,B} = 1 + SIG^A[p]_{A,B}$
	\item $SIG^B[2]_{A,B} = M + SIG^A[3]_{A,B}$
	\item $SIG^B[3]_{A,B} = m + SIG^A[2]_{A,B}$
	\end{enumerate}
as required by Statement 1 of the Lemma.  By Statement (b) above, either the signature buffers along $E(A,B)$ do not change through the end of the transmission, or the next change necessarily occurs as on ({\bf \ref{routingRules2M}.48-49}).  In the former case, the Claim certainly remains true.  In the latter case, let $\mathtt{t}'$ denote the time that ({\bf \ref{routingRules2M}.48}) is next reached.  Notice that $\mathtt{t}' > \mathtt{t}$, as Statement (b) above guarantees ({\bf \ref{routingRules2M}.48}) is reached {\it after} ({\bf \ref{routingRules2M}.74}), and by examining the pseudo-code, this cannot happen until at least the next round after $\mathtt{t}$.  In particular, the values received on ({\bf \ref{routingRulesM}.07}) of round $\mathtt{t}'$ necessarily reflect the most recent values of $SIG^B$ (i.e.\ $B$'s signature buffers have already been updated as on ({\bf \ref{routingRules2M}.74-75}) when $B$ sends $A$ the communication on ({\bf \ref{routingRules2M}.11})).  Consequently, $A$ will change $SIG^A[1]$, $SIG^A[2]$, and $SIG^A[p]$ to the values $B$ is storing in $SIG^B[1]$, $SIG^B[3]$, and $SIG^B[p]$, respectively.  Therefore, the claim (and hence the lemma) will be true provided we can show that when $A$ updates $SIG^A[3]$ as on ({\bf \ref{routingRules2M}.49}), that the new value for $SIG^A[3]$ equals the value stored in $SIG^B[2]$.  Since before ({\bf \ref{routingRules2M}.49}) is reached, we have by the induction hypothesis that $SIG^B[2]_{A,B} = M + SIG^A[3]_{A,B}$, it is enough to show that when $SIG^A[3]$ is updated on ({\bf \ref{routingRules2M}.49}), that the value of $H_{FP}$ there equals $M$.  We argue that this by showing $H_{FP}$ will not change from line ({\bf \ref{routingRules2M}.60}) of round $\mathtt{t}$ (when $M$ was set to $H_{FP}$) through line ({\bf \ref{routingRules2M}.49}) of round $\mathtt{t}'$.  To see this, notice that the only possible places $H_{FP}$ can change during a transmission are lines ({\bf \ref{routingRules2M}.51}), ({\bf \ref{routingRules2M}.53}), and ({\bf \ref{routingRules2M}.56}).  Clearly, ({\bf \ref{routingRules2M}.51}) cannot be reached between these times, since ({\bf \ref{routingRules2M}.49}) is not reached during these times.  Also, Statement 3 of Lemma \ref{obsBelowHere} implies that ({\bf \ref{routingRules2M}.56}) cannot be reached between these times either.  Finally, ({\bf \ref{routingRules2M}.53}) cannot be reached, since $RR$ will be set to $\mathtt{t}$ on ({\bf \ref{routingRules2M}.78}) of round $\mathtt{t}$, and by statement (b), ({\bf \ref{routingRules2M}.78}) cannot be reached again until after ({\bf \ref{routingRules2M}.49}) is reached in round $\mathtt{t}'$, and hence $RR$ will be equal to $\mathtt{t}$ from ({\bf \ref{routingRules2M}.78}) of round $\mathtt{t}$ through ({\bf \ref{routingRules2M}.49}) of round $\mathtt{t}'$.  Also, $FR$ will not change between these times (also by Statement 3 of Lemma \ref{obsBelowHere}), and since the only non-$\bot$ value $FR$ is ever set to is the current round as on ({\bf \ref{routingRules2M}.56}), we have that $FR \leq \mathtt{t}$.  Putting these facts together, we have that for all times between line ({\bf \ref{routingRules2M}.60}) of round $\mathtt{t}$ through line ({\bf \ref{routingRules2M}.49}) of round $\mathtt{t}'$, either $A$ does not receive $RR$ (in which case $RR = \bot$ when ({\bf \ref{routingRules2M}.52}) is reached) or $A$ receives $RR$, which as noted obeys $RR = \mathtt{t} \geq FR$.  In either case, ({\bf \ref{routingRules2M}.52}) will fail, and ({\bf \ref{routingRules2M}.53}) cannot be reached.
\end{proof}
\begin{lemma} \label{honestP} For any transmission $\mathtt{T}$, recall that $\mathcal{P}_{\mathtt{T}}$ denotes the list of nodes that {\bf \em participated} in that transmission, and it is set at the end of each transmission on ({\bf \ref{routingRules5M}.187}).  For any honest (not corrupt) node $A \in G$, during any transmission $\mathtt{T}$, $A$ will not exchange any codeword packets with any node that is not put on $\mathcal{P}_{\mathtt{T}}$ at the end of the transmission.
\end{lemma}
\begin{proof} Restating the lemma more precisely, for any node $N$ that is NOT put on $\mathcal{P}_{\mathtt{T}}$ as on ({\bf \ref{routingRules5M}.187}) and for any honest node $A \in G$, then along (directed) edge $E(A,N)$, $A$ will never reach line ({\bf \ref{routingRulesM}.60}), and along (directed) edge $E(N,A)$, $A$ will never reach lines ({\bf \ref{routingRules2M}.67-82}).  Fix a transmission $\mathtt{T}$ in which ({\bf \ref{routingRules2M}.187}) is reached (i.e.\ a node is not eliminated as on ({\bf \ref{routingRules5M}.169-177}) of $\mathtt{T}$), let $N \notin \mathcal{P}_{\mathtt{T}}$ be any node {\it not} put on $\mathcal{P}_{\mathtt{T}}$ on ({\bf \ref{routingRules5M}.187}) of $\mathtt{T}$, and let $A \in G$ be an honest node.  Since $N \notin \mathcal{P}_{\mathtt{T}}$, we have that either $N \in EN$ or $N \in BL$ when ({\bf \ref{routingRules5M}.187}) is reached.  Since no nodes can be {\it added to} $EN$ or $BL$ from the outset of $\mathtt{T}$ through line ({\bf \ref{routingRules5M}.187}) of $\mathtt{T}$, we must have that $N \in EN$ or $N \in BL$ as of either line ({\bf \ref{routingRules5M}.188}) or ({\bf \ref{routingRules5M}.170}) of the previous transmission.  Therefore, either $(N, \mathtt{T})$ or $(N, \mathtt{T}', \mathtt{T})$ is added to the {\it SOT} broadcast of transmission $\mathtt{T}$ (on ({\bf \ref{routingRules5M}.176}) or ({\bf \ref{routingRules5M}.200}) of transmission $\mathtt{T}-1$), indicating $N$ is an eliminated/blacklisted node.  If $A$ has not received the full {\it Start of Transmission} (SOT) broadcast for $\mathtt{T}$ yet, then the lemma is true by Lemma \ref{sigBuffsInit}.  If on the other hand $A$ has received the full {\it SOT} broadcast, then in particular $A$ has received the parcel indicating that $N$ is either eliminated or blacklisted.  Thus, by lines ({\bf \ref{routingRules2M}.59}), ({\bf \ref{routingRulesM}.31-33}), ({\bf \ref{routingRules2M}.63}) and ({\bf \ref{routingRulesM}.35-37}), $A$ will not transfer any packets with $N$.
\end{proof}
\begin{lemma} \label{finalBC} The receiver's {\it end of transmission} broadcast takes at most $n$ rounds to reach the sender.  In other words, the sender will have always received the {\it end of transmission} broadcast by the time he enters the {\bf \em Prepare Start of Transmission Broadcast} segment on ({\bf \ref{routingRulesM}.29}).\end{lemma}
\begin{proof} By the conforming assumption, for every round $\mathtt{t}$ of every transmission there is a path $\mathsf{P}_{\mathtt{t}}$ between the sender and receiver consisting of edges that are always up and nodes that are not corrupt.  We consider the final $n$ rounds of any transmission, and argue that for each round, either the sender already knows the end of transmission parcel $\Theta$, or there is a {\it new} honest node $N \in G$ that learns $\Theta$ for the first time.  Since the latter case can happen at most $n-1$ times (the receiver already knows $\Theta$ when there are $n$ rounds remaining, see ({\bf \ref{routingRulesM}.28}) and ({\bf \ref{routingRules5M}.178-179})), it must be that the sender has learned $\Theta$ by the end of the transmission.  Therefore, let $4D-n < \mathtt{t} \leq 4D$ be one of the last $n$ rounds of some transmission.  If the sender already knows $\Theta$, then we are done.  Otherwise, let $\mathsf{P}_{\mathtt{t}} = N_0 N_1 \dots N_L$ (here $N_0 = S$ and $N_L = R$) denote the active honest path for round $\mathtt{t}$ that connects the sender and receiver.  Since $S$ does not know $\Theta$ but $R$ does, there exists some index $0 \leq i < L$ such that $N_i$ does not know $\Theta$ but $N_{i+1}$ does know $\Theta$.  Since edge $E(N_{i}, N_{i+1})$ is active and the nodes at both ends are honest (by choice of $\mathsf{P}_{\mathtt{t}}$), node $N_{i+1}$ will send $N_i$ a broadcast parcel on ({\bf \ref{routingRulesM}.15}).  Looking at the manner in which broadcast parcels are chosen ({\bf \ref{routingRules3M}.115}), it must be that $N_{i+1}$ will send $\Theta$ to $N_i$ in round $\mathtt{t}$, and hence $N_i$ will learn $\Theta$ for the first time, which was to be showed.
\end{proof}
\begin{lemma} \label{decoding} If the receiver has received at least $D-6n^3$ {\it distinct} packets corresponding to the current codeword, he can decode the codeword (except with negligible probability of failure).
\end{lemma}
\begin{proof} Fact $1'$ guarantees that if the receiver obtains $D-6n^3$ distinct packets corresponding to a codeword, then he can decode.  Since all codeword packets are signed by the sender to prevent modifying them, the security of the signature scheme guarantees that any properly signed codeword packet the receiver obtains will be legitimate (except with negligible probability of failure).
\end{proof}
\begin{lemma} \label{rParticipates} For every transmission $\mathtt{T}$: $S, R \in \mathcal{P}_{\mathtt{T}}$.\end{lemma}
\begin{proof} The participating list $\mathcal{P}_{\mathtt{T}}$ is set at the end of every transmission on line ({\bf \ref{routingRules5M}.187}).  By looking at the code there, we must show that $S, R \notin EN \cup BL$ at the end of any transmission.  That an honest node can never be identified as corrupt and eliminated is the content of the proof of Theorem \ref{mainAdThm}, so $S, R \notin EN$.  Since $S$ is never put on the blacklist ({\bf \ref{routingRules5M}.188}), it remains to show $R \notin BL$ when ({\bf \ref{routingRules5M}.187}) is reached.  Since nodes are removed from the blacklist on line ({\bf \ref{routingRules4M}.166}) and not put on it again until ({\bf \ref{routingRules5M}.188}), it is enough to show that if $R$ is ever placed on the blacklist at the end of some transmission $\mathtt{T}-1$, then it will be removed as on ({\bf \ref{routingRules4M}.166}) of transmission $\mathtt{T}$.  If $R$ is ever placed on the blacklist, we argue that: 1) $R$ will learn what status report parcels the sender requires of it after at most $2n^2$ rounds; and 2) $S$ will receive all of these parcels by at most $4n^3$ rounds later.  Therefore, $R$ will necessarily be removed from the blacklist by round $4n^3+2n^2 < 4D$ (since $D \geq 6n^3$), as required.  To prove 1), first note that all honest nodes remove the receiver's {\it end of transmission} parcel for $\mathtt{T}-1$ at the very end of $\mathtt{T}-1$ ({\bf \ref{routingRules5M}.203}).  Therefore, no honest node will have any {\it End of Transmission Parcel} in its broadcast buffer at any point during $\mathtt{T}$ {\it until} one is created for the current transmission on ({\bf \ref{routingRules5M}.178-179}).  Therefore, for the first $n^3$ rounds, the sender's {\it SOT} broadcast will have top priority in terms of sending/receiving broadcast parcels ({\bf \ref{routingRules3M}.115}).  Since $S$ and $R$ are connected by an active honest path at each round, we follow the proof as in Lemma \ref{finalBC} to argue that for every round between the outset of $\mathtt{T}$ and round $n^3$, either $R$ has learned the full {\it SOT} broadcast, or there is an honest node that is learning a {\it new} {\it SOT} broadcast parcel for the first time.  Since there are (at most) $n$ nodes, and the {\it SOT} broadcast has at most $2n$ parcels (see proof of Lemma \ref{mProperDomains}, and Statement 2 of the Broadcast Buffer therein), it takes at most $2n^2$ rounds for $R$ to receive the full {\it SOT} broadcast, and hence to learn it has been blacklisted.  This proves 1).

\hspace{.5cm}Upon receipt of this information, $R$ adds the necessary information (i.e.\ its status report) to its broadcast buffer ({\bf \ref{routingRules4M}.137-145}).  Looking at the proof of Theorem \ref{maxWasted} and in particular Claim 2 within the proof, edges along the active honest path can take at most $4n^3<4D$ rounds to communicate across their edges the broadcast information of priorities 1-6 on lines ({\bf \ref{routingRules3M}.115}), and since the receiver is connected to the sender {\it every} round via some active honest path (by the conforming assumption), its requested status report information will necessarily reach the sender within $4n^3$ rounds, proving 2).
\end{proof}
\begin{lemma} \label{partListNotEmpty} For any transmission $\mathtt{T}$, if $\mathcal{P}_{\mathtt{T}} = \{S, R\}$, then the transmission was necessarily successful.\end{lemma}
\begin{proof} $\mathcal{P}_{\mathtt{T}}$ is set on line ({\bf \ref{routingRules5M}.187}).  Since the only place the sender adds nodes to the blacklist is on ({\bf \ref{routingRules5M}.188}), which happens at the very end of each transmission, and because the hypothesis states that every non-eliminated node except for $S$ and $R$ is on the blacklist when line ({\bf \ref{routingRules5M}.187}) of transmission $\mathtt{T}$ is reached, it must be the case that transmission $\mathtt{T}$ began with every non-eliminated node on the blacklist, with the possible exception of the receiver (and the sender who is never blacklisted).  Since all internal nodes are still blacklisted by the end of the transmission, the sender will never transfer any packets to any node other than $R$ during transmission $\mathtt{T}$ (line ({\bf \ref{routingRules2M}.59}) will always fail for any other node, see ({\bf \ref{routingRulesM}.31-33})).  Theorem \ref{maxWasted} indicates there are at most $4n^3$ rounds that are wasted, and since the only edge the sender can ever use to transfer codeword packets during $\mathtt{T}$ is $E(S,R)$, the conforming assumption implies edge $E(S,R)$ is active {\it every} round of $\mathtt{T}$.  We may therefore view the graph as reduced to a single edge connecting $S$ and $R$ (see Lemma \ref{honestP}), where there are at least $4D-4n^3 > 3D$ (non-wasted) rounds per transmission.  Since both $S$ and $R$ are honest, correctness is guaranteed as in the edge-scheduling protocol by Lemma \ref{similar}.  In particular, the transmission will necessarily be successful.
\end{proof}
\begin{lemma} \label{removeOnce} No honest node will accept more than one distinct parcel (per node $\widehat{N}$ per transmission) indicating that $\widehat{N}$ should be removed from the blacklist.
\end{lemma}
\begin{proof} Line ({\bf \ref{routingRules3M}.110}) guarantees that any node $A$ will only accept the parcel if it has already received the sender's {\it start of transmission} broadcast corresponding to the current transmission.  In particular, this means that $A$ has received an updated blacklist (and a list of eliminated nodes) before it accepts any removals from the blacklist.  Therefore, in some transmission $\mathtt{T}$, if $A$ ever does accept the information that a node $\widehat{N}$ should be removed from the blacklist, then this information will not become out-dated until (if) $\widehat{N}$ is added to the blacklist again, which can happen at the earliest at the very end of transmission ({\bf \ref{routingRules5M}.188}).  Therefore, after receiving the information for the first time that $\widehat{N}$ should be removed, the comments on line ({\bf \ref{routingRules4M}.123}) will guarantee $A$ will not accept additional blacklist information regarding $\widehat{N}$ until the following transmission, proving the lemma.
\end{proof}
\begin{lemma} \label{completeKnowledgeOnce} For any node $\widehat{N} \in G$, after receiving the complete {\it SOT} broadcast, an honest node $N$ will transmit along each edge at most once per transmission the fact that it knows $\widehat{N}$'s complete status report.\end{lemma}
\begin{proof} Each parcel stored in $N$'s broadcast buffer $BB$ is accompanied by a list of which edges the parcel has been successfully transmitted across (see comments on line ({\bf \ref{routingRules4M}.123})).  Therefore, as long as the parcel is not deleted from the broadcast buffer, line ({\bf \ref{routingRules3M}.115}) guarantees that each parcel of broadcast information will only pass along each edge once, as required.  Therefore, it remains to prove the lemma in the case that the relevant broadcast parcel is deleted at some point in a transmission.  Fix a transmission $\mathtt{T}$ and an arbitrary $\widehat{N} \in G$.  Since broadcast parcels of the relevant type (i.e.\ that $N$ has $\widehat{N}$'s complete status report) are only removed on ({\bf \ref{routingRules4M}.139}) and ({\bf \ref{routingRules4M}.149}), we need only consider the case that ({\bf \ref{routingRules4M}.149}) is reached in transmission $\mathtt{T}$ (the former line can only be reached as part of the {\it SOT} broadcast, and therefore lies outside the hypotheses of the lemma).  In particular, we will show that if ({\bf \ref{routingRules4M}.149}) deletes from $N$'s broadcast buffer the parcel indicating that $N$ knows $\widehat{N}$'s complete status report, then $N$ will never again add a parcel of this form to its broadcast buffer (as on ({\bf \ref{routingRules4M}.155})) for the remainder of $\mathtt{T}$.  But this is immediate, since if $N$ removes this parcel from $BB$ on ({\bf \ref{routingRules4M}.149}) of $\mathtt{T}$, then $\widehat{N}$ must have been removed from the blacklist (see ({\bf \ref{routingRules4M}.147})), and since $\widehat{N}$ cannot be re-added to the blacklist until the end of $\mathtt{T}$ ({\bf \ref{routingRules5M}.188}), line ({\bf \ref{routingRules4M}.152}) (of $N$'s code, with the $\widehat{N}$ that appears there equal to the $\widehat{N}$ used in the present notation) cannot be satisfied for the remainder of $\mathtt{T}$, and hence ({\bf \ref{routingRules4M}.155}) cannot be reached.  This proves that once the parcel is deleted, it cannot be later added in the same transmission, proving the lemma.
\end{proof}
\section{Conclusion and Open Problems}
\indent \indent In this paper, we have described a protocol that is secure simultaneously against conforming node-controlling and edge-scheduling adversaries.  Our results are of a theoretical nature, with rigorous proofs of correctness and guarantees of performance.  Surprisingly, our protocol shows that the additional protection against the node-controlling adversary, on top of protection against the edge-scheduling adversary, can be achieved without any additional asymptotic cost in terms of throughput.

While our results do provide a significant step in the search for protocols that work in a dynamic setting (edge-failures controlled by the edge-scheduling adversary) where some of the nodes are susceptible to corruption (by a node-controlling adversary), there remain important open questions.  The original Slide protocol\footnote{In \cite{KO}, it was shown how to modify the Slide protocol so that it only requires $O(n \log n)$ memory per internal node.  We did not explore in this paper if and/or how their techniques could be applied to our protocol to similarly reduce it by a factor of $n$.} requires each internal node to have buffers of size $O(n^2 \log n)$, while ours requires $O(n^4 \log n)$, though this can be slightly improved with additional assumptions\footnote{If we are given an a-priori bound that a path-length of any conforming path is at most $L$, the $O(n^4\log n)$ can be somewhat reduced to $O(Ln^3\log n)$.}. In practice, the extra factor of $n^2$ may make our protocol infeasible for implementation, even for overlay networks. While the need for signatures inherently force an increase in memory per node in our protocol verses the original Slide protocol, this is not what contributes to the extra $O(n^2)$ factor.  Rather, the only reason we need the extra memory is to handle the third kind of malicious behavior, which roughly corresponds to the mixed adversarial strategy of a corrupt node replacing a valid packet with an old packet that the node has duplicated.  Recall that in order to detect this, for {\it every} packet a node sees and for every neighbor, a node must keep a (signed) record of how many times this packet has traversed the adjacent edge (the $O(n^3)$ packets per codeword and $O(n)$ neighbors per node yield the $O(n^4)$ bound on memory).  Therefore, one open problem is finding a less memory-intensive way to handle this type of adversarial behavior.

Our model also makes additional assumptions that would be interesting to relax.  In particular, it remains an open problem to find a protocol that provides efficient routing against a node-controlling and edge-scheduling adversary in a network that is fully {\it asynchronous} (without the use of timing assumptions, which can be used to replace full synchrony in our solution) and/or does not restrict the adversaries to be {\it conforming}. As mentioned in the Introduction, if the adversary is not conforming, then he can simply permanently disconnect the sender and receiver, disallowing any possible progress.  Therefore, results in this direction would have to first define some notion of {\it connectedness} between sender and receiver, and then state throughput efficiency results in terms of this definition.
\end{document}